%% file: CPL.tex
\documentclass[a4paper,10pt]{article}
\usepackage{a4wide}
\usepackage[utf8]{inputenc}
\usepackage{amsmath}
\usepackage{amsthm}
\usepackage{amssymb}
\usepackage{xcolor}
\usepackage{stmaryrd}
\usepackage{thmtools,thm-restate}
\usepackage{hyperref}

\usepackage{caption}
\usepackage{subcaption}

\usepackage[scr=boondoxo,scrscaled=1.05]{mathalfa}

\usepackage{bussproofs}
\EnableBpAbbreviations

\usepackage{graphicx}
\usepackage{adjustbox}

\usepackage{listings}

\usepackage{blindtext}
\input{macros}

\title{On Counting Propositional Logic\thanks{This work has been partially supported by the ErC CoG ``DIAPASoN'', GA 818616.}}
\author{Melissa Antonelli\footnote{University of Bologna \& INRIA Sophia Antipolis, \texttt{melissa.antonelli2@unibo.it}}
        \and Ugo Dal Lago\footnote{University of Bologna \& INRIA Sophia Antipolis, \texttt{ugo.dallago@unibo.it}}
        \and Paolo Pistone\footnote{University of Bologna \& INRIA Sophia Antipolis, \texttt{paolo.pistone2@unibo.it}}}

\date{}							

\begin{document}
\maketitle

\begin{abstract}
  We study counting propositional logic as an extension of
  propositional logic with counting quantifiers. We prove that the
  complexity of the underlying decision problem perfectly matches the
  appropriate level of Wagner's counting hierarchy, but also that the
  resulting logic admits a satisfactory proof theoretical
  treatment. From the latter, a type system for a probabilistic
  $\lambda$-calculus is derived in the spirit of the Curry-Howard
  correspondence, showing the potential of counting propositional
  logic as a useful tool in several fields of theoretical computer
  science.
  \end{abstract}

\section{Introduction}\label{introduction}
\input{1Introduction}

\section{On Counting Quantifiers, Computational 
Complexity, and Type Systems}\label{section2}
\input{2Counting}

\section{Univariate Counting Propositional Logic}
\label{section3}
\input{section3}

\section{The Multivariate Case}\label{section4}
\input{section4}

\section{Relating to the Counting Hierarchy}\label{section5}
\input{5Complexity}

\section{Type Systems from Counting Propositional Logic}\label{section6}
\input{section6}
\section{Arithmetic}\label{section7}
\input{7Arithmetic}

\section{Related Works}\label{section8}
\input{8RelatedWork}

\section{Conclusion}\label{conclusion}
\input{9Conclusion}

\bibliographystyle{plain}
\bibliography{CPL}

\newpage
\appendix
\section*{Appendix}\label{appendix}
\input{appendix}

\end{document}

%% file: macros.tex
\usepackage{amsthm}
\newtheorem{example}{Example}
\theoremstyle{definition}
\newtheorem{definition}{Definition}[section]
\newtheorem{notation}{Notation}
\newtheorem{remark}{Remark}
\newtheorem{theorem}{Theorem}
\newtheorem{lemma}[theorem]{Lemma}
\newtheorem{proposition}[theorem]{Proposition}
\newtheorem{corollary}{Corollary}

\newcommand{\valid}{valid}

\newcommand{\cc}[1]{\textsf{#1}}
\newcommand{\pc}[1]{\mathtt{#1}}
\newcommand{\lc}[1]{\textsf{\textbf{#1}}}
\newcommand{\hc}[1]{\mathsf{#1}}

\newcommand{\PPLc}{\lc{CPL}_{0}}
\newcommand{\PPL}{\lc{CPL}}
\newcommand{\IPPL}{\lc{mCPL}}
\newcommand{\GPPLc}{\mathbf{GsCPL_{0}}}
\newcommand{\GPPL}{\mathbf{GsCPL}}

\newcommand{\PL}{\lc{PL}}

\newcommand{\PTIME}{\cc{P}}
\newcommand{\NPTIME}{\cc{NP}}
\newcommand{\coNP}{\cc{coNP}}
\newcommand{\PP}{\cc{PP}}
\newcommand{\BPP}{\cc{BPP}}

\newcommand{\PSpace}{\cc{PSPACE}}
\newcommand{\PPSpace}{\cc{PPSPACE}}

\newcommand{\sharpP}{\sharp \PTIME}

\newcommand{\Psharp}{\cc{P}^{\SSAT}}
\newcommand{\MajSat}{\pc{MajSAT}}
\newcommand{\MMSat}{\pc{MajMajSAT}}
\newcommand{\TQBF}{\pc{TQBF}}
\newcommand{\SSAT}{\sharp\pc{SAT}}

\newcommand{\PH}{\hc{PH}}
\newcommand{\CH}{\hc{CH}}

\newcommand{\Nat}{\mathbb{N}}

\newcommand{\eval}{\theta}

\newcommand{\fone}{A}
\newcommand{\ftwo}{B}
\newcommand{\fthree}{C}
\newcommand{\ffour}{D}

\newcommand{\predone}{P}

\newcommand{\ssetO}{\Phi}
\newcommand{\ssetT}{\Psi}

\newcommand{\seqO}{\phi}

\newcommand{\lone}{L}

\newcommand{\bone}{\mathscr b}
\newcommand{\btwo}{\mathscr c}
\newcommand{\bthree}{\mathscr d}
\newcommand{\bfour}{\mathscr e}
\newcommand{\bvar}{\mathscr x}

\newcommand{\twoOm}{2^{\omega}}

\newcommand{\at}[1]{\mathbf{i}_{#1}}
\newcommand{\atm}[2]{\mathbf{#1}_{#2}}

\newcommand{\TTT}{\lambda\lc{CPL}}
\newcommand{\Lnu}{\Lambda_{\nu}}

\newcommand{\rred}{\rightarrowtriangle}

\newcommand{\pperm}{\rightarrow_{\mathsf p}}

\newcommand{\FLIP}{\mathsf{Flip}}
\newcommand{\bb}[1]{{\color{blue}#1}}

\newcommand{\supp}{\mathrm{supp}}

\newcommand{\To}{\Rightarrow}
\newcommand{\FF}[1]{\mathfrak{#1}}
\newcommand{\PROB}{\mathbf{P}}
\newcommand{\RED}{\mathrm{Red}}

\newcommand{\TB}{\mathsf{R}_{\bot}}
\newcommand{\TID}{\mathsf{R}_{\mathrm{id}}}
\newcommand{\TU}{\mathsf{R}_{\cup}}
\newcommand{\TLA}{\mathsf{R}_{\lambda}}
\newcommand{\TA}{\mathsf{R}_{@}}
\newcommand{\TL}{\mathsf{R}_{\oplus}^{\mathrm l}}
\newcommand{\TR}{\mathsf{R}_{\oplus}^{\mathrm r}}
\newcommand{\TN}{\mathsf{R}_{\nu}}

\newcommand{\TBx}{\mathsf{R}_{\bot}}
\newcommand{\TIDx}{{\mathbf{Ax}}}
\newcommand{\TUx}{\mathsf{R}_{\vee}}
\newcommand{\TLAx}{\mathsf{R}_{\To}^{\text E}}
\newcommand{\TAx}{\mathsf{R}_{\To}^{\text I}}
\newcommand{\TNx}{\mathsf{R}_{\BOX}}

\newcommand{\model}[1]{\llbracket #1\rrbracket}

\newcommand{\atom}[1]{\mathbf #1}

\newcommand{\midd}{\; \; \mbox{\Large{$\mid$}}\;\;}

\newcommand{\BOX}{\mathbf C}
\newcommand{\DIA}{\mathbf D}
\newcommand{\Cyl}[1]{\mathrm{Cyl}(#1)}

\newcommand{\VAL}{\mathtt{Val}}
\newcommand{\BOOL}{\mathtt{Bool}}

\newcommand{\FN}{\mathrm{FN}}
\newcommand{\Var}[1]{\mathrm{FN}(#1)}

\newcommand{\PROJ}{\Pi}
\newcommand{\EXT}[2]{ (#1)^{\Uparrow#2}}

\newcommand{\Pto}{\rightarrowtail}
\newcommand{\Pfrom}{\leftarrowtail}

\newcommand{\Era}{\Pto}
\newcommand{\Ela}{\Pfrom}

\makeatletter
\providecommand*{\Dashv}{%
  \mathrel{%
    \mathpalette\@Dashv\vDash
  }%
}
\newcommand*{\@Dashv}[2]{%
  \reflectbox{$\m@th#1#2$}%
}
\makeatother

\newcommand{\Rmur}{\mathsf{R_{\mu}^{\Era}}}
\newcommand{\Rmul}{\mathsf{R_{\mu}^{\Ela}}}

\newcommand{\AxO}{\mathsf{Ax1}}
\newcommand{\AxT}{\mathsf{Ax2}}
\newcommand{\Rcup}{\mathsf R_\cup^{\Era}}
\newcommand{\Rcap}{\mathsf R_\cap^{\Ela}}
\newcommand{\Rnla}{\mathsf R_{\neg}^{\Ela}}
\newcommand{\Rnra}{\mathsf R_{\neg}^{\Era}}
\newcommand{\Rvla}{\mathsf R_{\vee}^{\Ela}}
\newcommand{\Rvra}{\mathsf R_{\vee}^{\Era}}
\newcommand{\Rwla}{\mathsf R_{\wedge}^{\Ela}}
\newcommand{\Rwra}{\mathsf R_{\wedge}^{\Era}}
\newcommand{\Rbl}{\mathsf R_{\BOX}^{\Ela}}
\newcommand{\Rbr}{\mathsf R_{\BOX}^{\Era}}
\newcommand{\Rdl}{\mathsf R_{\DIA}^{\Ela}}
\newcommand{\Rdr}{\mathsf R_{\DIA}^{\Era}}

\newcommand{\dec}{\leadsto_{0}}
\newcommand{\Dec}{\leadsto}

\newcommand{\cn}{\mathsf{cn}}
\newcommand{\ms}{\mathsf{ms}}

\newcommand{\prob}{\mathsf{prob}}

\newcounter{number}

\newcommand{\op}[1]{\triangle_{#1}}

\newcommand{\set}{S}

\newcommand{\True}{\mathbf{T}}
\newcommand{\False}{\mathbf{F}}

\newcommand{\GsCPLc}{\mathbf{GsCPL_{0}}}

\newcommand{\PNF}{\mathscr{PNF}}

%% file: 1Introduction.tex

Among the many intriguing relationships existing between logic and
computer science, we can certainly mention the ones between classical
propositional logic, on the one hand, and computational complexity, the
theory of programming languages, and many other branches of computer science, on the other.
As it is well known, indeed, classical propositional logic provided the 
first example of an $\NPTIME$-complete problem~\cite{Cook}.
Moreover, formal systems for classical and
intuitionistic propositional logic correspond to type systems for
$\lambda$-calculi and related formalisms~\cite{Girard, SorensenUrzyczyn}.

These lines of research evolved in several directions, resulting in 
active sub-areas of computer science, in which variations of
propositional logic have been put in relation with complexity classes
other than $\PTIME$ and $\NPTIME$ or with type systems other than
simple types.
For example, the complexity of deciding \emph{quantified} 
propositional logic formulas was proved to correspond to the polynomial
hierarchy~\cite{MeyerStockmeyer72,MeyerStockmeyer73, Stockmeyer77, Wrathall,AllenderWagner,BB}.
As another example, proof systems for propositional \emph{linear} 
logic or \emph{bunched} logic have inspired resource-conscious type systems in which duplication and sharing are taken into account and appropriately dealt with through the type system~\cite{OHearn,Wadler}.

Nevertheless, some developments in both computational complexity 
and programming language theory have not found a precise 
counterpart in propositional logic, at least so far. 
One such development has to do with counting classes as introduced 
by Valiant~\cite{Valiant} and Wagner~\cite{Wagner84,Wagner,Wagner86}. 
The counting hierarchy is deeply related to randomized complexity 
classes, such as $\PP$, and has been treated logically by means of 
tools from descriptive complexity and finite model
theory~\cite{Kontinen}. 
Still, there is no counterpart for these classes in the sense of propositional logic, 
at least to the best of the authors' knowledge.
On the other hand, type systems for randomized $\lambda$-calculi 
and guaranteeing various forms of termination properties have
been introduced in the last years~\cite{DalLagoGrellois,BreuvartDalLago,ADLG}.
Again, although these systems are related to simple, intersection, 
and linear dependent types, 
no Curry-Howard correspondence is known for them.

This paper aims at bridging the gap by introducing a new quantified
propositional logic, 
called counting propositional logic ($\PPL$, for short). 
The quantifiers of $\PPL$ are inherently quantitative, 
as they are designed to \emph{count} the number of values 
of the bound propositional variables satisfying the argument formula. 
We study the proof theory of $\PPL$ together with its relations to 
computational complexity and programming language theory, 
obtaining two main results. 
The first achievement consists in showing that  
deciding prenex normal forms is complete for the counting
hierarchy's appropriate level, 
in the spirit of the well-known correspondence between 
quantified propositional logic and the polynomial hierarchy.
The second result is the introduction of a type system for a
probabilistic $\lambda$-calculus based on counting quantifiers, 
together with the proof that this system guarantees the
expected form of termination property.
Beyond offering a conceptual bridge with type systems 
from the literature, 
our system effectively suggests that the Curry-Howard 
correspondence can have a probabilistic counterpart.

Along the way, many side results are proved, and a somehow 
unusual route to our two main achievements is followed. 
After motivating the introduction of counting propositional logic 
in Section~\ref{section2}, two sections are devoted to its syntax,
semantics and proof theory. 
Specifically, Section~\ref{section3} deals with
univariate $\PPL$, namely the fragment in which counting
quantification is nameless as there is only one, 
global counting variable. 
In this case, the underlying model and proof theory turn out
to be simple, 
but the correspondence with computational complexity is limited 
to the class $\PTIME^{\sharpP}$. 
A sound and complete proof system in the form of a one-sided,
single-succedent sequent calculus on labelled formulas 
is the main result here. 
The multivariate logic, in which many variables are available, 
is treated in Section~\ref{section4} by generalizing the 
corresponding results from Section~\ref{section3}.
In Section~\ref{section5}, the complexity of the decision problem 
for multivariate $\PPL$ is proved to correspond to the whole 
counting hierarchy. 
The proof proceeds by a careful analysis of prenex normal forms 
for the logic, 
which by construction have precisely the shape one needs to 
match Wagner's complete problems~\cite{Wagner}.
Section~\ref{section6} is devoted to the introduction and discussion 
of a type system whose type derivations can be seen as proofs 
in $\PPL$, and in which typability ensures a probabilistic 
termination property. 
Finally, an extension of the syntax for formulas of $\PPL$
arithmetic is discussed in Section~\ref{section7}.

%% file: 2Counting.tex

When considering the class $\PL$ of propositional logic formulas,
checking any such formula for satisfiability is the paradigmatic 
$\NPTIME$-complete problem, 
while the subclass of $\PL$ consisting of tautologies is, dually, 
$\coNP$-complete. 
How could we capture these two classes by means of one 
single logical concept, rather than having to refer to satisfiability 
and validity? 
The answer consists in switching to quantified propositional
logic~\cite{MeyerStockmeyer72,MeyerStockmeyer73,Stockmeyer77,Wrathall,BB}, 
in which existential and universal quantification over propositional 
variables are available \emph{from within} the logic. 

Satisfiability corresponds to the truth of closed, existentially 
quantified formulas in the form 
$\exists x_{1}\dots \exists x_{n} \fone$ 
(where $\fone$ is quantifier-free), 
while validity stands for the truth of closed, universally quantified
formulas in the form $\forall x_{1}\dots \forall x_{n} \fone$ (where
$\fone$ is quantifier-free). 
For example, the formula
$$
\exists x_{1}\exists x_{2}\exists x_{3} \big ((x_{1} \land \lnot x_{2})\lor ( x_{2}\land  \lnot x_{3})\lor (x_{3}\land \lnot x_{1} )\big)
$$
expresses the fact that the formula $(x_{1} \land \lnot
x_{2})\lor ( x_{2}\land \lnot x_{3})\lor (x_{3}\land \lnot x_{1} )$ is
satisfied by \emph{at least} one model.
As it is well known, all this can be generalized to the whole
polynomial hierarchy, $\PH$, where each level of the hierarchy 
is characterized by a quantified propositional formula
(in prenex normal form) with the corresponding
number of quantifier alternations. 
Here, universal and existential quantifications play the
role of the acceptance condition in the machines defining the
corresponding complexity class.

But what happens if other kinds of quantifications, 
still on propositional variables, 
replace the standard universal and existential quantifications?
Would it be possible to be more \emph{quantitative} in the
sense of, say, Wagner's counting quantifiers~\cite{Wagner}?
Apparently, this question does not make much sense,
as the variables one would like to quantify over 
only take \emph{two} possible values.
However, this problem can be overcome by permitting to
simultaneously quantify over \emph{all} the propositional variables
which are free in the argument formula. 
In other words, the interpretation of a formula $\fone$ is not the 
single truth-value that $\fone$ gets in one given valuation, 
but the set of \emph{all} valuations making $\fone$ true.
If we let $\model \fone$ denote such a set, then we can
interpret a quantified formula of the form $\BOX^{q} A$ as saying that
$\fone$ is true in \emph{at least a fraction $q\cdot 2^n$} of the possible
$2^n$ assignments to its $n$ variables.
As an example, one could assert that
a formula like
$$\BOX^{\frac{1}{2}}
\big ((x_{1} \land \lnot x_{2})\lor ( x_{2}\land  \lnot x_{3})\lor (x_{3}\land \lnot x_{1} )\big)
$$
is valid because $(x_{1} \land \lnot x_{2})\lor (
x_{2}\land \lnot x_{3})\lor (x_{3}\land \lnot x_{1} )$
is true in at least $4$ (actually, in $6$) of its $2^3=8$ models.

By the way, the aforementioned counting quantifier is very reminiscent
of the operators on classes of languages 
which Wagner introduced in his
seminal works on the counting hierarchy~\cite{Wagner84,Wagner}. 
Still, the logical status of
such counting operations has not been investigated further, at least
not in the realm of propositional logic (see Section~\ref{section8}
for more details).
This is precisely one part of the study developed in this paper: 
the syntax, semantics, and proof theory of the resulting logic, 
called counting propositional logic, 
are fully explored, together with
connections with computational complexity.
A glimpse of this is given in
Section~\ref{section2A} below, while Section~\ref{section2B} is
devoted to the applications of our counting logic to 
programming language theory,
which will be further studied in Section~\ref{section7}.
The development of a Curry-Howard-style correspondence in 
framework of randomized computation is the main 
result of this work.

\subsection{From Counting Quantifiers to the Counting Hierarchy, and Back}\label{section2A}
In computational complexity theory, the problem of checking if a
formula of $\PL$ is true in at least half of its models is
called $\MajSat$.
This problem is known to be complete for $\PP$, 
the class of decision problems which can be solved in
polytime with an error probability strictly less than $\frac{1}{2}$.
$\PP$ is also related to $\Psharp$, the class of counting 
problems associated with the decision problems in $\NPTIME$,
since $\PTIME^{\PP}=\Psharp$ 
(and both these classes contain the whole
polynomial hierarchy, by Toda's Theorem~\cite{Toda89,Toda91}).
In Section~\ref{section3}, we introduce the \emph{univariate} counting
propositional logic, $\PPLc$, a logic containing a counting quantifier
$\BOX^{q}$ (where $q$ is a rational number between $0$ and $1$),
such that, as previously explained, $\BOX^{q}\fone$ 
intuitively means that ``$\fone$ is
true in at least $q\cdot n$ of its $n$ models''. 
We prove that, in addition to providing a purely logical description 
of $\PP$-complete problems, such as $\MajSat$, 
$\PPLc$-validity can be decided in $\PTIME^{\PP}=\Psharp$.

However, our study of counting quantifiers does not stop here.  
As it is known, a full hierarchy of complexity classes 
comprising counting problems, 
the so-called \emph{counting hierarchy}, 
was (independently) introduced in~\cite{Wagner84,Wagner} 
and~\cite{AllenderWagner} in analogy with the polynomial
hierarchy, by letting
$$
\CH_{0}=\PTIME\qquad \qquad \qquad \qquad \qquad\CH_{n+1}=\PP^{\CH_{n}}.
$$ 
A typical problem belonging to this hierarchy (in fact one which is
complete for $\CH_{2}=\PP^{\PP}$) is the problem $\MMSat$
of, given a formula of $\PL$, call it $\fone$, 
containing two disjoint sets $\mathbf x$ and $\mathbf y$ of variables,
determining whether for the majority of the valuations of $\mathbf x$, 
the majority of the valuations of $\mathbf y$ makes $\fone$ true.
In Section~\ref{section4}, we introduce the \emph{multivariate}
counting propositional logic, $\PPL$, a logic in which propositional
variables are grouped into disjoint sets using \emph{names},
$a,b,c,\dots$, and where one can form \emph{named} counting
quantified formulas, such as $\BOX^{q}_{a}\fone$, 
with the intuitive meaning that ``$\fone$ is
true in at least $q\cdot n$ of the $n$ valuations of the variables of name
$a$''. 
A problem like $\MMSat$ is thus captured by
formulas of $\PPL$ in the form $\BOX^{q}_{a}\BOX^{r}_{b}\fone$.
Our main result, in Section~\ref{section5}, is a
characterization of the full counting hierarchy by $\PPL$.
This amounts at 
(1) proving that any formula of $\PPL$ is equivalent to one in 
\emph{prenex form} 
$\BOX^{q_{1}}_{a_{1}}\dots \BOX^{q_{k}}_{a_{k}}\fone$, 
where $\fone$ is quantifier-free, and 
(2) showing that the prenex formulas with $k$ nested counting
quantifiers characterize the level $k$ of the counting hierarchy.

\subsection{A Logical Foundation of Type Systems for Randomized Calculi}\label{section2B}

Propositional formulas are not only instances of problems.
Indeed, they can also be seen as abstractions specifying 
the intended behavior of higher-order programs. 
As an example, the formula 
$(\fone\rightarrow\ftwo)\rightarrow\fthree$ can be seen 
as the type of a program $t$ taking in input a function of type
$(\fone\rightarrow\ftwo)$ and producing in output an object of type
$\fthree$. 
This analogy between formulas and types can be generalized
to programs and proofs, becoming a formal correspondence, 
known as the Curry-Howard correspondence.
If the underlying program has probabilistic features,
it can well happen that, for example, 
the program $t$ above produces in output an element of type
$\fthree$ \emph{only half of the times}, 
or that $t$ requires that the input function takes in input an object 
of type $\fone$ only \emph{with probability at most equal to} 
$\frac{3}{4}$. 
Indeed, some type systems specifically designed for
so-called randomized $\lambda$-calculi exist in which types are
enriched with quantitative information about the probability of
behaving according to the underlying type, 
this way somehow lifting the probabilistic monad up to the 
level of types~\cite{DalLagoGrellois,BreuvartDalLago,ADLG}.

Can all this be justified from a logical viewpoint, in the spirit of
the aforementioned Curry-Howard correspondence? 
The literature does not offer much in this direction.
Thanks to $\PPL$, however, the answer to the question above 
becomes positive, as we will show in Section~\ref{section6}. 
Just to give you a glimpse of what waits for you there, 
if we ``enrich'' the formula 
$(\fone\rightarrow\ftwo)\rightarrow\fthree$,
with counting quantifiers, for example, 
in the following way,
\begin{equation}\label{equ:exampletype}
(\BOX^{\frac{3}{4}}_a\fone\rightarrow\ftwo\big)\rightarrow\BOX^{\frac{1}{2}}_b\fthree
\end{equation}
we obtain a type which can be used to capture probabilistic behaviors.
In fact, the subscripts $a$ and $b$ very nicely correspond 
to probabilistic events, in the sense of the recent work by 
Dal Lago et al. on decomposing probabilistic 
$\lambda$-calculi~\cite{DLGH}. 
Remarkably, the resulting type system not only enjoys 
subject reduction, 
but can also be proved to normalize in the expected way, 
for example any term typable as~(\ref{equ:exampletype}) 
(when applied to an argument of the appropriate type) 
normalizes to a value having type $\fthree$ 
with probability at least equal to $\frac{1}{2}$.

%% file: section3.tex

In this section, we introduce propositional counting logic in a
simplistic version, called $\PPLc$.  
We present a sound and complete proof system for $\PPLc$ 
and show that derivability in this logic is in the class $\Psharp$. 
In the next Section~\ref{section4}, we will introduce the whole
propositional counting logic, $\PPL$, 
as a generalization of $\PPLc$ which additionally 
provides a characterization of each layer of the 
counting hierarchy, Section~\ref{section5}.

\subsection{From Valuations to Measurable Sets}
In standard propositional logic, the interpretation of a formula
containing propositional atoms, $p_{1},\dots, p_{n}$, call it $\fone$,
is a truth-value, $\model{\fone}_{\eval}$, depending on a valuation
$\eval: \{p_{1},\dots, p_{n}\}\to 2$ (where $2=\{0,1\}$). 
Yet, what if one defined the semantics of $\fone$ as consisting 
of \emph{all} valuations making $\fone$ true? 
Since propositional formulas can have an arbitrary number 
of propositional values, 
it makes sense to consider maps $f: \omega \to 2$
, instead of a finite valuation.
Hence, the interpretation of a formula, call it $\fone$, 
should be the set $\model \fone\subseteq 2^{\omega}$ 
made of all $f\in 2^{\omega}$ ``making $\fone$ true''.
In fact, it is not difficult to see that such sets
are \emph{measurable sets} of the standard Borel algebra
$\mathcal{B}(2^{\omega})$. First, an atomic proposition $p_{i}$ is
interpreted as the \emph{cylinder set},~\cite{Billingsley}
\begin{equation}
\Cyl i= \{f\in 2^{\omega}\mid f(i)=1\}
\end{equation}
of all $f\in 2^{\omega}$, such that $f(i)=1$.
Then, the interpretation of
non-atomic propositions is provided by the standard $\sigma$-algebra
operations of complementation, finite intersection and finite union.

In this setting, we define two counting operators, 
by adapting the notion of counting quantifier presented
in~\cite{Wagner84,Wagner}. 
In doing so, given a formula $\fone$, it becomes possible to
express that $\fone$ is satisfied by, 
for example \emph{at least half} or 
by \emph{strictly less than half} of its models.
Semantically, this amounts at checking that 
$\mu(\model \fone)\geq\frac{1}{2}$ or $\mu(\model \fone)< \frac{1}{2}$,
where $\mu$ denotes the standard Borel measure
on $\mathcal B(2^{\omega})$.
Generally speaking, for all rational
numbers $q\in \mathbb{Q}_{[0,1]}$, 
we can define two new formulas 
$\BOX^{q}\fone$ and $\DIA^{q}\fone$ 
(depending on no propositional atom) with the
intuitive meaning that $\BOX^{q}\fone$ is \valid \ when
$\mu(\model{\fone})\geq q$ and 
that $\DIA^{q}\fone$ is \valid \ when $\mu(\model{\fone})<q$.
This leads to the following definition:
\begin{definition}[Semantics of $\PPLc$]\label{semantics}
The \emph{formulas} of $\PPLc$ are defined by the grammar below:
$$
\fone, \ftwo := \atom{i} \midd  \lnot \fone \midd \fone\land \ftwo\midd \fone\lor \ftwo \midd \BOX^{q}\fone\midd \DIA^{q}\fone
$$
where $i$ is an arbitrary natural number and $q$ is a rational number
from the interval $[0,1]$. 
For each formula $\fone$ of $\PPLc$, its \emph{interpretation} 
$\model{\fone}\in \mathcal B(2^{\omega})$ 
is a measurable set inductively defined as follows:
\begin{align*}
\model{\atom{i}}& = \Cyl i \\
\model{\lnot \fone}&= 2^{\omega}-\model{\fone} \\
\model{\fone\land \ftwo}& = \model{\fone}\cap \model{\ftwo} \\
\model{\fone \lor \ftwo}& = \model{\fone}\cup \model{\ftwo} \\
\model{\BOX^{q}\fone}&= 
	\begin{cases}
	\twoOm & \text{ if }\mu(\model \fone)\geq q \\
	\emptyset & \text{ otherwise }
	\end{cases}\\
\model{\DIA^{q}\fone}&= 
	\begin{cases}
	\twoOm & \text{ if }\mu(\model \fone)< q \\
	\emptyset & \text{ otherwise }
	\end{cases}
\end{align*}
\end{definition}
\noindent
The notion of logical equivalence is defined in a standard way: 
two formulas, call them $\fone$ and $\ftwo$, 
are \emph{logically equivalent}, $\fone \equiv \ftwo$, 
when $\model{\fone}=\model{\ftwo}$.
A formula, call it $\fone$, is \emph{\valid}
when $\model \fone=\twoOm$, and \emph{invalid}
when $\model{\fone}=\emptyset$.
For all classical tautologies $\fone$, 
each formulas of the form $\BOX^{q}\fone$ is \valid, whereas
for any formula of $\PPLc$, call it $\ftwo$, formulas of the form
$\BOX^{0}\ftwo$ are \valid, see Lemma~\ref{lemma:box0}
in Appendix~\ref{appendix3.1}.

Counting quantifiers are somehow non-standard.
Indeed, they are inter-definable but not dual, as:\footnote{For the detailed proof, see Lemma~\ref{lemma:interdef} in Appendix~\ref{appendix3.1}.}
\begin{equation}\label{equ:interdefinability}
\BOX^{q}\fone \equiv \lnot \DIA^{q}\fone \qquad \qquad \DIA^{q}\fone \equiv \lnot\BOX^{q}\fone.
\end{equation}
Remarkably, using the these quantifiers, 
it is also possible to express that a formula is satisfied 
by \emph{strictly more than $q$} or by \emph{no more than $q$},
for example, given a formula $\fone$, this can formalized as 
(respectively) $\DIA^{1-q} \lnot A $ and $\BOX^{1-q}\lnot A$.
In order to further clarity the semantics of $\PPLc$, 
let us consider a few examples.
\begin{example}
Let $\fone=\ftwo \lor \fthree$, where $\ftwo,\fthree$ are respectively
the formulas $\atom 0\land \lnot \atom 1$ and $\lnot \atom 0 \land\atom  1$.
As it can easily
be checked, the two measurable sets $\model{\ftwo}$ and
$\model{\fthree}$ are disjoint and have both measure $\frac{1}{4}$.
Hence
$\mu(\model{\fone})=\mu(\model{\ftwo})+\mu(\model{\fthree})=\frac{1}{2}$,
which means that the formula $\BOX^{\frac{1}{2}}\fone$ is \valid.
\end{example}
\begin{example}
Let $\fone=\ftwo\lor\fthree$, where $\ftwo,\fthree$ are respectively
the formulas $( \atom 0\land\lnot  \atom 1)\vee \atom 2$ and $( \lnot \atom 0 \land\atom  1)\lor \atom 2$.  
In this case the two
sets $\model{\ftwo}$ and $\model{\fthree}$ have  
both measure $\frac{5}{8}$ 
(in fact, $5$ of their $8$ models satisfy them) but they are not disjoint, and
$\model{\ftwo}\cap \model{\fthree}=\Cyl 2$, 
which has measure $\frac{1}{2}$. 
Hence,
$\mu(\model{\fone})=\mu(\model{\ftwo})+\mu(\model{\fthree})-\mu\big(\Cyl
2\big)=\frac{3}{4}$, which means that the formula $\DIA^{1}\fone$ is \valid.
\end{example}

\subsection{The Proof Theory of $\PPLc$}
In order to obtain a sound and complete proof system for $\PPLc$, 
a one-sided, single-succedent, and
labelled sequent calculus is introduced.
Notice that, although the quantitative semantics for our counting logic 
appears as very natural, 
the definition of a complete deductive system is not so trivial.
The language of our calculus is defined by
sequents made of \emph{labelled formulas}, 
$\bone\Pto \fone$, $\bone \Pfrom \fone$, 
where $\fone$ is a formula of $\PPLc$ and $\bone$ is Boolean formula.
Intuitively, a labelled formula $\bone\Pto \fone$ is true when 
the (measurable) set of valuations making $\bone$ true 
is included in $\model \fone$, and similarly, 
a labelled formula $\bone\Pfrom \fone$ is true when 
$\model \fone$ is included in the set of valuations making
$\bone$ true.

In order to formally define the language of this calculus, 
we will start by introducing Boolean formulas.
\begin{definition}[Boolean Formulas]\label{Bool}
\emph{Boolean formulas}, starting from a countable set of atoms,
$\bvar_{0},\bvar_{1},\bvar_{2},\dots$, are defined by the following grammar:
$$
\bone, \btwo:= \bvar_{i} \midd \top\midd \bot\midd  \lnot \bone \midd \bone \land \btwo \midd \bone \lor \btwo
$$
The
interpretation $\model \bone\in \mathcal B(2^{\omega})$ of a Boolean
formula $\bone$ is defined inductively as follows:
\begin{align*}
\model{\bvar_{i}} &= \Cyl{i} &&&
\model{\neg \bone} &= \twoOm \ – \ \model{\bone} \\
\model{\top} &= 2^{\omega} &&&
\model{\bone \land \btwo} &= \model \bone \cap \model \btwo \\
\model{\bot} &= \emptyset &&&
\model{\bone \lor \btwo} &= \model \bone \cup \model \btwo
\end{align*}
\end{definition}
\noindent Let $ \bone \vDash \btwo$ and $ \bone \Dashv \vDash \btwo$
be shorthands for, respectively, 
$\model \bone\subseteq \model \btwo$ and 
$\model \bone =\model \btwo$.  
Then, \emph{external hypotheses} and \emph{labelled formulas}
can be introduced.
\begin{definition}[External Hypotheses]
An \emph{external hypothesis} is an expression of one of the following forms: $\bone \vDash \btwo$, $\bone \Dashv \vDash \btwo$ or $\mu(\model{\bone}) \triangleright q$, where $\triangleright \in \{\ge, \leq, <, >, =\}$, $\bone, \btwo$ are Boolean formulas, and $q \in \mathbb{Q}_{[0,1]}$.
\end{definition}
\begin{definition}[Labelled Formulas]
A \emph{labelled formula}, call it \lone, 
is an expression of the form
$\bone \Pto \fone$ or $\bone \Pfrom \fone$, 
where $\bone$ is a Boolean formula and $\fone$ is a formula of $\PPLc$. 
\end{definition}
\noindent
The measure of the interpretation of a Boolean formula, $\mu(\model \bone)$, can be related to the number $\SSAT(\bone)$ of the valuations making $\bone$ true as follows:
\begin{lemma}\label{lemma:counti}
For each Boolean formula, call it $\bone$, with propositional variables,
$\bvar_{0},\dots, \bvar_{n-1}$,
$$
\mu(\model \bone)={\SSAT(\bone)}\cdot{2^{-n}}.
$$
\end{lemma}
\begin{proof}
To any valuation $\eval: \{\bvar_{0},\dots, \bvar_{n-1}\}\to 2$ we can
associate a measurable set $X(\eval)\in \mathcal B(2^{\omega})$ by
letting
$$
X(\eval)= \{f\mid \forall i<n, f(i)=\eval(\bvar_{i})\}= \bigcap_{i=0}^{n-1}\Cyl{i}^{\eval(\bvar_{i})}
$$
where $\Cyl{i}^{\eval(\bvar_{i})}$ is $\Cyl{i}$ if $\eval(\bvar_{i})=1$ and
$\overline{\Cyl{i}}$ otherwise. 
One can check by induction that for
all Boolean formula $\bone$, $\model \bone= \bigcup_{\eval\vDash
\bone}X(\eval)$.  
Then, since for all distinct $\eval,$ and $\eval'$,
$X(\eval)\cap X(\eval')=\emptyset$, we deduce
$\SSAT(\bone)\cdot 2^{-n}=\sum_{\eval \vDash \bone}2^{-n}=
\sum_{\eval\vDash \bone}\mu(X(\eval))=
\mu(\bigcup_{\eval\vDash \bone}X(\eval))=\mu( \model \bone)$.
\end{proof}

A \emph{labelled sequent} of $\PPLc$ is a sequent of the form $\vdash \lone$.  
The proof system for $\PPLc$ is defined by the rules displayed 
in Fig.~\ref{fig:pplcrules}. 
We let $\vdash_{\PPLc}\lone$ denote the fact that
$\vdash \lone$ is derivable by those rules.  
Notice that some of the rules have semantic statements among their
premisses, something which can be seen as unsatisfactory. 
However, as it will become clear in Subsection~\ref{Subsection3C},
 such premisses make sense from a computational viewpoint: 
 they correspond to the idea that, when searching for a proof, 
 one might need to call for an \emph{oracle}, 
 a central concept in the definition of counting complexity classes.
 The presence of such external hypothesis does not make our proof
 theory trivial and, in fact, proof search is sometimes convoluted.
 Derivations are defined in the standard way.
We illustrate in Fig.~\ref{fig:examplepplc} a 
simple example of derivation in $\PPLc$. 
\begin{figure}
\begin{center}
\fbox{
\begin{minipage}{.97\textwidth}
\small
\begin{center}
Initial Sequents
\end{center}
\normalsize
\begin{minipage}{\linewidth}
\adjustbox{scale=1,center}{
$
\AxiomC{$\bone  \vDash x_{n}$}
\RightLabel{$\AxO$}
\UnaryInfC{$\vdash \bone  \Pto \atom{n}$}
\DP
\qquad \qquad \qquad
\AxiomC{$x_{n}\vDash \bone$}
\RightLabel{$\AxT$}
\UnaryInfC{$\vdash \bone  \Pfrom \atom{n}$}
\DP
$}
\end{minipage}

\small
\vskip2mm
\begin{center}
Union and Intersection Rules
\end{center}
\normalsize

\adjustbox{scale=1,center}{
$
\AxiomC{$\vdash \btwo \Pto \fone$}
\AxiomC{$\vdash \bthree \Pto \fone$}
\AxiomC{$ \bone  \vDash \btwo\lor \bthree $}
\RightLabel{$\Rcup$}
\TrinaryInfC{$\vdash \bone  \Pto \fone$}
\DP
$}

\vskip4mm
\adjustbox{scale=1,center}{
$
\AxiomC{$\vdash \btwo \Pfrom \fone$}
\AxiomC{$\vdash \bthree  \Pfrom \fone$}
\AxiomC{$\btwo\land \bthree\vDash   \bone $}
\RightLabel{$\Rcap$}
\TrinaryInfC{$\vdash \bone \Pfrom \fone$}
\DP
$}

\vskip2mm
\small
\begin{center}
Logical  Rules
\end{center}
\normalsize

\begin{minipage}[t]{\linewidth}
\begin{minipage}[t]{0.4\linewidth}
\begin{prooftree}
\AxiomC{$\vdash \btwo \Pfrom \fone$}
\AxiomC{$\bone \vDash   \lnot \btwo$}
\RightLabel{$\Rnra$}
\BinaryInfC{$\vdash \bone  \Pto \neg \fone $}
\end{prooftree}
\end{minipage}
\hfill
\begin{minipage}[t]{0.5\linewidth}
\begin{prooftree}
\AxiomC{$\vdash \btwo \Pto \fone$}
\AxiomC{$ \lnot{\btwo} \vDash \bone $}
\RightLabel{$\Rnla$}
\BinaryInfC{$\vdash \bone  \Pfrom \neg \fone$}
\end{prooftree}
\end{minipage}
\end{minipage}

\bigskip
\begin{minipage}{\linewidth}
\begin{minipage}[t]{0.4\linewidth}
\begin{prooftree}
\AxiomC{$\vdash \bone \Pto \fone$}
\RightLabel{$\mathsf{R1}_{\vee}^{\Era}$}
\UnaryInfC{$\vdash \bone  \Pto \fone \vee \ftwo$  }
\end{prooftree}
\end{minipage}
\hfill
\begin{minipage}[t]{0.5\linewidth}
\begin{prooftree}
\AxiomC{$\vdash  \bone  \Era \ftwo$}
\RightLabel{$\mathsf{R2}^{\Era}_{\vee}$}
\UnaryInfC{$\vdash \bone  \Pfrom \fone \vee \ftwo$  }
\end{prooftree}
\end{minipage}
\end{minipage}

\begin{prooftree}
\AxiomC{$\vdash\bone\Ela\fone$}
\AxiomC{$\vdash\bone\Ela\ftwo$}
\RightLabel{$\Rvla$}
\BinaryInfC{$\vdash\bone\Ela\fone\vee\ftwo$}
\end{prooftree}

\begin{prooftree}
\AxiomC{$\vdash\bone\Era\fone$}
\AxiomC{$\vdash\bone\Era\ftwo$}
\RightLabel{$\Rwra$}
\BinaryInfC{$\vdash\bone\Era\fone\wedge\ftwo$}
\end{prooftree}

\begin{minipage}[t]{\linewidth}
\begin{minipage}[t]{0.4\linewidth}
\begin{prooftree}
\AxiomC{$\vdash \bone \Ela \fone$}
\RightLabel{$\textsf{R1}_{\wedge}^{\Ela}$}
\UnaryInfC{$\vdash \bone  \Ela \fone \wedge \ftwo$ }
\end{prooftree}
\end{minipage}
\hfill
\begin{minipage}[t]{0.5\linewidth}
\begin{prooftree}
\AxiomC{$\vdash \bone  \Pfrom \ftwo$ }
\RightLabel{$\mathsf{R2}_{\wedge}^{\Ela}$}
\UnaryInfC{$\vdash \bone \Pfrom \fone \wedge \ftwo$}
\end{prooftree}
\end{minipage}
\end{minipage}

\small
\begin{center}
Counting Rules
\end{center}
\normalsize

\begin{minipage}{\linewidth}
\begin{minipage}[t]{0.4\linewidth}
\begin{prooftree}
\AxiomC{$\mu(\model \bone) = 0$}
\RightLabel{$\Rmur$}
\UnaryInfC{$\vdash \bone  \Pto \fone$}
\end{prooftree}
\end{minipage}
\hfill
\begin{minipage}[t]{0.5\linewidth}
\begin{prooftree}
\AxiomC{$\mu(\model \bone) = 1$}
\RightLabel{$\Rmul$}
\UnaryInfC{$\vdash \bone  \Pfrom \fone$}
\end{prooftree}
\end{minipage}
\end{minipage}

\bigskip
\begin{minipage}{\linewidth}
\begin{minipage}[t]{0.4\linewidth}
\begin{prooftree}
\AxiomC{$\vdash \btwo \Pto \fone$}
\AxiomC{$\mu(\model \btwo) \ge q$}
\RightLabel{$\Rbr$}
\BinaryInfC{$\vdash \bone \Pto \BOX^{q}\fone$}
\end{prooftree}
\end{minipage}
\hfill
\begin{minipage}[t]{0.5\linewidth}
\begin{prooftree}
\AxiomC{$\vdash \btwo \Pfrom \fone$}
\AxiomC{$\mu(\model \btwo) < q$}
\RightLabel{$\Rbl$}
\BinaryInfC{$\vdash  \bone \Pfrom$ $\BOX^{q}\fone$}
\end{prooftree}
\end{minipage}
\end{minipage}

\bigskip
\begin{minipage}{\linewidth}
\begin{minipage}[t]{0.4\linewidth}
\begin{prooftree}
\AxiomC{$\vdash \btwo \Pfrom \fone$}
\AxiomC{$\mu(\model \btwo) < q$}
\RightLabel{$\Rdr$}
\BinaryInfC{$\vdash \bone \Pto \DIA^{q}$\fone}
\end{prooftree}
\end{minipage}
\hfill
\begin{minipage}[t]{0.5\linewidth}
\begin{prooftree}
\AxiomC{$\vdash \btwo \Pto \fone$}
\AxiomC{$\mu(\model \btwo) \ge q$}
\RightLabel{$\Rdl$}
\BinaryInfC{$\vdash \bone \Pfrom$ $\DIA^{q}$\fone}
\end{prooftree}
\end{minipage}
\end{minipage}

\end{minipage}}
\caption{\small{Proof system for $\PPLc$.}}
\label{fig:pplcrules}
\end{center}
\end{figure}

\begin{figure}
\begin{center}
\fbox{
\begin{minipage}{.97\textwidth}
\adjustbox{scale=1,center}
{$
\AXC{$x_{1}\vDash x_{1}$}
\RightLabel{$\AxO$}
\UIC{$\vdash x_{1} \Pto \atom{1}$}
\RightLabel{$\mathsf{R1}_{\vee}^{\Era}$}
\UIC{$\vdash x_{1} \Pto \atom{1}\vee \atom{2}$}
\AXC{$x_{2}\vDash x_{2}$}
\RightLabel{$\AxO$}
\UIC{$x_{2} \Pto \atom{2}$}
\RightLabel{$\mathsf{R2}^{\Era}_{\vee}$}
\UIC{$\vdash x_{2} \Pto \atom{1}\vee \atom{2}$}
\RightLabel{$\Rcup$}
\BIC{$\vdash x_{1} \lor x_{2} \Pto \atom{1}\vee \atom{2}$}
\AXC{$\mu(\model{x_{1}\lor x_{2}})\geq \frac{3}{4}$}
\RightLabel{$\Rbr$}
\BIC{$\vdash \top\Pto \BOX^{\frac{3}{4}}(\atom{1}\vee \atom{2})$}
\DP
$}
\end{minipage}}
\caption{\small{Derivation $\vdash \top\Pto \BOX^{\frac{3}{4}}(\atom{1}\lor\atom{2})$.}}
\label{fig:examplepplc}
\end{center}
\end{figure}

As anticipated, the given proof system is sound and complete with respects of the semantics of $\PPLc$. 
Validity of labelled formulas and sequents is defined as
follows:
\begin{definition}[Validity]\label{valB}
A labelled formula $\bone\Pto \fone$ $($resp. $\bone\Pfrom \fone)$
is \emph{valid}, denoted $\vDash \bone\Pto \fone$
(resp. $\vDash \bone\Pfrom \fone$), when
$\model \bone \subseteq \model \fone$
(resp. $\model \fone \subseteq \model \bone$).  
A sequent
$\vdash\lone$ is \emph{valid}, denoted $\vDash \lone$, 
when it contains a valid labelled formula.
\end{definition} 
\begin{remark}
For any formula of $\PPLc$, call it $\fone$, 
one can construct a Boolean formula $\bone_{\fone}$ 
such that $\model \fone= \model{\bone_{\fone}}$. The
only non-trivial cases are $\fone=\BOX^{q}\ftwo$ and
$\fone=\DIA^{q}\ftwo$. In the first case we let $\bone_{\fone}=\top$
if $\mu(\model \ftwo)\geq q$ and $\bone_{\fone}=\bot$ if
$\mu(\model \ftwo)<q$, and similarly for the second case.
\end{remark}

The soundness of the proof system can be proved by a standard
induction on derivation height.\footnote{The complete proof is offered in the corresponding Proposition~\ref{soundnessAP} in Appendix~\ref{appendix3.2}.}
\begin{restatable}[Soundness]{proposition}{primeSoundness}\label{soundnessAP}
If $\vdash_{\PPLc} \lone$ holds, then $\vDash \lone$ holds.
\end{restatable}
\noindent
The completeness proof is more complex.
It relies on the introduction of a \emph{decomposition relation}, 
$\Dec$, between finite sets of sequents. 
Reduction $\Dec$ is in its turn defined basing on a
relation $\dec$ between sequents and 
finite sets of sequents,
which is generated by clauses like the following ones:
\begin{align*}
\text{if }\bone \vDash \btwo\lor \bthree, \quad  \vdash \bone \Pto \fone\lor \ftwo \ & \leadsto_{0} \ \{\vdash \btwo\Pto \fone, \vdash \bthree\Pto \ftwo\} \\
 \quad  \vdash \bone \Pto \fone\land \ftwo \ &\leadsto_{0} \ \{\vdash \bone\Pto \fone, \vdash \bone\Pto \ftwo\} \\
\text{if } \mu(\model{\btwo})\geq q, \quad \vdash \bone\Pto \BOX^{q}\fone \ & \leadsto_{0} \ \{\vdash \btwo\Pto \fone\}\\
\text{if }\mu(\model{\btwo})< q, \quad  \vdash \bone\Pto \DIA^{q}\fone \ & \leadsto_{0} \ \{\vdash \btwo\Ela \fone\}
\end{align*}
\noindent
and letting $\ssetO\Dec \ssetT$ whenever, $\ssetO=\ssetO'\cup\{ \vdash \lone\}$,
$\vdash \lone \dec \ssetT'$ and $\Psi=\ssetO'\cup \ssetT'$.\footnote{The 
complete definition of $\dec$ and $\Dec$ is presented in the Appendix~\ref{appendix3} and specifically, see respectively, Definition~\ref{decomposition} and Definition~\ref{Decomposition}.}
\begin{restatable}{lemma}{primeDec}\label{sNorm}
The reduction $\Dec$ is strongly normalizing.
\end{restatable} 
\noindent
Lemma~\ref{sNorm} is proved in a standard way by
defining a measure $\ms(\ssetO)$ over finite sets of sequents and showing
that whenever $\ssetO\Dec \ssetT$, $\ms(\ssetT)<\ms(\ssetO)$.
For a sequent $\seqO$ = $\vdash \lone$, 
let $\ms(\seqO)$ be the size of $\lone$, which is $\cn(\lone)$.
Then, $\ms(\ssetO)$ is defined as $\sum_{i}^{k}3^{\ms(\seqO_{i})}$,
where $\ssetO=\{\seqO_{1},\dots, \seqO_{k}\}$.\footnote{For further details on the proof, see Appendix~\ref{appendix3.2}. Specifically, the precise notions of 
number of connectives and of set measure can be found as 
Definition~\ref{def:cn} and Definition~\ref{def:ms},
while the proof of $\Dec$ strong normalization is presented as
Lemma~\ref{sNorm} in~\ref{appendix3.2}.}
For a finite set of sequents $\ssetO=\{\seqO_{1},\dots, \seqO_{k}\}$, let
$\vDash \ssetO$ 
(resp., $\vdash_{\PPLc}\ssetO$) indicate that 
$\vDash \seqO_{i}$
(reps., $\vdash_{\PPLc}\seqO_{i}$) holds for all $i=1,\dots,k$. 
The fundamental ingredients of the completeness proof 
are expressed by the following three properties:
\begin{lemma}\label{lemma:decomposition2}
Concerning $\PPLc$, the following properties hold:
\begin{itemize}
\itemsep0em
\item[i] if $\ssetO$ is $\Dec$-normal, then $\vDash \ssetO$ if and only if
 $\vdash_{\PPLc}\ssetO$
\item[ii.] if $\vDash \ssetO$, there is a $\ssetT$ such that $\ssetO \Dec \ssetT$ and $\vDash \ssetT$
\item[iii.] if $\vdash_{\PPLc}\ssetT$ and $\ssetO \Dec \ssetT$, then $\vdash_{\PPLc}\ssetO$.
\end{itemize}
\end{lemma}
\noindent
Using the properties above, one can check that if $\vDash \lone$
holds, then by ii. $\vDash \ssetO$ holds for some $\leadsto$-normal form
$\ssetO$ of $\vdash \lone$; then by i., $\vdash_{\PPLc}\ssetO$ holds and by iii. we
conclude that $\vdash_{\PPLc}\lone$.\footnote{The full proof of completeness is presented in Appendix~\ref{appendix3.2}. Specifically,
i. is established as Lemma~\ref{NormValDer},
 ii. is proved as Lemma~\ref{ExVal},
 and iii. corresponds to Lemma~\ref{DerPres}.}
\begin{restatable}[Completeness]{proposition}{primeCompleteness}\label{completeness}
If $\vDash \lone$ holds, then $\vdash_{\PPLc}\lone$ holds.
\end{restatable}
\begin{remark}\label{rem:cutrule}
A remarkable consequence of the completeness theorem is that the following cut-rule comes immediately out to be derivable in $\PPLc$:
$$
\AXC{$\vdash \btwo\Pto \lnot A\vee B$}
\AXC{$\vdash \bthree\Pto A$}
\AXC{$\bone\vDash \btwo\land \bthree$}
\RL{$\mathsf{Cut}$}
\TIC{$\vdash \bone \Pto B$}
\DP
$$
\end{remark}

\subsection{$\PPLc$-validity is in $\Psharp$}\label{Subsection3C}

As the right-hand premisses of the counting rules 
$\Rbr/\Rbl$ and $\Rdr/ \Rdl$ suggest, 
the inference of formulas such as $\BOX^{q}\fone$ or
$\DIA^{q}\fone$ can be seen as obtained by invoking
an \emph{oracle}, providing a suitable measurement. 
Actually, the needed forms of measurement are expressed 
as  $\mu(\model \bone)$, 
where $\bone$ is a Boolean formula.
As shown by Lemma~\ref{lemma:counti}, 
these measurements correspond to actually \emph{counting} 
the number of valuations satisfying such formulas.
We can make this intuition precise by showing that validity
in $\PPLc$ can be decided by a polytime algorithm having access to an oracle for the problem $\SSAT$ of counting the models of a Boolean formula.

Let us define a formula of $\PPLc$, call it $\fone$, 
\emph{closed} if it is either of the form
$\BOX^{q}\ftwo$ or $\DIA^{q}\ftwo$ or is a negation, 
conjunction or disjunction of closed formulas. 
It can be easily checked by induction on the structure
of formulas that for all closed formula, its interpretation
is either $2^{\omega}$ or $\emptyset$, see Lemma~\ref{app3.3} in Appendix~\ref{appendix3.3}.
We define by mutual recursion two polytime algorithms 
$\BOOL$ and $\VAL$ such that, for each formula of $\PPLc$, 
call it $\fone$,  $\BOOL(\fone)$ is the Boolean formula 
$\bone_{\fone}$, and, for all closed formula $\fone$, 
$\VAL(\fone)=1$ if and only if $\model \fone=2^{\omega}$ 
and $\VAL(\fone)=0$ if and only if $\model \fone=\emptyset$.
The algorithm $\VAL(\fone)$ will make use of a $\SSAT$ oracle.
The two algorithms are defined as follows:

\medskip
\noindent
\adjustbox{scale=0.95}{
\begin{minipage}{\linewidth}
\begin{align*}
\BOOL(n) & = x_{n} \\
\BOOL(\fone_{1}\land \fone_{2}) & = \BOOL(\fone_{1}) \land \BOOL(\fone_{2}) \\
\BOOL(\fone_{1}\lor \fone_{2}) & = \BOOL(\fone_{1}) \lor \BOOL(\fone_{2}) \\
\BOOL(\lnot \fone_{1}) & = \lnot \BOOL(\fone_{1})  \\
\BOOL(\BOX^{q}\fone_{1}) & = \VAL(\BOX^{q}\fone_{1})\\
\BOOL(\DIA^{q}\fone_{1}) & = \VAL(\DIA^{q}\fone_{1})
\end{align*}
\begin{align*}
\VAL(A_{1}\land \fone_{2}) & = \VAL(\fone_{1}) \texttt{ AND } \VAL(\fone_{2}) \\
\VAL(\fone_{1}\lor \fone_{2}) & = \VAL(\fone_{1}) \texttt{ OR } \VAL(\fone_{2}) \\
\VAL(\lnot \fone_{1}) & = \texttt{NOT }  \VAL(\fone_{1})  \\
\begin{matrix}
\VAL(\BOX^{q}\fone_{1}) \\ \ \\ \ 
\end{matrix}  & \  
\begin{matrix}
=& \texttt{let } \bone =\BOOL(\fone_{1})\ \texttt{in}\\
& \texttt{let } n= \sharp \texttt{Var}(\bone) \ \texttt{in } \\
&  \frac{\SSAT(\bone)}{2^{n}} \geq q 
 \end{matrix}\\
 \begin{matrix}
\VAL(\DIA^{q}\fone_{1}) \\ \ \\ \ 
\end{matrix}  & \  
\begin{matrix}
=& \texttt{let } \bone =\BOOL(\fone_{1})\ \texttt{in}\\
& \texttt{let } n= \sharp \texttt{Var}(\bone) \ \texttt{in } \\
&  \frac{\SSAT(\bone)}{2^{n}} < q 
 \end{matrix}\end{align*}
 \end{minipage}
 }
 \medskip

\noindent
where $\sharp\texttt{Var}(\bone)$ is the number of propositional variables in $\bone$.
We then immediately deduce:
\begin{proposition}
$\PPLc$-validity is in $\Psharp$.
\end{proposition}
\noindent
Can we say anything about the \emph{completeness} of $\PPLc$-validity?
Under Turing reductions, the latter is undoubtedly a complete problem
for $\Psharp$, since an oracle for $\PPLc$-validity can, of course,
solve any instance of $\SSAT$. Unfortunately, the same cannot be
said about many-one reductions.

%% file: section4.tex

After having introduced the basic counting logic $\PPLc$ 
and discussed its connections with the class $\Psharp$, 
let us now take into account a more expressive logic $\PPL$, 
that will be shown in Section~\ref{section5} to yield
a characterization of the full counting hierarchy.

\subsection{A Logic for Relative Counting Problems}

As discussed in Section~\ref{section2}, the counting problems 
coming from $\PL$ are not restricted to those in the class $\Psharp$. 
Indeed, behind counting the models of a formula, one can consider
multivariate counting problems, such as $\MMSat$, which
involves the relation between valuations of \emph{different} groups of
variables.
In order to address this kind of problems in our logic, we need counting
quantifiers to relate to groups of variables. 
To
this end we introduce a countable alphabet of \emph{names},
$a,b,c,\dots$, so that any logical atom, such as $\atom{i}$, 
becomes named, $\at{a}$, and, equally, 
counting quantifiers turn into named ones, 
$\BOX^{q}_{a}\fone$ or $\DIA^{q}_{a}\fone$.
\begin{definition}[Grammar of $\PPL$]
The \emph{formulas} of $\PPL$ are defined by the grammar below:
$$
A,B:= \ \at{a}\midd \lnot \fone\midd \fone\land \ftwo\midd \fone\lor \ftwo\midd \BOX^{q}_{a}\fone\midd \DIA^{q}_{a}\fone
$$
where $\at{}$ is a natural number, $a$ is a name, 
and $q$ is a rational
number included in the interval $[0,1]$.
\end{definition}
\noindent
The intuitive meaning of \emph{named quantifiers} is that 
they count models \emph{relative} to the corresponding 
bounded variables.
The name $a$ is considered bound in the formulas 
$\BOX^{q}_{a}A$ and $\DIA^{q}_{a}\fone$.
Given a formula of $\PPL$, call it $\fone$, we denote 
the set of names which occur \emph{free} in it, $\FN(\fone)$.

The semantics of these named formulas is slightly subtler 
than that of $\PPLc$.
Since we must take names into account, the interpretation of
a formula $\fone$ of $\PPL$ will no longer be defined as a 
measurable set $\model\fone\in \mathcal{\ftwo}(2^{\omega})$. 
Instead, for all finite set of names $X\supseteq \FN(\fone)$, 
we will define a measurable set $\model A_{X}$
belonging to the Borel algebra $\mathcal{B}((2^{\omega})^{X})$. 
To describe the semantics of $\PPL$ we first need to introduce a
fundamental operation:
\begin{definition}[$f$-projection]
Let $X,Y$ be two disjoint finite sets of names and $f\in (2^{\omega})^{X}$. For all $\mathcal X\subseteq (2^{\omega})^{X\cup Y}$, 
the \emph{$f$-projection of $\mathcal X$} is the set:
$$
\PROJ_{f}(\mathcal X)=\{ g\in (2^{\omega})^{Y}\mid f+g\in \mathcal X\}\subseteq (2^{\omega})^{Y}
$$ 
where $(f+g)(\alpha)$ is $f(\alpha)$, if $\alpha\in X$ and $g(\alpha)$ if $\alpha\in Y$.
\end{definition}
\noindent
Suppose $X$ and $Y$ are disjoint sets of names, 
with $\FN(\fone)\subseteq X\cup Y$.  
Then, if we fix a valuation $f\in (2^{\omega})^{X}$ of the
variables of $A$ with names in $X$, the set 
$\PROJ_{f}(\model \fone_{X\cup Y})$ 
describes the set of valuations of the variables of $\fone$ 
with names in $Y$ which extend $f$.
We can now define the set $\model \fone_{X}$ formally:
\begin{definition}[Semantics of $\PPL$]\label{PPLsem}
For all formula $\fone$ of $\PPL$ and finite set of names
$X\supseteq \FN(A)$, the set $\model \fone_{X}$ is inductively defined as
follows:
\begin{align*}
\model{\at{a}}_{X}& = \{ f\mid f(a)(i)=1\} \\
\model{\lnot \fone}_{X}& = (\twoOm)^{X} \ - \ \model \fone \\
\model{\fone\land \ftwo}_{X}&= \model{\fone}_{X}\cap \model{\ftwo}_{X} \\
\model{\fone\lor \ftwo}_{X}&= \model{\fone}_{X}\cup \model{\ftwo}_{X} \\
\model{ \BOX^{q}_{a}\fone}_{X}&= \{ f \mid  \mu ( \PROJ_{f}(\model \fone_{X\cup\{a\}}))\geq q\} \\
\model{ \DIA^{q}_{a}\fone}_{X}&= \{ f \mid  \mu ( \PROJ_{f}(\model \fone_{X\cup\{a\}}))< q\} 
\end{align*}
\end{definition}
\noindent
We will show below that $\model \fone_{X} \in \mathcal
B((2^{\omega})^{X})$, which is that the interpretation of a formula of
$\PPL$ is always a measurable set.
To have a grasp of the semantics of named quantifiers, consider the following example:

\begin{example}
Let $\fone$ be the following formula of $\PPL$:
$$
\fone = (\atom{0}_{a} \wedge (\neg \atom{0}_{b} \wedge 
\atom{1}_{b})) \vee (\neg\atom{0}_{a}
\wedge \atom{0}_{b} \wedge \neg \atom{1}_{b})
\vee ((\neg \atom{0}_{a}\wedge \atom{1}_{a})
\wedge \atom{1}_{b}).
$$
The valuations $f \in (\twoOm)^{\{b\}}$
belonging to $\model{\BOX^{\frac{1}{2}}_{a}\fone}_{\{b\}}$
are those which can be extended to models 
of $\fone$ in at least $\frac{1}{2}$ of the cases.
In such a simple situation, we can list all possible
cases:
\begin{enumerate}
\itemsep0em
\item if $f(b)(0)=f(b)(1)=1$, then $\fone$ has $\frac{1}{4}$
chances of being true, since both $\neg \atom{0}_{a}$ and
$\atom{1}_{a}$ must be true,

\item if $f(b)(0)=1, f(b)(1)=0$, then $\fone$ has $\frac{1}{2}$ chances of being true, since $\neg\atom{0}_{a}$ must be true
\item if $f(b)(0)=0, f(b)(1)=1$, then $\fone$
has $\frac{3}{4}$ chances of being true, since either $\atom{0}_{a}$ or both $\neg \atom{0}_{a}$ and $\atom{1}_{a}$ must be true
\item if $f(b)(0)=f(b)(1)=0$, then $\fone$ has 0 chances of being true.
\end{enumerate}
Thus, $\model{\BOX^{\frac{1}{2}}_{a}\fone}_{\{b\}}$
only contains the valuations which agree
with cases 2. and 3. 
Therefore, $\model{\BOX_{b}^{\frac{1}{2}}\BOX^{\frac{1}{2}}_{a}\fone}_{\emptyset}=\twoOm$,
which is 
$\BOX^{\frac{1}{2}}_{b}\BOX^{\frac{1}{2}}_{a}\fone$
is true, since half of the valuations
of $b$ has at least $\frac{1}{2}$ chances of being
extended to a model of $\fone$.
\par
Let us now consider a slightly different case, being $\fone$ the following formula:
$$
\fone = \big(\atom{0}_{a} \wedge (\neg \atom{0}_{b} \wedge \atom{1}_{b})\big) \vee \big(\neg \atom{0}_{a} \wedge \atom{0}_{b}
\wedge \neg \atom{1}_{b}\big)
\vee 
\big((\neg\atom{0}_{a}\vee \atom{1}_{a}) \wedge \atom{1}_{b}\big).
$$
Cases 2. and 4. are equivalent, but 1. and 3. are not. 
Indeed,
\begin{enumerate}
\itemsep0em
\item if $f(b)(0)=f(b)(1)=1$, then $\fone$ has $\frac{3}{4}$
chances of being true, since either $\neg\atom{0}_{a}$ or
$\atom{1}_{a}$ must be true.
\item if $f(b)(0)=1,f(b)(1)=0$, then $\fone$ has $\frac{1}{2}$ chances of being true, since, as before, $\neg\atom{0}_{a}$ must be true.
\item if $f(b)(0)=0, f(b)(1)=1$, then $\fone$ has 1 chances of being true, since either $\atom{0}_{a}$ or either $\neg \atom{0}_{a}$ or $\atom{1}_{a}$ must be true.
\item if $f(b)(0)=f(b)(1)=0$, then $\fone$ has 0 chances of
being true.
\end{enumerate}
In this case $\model{\BOX^{\frac{1}{2}}_{a}\fone}_{\{b\}}$
contains the valuations which agree with cases 1., 2., and 3. Thus, again, $\model{\BOX^{\frac{1}{2}}_{b}\BOX^{\frac{1}{2}}_{a}\fone}_{\emptyset}=\twoOm$, but
also $\model{\BOX^{\frac{3}{4}}_{b}\BOX^{\frac{1}{2}}_{a}\fone}_{\emptyset}=\twoOm$.
\end{example}

\subsection{Characterizing the Semantics of $\PPL$ using Boolean Formulas}

The definition of the sets $\model{\BOX^{q}_{a}\fone}_{X}$ and
$\model{\DIA^{q}_{a}A}_{X}$ is not very intuitive at first glance. 
We now provide an alternative characterization of these sets 
by means of \emph{named} Boolean formulas, 
and show that they are measurable.
\begin{definition}[Named Boolean Formulas]
\emph{Named Boolean formulas} are defined by the following grammar:
$$
\bone,\btwo:= \bvar_{i}^{a}\midd \top\midd \bot\midd \lnot \bone\midd \bone\land \btwo\midd \bone\vee \btwo
$$
The interpretation $\model \bone_{X}$ of a Boolean formula $\bone$
 with $\FN(\bone)\subseteq X$ is defined in a straightforward way,
 following Definition~\ref{PPLsem}.\footnote{ 
 For further details, see Definition~\ref{booleanFormulas} in Appendix~\ref{appendix4.1}.} 
\end{definition}
\begin{definition}[$a$-decomposition]
Let $\bone$ be a named Boolean formula with free names in
$X\cup \{a\}$. 
An \emph{$a$-decomposition of $\bone$} is a Boolean
formula $\btwo= \bigvee_{i=0}^{k-1}\bthree_{i}\land \mathscr e_{i}$
such that:
\begin{itemize}
\itemsep0em
\item $\model{\btwo}_{X\cup\{a\}}=\model{\bone}_{X\cup\{a\}}$
\item $\FN(\bthree_{i})\subseteq  \{a\}$ and $\FN(\bfour_{i})\subseteq X$
\item if $i\neq j$, then $\model{\bfour_{i}}_{X}\cap\model{\bfour_{j}}_{X}=\emptyset$.
\end{itemize}
\end{definition}
\noindent
The following lemma can be proved by induction.
\begin{lemma}
Any named Boolean formula $\bone$ with $\FN(\bone)\subseteq
X\cup \{a\}$ admits an $a$-decomposition in $X$.
\end{lemma}

It is worth observing that while an $a$-decomposition of $\bone$
always exists, it needs not be feasible to find it, since this formula
can be of exponential length with respect to $\bone$.
Yet, $a$-decompositions can be used to show that the
interpretation of a quantified formula is a finite union of 
measurable sets:
\begin{restatable}[Fundamental Lemma]{lemma}{primeFundamental}\label{lemma:fundamental}
Let $\bone$ be a named Boolean formula with $\FN(\bone)\subseteq
X\cup \{a\}$ and let
$\btwo= \bigvee_{i=0}^{k-1}\bthree_{i}\land \bfour_{i}$ be an
$a$-decomposition of $b$. Then, for all $q\in [0,1]$,
\begin{align*}
\{f \in (2^{\omega})^{X}\mid \mu(\PROJ_{f}(\model \bone_{X\cup\{a\}}))\geq q\} &=
\bigcup \{ \model{\bfour_{i}}_{X}\mid \mu(\model{\bthree_{i}}_{\{a\}})\geq q\} \\
\{f \in (2^{\omega})^{X}\mid \mu(\PROJ_{f}(\model \bone_{X\cup\{a\}}))< q\} &=
\bigcup \{ \model{\bfour_{i}}_{X}\mid \mu(\model{\bthree_{i}}_{\{a\}})< q\}. 
\end{align*}
\end{restatable}
\begin{proof}
We only prove the first equality, the second one being proved in a
similar way.
First note that if $r=0$, then both sets are equal to 
$(2^{\omega})^{X}$, so we can suppose $r>0$.
\begin{description}
\item[($\subseteq$)]
Suppose $\mu(\Pi_{f}(\llbracket \bone\rrbracket_{X\cup \{a\}}))\geq r$. Then $\Pi_{f}(\llbracket \bone\rrbracket_{X\cup \{a\}})$ is non-empty and from $\bone\equiv\bigvee_{i}^{k}\bthree_{i}\land \bfour_{i}$
we deduce that there exists $i\leq k$ such that 
$f\in \llbracket \bfour_{i}\rrbracket_{X}$  and 
for all $g\in \llbracket \bthree_{i}\rrbracket_{\{a\}}$, $f+g\in 
\llbracket \bthree_{i}\land \bfour_{i}\rrbracket_{X\cup \{a\}}$.
This implies then that $\llbracket \bthree_{i}\rrbracket_{\{a\}}\subseteq \Pi_{f}(\llbracket \bone\rrbracket_{X\cup \{a\}})$.
Moreover, since the sets $\model{\bfour_{i}}_{X}$ are pairwise disjoint, for all $j\neq i$, $f\notin \llbracket \bfour_{j}\rrbracket_{X}$, which implies that 
$\Pi_{f}(\llbracket \bone\rrbracket_{X\cup \{a\}})\subseteq \llbracket \bthree_{i}\rrbracket_{\{a\}}$.
Hence
$\Pi_{f}(\llbracket \bone\rrbracket_{X\cup \{a\}})=  \llbracket \bthree_{i}\rrbracket_{\{a\}}$, which implies
$\mu( \llbracket \bthree_{i}\rrbracket_{\{a\}})\geq r$.

\item[($\supseteq$)]
If $f\in \llbracket \bfour_{i}\rrbracket_{X}$, where 
$\mu(\llbracket \bthree_{i}\rrbracket_{\{a\}})\geq r$, then since $\bthree_{i}\land \bfour_{i}\vDash^{X\cup\{a\}} \bone$,
$\mu\big(\Pi_{f}(\llbracket \bone\rrbracket_{X\cup \{a\}})\big)$ $\geq$ 
$\mu\big(\Pi_{f}(\llbracket \bthree_{i}\land \bfour_{i}\rrbracket_{X\cup \{a\}})\big)
= \mu(\llbracket \bthree_{i}\rrbracket_{\{a\}})
\geq r$.
\end{description}
\end{proof}
\noindent
The Lemma~\ref{lemma:fundamental} above can be used 
to associate with any formula $\fone$ 
(with $\FN(\fone)\subseteq X$) a Boolean formula 
$\bone_{\fone}$, such that 
$\model\fone_{X}=\model{\bone_{\fone}}_{X}$. 
The crucial steps of
$\model{\BOX^{q}_{a}\fone}_{X}$ and 
$\model{\DIA^{q}_{a}\fone}_{X}$ are
handled using the fact that these sets are finite unions of measurable
sets of the form $\model{\bfour_{i}}_{X} $, where
$\bigvee_{i=0}^{k-1}\bthree_{i}\land \bfour_{i}$ is an
$a$-decomposition of $\bone_{A}$.
This shows in particular that
all sets $\model{A}_{X}$ are measurable.

\subsection{The Proof Theory of $\PPL$}

Using the Fundamental Lemma, it is not difficult to devise a
complete proof system for $\PPL$, 
in analogy with the one defined for $\PPLc$. 
A \emph{named sequent} is an expression of the form
$\vdash^{X}\lone$, where $\lone$ is a labelled formula,
$\bone\Pto \fone, \bone\Pfrom \fone$, such that  
$\FN(\bone) \cup \FN(\fone)\subseteq X$. 
Most of the rules of the proof system for $\PPL$ 
are straightforward extensions of those of $\PPLc$. 
We illustrate in
Fig.~\ref{fig:cpl} the initial sequents and the counting rules.
These are the only rules which are substantially different 
with respect to the corresponding $\PPLc$ ones.\footnote{For the complete proof system for $\PPL$, see Definition~\ref{CPLsystem} 
in Appendix~\ref{appendix4.2}.}  
\begin{figure}[ht!]
\framebox{
\parbox[t][6.5cm]{15cm}{
\begin{minipage}{\textwidth}
\small
\begin{center}
Initial Sequents
\end{center}
\normalsize
\begin{minipage}{\linewidth}
\adjustbox{scale=0.8,center}{
$
\AxiomC{$\vDash a \in X$}
\AxiomC{$\bone \vDash^{X} \bvar_{i}^{a}$}
\RightLabel{$\AxO$}
\BinaryInfC{$\vdash^{X} \bone \Era \atm{i}{a}$}
\DP
\qquad \qquad \qquad
\AxiomC{$\vDash a \in X$}
\AxiomC{$\bvar_{i}^{a} \vDash^{X} \bone$}
\RightLabel{$\AxT$}
\BinaryInfC{$\vdash^{X} \bone  \Ela \atm{i}{a}$}
\DP
$}
\end{minipage}

\small
\begin{center}
Counting Rules
\end{center}
\normalsize

\begin{minipage}{\linewidth}
\adjustbox{scale=0.8,center}{
$\AxiomC{$\mu(\model{\bone}_{X}) = 0$}
\RightLabel{$\Rmur$}
\UnaryInfC{$\vdash^{X} \bone  \Pto \fone$}
\DP
\qquad \qquad \qquad \qquad
\AxiomC{$\mu(\model{\bone}_{X}) = 1$}
\RightLabel{$\Rmul$}
\UnaryInfC{$\vdash^{X} \bone  \Pfrom \fone$}
\DP
$}
\end{minipage}

\bigskip
\begin{minipage}{\linewidth}
\adjustbox{scale=0.8,center}{
$
\AxiomC{$\vdash^{X\cup\{a\}} \btwo \Pto \fone$}
\AxiomC{$\bone \vDash^{X} \bigvee_{i}\{\bfour_{i} \ | \ \mu(\model{\bthree_{i}}_{\{a\}}) \ge q\}$}
\RightLabel{$\Rbr$}
\BinaryInfC{$\vdash^{X} \bone \Pto \BOX^{q}_{a}\fone$}
\DP
\quad
\AxiomC{$\vdash^{X\cup\{a\}} \btwo \Ela \fone$}
\AxiomC{$\bigvee_{i}\{\bfour_{i} \ | \ \mu(\model{\bthree_{i}}_{\{a\}}) \ge q\} \vDash^{X} \bone$}
\RightLabel{$\Rbl$}
\BinaryInfC{$\vdash^{X} \bone \Ela \BOX^{q}_{a}\fone$}
\DP
$}
\end{minipage}

\bigskip

\begin{minipage}{\linewidth}
\adjustbox{scale=0.8,center}{
$
\AxiomC{$\vdash^{X\cup\{a\}} \btwo \Ela \fone$}
\AxiomC{$\bone \vDash^{X} \bigvee_{i}\{\bfour_{i} \ | \ \mu(\model{\bthree_{i}}_{\{a\}}) < q\}$}
\RightLabel{$\Rdr$}
\BinaryInfC{$\vdash^{X} \bone \Era \DIA^{q}_{a}\fone$}
\DP
\quad
\AxiomC{$\vdash^{X\cup\{a\}} \btwo \Era \fone$}
\AxiomC{$\bigvee_{i}\{\bfour_{i} \ | \ \mu(\model{\bthree_{i}}_{\{a\}}) < q\} \vDash^{X} \bone$}
\RightLabel{$\Rdl$}
\BinaryInfC{$\vdash^{X} \bone \Ela \DIA^{q}_{a}\fone$}
\DP
$}
\end{minipage}
\small
\begin{center}
\footnotesize{(where $\bigvee_{i}\bfour_{i}\wedge \bthree_{i}$ is an $a$-decomposition of $\btwo$)}
\end{center}
\end{minipage}
}}
\caption{\small{$\PPL$ Initial Sequents and Counting Rules.}}
\label{fig:cpl}
\end{figure}

Similarly to $\PPLc$, given a named Boolean formula $\bone$, 
a formula of $\PPL$ $\fone$, and a set $X$, 
such that $\FN(\bone)\cup\FN(\fone) \subseteq X$, 
the labelled formula $\bone\Pto \fone$
(resp.~$\bone\Pfrom \fone$) is \emph{$X$-valid} when
$\model \bone_{X}\subseteq \model \fone_{X}$ (resp.~
$\model \fone_{X} \subseteq \model \bone_{X}$).
A named sequent $\vdash^{X}\lone$ is \emph{valid} 
(denoted $\vDash^{X}\lone$) when $\lone$ is $X$-valid. 
We write $\fone\equiv_X\ftwo$
when $\model{\fone}_{X}=\model{\ftwo}_{X}$.
The soundness of the proof system can be again
proved by standard induction. 
Completeness is established with an argument very similar to that for $\PPLc$.\footnote{For further details, see Appendix~\ref{appendix4.2}. Specifically,
soundness is proved as Theorem~\ref{theorem:soundness}
and completeness as Theorem~\ref{CPLcompleteness}.}
\begin{theorem}[Soundness and Completeness]
$\vdash_{\PPL}^{X}\lone$ holds if and only if $\vDash^{X}\lone$ holds.
\end{theorem}
\noindent
From completeness one can deduce the admissibility of a cut-rule defined as in Remark~\ref{rem:cutrule}.

%% file: 5Complexity.tex

In his seminal works~\cite{Wagner84,Wagner,Wagner86}, 
Wagner observed that one can define canonical complete 
problems for the each level of counting hierarchy, 
in the spirit of similar results for the
polynomial hierarchy~\cite{MeyerStockmeyer72, MeyerStockmeyer73, Stockmeyer77, Wrathall, BB}. 
Nevertheless, there are some crucial differences. 
First of all, satisfiability should be checked
against \emph{bunches} of variables, 
rather than on \emph{one} variable. 
Moreover, the verifier is allowed to require the
number of models to be higher than a certain threshold, 
rather than merely being minimal or maximal. 
The resulting operators on classes of languages are 
called \emph{counting quantifiers}.

In some sense, $\PPL$ can be seen as a way to internalize 
Wagner's construction inside a proper logical system. 
Counting quantifiers, however, can be used deep inside formulas, 
rather than just at top-level. 
Does all this make us lose the connection to the counting hierarchy? 
The answer to the question is negative and is the topic of
this section. 
More specifically, formulas can be efficiently put in prenex form, 
as we prove in Section~\ref{sect:prenexform} below. 
Moreover, there is a way of getting rid of the $\DIA$ quantifier, 
which has no counterpart in Wagner's problems, 
as shown in Section~\ref{sect:simpleprenexform}.

\subsection{Prenex Normal Forms}\label{sect:prenexform}

The prenex normal form for formulas of $\PPL$ is defined in 
a standard way:
\begin{definition}[PNF]
  A formula of $\PPL$ is an 
  \emph{$n$-ary prenex normal form} 
  (or simply a \emph{prenex normal form}, PNF for short) 
  if it can be written as $\op{1}\ldots\op{n}\fone$ where, 
  for every $i \in \{1,\dots,n\}$, 
  the operator $\op{i}$ is either in the form $\BOX^{q}_{a}$ or
  $\DIA^{q}_{a}$ (for arbitrary $a,$ and $q$), 
  and $\fone$ does not contain any counting operator. 
  In this case, the formula $\fone$ is said
  to be the \emph{matrix} of the PNF.
\end{definition}
\noindent
To convert a formula of $\PPL$ into an equivalent
formula in \emph{in prenex normal form}, 
some preliminary lemmas are needed. 
First, it is worth noticing that for every formula of $\PPL$ $\fone$,
name $a$, and finite set $X$, 
with $\Var{\fone} \subseteq X$ and $a \not \in X$, 
if $q$ = 0, then 
$\model{\BOX^{q}_{a}\fone}_{X}$ = ($\twoOm$)$^{X}$ 
and $\model{\DIA^{q}_{a}\fone}_{X}$ = $\emptyset^{X}$.\footnote{For the proof, see Lemma~\ref{Lq0} in Appendix~\ref{appendix5.1}.}

The following lemma states that counting quantifiers occurring
inside any conjunction or disjunction can somehow be extruded from it.\footnote{For further details, see the corresponding Lemma~\ref{lemma:commutations}, in Appendix~\ref{appendix5.1}.} 
\begin{restatable}{lemma}{primeCommutations}\label{lemma:commutations}
  Let $a \not\in \FN(\fone)$ and $q>0$. 
  Then, for every $X$ such that 
  $\Var{\fone} \cup \Var{\ftwo} \subseteq X$, and $a \not \in X$,
  the followings hold:
\begin{align*}
  \fone \wedge \BOX^{q}_{a}\ftwo &\equiv_{X}\BOX^{q}_{a}(\fone \wedge \ftwo) &&&
\fone \vee \BOX^{q}_{a}\ftwo &\equiv_{X}\BOX^{q}_{a}(\fone \vee \ftwo) \\
\fone \wedge \DIA^{q}_{a}\ftwo &\equiv_{X} \DIA^{q}_{a}(\neg \fone \vee \ftwo) &&&
\fone \vee \DIA^{q}_{a}\ftwo &\equiv_{X} \DIA_{a}^{q}(\neg\fone \wedge \ftwo).
\end{align*}
\end{restatable}
\noindent
Remarkably, a similar lemma does \emph{not} hold for $\PPLc$.

What about negation? Equation~\eqref{equ:interdefinability} 
can be generalized to $\PPL$, in order to get rid of negations, 
which lie between any occurrence of a counting quantifier 
and the formula's root.\footnote{For the proof, see the corresponding Lemma~\ref{lemma:commutations}, in Appendix~\ref{appendix5.1}.}
\begin{restatable}{lemma}{primeDuality}\label{lemma:pseudoduality}
For every $q \in \mathbb{Q}_{[0,1]}$, name $a$, and $X$ 
such that $\Var{\fone} \subseteq X\cup\{a\}$ and $a \not \in X$, then:
$$
  \lnot \DIA^{q}_a\fone \equiv_{X} \BOX^{q}_a\fone \qquad \qquad \lnot\BOX^{q}_a\fone \equiv_{X} \DIA^{q}_a\fone.
  $$
\end{restatable}
\noindent
Please notice that $\BOX$ and $\DIA$, contrarily to standard
$\forall$ and $\exists$, are \emph{not} dual operators. 

As the equivalences from Lemma~\ref{lemma:commutations} 
and Lemma~\ref{lemma:pseudoduality}
can be oriented, see Appendix~\ref{appendix5.1}, 
we conclude the following:
\begin{restatable}{proposition}{primeProposition}\label{PNF}
  For every formula of $\PPL$, $\fone$, there is a PNF $\ftwo$
  such that, for every $X$, with $\Var{\fone}$ $\cup$ $\Var{\ftwo} \subseteq X$:
  $$
  \fone\equiv_{X}\ftwo.
  $$
   Moreover, $\ftwo$ can be computed
  in polynomial time from $\fone$.
\end{restatable}

\subsection{Positive Prenex Normal Form}\label{sect:simpleprenexform}
Prenex normal forms for $\PPL$ are already very close 
to what we need,
but there is one last step to make, 
namely getting rid of the $\DIA$ operator, which does not 
have any counterpart in Wagner's construction. 
In other words, we need prenex normal forms \emph{of a
special kind}.
\begin{definition}[PPNF]
  A formula of $\PPL$ is said to be in \emph{positive prenex normal form} (PPNF, for short) if it is in PNF and no $\DIA$-quantifier 
  occurs in it.
\end{definition}
How could we get rid of $\DIA$? The idea is to turn any instance of
$\DIA$ into one of $\BOX$ by applying the second equality from Lemma~\ref{lemma:pseudoduality},
and then to use the fact that $\BOX$ enjoys a weak form of 
self duality, as shown by the following lemma:
\begin{restatable}[Epsilon Lemma]{lemma}{primeEpsilon}\label{lemma:epsilon}
  For every formula of \PPL, call it $\fone$, 
  and for every $q \in \mathbb{Q}_{[0,1]}$, 
  there is a $p \in \mathbb{Q}_{[0,1]}$ such that, 
  for every $X$, with $\Var{\fone} \subseteq X$ and 
  $a \not \in X$, the following holds:
  $$
  \neg \BOX^{q}_{a}\fone \equiv_{X} \BOX^{p}_{a}\neg \fone.
  $$
  Moreover, $p$ can be computed from $q$ in polynomial
  time.
\end{restatable}

\begin{proof}[Proof Sketch of Lemma~\ref{lemma:epsilon}]
To establish Lemma~\ref{lemma:epsilon} we need a few preliminary lemmas which we will here assumed to be already established.\footnote{The detailed proof of Lemma~\ref{lemma:epsilon} and of all the auxiliary lemmas can be found in Appendix~\ref{appendix5.2}.}
Let $\bone_{A}$ be a Boolean formula satisfying $\llbracket \bone\rrbracket_{X\cup\{a\}}=\llbracket \bone_{A}\rrbracket_{X\cup\{a\}}$. 
Let $\bone_{A}$ be $a$-decomposable as $\bigvee_{i}^{n}\bthree_{i}\land \bfour_{i}$ and let $k$ be maximum such that $\bvar_{k}^{a}$ occurs in $\bone_{A}$.
Let $[0,1]_{k}$ be the set of those rational numbers $r\in[0,1]$ which can be written as a finite sum of the form $\sum_{i=0}^{k}b_{i}\cdot 2^{–i}$.
For all $i=0,\dots,n$, $\mu(\llbracket \bthree_{i}\rrbracket_{\{a\}})\in[0,1]_{k}$  where $b_{i}\in \{0,1\}$ (see Lemma~\ref{lemma:k} in Appendix~\ref{appendix5.2}).
Hence, for all $f:X\to 2^{\omega}$, also $\mu\big(\Pi_{f}(\llbracket A\rrbracket_{X\cup \{a\}})\big)\in[0,1]_{k}$, since $\Pi_{f}\big(\llbracket A\rrbracket_{X\cup \{a\}})\big)$ coincides with the unique $\llbracket \bthree_{i}\rrbracket_{\{a\}}$ such that $f\in \llbracket \bfour_{i}\rrbracket_{X}$, by~Lemma~\ref{lemma:fundamental}.
Let now $\epsilon$ be $2^{–(k+1)}$ if $q\in[0,1]_{k}$ and $q\neq 1$, let $\epsilon$ be $–2^{–(k+1)}$ if $q=1$ and let $\epsilon=0$ if $q\notin [0,1]_{k}$. In all cases $q+\epsilon\notin [0,1]_{k} $ so, by means of some auxiliary results (specifically by Lemma~\ref{lemma:leq} and Lemma~\ref{lemma:pif}),
it is possible to conclude: 
$$
\model{\neg\BOX^{q}_{a}\fone}_{X}=\model{\BOX^{1–(q+\epsilon)}_{a}\neg\fone}_{X}.
$$
\end{proof}
\noindent
Actually, the value of $p$ is very close to $1-q$, the difference between
the two being easily computable from the formula $\fone$. 
Hence, any negation occurring in the counting prefix of a 
PNF can be pushed back into
the formula. 
Summing up:
\begin{restatable}{proposition}{primePPNF}\label{PNF}
  For every formula of \PPL, call it $\fone$, there is a PPNF $\ftwo$
  such that for every $X$, with $\Var{\fone} \cup \Var{\ftwo} \subseteq X$,
  $$
  \fone\equiv_{X}\ftwo.
  $$
  Moreover, $\ftwo$ can be computed
  from $\fone$ in polynomial time.
\end{restatable}

\subsection{$\PPL$ and the Counting Hierarchy}

In~\cite{Wagner}, Wagner not only introduced counting quantifiers 
(as operators on classes of languages) and hierarchy, 
but he also presented complete problems for each level of $\CH$. 
Below, we state a slightly weaker version of Wagner's Theorem, 
which however perfectly fits our needs. 
For the sake of keeping our exposition self-contained,
we also introduce preliminary definitions, 
despite their being standard.

Assume $\mathcal{L}$ is a subset of $\set^n$, where $\set$ is a set.
Now, suppose that $1\leq m<n$, and $b\in\Nat$. 
We thus define $\BOX_m^b\mathcal{L}$ as the following subset
of $\set^{n-m}$:
$$
\left\{(a_n,\ldots,a_{m+1})\mid \#(\{(a_m,\ldots,a_1)\mid(a_n,\ldots,a_1)\in\mathcal{L}\})\geq b\right\}.
$$
Now, take $\PL$ as a set comprising the language of propositional logic
formulas, including the propositional constants $\True$ and $\False$.
For any natural number $n\in\Nat$, let $\mathcal{TF}^n$ be the
subset of $\PL^{n+1}$ containing all tuples
in the form $(\fone,t_1,\ldots,t_n)$, where $\fone$ is a propositional
formula in CNF with at most $n$ free variables, and $t_1,\ldots,t_n\in\{\True,\False\}$, which
render $\fone$ true. Finally, for every $k\in\Nat$, we indicate as $\mathit{W}^k$
the language consisting of all (binary encodings of) tuples
in the form $(\fone,m_1,\ldots,m_k,b_1,\ldots,b_k)$ and such that
$\fone\in\BOX_{m_1}^{b_1}\cdots\BOX_{m_k}^{b_k}\mathcal{TF}^{\sum m_i}$.
 \begin{theorem}[Wagner, Th. 7~\cite{Wagner}]
  For every $k$, the language $W^k$ is complete for $\CH_k$.
 \end{theorem}
\noindent
Please notice that the elements of $W^k$ can be seen as
alternative representations for $\PPL$'s PPNFs, 
once any $m_i$ is replaced by $\min\{1,\frac{m_i}{2^{b_i}}\}$. 
As a consequence:
 \begin{corollary}
 The closed and valid $k$-ary PPNFs, whose matrix is in CNF, 
 define a complete set for $\CH_k$.
 \end{corollary}

%% file: section6.tex

\subsection{Counting Proofs as Programs}
As is well known, the so-called Curry-Howard correspondence relates proofs with typable programs.
When this correspondence holds, any provable sequent 
$\vdash A$
can be replaced with a type judgment
$\vdash t:A$, 
where $t$ is a program encoding the computational content of some proof of $\vdash A$.
Now, what is the computational content of the proofs of the system $\PPL$? 
A natural way to interpret such proofs is as programs depending on some \emph{probabilistic event}: if $\vdash \bone\Pto A$ is provable in, say, $\PPLc$, then a program $\vdash t: \bone\Pto A$ should be one depending on some event $f\in 2^{\omega}$ and such that, whenever $f$ satisfies $\bone$, then $t$ yields a proof of $A$.

 In the literature on higher-order probabilistic type systems, one usually considers the extension of the standard $\lambda$-calculus with a probabilistic choice operator $t\oplus u$, which reduces to either $t$ or $u$ with probability $\frac{1}{2}$. 
 However, this extension comes with the undesired consequence the confluence property is lost.
 This is due to the basic observation that probabilistic choice does not commute with duplication: should two copies of $t\oplus u$ follow the \emph{same} probabilistic choice or behave as \emph{independent} choices? 
A solution to this problem, proposed in \cite{DLGH}, is to decompose the operator $\oplus$ into a \emph{generator} $\nu a.t$ and a  \emph{choice operator} $\oplus_{a}$, so that
$$t\oplus u= \nu a. t\oplus_{a}u.$$
In this way one obtains a confluent system, in which two copies of $t\oplus_{a}u$ behave following the same choice, while two copies of
$\nu a.t\oplus_{a}u$ behave following independent choices.

This decomposition of the probabilistic choice operator yields a surprising bridge with counting propositional logic.
In fact, suppose $\bone$ is true in at least half of its models (i.e. $\mu(\model \bone)\geq \frac{1}{2}$). As we said, a proof of $\vdash \bone \Pto A$ should correspond to a program depending on the choice of a valuation $f$ for $\bone$, that we can think of as $t\oplus_{a}u$.
Then, the proof of $\vdash \btwo \Pto \BOX^{1/2}A$, that we obtain using the rules for the counting quantifier corresponds to a program which \emph{no more depends} on $f$, that we can think of as $\nu a.t\oplus_{a}u$. In other words, it makes perfect sense to consider the following typing rule:
$$
\AXC{$\vdash t: \bone \Pto A$}
\AXC{$\mu(\model \bone)\geq q$}
\BIC{$\vdash \nu a.t: \btwo \Pto \BOX^{q}A$}
\DP
$$
This idea leads then to look for a suitable type system for a probabilistic $\lambda$-calculus in the style of \cite{DLGH}, which somehow reflects the structure of counting propositional logic. 
This is what we are going to present below.

\subsection{The Type System $\TTT$}

We introduce a probabilistic type system $\TTT$ with counting quantifiers for a $\lambda$-calculus enriched with generators and choice operators as in \cite{DLGH}.
The set $\Lnu$ of $\lambda$-terms with \emph{generators} and \emph{choice}  operators is defined by the grammar:
$$
t,u:=   x\midd \lambda x.t\midd tu \midd t\oplus^{i}_{a}u \midd \nu a.t.
$$ 
Note that the choice operator $t\oplus_{a}^{i}u$ is indexed by both a name and an index $i\in \mathbb N$, since, as for $\PPL$, we take the Borel algebra $\mathcal B(2^{\omega})$ as the underlying event space. 
We let $\FN(t)$ indicate the set of free names of $t$.
A term $t$ is \emph{name-closed} when $\FN(t)=\emptyset$. 
If $\FN(t)\subseteq X$, then for all $f\in (2^{\omega})^{X}$, we let the \emph{application of $f$ to $t$} be the name-closed term $\pi^{f}(t)$ defined by induction as follows:

\adjustbox{scale=1}{
\begin{minipage}{\linewidth}
\begin{align*}
  \pi^{f}(x) & = x \\
\pi^{f}(\lambda x.t)&= \lambda x.\pi^{f}(t)\\
 \pi^{f}(tu)& =\pi^{f}(t)\pi^{f}(u)  \\
\pi^{f}(t_{0}\oplus^{i}_{a} t_{1}) & = \begin{cases} 
  \pi^{f}(t_{f(i)}) & \text{ if } a\in X \\
  \pi^{f}(t_{0})\oplus^{i}_{a} \pi^{f}(t_{1}) & \text{ if }a \notin X
  \end{cases} \\
  \pi^{f}(\nu a.t) & = \nu a.\pi^{f}(t).
\end{align*}
\end{minipage}
}
\medskip

$\lambda$-terms are equipped with a reduction $\rred$ generated by usual $\beta$-reduction, $\to_{\beta}$, and \emph{permutative reduction}, $\to_{\mathsf{p}}$, illustrated in Fig.~\ref{fig:permutations}, where for a term $t$, a name $a$, and an index $i$, $(a,i)<(b,j)$ holds when either $a\neq b$ and $b$ is bound within the scope of $a$, or $a=b$ and $i<j$.
Let $\mathcal{HN}\subseteq \mathcal{V}$ be the set of pseudo-values which are in \emph{head normal form}, that is, which are of the form $\lambda x_{1}\dots\lambda x_{n}.yu_{1}\dots u_{n}$.
\begin{figure}[t]
\fbox{
\begin{minipage}{0.97\textwidth}
\begin{align*}
 t\oplus_{a}^{i}t & \to_{\mathsf{p}} t  \\
 (t\oplus_{a}^{i}u)\oplus_{a}^{i} v & \to_{\mathsf{p}} t\oplus_{a}^{i}v \\
 t\oplus_{a}^{i}(u\oplus_{a}^{i}v) & \to_{\mathsf{p}} t\oplus_{a}^{i}v \\
 \lambda x.(t\oplus_{a}^{i}u) & \to_{\mathsf{p}} (\lambda x.t)\oplus_{a}^{i} (\lambda x.u) \\
 (t\oplus_{a}^{i}u)v & \to_{\mathsf{p}} (tu)\oplus_{a}^{i}(uv) \\
 t(u\oplus_{a}^{i}v) & \to_{\mathsf{p}} (tu)\oplus_{a}^{i}(tv) \\
 (t\oplus_{a}^{i}u)\oplus_{b}^{j} v & \to_{\mathsf{p}} (t\oplus_{b}^{j}v)\oplus_{a}^{i}(u\oplus_{b}^{j}v) & ((a,i)< (b,j)) \\
 t\oplus_{b}^{j}(u\oplus_{a}^{i} v) & \to_{\mathsf{p}} (t\oplus_{b}^{j}u)
\oplus_{a}^{i}(t\oplus_{b}^{j} v) & ((a,i)<(b,j)) \\
\nu b. (t\oplus_{a}^{i} u) & \to_{\mathsf{p}} (\nu b.t)\oplus_{a}^{i}(\nu b.u) &  (a\neq b) \\
\nu a.t & \to_{\mathsf{p}} t &  (a\notin \FN(t)) \\
\lambda x.\nu a.t & \to_{\mathsf{p}} \nu a. \lambda x.t \\
(\nu a.t)u  & \to_{\mathsf{p}}  \nu a.(tu)
 \end{align*}
\end{minipage}
}
\caption{Permutative reductions of $\Lnu$-terms.}
\label{fig:permutations}
\end{figure}
The following facts can be proved as in \cite{DLGH}.
\begin{proposition}\label{prop:confluence}
The reduction relations $\to_{\beta}, \to_{\mathsf p},$ and $\rred$ are all confluent.
Moreover, $\to_{\mathsf p}$ is strongly normalizing.
\end{proposition}

\begin{remark}
The $\lambda$-calculus in \cite{DLGH} slightly differs from the present one, since choice operators do not depend on indexes $i \in \mathbb N$. However, one can define a bijective translation between the two languages, so that the permutative reductions in \cite{DLGH} translate into those in Fig.~\ref{fig:permutations}.
%
%
%
%

\end{remark}

A consequence of Proposition \ref{prop:confluence} is that any term admits a (unique) \emph{permutative normal form} ($\PNF$, for short). We let $\mathcal T$ indicate the set of $\PNF$.  As shown in detail in Appendix \ref{appendix6}, a name-closed term $t\in \mathcal T$ can be of two types: 
\begin{itemize}
\item a \emph{pseudo-value}, that is a variable, a $\lambda$ or an application (we let $\mathcal V\subseteq \mathcal T$ indicate the set of pseudo-values)
\item a term $\nu a.t$, where $t$ is a tree whose leaves form a finite subset of $\mathcal T$ (the \emph{support} of $t$, $\supp(t)$) and formed by concatenating sums of the form $t\oplus_{a}^{i}u$ (see Def.~\ref{def:satree}). 
\end{itemize}
Using this decomposition, any $\PNF$ $t\in \mathcal T$ can be associated with a \emph{distribution} of pseudo-values $\mathcal D_{t}: \mathcal V\to[0,1]$ by letting $\mathcal D_{t}(v)=\delta_{t}$, when $t\in \mathcal V$, and otherwise
\begin{align*}
\mathcal D_{t}(v) & = 
\sum_{t'\in \supp(t)} \mathcal D_{t'}(v)\cdot \mu \big (\{ f\in 2^{\omega}\mid \pi^{f}(t)= t'\}\big ).
\end{align*} 
The real number $\mathcal D_{t}(v)\in[0,1]$ can be interpreted as the probability that the term $v$ is found by applying to $t$ a chain of random choices.

We now introduce \emph{simple types with counting} $\sigma,\tau$ by the following grammar:
\begin{align*}
\sigma,\tau:= o \mid  \FF s \To \sigma  \qquad \qquad  \FF s, \FF t  := \BOX^{q}\sigma.
\end{align*}
A \emph{type context} $\Gamma$ is a finite set of declarations of the form $x:\FF s$, with the $x$ pairwise distinct. 
A \emph{labelled type} is an expression of the form $\bone \Pto \sigma$, where $\bone$ is a Boolean formula.
A \emph{type judgement} is of the form
$
\Gamma \vdash^{X,q} t: \bone \Pto \sigma
$
where the free names in both $t$ and $\bone$ are included in $X$, and $q$ is a rational number in $[0,1]$.
The rules of the system $\TTT$ are illustrated in Fig.~\ref{fig:typingrules}, where for a Boolean formula $\bone $ with $\FN(\bone)\subseteq X\cup\{a\}$ (with $a\notin X$), a \emph{weak $a$-decomposition} of $\bone$ is any Boolean formula $\bigvee_{i}^{k}\btwo_{i}\land \bthree_{i}$ equivalent to $\bone$, such that $\FN(\btwo_{i})\subseteq\{a\}$, $ \FN(\bthree_{i})\subseteq X$, and the $\btwo_{i}$ are satisfiable.

\begin{figure}[htp!]
\framebox{
\parbox[t][11cm]{15cm}{
\begin{minipage}{0.95\linewidth}
\begin{center}

\begin{center}
Initial Rules
\end{center}

\begin{minipage}{\linewidth}
\begin{minipage}[t]{0.4\linewidth}
\begin{prooftree}
\AxiomC{$\mathrm{FV}(t)\subseteq \Gamma,\FN(t)\subseteq X$}
\RightLabel{$\TB$}
\UnaryInfC{$\Gamma\vdash^{X,r}t: \bot  \Pto \sigma$}
\end{prooftree}
\end{minipage}
\hfill
\begin{minipage}[t]{0.45\linewidth}
\begin{prooftree}
\AxiomC{$\FN(\bone)\subseteq X$}
\RightLabel{$\TID$}
\UnaryInfC{$\Gamma, x:\BOX^{q}\sigma\vdash^{X,q\cdot s} x: \bone  \Pto \sigma$}
\end{prooftree}
\end{minipage}
\end{minipage}

\vskip4mm

\begin{center}
Union Rule
\end{center}
$
\AXC{$\Gamma \vdash^{X,r}t: \btwo \Pto \sigma$}
\AXC{$\Gamma \vdash^{X,r}t: \bthree \Pto \sigma$}
\AXC{$\bone\vDash^{X} \btwo \vee \bthree$}
\RL{$\TU$}
\TIC{$\Gamma \vdash^{X,r} t:\bone\Pto \sigma$}
\DP
$

\vskip4mm
\begin{center}
Abstraction/Application Rules
\end{center}
$
\AXC{$\Gamma, x:\FF s \vdash^{X,q} t: \bone\Pto\sigma$}
\RL{$\TLA$}
\UIC{$\Gamma \vdash^{X,q} \lambda x.t:\bone\Pto \FF s\To \sigma$}
\DP
$
\vskip4mm
$
\AXC{$\Gamma \vdash^{X,q} t: \btwo\Pto \BOX^{r}\sigma \To \tau$}
\AXC{$\Gamma \vdash^{X,r} u: \bthree \Pto \sigma$}
\AXC{$\bone\vDash \btwo\land \bthree$}
\RL{$\TA$}
\TIC{$\Gamma\vdash^{X,q} tu: \bone \Pto \tau$}
\DP
$
\vskip4mm
\begin{center}
Choice Rules
\end{center}

\begin{minipage}{\linewidth}
\begin{minipage}[t]{0.5\linewidth}
$
\AXC{$\Gamma \vdash^{X\cup\{a\},q} t:\btwo\Pto \sigma$}
\AXC{$\bone\vDash\bvar_{i}^{a}\land \btwo$}
\RL{$\TL$}
\BIC{$\Gamma \vdash^{X\cup\{a\},q} t\oplus^{i}_{a}u:  \bone\Pto \sigma$}
\DP
$
\end{minipage}
\hfill
\begin{minipage}[t]{0.45\linewidth}
$
\AXC{$\Gamma \vdash^{X\cup\{a\},q} u:\btwo\Pto \sigma$}
\AXC{$\bone\vDash\lnot \bvar_{i}^{a}\land \btwo$}
\RL{$\TR$}
\BIC{$\Gamma \vdash^{X\cup\{a\},q} t\oplus^{i}_{a}u:\bone\Pto \sigma$}
\DP
$
\end{minipage}
\end{minipage}

\begin{center}
Generator Rule
\end{center}

%
%

$
\AXC{$\Gamma \vdash^{X\cup\{a\},q} t:\btwo\Pto \sigma$}
\AXC{$\left( \mu(\model{\btwo_{i}}_{\{a\}})\geq r\right)_{i\leq k}$}
\AXC{$\bone \vDash \bigvee_{i}^{k}\bthree_{i}$}
\RL{$\TN$}
\TIC{$\Gamma \vdash^{X,q\cdot r} \nu a.t: \bone\Pto  \sigma$}
\DP
$
\medskip

\footnotesize{
(where $\bigvee_{i}\btwo_{i}\land \bthree_{i}$ is a weak $a$-decomposition of $\btwo$)}

\end{center}
\end{minipage}
}
}
\caption{Typing rules of $\TTT$.}
\label{fig:typingrules}
\end{figure}


The following standard property is proved in detail in Appendix \ref{appendix6}.
\begin{restatable}{proposition}{subjectreduction}\label{prop:subject}
If $\Gamma\vdash^{X,r}t:\bone \Pto \sigma$ and $t\rred u$, then
$\Gamma\vdash^{X,r}u:\bone \Pto \sigma$.
\end{restatable}

\begin{example}
The term $t= \nu a.(I\oplus_{a}^{0}\Omega)(I\oplus_{a}^{0}\Omega)$ is typable ``with probability $\frac{1}{2}$'', where $I=\lambda x.x$, reduces to a normal form with $\PROB \geq \frac{1}{2}$.
In fact, for all types $\sigma$ and $q\in[0,1]$, we can prove $\vdash^{\{a\},q} I\oplus^{0}_{a}\Omega: \bvar_0^a\Pto (\BOX^{q}\sigma)\To \sigma$ (Fig.~\ref{fig:example1}).
Then, we can use this
to deduce $\vdash^{\frac{1}{2}}t: \top\Pto \rho$ 
(Fig.~\ref{fig:example2}, where $\rho=(\BOX^{q}o)\To o$).
Instead, the term $t= \big(\nu a.(I\oplus_{a}^{0}\Omega)\big)\big(\nu a.(I\oplus_{a}^{0}\Omega)\big)$ is typable ``with probability $\frac{1}{4}$'' (Fig.~\ref{fig:example3}).

\end{example}

\begin{figure*}
\fbox{
\begin{minipage}{.97\textwidth}
\begin{subfigure}{0.95\textwidth}
\adjustbox{scale=0.8, center}{
$
\pi_{\sigma,a,q} \ = \  
\AXC{$\FN(\bot)\subseteq\{a\}$}
\RL{$\TID$}
\UIC{$x:\BOX^{q}\sigma\vdash^{\{a\},q}x: \top \Pto \sigma$}
\RL{$\TLA$}
\UIC{$\vdash^{\{a\},q} \lambda x.x: \top\Pto (\BOX^{q}\sigma)\To \sigma$}
\RL{$\TL$}
\UIC{$\vdash^{\{a\},q} (\lambda x.x)\oplus^{0}_{a}\Omega: \bvar_0^a\Pto (\BOX^{q}\sigma)\To \sigma$}
\DP
$
}
\caption{}
\label{fig:example1}
\end{subfigure}

\vskip4mm

\begin{subfigure}{0.95\textwidth}
\adjustbox{scale=0.8, center}{
$
\AXC{$\pi_{\rho,a,1}$}
\noLine
\UIC{$
\vdash^{\{a\},1} I\oplus^{0}_{a}\Omega: \bvar_0^a\Pto (\BOX^{1}\rho)\To \rho$}
\AXC{$\pi_{o ,a,1}$}
\noLine
\UIC{$\vdash^{\{a\},1} I\oplus^{0}_{a}\Omega: \bvar_0^a\Pto \rho$}
\RL{$\TA$}
\BIC{$\vdash^{\{a\},1} (I\oplus^{0}_{a}\Omega)(I\oplus^{0}_{a}\Omega):
\bvar_0^a\Pto \rho$}
\AXC{$\mu(\model{\bvar_0^a}_{\{a\}})\geq 1/2$}
\RL{$\TN$}
\BIC{$\vdash^{\emptyset,1/2} \nu a.(I\oplus^{0}_{a}\Omega)(I\oplus^{0}_{a}\Omega):
\top\Pto \rho$}
\DP
$
}
\caption{}
\label{fig:example2}
\end{subfigure}

\vskip4mm

\begin{subfigure}{0.95\textwidth}
\adjustbox{scale=0.7, center}{
$
\AXC{$\pi_{\rho,a,1/2}$}
\noLine
\UIC{$\vdash^{\{a\},1/2} I\oplus^{0}_{a}\Omega: \bvar_0^a\Pto (\BOX^{1/2}\rho)\To \rho$}
\AXC{$\mu(\model{\bvar_0^a}_{\{a\}})\geq 1/2$}
\RL{$\TN$}
\BIC{$\vdash^{\emptyset,1/4}\nu b. (I\oplus^{0}_{b}\Omega): \top\Pto (\BOX^{1/2}\rho)\To \rho$}
\AXC{$\pi_{o ,a,1}$}
\noLine
\UIC{$\vdash^{\{a\},1} I\oplus^{0}_{a}\Omega: \bvar_0^a\Pto \rho$}
\AXC{$\mu(\model{\bvar_0^a}_{\{a\}})\geq 1/2$}
\RL{$\TN$}
\BIC{$\vdash^{\emptyset,1/2}\nu a. (I\oplus^{0}_{a}\Omega):
\top\Pto \rho$}
\RL{$\TA$}
\BIC{$\vdash^{\emptyset,1/4} \nu a.(I\oplus^{0}_{a}\Omega)(\nu a.(I\oplus^{0}_{a}\Omega)):
\top\Pto \rho$}
\DP
$
}
\caption{}
\label{fig:example3}
\end{subfigure}
\end{minipage}
}
\caption{\small{Type derivations in $\TTT$ (for readability we omit some external hypotheses, as they can be guessed from the context.)}}
\label{fig:stlcexamples}
\end{figure*}

Observe that in $\TTT$ one can type terms which are \emph{not} normalizing. For instance, we can prove
$\vdash^{\emptyset,1} \Omega: \bot \Pto \sigma$ and 
$\vdash^{\emptyset, 1}\lambda x.x\Omega : \top \Pto \sigma$,
 for any type $\sigma$,  
where
$\Omega=(\lambda x.xx)\lambda x.xx$. 
This shows that our type system provides a rather different kind of information with respect to the one in \cite{DLGH}, which ensures strong normalization. 
In fact, a typing judgement
$\Gamma\vdash^{X,r} t:\bone\Pto \sigma$ should be interpreted as some quantitative statement about the \emph{probability} that $t$ reaches a normal form. But how can we define this probability, precisely?

In usual randomized $\lambda$-calculi, program execution is defined so as to be inherently probabilistic: for example a term $t\oplus u$ can reduce to either $t$ or $u$, with probability $\frac{1}{2}$. In this way, chains of reduction can be described as \emph{stochastic Markovian sequences} \cite{Puterman1994}, leading to formalize the idea of \emph{normalization with probability} $r\in[0,1]$ (see \cite{Bournez2002}).
%
Instead, in the language $\Lnu$ reduction has nothing probabilistic by itself (it is even confluent), so we have to look for another approach. 
Since a $\PNF$, call it $t$, can be interpreted, \emph{statically}, as a  distribution of pseudo-values, one can ask what is the probability that a pseudo-value randomly found from $t$ {is} in normal form. This leads to the following:

\begin{definition}
For any $t\in \mathcal T$ and $r\in[0,1]$, $t$ is \emph{in normal form with $\PROB\geq r$} when
$$
\sum_{v\in \mathcal{HN}} \mathcal D_{t}(v) \geq r.
$$
\end{definition}
\noindent
For example, the term $\nu a. \big(\Omega \oplus_{a}^{0}( \lambda x.x\Omega)\big)\oplus_{a}^{1}(\lambda x.x)$ is in normal form with probability $\frac{3}{4}$: three of the four possible choices for $(a,0)$ and $(a,1)$ lead to a head normal form. Instead, the term $\nu a. \Big(\Omega\oplus_{a}^{0}\big( \nu b. \Omega \oplus_{b}^{0} (\lambda x.x)\big)\Big)$ is in normal form with probability $\frac{1}{4}$: only one possible choice for $(a,0)$ and one possible choice for $(b,0)$ lead to a head normal form.

Our main result is that typing in $\TTT$ captures reduction to such probabilistic normal forms.

%
%
%
\begin{restatable}{theorem}{normalization}\label{thm:normalization}
If $\Gamma \vdash^{X,r}t:\bone\Pto \sigma$ holds, then for all $f\in \model\bone_{X}$, 
$\pi^{f}(t)$ reduces to a term which is in normal form with $\PROB\geq r$.
%
\end{restatable}
\noindent
Theorem~\ref{thm:normalization} is proved by adapting the standard method of \emph{reducibility predicates}. Just to give the reader an idea of how this method can be adapted to counting simple types, we illustrate the definition, for $\sigma$, $X$, $r\in[0,1]$, and $S\in \mathcal B((2^{\omega})^{X}) $ of the set 
$\RED_{\sigma}^{X,r}(S) \subseteq \Lnu$, given as follows:  
\begin{align*}
\RED_{o}^{X,r}(S)&= \{ t \mid   \forall f\in S   \ \pi^{f}(t)\text{ reduces to a normal form with }\PROB\geq r\}\\ 
\RED_{(\BOX^{q}\sigma)\To \tau}^{X,r}(S)& =\{t \mid   
\forall S'\subseteq S,  \forall u\in \RED_{\sigma}^{X,q}(S'), \  tu\in \RED_{\tau}^{X,r}(S')
\}
\end{align*}
One can then establish that (1) whenever $t\in \RED_{\sigma}^{X,r}(S)$, for all $f\in S$, $\pi^{f}(t)$ reduces to a term in normal formal with $\PROB\geq r$, and (2) that whenever $\Gamma\vdash^{X,r}t:\bone\Pto \sigma$ is derivable, then $t\in \RED_{\sigma}^{X,r}(\model{\bone}_{X})$. All details can be found in the Appendix~\ref{appendix6}.

\subsection{The Curry-Howard Correspondence for $\TTT$}

In the usual correspondence between the simply typed $\lambda$-calculus and propositional logic, the typing derivations can be seen as the result of \emph{decorating} logical proofs with $\lambda$-terms. 
In a similar way, the type derivations of $\TTT$ can be seen as the result of decorating (a certain class of) $\PPL$-proofs with terms from $\Lnu$. 
Actually, we establish the correspondence by considering proofs
in a \emph{multi-succedent} system for $\PPL$.
This sequent calculus is equivalent to the one presented in Section~\ref{section4} and is provably sound and complete with respect to the counting semantics.
Specifically, it is obtained as a straightforward multi-succedent
extension of the given, single-succedent version.\footnote{For example,
the single-succedent version of rule $\AxO$ 
turns into the following multi-succedent rule:
\begin{prooftree}
\AxiomC{$\vDash a\in X$}
\AxiomC{$\bone\vDash^{X}\bvar_{i}^{a}$}
\RightLabel{$\AxO$}
\BinaryInfC{$\vdash^{X} \Delta, \bone\Era\atom{i}_{a}$}
\end{prooftree}
Remarkably, it is equally possible to establish the correspondence between $\lambda\PPL$ and $\PPL$, 
by considering the single-succedent rule system.
Nevertheless, the embedding for the abstraction rule becomes somehow trivial.
}

To give the reader a glimpse of how a translation from $\TTT$ to $\PPL$ can be defined, we first associate simple types $\sigma,\tau$ with formulas $\sigma^{*},\tau^{*}$ as follows:
\begin{align*}
o^{*}& = \FLIP  \\
 \qquad  (\FF s\To \sigma)^{*} &= \lnot \FF s^{*} \vee \sigma^{*} \\
(\BOX^{q}\sigma)^{*} & = \BOX^{q}_{a}\sigma^{*}
\end{align*}
where $\FLIP$ is the (classically provable) closed formula $\BOX^{1}_{a}(\mathbf 0_{a}\lor \lnot \mathbf 0_{a})$, and where $a$ is some fixed name. 
The idea now is that a type derivation of $\Gamma \vdash^{X,r} t: \bone \Pto \sigma$ yields a proof in $\PPL$ of the labelled sequent 
$\vdash^{X}\bone \Pto \BOX^{r}_{a}\sigma^{*}, \bot\Pfrom \Gamma^{*}$, 
where $\bot\Pfrom \Gamma^{*}$ is 
$\bot\Pfrom (\FF{s}_{1})^{*},\dots, \bot\Pfrom (\FF{s}_{n})^{*}$, 
for $\Gamma=\{\FF{s}_{1}:\sigma_{1},\dots, \FF{s}_{n}:\sigma_{n}\}$, and $a$ is some fresh name. 

One can check then that the abstraction rule translates into one of the disjunction rules of $\PPL$ (together with an application of Lemma \ref{lemma:commutations}), and that the application rule translates into the cut-rule. 
Furthermore, the generator rule of $\TTT$ essentially translates into the rule $\Rbr$ of $\PPL$ (see the Appendix~\ref{appendix6} for all details).
%
%
%
%
%
More generally, all rules of $\TTT $ translate into the application of one or more rules of $\PPL$, leading to the following: 

\begin{proposition}
If $\Gamma\vdash^{X,r}t: \bone \Pto \sigma$ is derivable in $\TTT$, then 
$\vdash^{X}_{\PPL}\bone \Pto \BOX^{r}_{a}\sigma^{*}, \bot\Pfrom \Gamma^{*}$.

\end{proposition}

Is it possible to go in the other direction, that is, from proofs to programs? 
Recall that simply typed systems for the $\lambda$-calculus usually correspond to proof systems for \emph{minimal logic}. As $\PPL$ was defined as a \emph{classical} system, one can hardly hope to translate \emph{any} proof from $\PPL$ into a $\lambda$-term. However, it is possible to define a minimal fragment of $\PPL$ in such a way that proofs in $\IPPL$ translate into typing derivations of $\TTT$.
For the interested reader, we present this system, called $\IPPL$, in the Appendix~\ref{appendix6}.

\begin{remark}
The Curry-Howard correspondence usually includes also a correspondence between $\lambda$-term reduction and normalization procedures.
We already observed that the cut-rule is admissible in $\PPLc$ and $\PPL$. However, the investigation of cut-elimination procedures for these systems is among the many aspects of counting quantifiers that we still have to explore.
\end{remark}

%% file: 7Arithmetic.tex

This paper introduced a new quantified propositional logic, 
capable of characterizing Wagner's counting hierarchy and 
of inspiring the design of type systems for probabilistic 
$\lambda$-calculi. 
But this is not the end of the story. 
As we will further discuss in the following paragraphs,
counting quantifiers also make sense in other settings, 
particularly in the context of first-order logic, 
as witnessed by investigations on the descriptive complexity 
of counting classes~\cite{Kontinen}.
In this section, we want to spell out, without going into the details,
how the ideas behind counting propositional logic make sense in
arithmetic too. 
Remarkably, the resulting logical system becomes much
more powerful and potentially very interesting from a computational
viewpoint.\footnote{For further details, see~\cite{ADLP}.}

Observe first although it is convenient and elegant to 
formulate the semantics of $\PPL$ by referring to the 
Cantor space, this might look unnecessary.
In fact, the nature of the underlying system makes it
\emph{impossible} to consider formulas with more than 
a constant number of indexes. 
But what if natural numbers become first-class citizens, 
like in arithmetic?
In fact, one may naturally view $\PPL$-quantification as a quantification \emph{over a first-order predicate}, 
$\BOX_{P}^{q}\fone$ such that occurrences of the atoms 
$\atom i_{a}$ in $\fone$ are replaced by the first-order 
atom $\predone(t)$, where $t$ is now an arithmetical term.
It would make perfect sense then 
to write sentences as the following one:
$$
\BOX^q_\predone \forall x.\predone(x).
$$
For obvious reasons, the formula is \valid \ only when $q$ is $0$, since the infinite sequences in $\twoOm$ in which
$0$ never occurs form a singleton, which has measure $0$. Dually,
the following formula turns out to be \valid \ for every value of $q$:
$$
\BOX^q_\predone \exists x.\predone(x).
$$

More interesting formulas can also be built. For example,
let $F(x,y)$ and $G(x,y)$ be arithmetical formulas expressing respectively that 
``$y$ is strictly smaller than the length of $x$'' (where $x$ is seen as a code for a finite binary sequence), and 
``the $y+1$-th bit of $x$ is 1''. Then, the following formula
\begin{equation*}\label{eq:monkey}
\forall x. \BOX^{1}_\predone \exists y.\forall z. F(x,z)\rightarrow \big(F(x,z)\land G(x,z) \leftrightarrow \predone(y+z)\big)
\end{equation*}
\noindent 
encodes the so-called \emph{infinite monkey theorem}, a classical
result from probability theory that states that a monkey randomly
typing on a keyboard has probability $1$ of ending up writing
 the \emph{Macbeth} (or any other fixed string) sooner or later: 
let $x$ be a binary encoding of the \emph{Macbeth}, 
and let the predicate $\predone(n)$ describe the random
value typed by the monkey at time $n$; 
then, with probability 1, there exists a time $y$ after which 
$\predone(n)$ will evolve exactly like $x$.

But could we give semantics to the resulting formulas? 
The answer is affirmative and surprising. 
Indeed, we can do it precisely along the same lines we 
followed in Section~\ref{section3}: 
instead of interpreting a predicate symbol, $\predone$,
as some function from natural numbers to truth-values, 
we may directly interpret $\predone$ as the set of 
\emph{all} functions yielding a measurable set 
$\model{\predone}\in \mathcal B(\twoOm)$, 
precisely as in $\PPLc$.
Since all other connectives of first-order arithmetic translate 
through the operations of a $\sigma$-algebra 
(with \emph{countable} unions and intersections playing a
fundamental role in interpreting the arithmetic quantifiers), 
any arithmetic formula $\fone$ can be interpreted
as a measurable set, $\model\fone\in \mathcal B(\twoOm)$. 
We are no longer in the finite world of propositional logic, 
so it makes no sense to \emph{count} such sets $\model\fone$; 
yet, it still makes perfect sense to \emph{measure} them. 
In particular, when $\fone$ is a formula with a unique 
predicate variable $\predone$, we can interpret
$\BOX^{q}_\predone \fone$ as being true precisely when 
$\mu(\model\fone)\geq q$.

In this way we can even \emph{prove} that formulas, 
such as the one formalizing the infinite monkey theorem, are \valid,
and their proofs will rely on some (non-trivial) facts from
measure theory. 
Moreover, along similar lines to what was shown in
Section~\ref{section6}, it should be possible to capture
some operational properties of probabilistic
programs, as for example~\emph{almost sure termination}, 
by means of arithmetical formulas. 
In sum, with counting propositional logic we just touched the 
tip of the iceberg of a seemingly new and exciting 
correspondence between logic and computation theory, 
that deserves to be explored much further.\footnote{These ideas are at the basis of the development of a first-order logic with measure quantifiers, called $\mathsf{MQPA}$, as presented in~\cite{ADLP}, where some first results are also studied.}

%% file: 8RelatedWork.tex

The counting hierarchy was (independently) defined 
in the 1980s by Wagner~\cite{Wagner84,Wagner,Wagner86} and
by Pareberry and Schnitger~\cite{PareberrySchnitger}.
It was conceived as Meyer and Stockmeyer's polynomial 
hierarchy~\cite{MeyerStockmeyer72,MeyerStockmeyer73},
which allows to express the complexity of 
many natural problems in which counting is involved, 
and which are not ``represented'' by $\PH$.
There are two main, equivalent~\cite{Toran91}, ways
of characterizing the counting hierarchy: 
the original characterization, in terms of 
\emph{alternating quantifiers},~\cite{Wagner}, 
and the oracle characterization~\cite{Toran88,Toran91}.
To the best of our knowledge, $\PPL$ is novel in its being an
enrichment of propositional logic by means of counting quantifiers. 
However, strong similarities exist between our $\BOX$ and
Wagner's \emph{counting quantifier}~\cite{Wagner},
and these similarities have been exploited in Section~\ref{section5}.
Furthermore, there is a worth-noticing analogy between the role of 
the operator $\BOX$ and usual propositional quantifiers: 
remarkably, Meyer and Stockmeyer~\cite{MeyerStockmeyer73} 
not only showed that $\TQBF$, the quantitative analogue of $\PPL$, 
is $\PSpace$-complete, but also that $\Sigma^{P}_{k}$-complete
problems in the polynomial hierarchy correspond to 
the satisfiability problem for quantified Boolean languages with 
$k$–1 quantifiers alternation, namely $\TQBF_{k}$.

It is worth mentioning that in the same years some
other ``probabilistic'' quantifiers were introduced. 
For example, Papadimitriou~\cite{Papadimitriou} characterizes 
$\PPSpace$ by alternating standard quantifiers with a \emph{probabilistic} quantifier, 
which expresses the fact  that more than the half of the strings of a
certain length satisfy the underlying predicate. 
Zachos and Heller characterized $\BPP$ by means of a 
\emph{random} quantifier~\cite{ZachosHeller}, and Zachos also
considered the relationship between classical and
probabilistic classes introducing other quantifiers, such as
the \emph{overwhelming} the
and \emph{majority} ones~\cite{Zachos88}. 
Remarkably, all these works concern counting quantification on 
(classes of) languages, 
rather than \emph{stricto sensu} logical operators. 
An exception is represented by Kontinen's work~\cite{Kontinen}, 
in which counting quantifiers are second-order quantifiers 
in an otherwise first-order logic, 
much in the style of descriptive complexity. 
In other words, one formula stands for a complexity
bounded algorithm, while here formulas are taken as data.

\bigskip
Despite the extensive literature on logical systems enabling 
(in various ways and for different purposes) some forms of probabilistic
reasoning, there is not much about logics tied to computational aspects, 
as $\PPL$ is. 
Most of the recent logical formalisms have been
developed in the realm of modal logic, starting from the work by
Nilsson~\cite{Nilsson86,Nilsson93}.
In particular, in the 1990s, 
Bacchus~\cite{Bacchus, Bacchus90a, Bacchus90b}
defined \emph{probability terms} by means of a modal operator $\prob$,
which effectively computes the probability of certain events,
and \emph{probability formulas}, 
which are equalities between probability
terms and numbers, such as $\prob(\alpha) = \frac{1}{2}$ 
(this, by the way, is not too different from the formula of 
$\PPLc$ $\BOX^{\frac{1}{2}}\fone$;
noticeably, Bacchus' $\prob$ yields terms, whereas $\BOX$ yields
formulas). 
In the same years also Fagin, Halpern and Megiddo~\cite{FHM}
 introduced an interesting form of probability logic (which was later
further studied in~\cite{FH94,Halpern90,Halpern03}).
In this logic, probability spaces are the underlying model, 
and can be accessed through so-called \emph{weight terms}.

Another class of probabilistic modal logics have been designed to model Markov chains and similar structures. Some of these logics are probabilistic extensions of \textbf{CTL}, the standard logic for model-checking (see, for example,~\cite{KOZEN1981328,Hansson1994,LEHMANN1982165}).
A notable example is \emph{Riesz modal logic}~\cite{lmcs:6054}, which admits a sound and complete proof system.
Differently from $\PPL$, in these logics modal operators have a dynamical meaning, as they describe transitions in a MDP.

Concerning proof theory, axiomatic proof systems have been provided 
for both logics~\cite{Bacchus,FHM}, in connection with well-known
modal systems such as \textbf{KD45}~\cite{Bacchus90c,FH91}. 
With the sole exception of Riesz modal logic, we are
not aware of sequent calculi for probability logic. 
By the way, our one-sided, sequent calculi are inspired by labelled
calculi such as \textbf{G3K*} and \textbf{G3P*}, 
as presented for example in~\cite{NegrivonPlato,GNS}, 
but, differently from them, our rules are not invertible.

\bigskip
Probabilistic models and randomized computation are pervasive 
in many areas of computer science and programming. 
From the 1950s on, the interest for probabilistic algorithms and
models started 
spreading~\cite{LMSS,Davis61,Carlyle,Rabin,Santos68,Santos69b,Simon81}
and probabilistic computation has been intensely 
studied by the TCS research 
community~\cite{Rabin,Santos69,ShapiroMoore,Santos71}. 
Nowadays, well-defined computational models, 
such as randomized variations on 
probabilistic automata[56], 
(both Markovian and oracle) Turing machines~\cite{Santos69,Santos71,Gill74,Gill77}, 
and $\lambda$-calculi, are available.

Also probabilistic programming languages have been extensively
studied from the 1970s on, and are still under
scrutiny.
Specifically,  foundational work on probabilistic functional languages
is \cite{SahebDjaromi}, where Saheb-Djahromy's the typed, 
higher-order $\lambda$-calculus \textsf{LCF} is introduced. 
Other early introduced calculi are, for example, 
Ramsey and Pfeffer's stochastic $\lambda$-calculus, \cite{RamseyPfeffer} and Park's $\lambda_{\gamma}$, \cite{Park}, 
see also \cite{KMAP}, \cite{PlessLuger}, \cite{PPT}. 
Operational semantics and properties of higher-order 
probabilistic programs have been studied in \cite{DLZ}.
A particularly fruitful approach to the study of higher-order probabilistic
computation consists in enriching general $\lambda$-calculus with a
probabilistic choice operator $\oplus$~\cite{DPHW,DLZ,EPT}. 
For these calculi, type systems ensuring various form of termination
properties have recently been 
introduced~\cite{DalLagoGrellois,BreuvartDalLago,ADLG}.
Our calculus is particularly inspired by the one introduced
in~\cite{DLGH}. 
Moreover, our probabilistic definition of normal forms is 
somehow reminiscent of the operational 
semantics from~\cite{leventis_2019}. 
In all these cases, however, no logically inspired
type system or Curry-Howard style correspondence is currently known.

%% file: 9Conclusion.tex

The aim of the present paper is not the introduction of counting
propositional logic \emph{per se}. 
After all, the idea of switching from qualitative to quantitative notions 
of quantification is not ours and has already been proved useful.
Indeed, the main contribution of this work is to define a logical
counterpart for \emph{randomized} computation and, so, to
generalizing some standard results developed in the deterministic framework. 
Indeed, in Section~\ref{section5}, we have shown that counting 
quantifiers can play nicely with propositional logic in characterizing 
the counting hierarchy, while, in Section~\ref{section6}, we have
designed type systems for the lambda-calculus in the
spirit of the Curry-Howard correspondence. 
In this way, logic can somehow catch up with some old 
and recent computational complexity and
programming language theory results.

Many problems and questions are still open, or could not have 
been described in this paper, due to space reasons.
For example, the proof theory of $\PPL$ has just been briefly delineated,
but the dynamics of the introduced formal systems 
certainly deserves to be investigated, despite not being
crucial for the results in this paper. 
The potential interest of injecting counting quantifiers 
(or, better, measure quantifiers) into the language of arithmetic has 
been briefly described in Section~\ref{section7}, 
and implicitly suggests several intriguing problems, 
which are still open. 
A first-order logic endowed with measure-quantifiers, 
called $\textbf{\textsf{MQPA}}$, has been presented in~\cite{ADLP}, 
where some first representation results concerning 
its connection with randomized computation has been established.  

%% file: appendix.tex

\section{Proofs from Section \ref{section3}}\label{appendix3}

\subsection{From Valuations to Measurable Sets}\label{appendix3.1}

\begin{lemma}\label{lemma:box0}
For every formula of $\PPLc$, call it $\fone$, 
the following holds:
$$
\model{\BOX^{0}\fone} = \twoOm \ \ \ \ \ \ \ \ \ \ \ \ \ \ \ \ \ \ \ \ \  \ \ \ 
\model{\DIA^{0}\fone} = \emptyset.
$$
\end{lemma}
\begin{proof}
For every $\fone$, $\mu$($\model{\fone}$) $\ge 0$ holds.
Thus,

\begin{minipage}{\linewidth}
\begin{minipage}[t]{0.4\linewidth}
\begin{align*}
\model{\BOX^{0}\fone} &= \begin{cases} \twoOm \ \ \ &\text{if } \mu(\model{\fone}) \ge 0 \\ \emptyset \ \ \ &\text{otherwise} \end{cases} \\
&= \twoOm
\end{align*}
\end{minipage}
\hfill
\begin{minipage}[t]{0.6\linewidth}
\begin{align*}
\model{\DIA^{0}\fone} &= \begin{cases} \twoOm \ \ \ &\text{if } \mu(\model{\fone}) < 0 \\ \emptyset \ \ \ &\text{otherwise} \end{cases} \\
&= \emptyset.
\end{align*}
\end{minipage}
\end{minipage}
\end{proof}

\begin{lemma}\label{lemma:interdef}
For every formula of $\PPLc$, call it $\fone$, and 
$q \in \mathbb{Q}_{[0,1]}$, the following holds:
$$
\BOX^{q}\fone \equiv \lnot \DIA^{q}\fone \ \ \ \ \ \ \ \ \ \ \ \ \ \ \ \ \ \ \ \ \  \ \ \  \DIA^{q}\fone \equiv \lnot\BOX^{q}\fone.
$$
\end{lemma}
\begin{proof}
The proof is based on Definition~\ref{semantics}:

\begin{minipage}{\linewidth}
\begin{minipage}[t]{0.4\linewidth}
\begin{align*}
\model{\neg \DIA^{q}\fone} &= \twoOm \ – \ \model{\DIA^{q}\fone} \\
&= \twoOm \ – \begin{cases} \twoOm \ \ &\text{ if } \mu(\model{\fone}) < q \\ \emptyset \ \ &\text{ otherwise }
\end{cases} \\
&= \begin{cases} \emptyset \ \ &\text{ if } \mu(\model{\fone}) < q \\ \twoOm \ \ &\text{ otherwise } \end{cases} \\
&= \model{\BOX^{q}\fone}
\end{align*}
\end{minipage}
\hfill
\begin{minipage}[t]{0.6\linewidth}
\begin{align*}
\model{\neg \BOX^{q}\fone} &= \twoOm \ – \ \model{\BOX^{q}\fone} \\
&= \twoOm \ – \ \begin{cases} \twoOm \ \ &\text{ if } \mu(\model{\fone}) \ge q \\ \emptyset \ \ &\text{ otherwise} \end{cases} \\
&= \begin{cases} \emptyset \ \ &\text{ if } \mu(\model{\fone}) \ge q \\ \twoOm \ \ &\text{ otherwise} \end{cases} \\
&= \model{\DIA^{q}\fone}.
\end{align*}
\end{minipage}
\end{minipage}
\end{proof}

\subsection{The Proof Theory of $\PPLc$}\label{appendix3.2}

From now on, let us denote the one-sided,
single-succedent proof system for $\PPLc$ as $\GsCPLc$.
Let us also call $\mu$-rules the two rules $\Rmur$ and $\Rmul$.
The notion of derivation in $\GsCPLc$ is defined in the 
standard way.
\begin{definition}[Derivation in $\GsCPLc$]
A \emph{derivation in $\GsCPLc$} is either an initial
sequent $–$ $\AxO$ or $\AxT$ $–$ 
or an instance of a $\mu$-rule $–$ $\Rmur$ or $\Rmul$ $–$, 
or it is obtained by applying a rule of $\GsCPLc$ 
to derivations concluding its premisses.
\end{definition}
\noindent
The definition of \emph{derivation height} is completely
canonical as well.
\begin{definition}[Derivation Height]
The \emph{height of a derivation in $\GsCPLc$} is 
the greatest number of successive applications of rules in it, 
where initial sequents and $\mu$-rules have height 0.
\end{definition}
\noindent
As anticipated, the proof of soundness is standard.
\primeSoundness*
\begin{proof}
The proof is by induction on the height of the derivation of 
$\vdash \lone$, call it $n$.
\begin{itemize}
\item \emph{Base case.} If $n$ = 0, $\vdash \lone$ is either an initial sequent or is derived by a $\mu$-rule. 
Both cases are trivial:
\par $\AxO.$ 
For some $\bone$ and $n\in\Nat$, $\lone = \bone \Era \atom{n}$. 
Then, the derivation is in the following form:
\begin{prooftree}
\AxiomC{$\bone \vDash \bvar_{n}$}
\RightLabel{$\AxO$}
\UnaryInfC{$\vdash \bone \Era \atom{n}$}
\end{prooftree}
But, $\bone \vDash \bvar_{n}$ means that 
$\model{\bone} \subseteq \model{\bvar_{n}}$, 
which is $\model{\bone} \subseteq \Cyl{n}$.
Furthermore, since $\model{\atom{n}}$ = $\Cyl{n}$, 
$\model{\bone} \subseteq \model{\atom{n}}$ holds.
For Definition~\ref{valB}, this means that 
$\vDash \bone \Era \atom{n}$, as desired.
\par $\AxT$. 
For some $\bone$ and $n\in\Nat$,
$\lone=\bone\Ela\atom{n}$. 
Then, the derivation is in the following form:
\begin{prooftree}
\AxiomC{$\bvar_{n}\vDash\bone$}
\RightLabel{$\AxT$}
\UnaryInfC{$\vdash\bone\Ela\atom{n}$}
\end{prooftree}
But, $\bvar_{n}\vDash \bone$ means $\model{\bvar_{n}}\subseteq \model{\bone}$, which is $\Cyl{n}\subseteq\model{\bone}$.
Moreover, $\model{\atom{n}}=\Cyl{n}$, so 
$\model{\atom{n}}\subseteq\model{\bone}$ holds.
For Definition~\ref{valB}, this means that
$\vDash\bone\Ela\atom{n}$, as desired.
\par $\Rmur$. 
For some $\bone$ and $\fone$, $\lone = \bone \Era \fone$.
Then, the derivation is in the following form:
\begin{prooftree}
\AxiomC{$\mu(\model{\bone}) = 0$}
\RightLabel{$\Rmur$}
\UnaryInfC{$\vdash\bone\Era\fone$}
\end{prooftree}
Since $\mu$($\model{\bone}$) = 0, for basic measure theory $\model{\bone} = \emptyset$. 
Thus, in particular, $\model{\bone} \subseteq \model{\fone}$, 
and so $\vDash \bone \Era \fone$. 
\par $\Rmul$. 
For some $\bone$ and $\fone$, $\lone$ = $\bone \Era \fone$.
Then, the derivation is in the following form:
\begin{prooftree}
\AxiomC{$\mu(\model{\bone}) = 1$}
\RightLabel{$\Rmul$}
\UnaryInfC{$\vdash\bone\Ela\fone$}
\end{prooftree}
Since $\mu$($\model{\bone}$) = 1, for basic measure theory, $\model{\bone}$ = $\twoOm$. 
Thus, in particular, 
$\model{\fone} \subseteq \model{\bone}$, 
and so $\vDash \bone \Ela \fone$.
\item \emph{Inductive case.} 
Let us assume that soundness holds for derivations of 
height up to $n$, 
and show that it holds for derivations of height $n$+1.
\par $\Rcup.$ 
If the last rule applied is an instance of $\Rcup$, 
the derivation has the following form:
\begin{prooftree}
\AxiomC{$\vdots$}
\noLine
\UnaryInfC{$\vdash \btwo \Era \fone$}
\AxiomC{$\vdots$}
\noLine
\UnaryInfC{$\vdash \bthree \Era \fone$}
\AxiomC{$\bone \vDash \btwo \vee \bthree$}
\RightLabel{$\Rcup$}
\TrinaryInfC{$\vdash \bone \Era \fone$}
\end{prooftree}
By IH, $\vdash \btwo \Era \fone$ and 
$\vdash \bthree \Era \fone$ are valid.
Furthermore, since $\bone \vDash \btwo \vee \bthree$, 
$\model{\bone} \subseteq \model{\btwo \vee \bthree}$ holds, 
that is 
$\model{\bone} \subseteq \model{\btwo} \cup \model{\bthree}$.
For IH and Definition~\ref{valB}, 
$\vDash \btwo \Era \fone$ and $\vDash \bthree \Era \fone$, 
which is $\model{\btwo} \subseteq \model{\fone}$ and 
$\model{\bthree} \subseteq \model{\fone}$.
For basic set theory, $\model{\btwo} \cup \model{\bthree} \subseteq \model{\fone}$
and, since 
$\model{\bone} \subseteq \model{\btwo} \cup \model{\bthree}$, 
also $\model{\bone} \subseteq \model{\fone}$ holds.
Thus, by Definition~\ref{valB}, $\vDash\bone\Era\fone$.

\par $\Rcap.$ If the last rule applied is an instance of $\Rcap$, the derivation has the following form:
\begin{prooftree}
\AxiomC{$\vdots$}
\noLine
\UnaryInfC{$\vdash \btwo \Ela \fone$}
\AxiomC{$\vdots$}
\noLine
\UnaryInfC{$\vdash \bthree \Ela \fone$}
\AxiomC{$\btwo \wedge \bthree \vDash \bone$}
\RightLabel{$\Rcap$}
\TrinaryInfC{$\vdash \bone \Ela \fone$}
\end{prooftree}
By IH, $\vdash \btwo \Ela \fone$ and 
$\vdash \bthree \Ela \fone$ are valid. 
Since $\btwo \wedge \bthree \vDash \bone$, 
for Definition~\ref{valB} $\model{\btwo \wedge \bthree} \subseteq \model{\bone}$, 
which is $\model{\btwo} \cap \model{\bthree} \subseteq \model{\bone}$.
For IH and Definition~\ref{valB}, $\vDash \btwo \Ela \fone$ and $\vDash \bthree \Ela \fone$, which is $\model{\fone} \subseteq \model{\btwo}$ and $\model{\fone} \subseteq \model{\bthree}$.
For basic set theory, $\model{\fone} \subseteq \model{\btwo} \cap \model{\bthree}$, and
so, given $\model{\btwo} \cap \model{\bthree} \subseteq \model{\bone}$, it can be concluded that $\model{\fone} \subseteq \model{\bone}$ holds.
Therefore, by Definition~\ref{valB}, $\vDash \bone \Ela \fone$.

\par $\Rnra.$ 
If the last rule applied is an instance of $\Rnra$, 
the derivation has the following form:
\begin{prooftree}
\AxiomC{$\vdots$}
\noLine
\UnaryInfC{$\vdash\btwo\Ela\fone$}
\AxiomC{$\bone \vDash \neg \btwo$}
\RightLabel{$\Rnra$}
\BinaryInfC{$\vdash\bone \Era \neg \fone$}
\end{prooftree}
Given the external hypothesis $\bone \vDash \neg \btwo$, 
$\model{\bone} \subseteq \big(\twoOm \ – \ {\model \btwo}\big)$ 
holds.
For IH, $\vdash \btwo \Ela \fone$ so,
for Definition~\ref{valB}, 
$\model{\fone} \subseteq \model{\btwo}$.
By basic set theory 
$\big(\twoOm \ – \ \model{\btwo}\big) \subseteq 
\big(\twoOm \ – \ \model{\fone}\big)$, 
that is $\big(\twoOm \ – \ \model{\btwo}\big) \subseteq \model{\neg\fone}$ and,
since 
$\model{\bone} \subseteq \big(\twoOm \ – \ {\model{\btwo}}\big)$, 
it is possible to conclude that
$\model{\bone} \subseteq \model{\neg \fone}$ holds.
But then, for Definition~\ref{valB} $\vDash \bone \Era \neg \fone$.

$\Rnla.$ 
If the last rule applied is an instance of $\Rnla,$ the derivation has the following form:
\begin{prooftree}
\AxiomC{$\vdots$}
\noLine
\UnaryInfC{$\vdash\btwo\Era\fone$}
\AxiomC{$\neg\btwo\vDash\bone$}
\RightLabel{$\Rnla$}
\BinaryInfC{$\vdash\bone\Ela\neg\fone$}
\end{prooftree}
Given the external hypothesis $\neg\btwo \vDash \bone$, 
then $\big(\twoOm \ – \ \model{\btwo}\big)\subseteq \model{\bone}$,
and, 
for IH, $\vDash \btwo \Era \fone$,
which is 
$\model{\btwo} \subseteq \model{\fone}$.
For basic set theory,
$\big(\twoOm \ – \ \model{\fone}\big) 
\subseteq \big(\twoOm \ – \ \model{\btwo}\big)$,
which is 
$\model{\neg \fone} \subseteq \big(\twoOm \ – \ {\model{\btwo}}\big)$, 
and,
since $\big(\twoOm \ – \ \model{\btwo}\big) \subseteq \model{\bone}$, 
also $\model{\neg \fone} \subseteq \model{\bone}$.
So, for Definition~\ref{valB}, $\vDash \bone \Ela \neg \fone$.

\par $\mathsf{R1_{\vee}^{\Era}}.$ 
If the last rule applied is an instance of $\mathsf{R1_{\vee}^{\Era}}$, 
the derivation has the following form:
\begin{prooftree}
\AxiomC{$\vdots$}
\noLine
\UnaryInfC{$\vdash \bone \Era \fone$}
\RightLabel{$\mathsf{R1_{\vee}^{\Era}}$}
\UnaryInfC{$\vdash \bone \Era \fone \vee \ftwo$}
\end{prooftree}
By IH, $\vDash \bone \Era \fone$, which is 
$\model{\bone} \subseteq \model{\fone}$.
For basic set theory, also 
$\model{\bone} \subseteq \model{\fone} \cup \model{\ftwo}$, which is 
$\model{\bone} \subseteq \model{\fone \vee \ftwo}$.
Therefore, $\vDash \bone \Era \fone \vee \ftwo$ as desired.

\par $\mathsf{R2_{\vee}^{\Era}}.$ 
If the last rule applied is an instance of $\mathsf{R2_{\vee}^{\Era}}$, 
the derivation has the following form:
\begin{prooftree}
\AxiomC{$\vdots$}
\noLine
\UnaryInfC{$\vdash \bone \Era \ftwo$}
\RightLabel{$\mathsf{R2_{\vee}^{\Era}}$}
\UnaryInfC{$\vdash \bone \Era \fone \vee \ftwo$}
\end{prooftree}
By IH, $\vDash \bone \Era \ftwo$, which is 
$\model{\bone} \subseteq \model{\ftwo}$.
For basic set theory, also 
$\model{\bone} \subseteq \model{\fone} \cup \model{\ftwo}$, which is 
$\model{\bone} \subseteq \model{\fone \vee \ftwo}$.
Therefore, $\vDash \bone \Era \fone \vee \ftwo$ as desired.

\par $\Rvla.$ If the last rule applied is an instance of $\Rvla$, the derivation has the following form:
\begin{prooftree}
\AxiomC{$\vdots$}
\noLine
\UnaryInfC{$\vdash \bone \Ela \fone$}
\AxiomC{$\vdots$}
\noLine
\UnaryInfC{$\vdash  \bone \Ela \ftwo$}
\RightLabel{$\Rvla$}
\BinaryInfC{$\vdash \bone \Ela \fone \vee \ftwo$}
\end{prooftree}
By IH, $\vdash \bone \Ela \fone$ and $\vdash \bone \Ela \ftwo$ are valid.
So, for Definition~\ref{valB}, $\vDash\bone\Ela\fone$ and $\vDash\bone\Ela\ftwo$, which is, 
$\model{\fone}\subseteq\model{\bone}$ and 
$\model{\ftwo}\subseteq\model{\bone}$.
For basic set theory, 
$\model{\fone}\cup\model{\ftwo}\subseteq \model{\bone}$,
which is $\model{\fone\vee\ftwo} \subseteq \model{\bone}$,
and so $\vDash \bone \Ela \fone \vee \ftwo$.

\par $\Rwra.$ 
If the last rule applied is an instance of $\Rwra$, 
the derivation has the following form:
\begin{prooftree}
\AxiomC{$\vdots$}
\noLine
\UnaryInfC{$\vdash\bone\Era\fone$}
\AxiomC{$\vdots$}
\noLine
\UnaryInfC{$\vdash\bone\Era\ftwo$}
\RightLabel{$\Rwra$}
\BinaryInfC{$\vdash\bone\Era\fone\wedge\ftwo$}
\end{prooftree}
By IH, $\vDash\bone\Era\fone$ and $\vDash\bone\Era\ftwo$, 
which is $\model{\bone} \subseteq \model{\fone}$ and $\model{\bone} \subseteq \model{\ftwo}$.
For basic set theory, 
$\model{\bone} \subseteq \model{\fone} \cap \model{\ftwo}$, 
which is $\model{\bone} \subseteq \model{\fone \wedge \ftwo}$.
Therefore, by Definition~\ref{valB},
$\vDash \bone \Era \fone \wedge \ftwo$.

\par $\mathsf{R1_{\wedge}^{\Ela}}.$ 
If the last rule applied is an instance of $\mathsf{R1_{\wedge}^{\Ela}}$,
the derivation has the following form:
\begin{prooftree}
\AxiomC{$\vdots$}
\noLine
\UnaryInfC{$\vdash \bone \Ela \fone$}
\RightLabel{$\mathsf{R1_{\wedge}^{\Ela}}$}
\UnaryInfC{$\vdash\bone \Ela \fone \wedge \ftwo$}
\end{prooftree}
By IH, $\vDash \bone \Ela \fone$,
which is $\model{\fone} \subseteq \model{\bone}$.
For basic set theory, 
$\model{\fone} \cap \model{\ftwo} \subseteq \model{\bone}$, which is
$\model{\fone \wedge \ftwo} \subseteq \model{\bone}$.
Therefore, $\vDash \bone \Ela \fone \wedge \ftwo$.

\par $\mathsf{R2_{\wedge}^{\Ela}}.$ 
If the last rule applied is an instance of $\mathsf{R2_{\wedge}^{\Ela}}$,
the derivation has the following form:
\begin{prooftree}
\AxiomC{$\vdots$}
\noLine
\UnaryInfC{$\vdash \bone \Ela \ftwo$}
\RightLabel{$\mathsf{R2_{\wedge}^{\Ela}}$}
\UnaryInfC{$\vdash\bone \Ela \fone \wedge \ftwo$}
\end{prooftree}
By IH, $\vDash \bone \Ela \ftwo$,
which is $\model{\ftwo} \subseteq \model{\bone}$.
For basic set theory, 
$\model{\fone} \cap \model{\ftwo} \subseteq \model{\bone}$, which is
$\model{\fone \wedge \ftwo} \subseteq \model{\bone}$.
Therefore, $\vDash \bone \Ela \fone \wedge \ftwo$.

\par $\Rbr.$ 
If the last rule applied is an instance of $\Rbr$, 
the derivation has the following form:
\begin{prooftree}
\AxiomC{$\vdots$}
\noLine
\UnaryInfC{$\vdash\btwo \Era \fone$}
\AxiomC{$\mu(\model{\btwo}) \ge q$}
\RightLabel{$\Rbr$}
\BinaryInfC{$\vdash\bone \Era \BOX^{q}\fone$}
\end{prooftree}
By IH, $\vDash \btwo \Era \fone$, which is $\model{\btwo}\subseteq\model{\fone}$, and $\mu(\model{\btwo}) \ge q$ holds.
For basic measure theory, also $\mu$($\model{\fone}$) $\ge q$.
Then, $\model{\BOX^{q}\fone}$ = $\twoOm$ and, so, 
(for every $\bone$) 
$\model{\bone} \subseteq \model{\BOX^{q}\fone}$, 
which is $\vDash \bone \Era \BOX^{q}\fone$.

\par $\Rbl.$ 
If the last rule applied is an instance of $\Rbl$, the derivation has the following form:
\begin{prooftree}
\AxiomC{$\vdots$}
\noLine
\UnaryInfC{$\vdash\btwo \Ela \fone$}
\AxiomC{$\mu(\model{\btwo}) < q$}
\RightLabel{$\Rbl$}
\BinaryInfC{$\vdash\bone \Ela \BOX^{q}\fone$}
\end{prooftree}
By IH, $\vDash \btwo \Ela \fone$, 
which is $\model{\fone}\subseteq\model{\btwo}$, 
and $\mu(\model{\btwo}) < q$ holds.
For basic set theory, also $\mu(\model{\fone})<q$.
Therefore, $\model{\BOX^{q}\fone}$ = $\emptyset$, and so 
$\model{\BOX^{q}\fone} \subseteq \model{\bone}$, 
which is $\vDash \bone \Ela \BOX^{q}\fone$.

\par $\Rdr.$ 
If the last rule applied is an instance of $\Rdr$, the derivation has the following form:
\begin{prooftree}
\AxiomC{$\vdots$}
\noLine
\UnaryInfC{$\vdash \btwo \Ela \fone$}
\AxiomC{$\mu(\model{\btwo}) < q$}
\RightLabel{$\Rdr$}
\BinaryInfC{$\vdash\bone \Era \DIA^{q}\fone$}
\end{prooftree}
By IH, $\vDash \btwo \Ela \fone$, 
which is $\model{\fone}\subseteq \model{\btwo}$, and 
$\mu(\model{\btwo}) < q$ holds.
For basic set theory, $\mu(\model{\fone}) < q$.
But then, $\model{\DIA^{q}\fone}$ = $\twoOm$ and so (for every $\bone$) $\model{\bone} \subseteq \model{\DIA^{q}\fone}$, which is $\vDash \bone \Era \DIA^{q}\fone$.

\par $\Rdl.$ If the last rule applied is an instance of $\Rdl$, the derivation has the following form:
\begin{prooftree}
\AxiomC{$\vdots$}
\noLine
\UnaryInfC{$\vdash\btwo \Era \fone$}
\AxiomC{$\mu(\model{\btwo}) \ge q$}
\RightLabel{$\Rdl$}
\BinaryInfC{$\vdash\bone \Ela \DIA^{q}\fone$}
\end{prooftree}
By IH, $\vDash\btwo \Era \fone$, which is 
$\model{\btwo}\subseteq\model{\fone}$, and 
$\mu(\model{\btwo}) \ge q$ holds.
For basic set theory, $\mu(\model{\fone})\ge q$. 
Therefore, $\model{\DIA^{q}\fone}$ = $\emptyset$ and, so $\model{\DIA^{q}\fone} \subseteq \model{\bone}$, which is $\vDash \bone \Ela \DIA^{q}\fone$.
\end{itemize}
\end{proof}

In order to prove the completeness of 
$\GsCPLc$ with respect to the semantics of $\PPLc$, 
some preliminary notions and lemmas have to be introduced.
\begin{definition}[Basic Formulas]
A \emph{basic formula} of $\GsCPLc$ is a labelled formula, 
$\bone \Era \fone$ or $\bone \Ela \fone$, 
such that its logical part $\fone$ is atomic.
\end{definition}
\begin{definition}[Regular Sequent]
A \emph{regular sequent} of $\GsCPLc$ is a sequent of the form $\vdash \lone$, such that $\lone$ is a basic formula.
\end{definition}
\noindent
The notion of decomposition rewriting reduction is
defined by the following decomposition rules.
\begin{definition}[Decomposition Rewriting Reduction, $\dec$]\label{decomposition}
The \emph{decomposition rewriting reduction}, $\dec$, from a  sequent to a set of  sequents (both in the language of $\GsCPLc$), is defined by the following decomposition rules:
\begin{align*}
\text{if } \bone \vDash \neg \btwo, \ \vdash \bone \Era \neg \fone \ \ &\dec \ \ \{\vdash \btwo \Ela \fone\} \\
\text{if } \neg\btwo \vDash \bone, \  \vdash \bone \Ela 
\neg \fone \ \ &\dec \ \ \{\vdash \btwo \Era \fone\} \\
\text{if } \bone \vDash \btwo \vee \bthree, \ \vdash \bone 
\Era \fone \vee \ftwo \ \ &\dec \ \ \{\vdash \btwo \Era \fone, \vdash \bthree \Era \ftwo\} \\
\vdash \bone \Ela \fone \vee \ftwo \ \ &\dec \ \ \{\vdash \bone \Ela \fone, \vdash \bone \Ela \ftwo\} \\
\vdash \bone \Era \fone \wedge \ftwo \ \ &\dec \ \ \{\vdash \bone \Era \fone, \vdash \bone \Era \ftwo\} \\
\text{if } \btwo \wedge \bthree \vDash \bone, \ \vdash \bone \Ela \fone \wedge \ftwo \ \ &\dec \ \ \{\vdash \btwo \Ela \fone, \vdash  \bthree \Ela \ftwo\} \\
\text{if } \mu(\model{\btwo}) \ge q, \ \vdash \bone \Era \BOX^{q}\fone \ \ & \dec \ \ \{\vdash \btwo \Era \fone\} \\
\text{if } \mu(\model{\btwo}) < q, \ \vdash \bone \Ela \BOX^{q}\fone \ \ &\dec \ \ \{\vdash \btwo \Ela \fone\} \\
\text{if } \mu(\model{\btwo}) < q, \ \vdash \bone \Era \DIA^{q}\fone \ \ &\dec \ \  \{\vdash \btwo \Ela \fone\} \\
\text{if } \mu(\model{\btwo}) \ge q, \ \vdash \bone \Ela \DIA^{q}\fone \ \ &\dec \ \ \{\vdash \btwo \Era \fone\} \\
\text{if } \mu(\model{\bone}) = 0, \ \vdash \bone \Era \fone \ \ &\dec \ \ \{\} \\
\text{if } \mu(\model{\bone}) = 1, \ \vdash \bone \Ela \fone \ \ &\dec \ \ \{\} \\
\text{if } \mu(\model{\bone}) \neq 0, \ \vdash \bone \Era \DIA^{0}\fone \ \ &\dec \ \ \{\vdash \bot\} \\
\text{if } \mu(\model{\bone}) \neq 1, \ \vdash \bone \Ela \BOX^{0}\fone \ \ &\dec \ \ \{\vdash \bot\}
\end{align*}
\end{definition}
\normalsize
\noindent
Remarkably, rewriting rules are so defined that each application of a decomposition reduction on an arbitrary sequent, 
$\vdash \lone$, leads to a set of sequents, 
$\{\vdash \lone_{1}, \dots, \vdash \lone_{n}\}$, 
such that for every $i \in \{1, \dots, n\}$, 
the number of connectives of $\vdash \lone_{i}$ is 
(strictly) smaller than that of $\vdash \lone$. 
Basing on $\dec$, it is possible to define a set-decomposition reduction, 
$\Dec$, from a set of sequents to a set of sequents.
\begin{definition}[Set Decomposition, $\Dec$]\label{Decomposition}
The \emph{set-decomposition reduction}, $\Dec$, from a set of sequents to another set of sequents (in $\GsCPLc$), is defined as follows:
\begin{prooftree}
\AxiomC{$\vdash \lone_{i} \dec \{\vdash \lone_{i_{1}}, \dots, \vdash \lone_{i_{m}}\}$}
\UnaryInfC{$\{\vdash \lone_{1}, \dots, \vdash \lone_{i}, \dots, \vdash \lone_{n}\} \Dec \{\vdash \lone_{1}, \dots, \vdash \lone_{i_{1}}, \dots, \vdash \lone_{i_{m}}, \dots, \vdash \lone_{n}\}$}
\end{prooftree}
\end{definition}
\noindent
In other words, $\Dec$ is the natural lifting of $\dec$ to a relation on sets. 

Before introducing the crucial notions of
normal form for these two reductions it is
worth mentioning that each predicate concerning one sequent 
can be naturally generalized to sets of sequents 
by stipulating that a predicate holds for the set when 
it holds for every sequent in the set. 
In order to make our presentation straightforward, 
let us introduce the following auxiliary notion of corresponding set.
\begin{definition}[Corresponding Set]
Given a sequent, $\vdash \lone$, 
we call its \emph{corresponding set} the set including only 
this sequent as its element, $\{\vdash \lone\}$.
\end{definition}
\noindent
The notion of $\Dec$-normal form is defined
as predictable: a set of sequents is in $\Dec$-normal form
when no set-decomposition reduction, $\Dec$, can be applied on it.
Otherwise said, the set is such that there is no sequent in it 
on which $\dec$ can be applied. 
\begin{definition}[$\Dec$-Normal Form]\label{def:normalform}
A sequent is a \emph{$\Dec$-normal form} if no decomposition rewriting reduction rule, $\dec$, can be applied on it.
A set of sequents is in \emph{$\Dec$-normal form} if it cannot be reduced by any $\Dec$ set-rewriting rule.
\end{definition}

The proof of completeness crucially relies on 
the given rewriting reductions,
$\dec$ and $\Dec$, and it is based on some auxiliary
steps: (i) validity is \emph{existentially} preserved 
through $\Dec$-decomposition, which is given a valid
sequent, there has (at least) one valid $\Dec$-normal form
(ii) each $\Dec$-normal form is valid if and only if it is derivable,
(iii) derivability is somehow ``back-preserved'', through 
$\Dec$-decomposition, which is, given a (set of) sequent(s)
which is $\Dec$-decomposed into another set of sequents,
if the $\Dec$-decomposed set is derivable,
then it is possible to construct a derivation for the original
(set of) sequent(s).

\bigskip
\textbf{$\Dec$ is strongly normalizing.}
First it is proved that $\Dec$ is strongly normalizing,
and so that each decomposition process terminates.
\begin{definition}[Number of Connectives, $\cn$]\label{def:cn}
The \emph{number of connectives of a labelled formula} $\lone$, 
$\cn(\lone)$, is the number of connectives of its logical part 
(equally labelled $\cn(\fone)$) and is inductively defined as follows:
\begin{align*}
\cn(\bone \Era \atom{n}) = \cn(\bone \Ela \atom{n}) = \cn(\atom{n}) &= 1 \\
\cn(\bone \Era \neg \fone) = \cn(\bone \Ela \neg \fone) = \cn(\neg\fone) &= 1 + \cn(\fone) \\ 
\cn(\bone \Era \fone^{*}\ftwo) = \cn(\bone \Ela \fone^{*}\ftwo) = \cn(\fone^{*}\ftwo) &= 1 + \cn(\fone) + \cn(\ftwo)\\
\cn(\bone \Era *\fone) = \cn(\bone \Ela *\fone) = \cn(*\fone) &= 1 + \cn(\fone)
\end{align*}
\normalsize
with $^{*} \in \{\wedge, \vee\}$ and $* \in \{\BOX^{q}, \DIA^{q}\}$.
Given a sequent $\vdash \lone$, the \emph{number of connectives of the sequent}, $\cn(\vdash \lone)$, is the 
number of connectives of its labelled formulas:
$$
\cn(\vdash \lone) = \cn(\lone).
$$
\end{definition}
\begin{definition}[Set Measure, $\ms$]\label{def:ms}
Given a set of sequents 
$\{\vdash \lone_{1}, \dots, \vdash \lone_{n}\},$ its \emph{measure}, 
$\ms(\{\vdash \lone_{1}, \dots, \vdash \lone_{n}\}$), 
is defined as follows:
\begin{align*}
\ms(\{\}) &= 0 \\
\ms\big(\{\vdash \lone_{1}, \dots, \vdash \lone_{n}\}\big) &= 3^{\cn(\vdash \lone_{1})} + \dots + 3^{\cn(\vdash \lone_{n})}.
\end{align*}
\end{definition}
\primeDec*
\begin{proof}
That every set of sequents is $\Dec$-strongly normalizable is proven by showing that, if $\{\vdash\lone_{1}, \dots, \vdash \lone_{m}\} \Dec \{\vdash \lone_{1}', \dots, \vdash \lone_{m}'\}$, then $\ms(\{\vdash \Delta_{1}, \dots, \vdash \Delta_{m}\}) > \ms(\{\vdash \Delta_{1}', \dots, \vdash \Delta_{m}'\})$. 
The proof is based on the exhaustive analysis of all
possible forms of $\Dec$-reduction applicable to the
given set, which is by dealing with all possible forms of
$\dec$-reduction of one of the $\vdash \lone_{i}$, with
$i\in\{1,\dots,m\}$ on which $\Dec$ is based.
Thus, we will take an arbitrary $\vdash\lone_{i}$ to be
the  ``active'' sequent of $\Dec$.
Let us consider all the possible forms of $\dec$ on which 
$\Dec$ can be based
\par Let $\lone_{} = \bone \Era \neg \fone$.
Assume that $\vdash\lone_{i}$ is the active sequent of the
$\Dec$-decomposition and that $\Dec$ is based on the following $\dec$:
$$
\vdash\bone \Era\neg \fone \ \ \ \dec \ \ \ \ 
\{\vdash \btwo \Ela \fone\}
$$
where $\bone\vDash\neg\btwo$. Thus:
\begin{align*}
\ms\big(\{\vdash \lone_{1},\dots,\vdash\btwo\Ela\fone,\dots,\vdash\lone_{m}\}\big) 
&\stackrel{\ref{def:cn}}{=}
3^{\cn(\lone_{1})}+\dots+3^{\cn(\btwo\Ela\fone)}
+\dots+3^{\cn(\lone_{n})} \\
&\stackrel{\ref{def:ms}}{=} 
3^{\cn(\lone_{1})}+\dots+3^{\cn(\fone)}+\dots+3^{\cn(\lone_{n})} \\
&\stackrel{\ref{def:ms}}{=} 
3^{\cn(\lone_{1})}+\dots+
3^{\cn(\neg\fone)–1}+\dots+3^{\cn(\lone_{n})} \\
&< 
3^{\cn(\lone_{1})}+\dots+3^{\cn(\bone\Era\neg\fone)}+\dots+3^{\cn(\lone_{n})} \\
&\stackrel{\ref{def:ms}}{=} 
\ms\big(\{\vdash\lone_{1},\dots,\vdash\bone\Era\neg\fone,\dots,\vdash \lone_{m}\}\big).
\end{align*}

\par $\lone_{i}=\bone\Ela\neg\fone$.
Assume that $\vdash\lone_{i}$ is the active sequent of the
$\Dec$-decomposition and that $\Dec$ is based on the following $\dec$:
$$
\vdash \bone \Ela\neg\fone \ \ \ \dec \ \ \ \{\vdash \btwo \Era \fone\}
$$
where $\neg\btwo\vDash\bone$. Thus,
\begin{align*}
\ms\big(\{\vdash\lone_{1},\dots,\vdash\btwo\Era\fone,\dots,\vdash\lone_{m}\}\big) 
&\stackrel{\ref{def:ms}}{=} 
3^{\cn(\lone_{1})}+\dots+3^{\cn(\btwo\Era\fone)}+\dots+3^{\cn(\lone_{n})} \\
&\stackrel{\ref{def:cn}}{=} 
3^{\cn(\lone_{1})}+\dots+3^{\cn(\fone)}+\dots+3^{\cn(\lone_{n})} \\
&\stackrel{\ref{def:cn}}{=}
3^{\cn(\lone_{1})}+\dots+3^{\cn(\neg\fone)–1}+\dots+3^{\cn(\lone_{n})} \\
&< 3^{\cn(\lone_{1})}+\dots+3^{\cn(\bone\Ela\neg\fone)}+\dots+3^{\cn(\lone_{n})} \\
&\stackrel{\ref{def:ms}}{=} \ms\big(\{\vdash\lone_{1},\dots,\vdash\bone\Ela\neg\fone,\dots,\vdash\lone_{m}\}\big).
\end{align*}

$\lone_{i}=\bone \Era \fone\vee\ftwo$. Assume that 
$\vdash\lone_{i}$ is the active sequent of $\Dec$ and that
$\Dec$ is based on the following $\dec$:
$$
\vdash\bone\Era\fone\vee\ftwo \ \ \ \dec \ \ \ \{\vdash\btwo\Era\fone,\vdash\bthree\Era\ftwo\}
$$
where $\bone\vDash\btwo\vee\bthree$. Then,
\begin{align*}
\ms\big(\{\vdash\lone_{i},\dots,\vdash\btwo\Era\fone,\vdash\bthree\Era\ftwo,\dots,\vdash\lone_{m}\}\big) 
&\stackrel{\ref{def:ms}}{=} 
3^{\cn(\lone_{1})}+\dots+3^{\cn(\btwo\Era\fone)}+3^{\cn(\bthree\Era\ftwo)}+\dots+3^{\cn(\lone_{m})} \\
&\stackrel{\ref{def:cn}}{=} 3^{\cn(\lone_{1})} + \dots + 3^{\cn(\fone)}
+ 3^{\cn(\fthree)} + \dots + 3^{\cn(\lone_{m})} \\
&< 3^{\cn(\lone_{1})}+ \dots + 3^{\cn(\fone)+\cn(\ftwo)+1} + \dots+ 3^{\cn(\lone_{m})} \\
&\stackrel{\ref{def:cn}}{=} 3^{\cn(\lone_{1})} +\dots+3^{\cn(\bone\Era\fone\vee\ftwo)}+\dots+3^{\cn(\lone_{m})} \\
&\stackrel{\ref{def:ms}}{=} \ms\big(\{\vdash\lone_{1},\dots,\vdash\bone\Era\fone\vee\ftwo,\dots,\vdash\lone_{m}\}\big).
\end{align*}

$\lone_{i}=\bone\Ela\fone\vee\ftwo$. 
Assume that $\vdash\lone_{i}$
is the active sequent of $\Dec$ and that $\Dec$
is based on the following $\dec$:
$$
\vdash \bone\Ela\fone\vee\ftwo \ \ \ \dec \ \ \ \
\{\vdash\bone\Ela\fone,\vdash\bone\Ela\ftwo\}
$$
Then,
\begin{align*}
\ms(\{\vdash\lone_{1},\dots,\vdash\bone\Ela\fone,
\vdash\bone\Ela\ftwo,\dots,\vdash\lone_{m}\}) 
&\stackrel{\ref{def:ms}}{=} 3^{\cn(\lone_{1})}
+\dots+3^{\cn(\bone\Ela\fone)}+3^{\cn(\bone\Ela\ftwo)}+
\dots+3^{\cn(\lone_{m})} \\
&\stackrel{\ref{def:cn}}{=} 
3^{\cn(\lone_{1})} + \dots + 3^{\cn(\fone)} + 3^{\cn(\ftwo)} + \dots + 3^{\cn(\lone_{m})} \\
&< 3^{\cn(\lone_{1})} + \dots + 3^{\cn(\fone)+\cn(\ftwo)+1} + \dots + 3^{\cn(\lone_{m})} \\
&\stackrel{\ref{def:cn}}{=} 
3^{\cn(\lone_{1})}+\dots + 3^{\cn(\bone\Ela\fone\vee\ftwo)}+\dots+3^{\cn(\lone_{m})} \\
&\stackrel{\ref{def:ms}}{=} 
\ms\big(\{\vdash\lone_{1},\dots,
\vdash\bone\Ela\fone\vee\ftwo,\dots,\vdash\lone_{m}\}\big)
\end{align*}

$\lone_{i}=\bone\Era\fone\wedge\ftwo$. 
The proof is equivalent to the one for $\bone\Ela\fone\vee\ftwo$.

$\lone_{i}=\bone\Ela\fone\wedge\ftwo$.
The proof is equivalent to the one for $\bone\Era\fone\vee\ftwo$.

$\lone=\bone\Era\BOX^{q}\fone$. 
Assume that $\vdash \lone_{i}$ is the active sequent of 
$\Dec$ decomposition and
that $\Dec$ is based on the following $\dec$:
$$
\vdash \bone\Era\BOX^{q}\fone \ \ \ \dec \ \ \ \{\vdash\btwo\Era\fone\}
$$
where $\mu(\model{\btwo})\ge q$. Then,
\begin{align*}
\ms\big(\{\vdash\lone_{1},\dots,\vdash\btwo\Era\fone,\dots,\vdash\lone_{m}\}\big) 
&\stackrel{\ref{def:ms}}{=}
3^{\cn(\lone_{1})}+\dots+3^{\cn(\btwo\Era\fone)}+\dots+
3^{\cn(\lone_{m})} \\
&\stackrel{\ref{def:cn}}{=} 3^{\cn(\lone_{1})} +
\dots + 3^{\cn(\fone)} + \dots + 3^{\cn(\lone_{m})} \\
&< 3^{\cn(\lone_{1})}+\dots +3^{\cn(\fone)+1}+\dots+
3^{\cn(\lone_{m})} \\
&\stackrel{\ref{def:cn}}{=} 3^{\cn(\lone_{1})}+\dots+
3^{\cn(\bone\Era\BOX^{q}\fone)+1}+\dots+3^{\cn(\lone_{m})} \\
&\stackrel{\ref{def:ms}}{=} \ms\big(\{\vdash\lone_{1},\dots,
\vdash\bone\Era\BOX^{q}\fone,\dots,\vdash\lone_{m}\}\big).
\end{align*}

$\lone_{i}=\bone\Ela\BOX^{q}\fone$, $\lone_{i}=\bone\Era\DIA^{q}\fone$,
$\lone_{i}=\bone\Ela\DIA^{q}\fone$.
The proof is equivalent to the one for 
$\bone\Era\BOX^{q}\fone$.

$\lone_{i}=\bone\Era\fone$.
Assume that $\vdash\lone_{i}$ is the active sequent of $\Dec$
decomposition and that $\Dec$ is based on the following
$\dec$:
$$
\vdash \bone \Era \fone \ \ \ \dec \ \ \ \{\}
$$
where $\mu(\model{\bone})=0$. 
Since for Definition~\ref{def:ms} and~\ref{def:cn},
$\cn(\{\})=0$ and $\cn(\vdash\bone\Ela\fone)>0$,
clearly $\cn(\{\})<\cn(\vdash\bone\Ela\fone)$.

$\lone_{i}=\bone\Ela\fone$, 
$\lone_{i}=\bone\Era\DIA^{0}\fone$,
$\lone_{i}=\bone\Ela\BOX^{0}\fone$.
The proof is equivalent to the other for
$\bone\Era\fone$.
\end{proof}

\bigskip
\textbf{$\Dec$-Normal Sequents.}
It is possible to show that $\Dec$-normal sequents are
valid if and only if they are derivable in $\GsCPLc$.
First, it is shown that sequents which are $\Dec$-normal
are regular. Then, it is proved that regular sequents 
are derivable if and only if they are valid.
Putting these two results together it is shown that
$\Dec$-normal sequents are valid only when they are derivable.
\begin{lemma}\label{lemma:NormalRegular}
If a (non-empty) sequent is $\Dec$-normal, then it is regular.
\end{lemma}
\begin{proof}
The proof is by contraposition. 
Given an arbitrary \emph{non-regular} sequent,
call it $\vdash\lone$, it is shown that it is not $\Dec$-normal.
Every possible form of lebelled formula $\lone$ is considered.

$\lone=\bone\Era\neg\fone$. 
Clearly, $\model{\bone}\subseteq\twoOm=\model{\neg\bot}$,
which is $\bone\vDash\neg\bot$.
Then, the following decomposition is well-defined:
$$
\vdash \bone\Era\neg\fone \ \ \ \dec \ \ \ \{\vdash\bot\Ela\fone\}.
$$

$\lone=\bone\Ela\neg\fone$.
Clearly, $\model{\neg\top}=\emptyset\subseteq\model{\bone}$,
which is $\neg\top \vdash\bone$.
Then, the following decomposition is well-defined:
$$
\vdash\bone\Ela\neg\fone \ \ \ \dec \ \ \ \{\vdash\top\Era \fone\}.
$$

$\lone=\bone\Era\fone\vee\ftwo$.
For every 
$\model{\bone}\subseteq\twoOm=\model{\top\vee\top}$,
which is $\bone\vDash(\top\vee\top)$.
Then, the following decomposition is well-defined:
$$
\vdash\bone\Era\fone\vee\ftwo \ \ \ \dec \ \ \ \{\vdash\top \Era\fone,\vdash\top\Era\ftwo\}.
$$

$\lone=\bone\Ela\fone\vee\ftwo$.
Then, the following decomposition is well-defined:
$$
\vdash\bone\Ela\fone\vee\ftwo \ \ \ \dec \ \ \ \{\vdash\bone\Ela\fone,\vdash\bone\Ela\ftwo\}.
$$

$\lone=\bone\Era\fone\wedge\ftwo$.
Then, the following decomposition is well-defined:
$$
\vdash\bone\Era\fone\wedge\ftwo \ \ \ \dec \ \ \ \{\vdash\bone\Era\fone,\vdash\bone\Era\ftwo\}.
$$

$\lone=\bone\Ela\fone\wedge\ftwo$.
For every 
$\bone,\model{\bot\wedge\bot}=\emptyset\subseteq\model{\bone}$,
which is $(\bot\wedge\bot)\vDash\bone$.
Then, the following decomposition is well-defined:
$$
\vdash \bone\Ela\fone\wedge\ftwo \ \ \ \dec \ \ \ \{\vdash\bot\Ela\fone,\vdash\bot\Ela\ftwo\}.
$$

$\lone=\bone\Era\BOX^{q}\fone$.
For every $q\in [0,1]$, 
$\mu(\model{\top})=\mu(\twoOm)\ge q$.
Then, the following decomposition is well-defined:
$$
\vdash\bone\Era\BOX^{q}\fone \ \ \ \dec \ \ \ \{\vdash\top\Era\fone\}.
$$

$\lone=\bone\Ela\BOX^{q}\fone$.
There are two possible cases, basing on the value of
$q\in[0,1]$:
\begin{itemize}
\item Let $q\neq0$. Then, 
$\mu(\model{\bot})=\mu(\emptyset)<q$ and the following
decomposition is well-defined:
$$
\vdash \bone\Ela\BOX^{q}\fone \ \ \ \dec \ \ \ \{\vdash\bot\Ela\fone\}.
$$
\item Let $q=0$. There are two possible sub-cases:
\begin{enumerate}
\item Let $\mu(\model{\bone})=1$. Then, the following
decomposition is well-defined:
$$
\vdash\bone\Ela\BOX^{0}\fone \ \ \ \dec \ \ \ \{\}.
$$
\item Let $\mu(\model{\bone})\neq 1$. Then, the following
decomposition is well-defined:
$$
\vdash\bone\Ela\BOX^{0}\fone \ \ \ \dec \ \ \ \{\vdash\bot\}.
$$
\end{enumerate}
\end{itemize}

$\lone=\bone\Era\DIA^{q}\fone$. There are two possible cases,
basing on the value of $q\in[0,1]$.
\begin{itemize}
\item Let $q\neq0$. Then $\mu(\model{\bot})=\mu(\emptyset)<q$
and the decomposition is well-defined:
$$
\vdash \bone\Era\DIA^{q}\fone \ \ \ \dec \ \ \ \{\vdash\bot\Ela\fone\}.
$$
\item Let $q=0$. There are two possible sub-cases:
\begin{enumerate}
\item Let $\mu(\model{\bone})=0$. Then, the following
decomposition is well-defined:
$$
\vdash\bone\Era\DIA^{0}\fone \ \ \ \dec \ \ \ \{\}.
$$
\item Let $\mu(\model{\bone})\neq0$. Then, the following
decomposition is well-defined:
$$
\vdash\bone\Era\DIA^{0}\fone \ \ \ \dec \ \ \ \{\vdash\bot\}.
$$
\end{enumerate}
\end{itemize}

$\lone=\bone\Ela\DIA^{q}\fone.$ For every $q\in[0,1]$,
$\mu(\model{\top\vee\top})=\mu(\twoOm)\ge q$.
Then, the following decomposition is well-defined:
$$
\vdash\bone\Ela\DIA^{q}\fone \ \ \ \dec \ \ \ \{\vdash\top\Era\fone\}.
$$
\end{proof}

\begin{lemma}\label{lemma:RegularValidDerivable}
A regular sequent is valid if and only if it is derivable in
$\GsCPLc$.
\end{lemma}
\begin{proof}
Let $\vdash\lone$ be a regular sequent, which is let $\lone$ be basic.
\begin{itemize}
\item[$\Rightarrow$] Assume that $\vdash \lone$ is valid.
Since $\lone$ is a basic formula, there are two main cases.
\par
$\lone=\bone\Era\atom{n}$, for arbitrary $\bone$ and $n\in\Nat$.
Since for hypothesis $\vDash\bone\Era\atom{n}$,
which is $\model{\bone}\subseteq\model{\atom{n}}$,
and $(\model{\atom{n}}=)\Cyl{n}=\bvar_{n}$, also
$\model{\bone}\subseteq\model{\bvar_{x}}$,
which is $\bone\vDash\bvar_{n}$. 
Therefore, $\vdash\bone\Era\atom{n}$ 
is derivable by means of $\AxO$:
\begin{prooftree}
\AxiomC{$\bone\vDash\bvar_{n}$}
\RightLabel{$\AxO$}
\UnaryInfC{$\vdash\bone\Era\atom{n}$}
\end{prooftree}
\par
$\lone=\bone\Ela\atom{n}$, for arbitrary $\bone$ and
$n\in\Nat$. Since for hypothesis $\vDash\bone\Ela\atom{n}$,
which is $\model{\atom{n}}\subseteq \model{\bone}$, and
$\model{\bvar_{n}}=\Cyl{n}(=\atom{n})$,
also $\model{\bvar_{n}}\subseteq\model{\bone}$,
which is $\bvar_{n}\vDash\bone$.
Therefore, $\vdash\bone\Ela\atom{n}$ is derivable by means
of $\AxT$:
\begin{prooftree}
\AxiomC{$\bvar_{n}\vDash\bone$}
\RightLabel{$\AxT$}
\UnaryInfC{$\vdash\bone\Ela\atom{n}$}
\end{prooftree}
\end{itemize}
Therefore, in both cases, $\vdash\lone$ is derivable in 
$\GsCPLc$.
\item[$\Leftarrow$] Assume that $\vdash\lone$ is derivable
in $\GsCPLc$. Then, by Proposition~\ref{soundnessAP}, $\vDash \lone$.
\end{proof}

\begin{corollary}\label{cor:NormalValidDerivable}
If a sequent is $\Dec$-normal, then it is valid
if and only if it is derivable in $\GsCPLc$.
\end{corollary}
\begin{proof}
By putting Lemma~\ref{lemma:NormalRegular} 
and~\ref{lemma:RegularValidDerivable} together.
\end{proof}
\noindent
All the given results can be naturally extended from sequents
to sets, obtaining in particular that,  
if a set of sequents is $\Dec$-normal, 
which is each of its sequents is $\Dec$-normal, 
then it is valid if and only if it is derivable,
namely its sequents are valid if and only if they are 
derivable.

\bigskip
\textbf{Existential Preservation of Validity.} 
It is possible to prove that the validity is existentially preserved
through $\Dec$-decomposition.
It is first needed to establish the following auxiliary lemma.
\begin{lemma}\label{label}
For every formula of $\PPLc$, call it $\fone$, 
there is an expression $\bone_{\fone}$, 
such that $\model{\fone} = \model{\bone_{\fone}}$.
\end{lemma}
\begin{proof}
The proof is by induction on the structure of $\fone$:
\begin{itemize}
\item \emph{Base case.} 
Let $\fone$ be atomic, which is $\fone \equiv \atom{n}$, 
for some $n\in\Nat$. 
By Definition~\ref{Bool}, $\model{\atom{n}}$ = $\Cyl{n}$, 
but also $\Cyl{n}$ = $\model{\bvar_{n}}$. 
Therefore, let us define $\bone_{\fone} = \bvar_{n}$, 
so that clearly $\model{\atom{n}}$ = $\model{\bvar_{n}}$.
\item \emph{Inductive case.} There are some possible cases.
\par $\fone = \neg\ftwo$. 
By IH, there is a $\bone_{\ftwo}$ such that $\model{\ftwo} = \model{\bone_{\ftwo}}$. 
So, let us define $\bone_{\fone}$ = $\neg \bone_{\ftwo}$, 
which is well-defined by Definition~\ref{Bool}. 
Indeed, $\model{\bone_{\fone}}$ = $\model{\neg \bone_{\ftwo}}$ = $\twoOm \ – \ \model{\bone_{\ftwo}}$ $\stackrel{\text{IH}}{=}$ $\twoOm \ – \ \model{\ftwo}$ = $\model{\neg\ftwo}$.
\par $\fone$ = $\ftwo \vee \fthree$. 
By IH, there are two $\bone_{\ftwo}$ and $\bone_{\fthree}$ such that $\model{\bone_{\ftwo}}$ = $\model{\ftwo}$ and $\model{\bone_{\fthree}}$ = $\model{\fthree}$. 
Thus, let us define $\bone_{\fone}$ = $\bone_{\ftwo}$ $\vee$ $\bone_{\fthree}$, which is well-defined by Definition~\ref{Bool}.
Indeed, as desired, $\model{\bone_{\fone}}$ = $\model{\bone_{\ftwo} \vee \bone_{\fthree}}$ = $\model{\bone_{\ftwo}} \cup \model{\bone_{\fthree}}$ $\stackrel{\text{IH}}{=}$ $\model{\ftwo} \cup \model{\fthree}$ = $\model{\ftwo \vee \fthree}$.
\par $\fone$ = $\ftwo \wedge \fthree$. By IH, there are two $\bone_{\ftwo}$ and $\bone_{\fthree}$ such that $\model{\bone_{\ftwo}}$ = $\model{\ftwo}$ and $\model{\bone_{\fthree}}$ = $\model{\fthree}$. Let us define $\bone_{\fone}$ = $\bone_{\ftwo}$ $\wedge$ $\bone_{\fthree}$, which is well-defined by Definition~\ref{Bool}. 
Indeed, $\model{\bone_{\fone}}$ = $\model{\bone_{\ftwo} \wedge \bone_{\fthree}}$ = $\model{\bone_{\ftwo}} \cap \model{\bone_{\fthree}}$ $\stackrel{\text{IH}}{=}$ $\model{\ftwo} \cap \model{\fthree}$ = $\model{\ftwo \wedge \fthree}$.
\par $\fone$ = $\BOX^{q}\ftwo$ or $\fone$ = $\DIA^{q}\ftwo$. In both cases, either $\model{\fone}$ = $\twoOm$ or $\model{\fone}$ = $\emptyset$. 
In the former case, let us define 
$\bone_{\fone}$ = $\top$, as $\model{\top}$ = $\twoOm$ = $\model{\fone}$. 
In the latter case, let us define $\bone_{\fone}$ = $\bot$, since $\model{\bot}$ = $\emptyset$ = $\model{\fone}$.
\end{itemize}
\end{proof}

\begin{lemma}[Existential Preservation of Validity]\label{ExVal}
Each valid sequent has a valid $\Dec$-normal form.
\end{lemma}
\begin{proof}
Let $\vdash\lone$ be an arbitrary, \emph{valid} sequent.
It is shown that $\vdash\lone$ has a valid $\Dec$-normal form.
The proof is by exhaustive inspection. There are two main cases.
\begin{itemize}
\item Let $\vdash\lone$ be such that no
$\dec$-reduction can be applied on it. 
Then, for Lemma~\ref{lemma:NormalRegular}, 
the sequent is either empty or regular, 
which is $\lone$ is a basic formula.
In both cases, the sequent is trivially $\Dec$-normal and,
for hypothesis, valid.
\item Let $\vdash\lone$ be $\Dec$-reducible. 
By Lemma~\ref{theorem:normalizing}, there is no infinite
reduction sequence.
Thus, it is sufficient to prove that each \emph{reduction step}
existentially preserves validity, which is that for every possible
$\Dec$, based on every possible $\dec$ reduction, 
if $\vDash \lone$, then there is a $\dec$ such that 
$\vdash\lone\dec\ssetO$, and $\ssetO$ is valid.
The proof is based on the exhaustive consideration
of all possible $\dec$ form.
\par $\lone=\bone\Era\neg\fone$. 
Let us consider a Boolean formula
defined as $\btwo=\neg\bone$. 
Thus, $\model{\neg\btwo}=\twoOm$ $–$ 
$\model{\btwo}=\twoOm$ $–$ 
($\twoOm$ $–$ $\model{\bone}$) = 
$\model{\bone}$ and so, in
particular $\bone\vDash\neg\btwo$.
Let us consider the following well-defined
$\dec$-decomposition (given $\bone\vDash\neg\btwo$):
$$
\vdash \bone\Era\neg\fone \ \ \ \dec \ \ \ \{\vdash\btwo\Ela\fone\}.
$$
For hypothesis $\bone\Era\neg\fone$, which is
$\model{\bone}\subseteq\model{\neg\fone}$.
By basic set theory, 
also $\model{\fone}\subseteq\model{\neg\bone}$.
Since for construction $\btwo=\neg\bone$,
$\model{\fone}\subseteq\model{\btwo}$ holds, which is
$\vDash\btwo\Ela\fone$, as desired. 
Therefore, there is a $\dec$-decomposition 
which preserves validity.
\par $\lone=\bone\Ela\neg\fone$. 
Let us consider a Boolean formula defined as 
$\btwo=\neg\bone$.
Thus, $\model{\btwo}=\model{\neg\bone}$ and, in particular,
$\neg\btwo\vDash\bone$.
Let us consider the following $\dec$-decomposition,
which is well-defined (given $\neg\btwo \vDash\bone$):
$$
\vdash\bone\Ela\neg\fone \ \ \ \dec \ \ \ \{\vdash\btwo\Era\fone\}.
$$
For hypothesis $\bone\Ela\neg\fone$, which is
$\model{\neg\fone}\subseteq\model{\bone}$, so for basic set
theory $\model{\neg\bone}\subseteq\model{\fone}$.
Since for construction $\model{\btwo}=\model{\neg\bone}$,
also $\model{\btwo}\subseteq\model{\fone}$, which
is $\vDash\btwo\Era\fone$ and 
the described $\dec$-decomposition is as desired.

$\lone=\bone\Era\fone\vee\ftwo$. Let us
consider two $\bone_{\fone}$ and $\bone_{\ftwo}$ such
that $\model{\bone_{\fone}}=\model{\fone}$ and 
$\model{\bone_{\ftwo}}=\model{\ftwo}$ (which must exists
by Lemma~\ref{label}).
For hypothesis $\vDash \bone\Era\fone\vee\ftwo$, which
is $\model{\bone}\subseteq\model{\fone}\cup\model{\ftwo}$.
Thus, for construction, also 
$\model{\bone}\subseteq\model{\bone_{\fone}}\cup\model{\bone_{\ftwo}}$,
which is $\bone\vDash\bone_{\fone}\vee\bone_{\ftwo}$.
Let us consider the following derivation, which
is well-defined (given $\bone\vDash\bone_{\fone}\vee\bone_{\ftwo}$):
$$
\vdash \bone\Era\fone\vee\ftwo \ \ \ \dec 
\ \ \ \{\vdash\bone_{\fone}\Era\fone,\vdash\bone_{\ftwo}\Era\ftwo\}.
$$
For construction $\model{\bone_{\fone}}=\model{\fone}$
and $\model{\bone_{\ftwo}}=\model{\ftwo}$ so, in 
particular, $\model{\bone_{\fone}}\subseteq\model{\fone}$
and $\model{\bone_{\ftwo}}\subseteq\model{\ftwo}$.
Therefore, $\vDash \bone_{\fone}\Era\fone$ and
$\vDash\bone_{\ftwo}\Era\ftwo$ and the given 
$\dec$-decomposition is as desired.

\par $\lone=\bone\Ela\fone\vee\ftwo$. 
Let us consider the following, well-defined 
$\dec$-decomposition:
$$
\vdash\bone\Ela\fone \vee \ftwo \ \ \ \dec \ \ \ \{\vdash\bone\Ela\fone,\vdash\bone\Ela\ftwo\}.
$$
For hypothesis $\vDash\bone\Ela\fone\vee\ftwo$, which
is $\model{\fone}\cup\model{\ftwo}\subseteq\model{\bone}$.
Then, by basic set theory, both $\model{\fone}\subseteq\model{\bone}$
and $\model{\ftwo}\subseteq\model{\bone}$,
which is $\vDash\bone\Ela\fone$ and $\vDash\bone\Ela\ftwo$,
as desired.
\par $\lone=\bone\Era\fone\wedge\ftwo$. 
The proof is equivalent
to the one for $\bone\Ela\fone\vee\ftwo$.
\par $\lone=\bone\Ela\fone\wedge\ftwo$. 
The proof is equivalent
to the one for $\bone\Era\fone\vee\ftwo$.
\par $\lone=\bone\Era\BOX^{q}\fone$. 
There are two main cases:
\begin{enumerate}
\item Let $\mu(\model{\bone})=0$. Then, the sequent
can be decomposed by means of the following, well-defined
$\dec$-decomposition:
$$
\vdash\bone\Era\BOX^{q}\fone \ \ \ \dec \ \ \ \{\}.
$$
$\{\}$ is vacuously valid so the given decomposition
is as desired.
\item Let $\mu(\model{\bone})\neq0$
For hypothesis $\vDash\bone\Era\BOX^{q}\fone$,
which is 
$\model{\bone}\subseteq\model{\BOX^{q}\fone}$.
Since $\model{\bone}\neq\emptyset$, also
$\model{\BOX^{q}\fone}\neq\emptyset$,
which is $\model{\BOX^{q}\fone}=\twoOm$,
and so $\mu(\model{\fone})\ge q$.
Let us consider a Boolean formula $\bone_{\fone}$,
such that $\model{\bone_{\fone}}=\model{\fone}$
(it can be constructed due to Lemma~\ref{label}).
Thus, $\mu(\model{\bone_{\fone}})\ge q$ and the
following decomposition is well-defined:
$$
\vdash\bone\Era\BOX^{q}\fone \ \ \ \dec \ \ \ 
\{\bone_{\fone} \Era\fone\}.
$$
Since for construction 
$\model{\bone_{\fone}}=\model{\fone}$,
in particular also 
$\model{\bone_{\fone}}\subseteq\model{\fone}$,
which is $\vDash \bone_{\fone}\Era\fone$ as
desired. 
\end{enumerate}
\par $\lone=\bone\Ela\BOX^{q}\fone.$ 
There are two main cases:
\begin{enumerate}
\item Let $\mu(\model{\bone})=1$.
Then, the sequent can be decomposed by
means of the following
well-defined $\dec$-decomposition:
$$
\vdash\bone\Ela\BOX^{q}\fone \ \ \ \dec \ \ \ \{\}.
$$
$\{\}$ is vacuously valid,
so the given decomposition is as desired.
\item Let $\mu(\model{\bone})\neq1$. 
For
hypothesis $\vDash\bone\Ela\BOX^{q}\fone$,
which is 
$\model{\BOX^{q}\fone}\subseteq\model{\bone}$.
Since $\model{\bone}\neq\twoOm$, 
$\model{\BOX^{q}\fone}=\emptyset$
and so $\mu(\model{\fone})<q$.
Let us consider a Boolean formula
$\bone_{\fone}$ such that 
$\model{\bone_{\fone}}=\model{\fone}$ (it
can be constructed due to Lemma~\ref{label}).
Thus, $\mu(\model{\bone_{\fone}})<q$, 
and the following decomposition is well-defined:
$$
\vdash\bone\Ela\BOX^{q}\fone \ \ \ \dec \ \ \ 
\{\bone_{\fone} \Ela\fone\}.
$$
Since for construction 
$\model{\bone_{\fone}}=\model{\fone}$,
in particular also 
$\model{\fone}\subseteq\model{\bone}_{\fone}$,
which is $\vDash\bone_{\fone}\Ela\fone$
as desired.
\end{enumerate}
\par 
$\lone=\bone\Era\DIA^{q}\fone$
or $\lone=\bone\Ela\DIA^{q}\fone$.
The proof is equivalent to the ones for 
$\bone\Ela\BOX^{q}\fone$ and 
$\bone\Era\BOX^{q}\fone$.
\par 
$\lone=\bone\Era\fone$ and $\mu(\model{\bone})=0$.
Then, the following decomposition is well-defined:
$$
\vdash\bone\Era\fone \ \ \ \dec \ \ \ \{\}.
$$
$\{\}$ is vacuously valid so the given decomposition
is as desired.
\par 
$\lone=\bone\Era\fone$ and 
$\mu(\model{\bone})=1$.
Then, the following decomposition
is well-defined:
$$
\vdash\bone\Ela\fone \ \ \ \dec \ \ \ \{\}.
$$
$\{\}$ is vacuously valid so the given
decomposition is as desired.
\end{itemize}
\noindent
There is no other possible case. 
(Notice that if 
$\lone=\bone\Era\DIA^{0}\fone$, for hypothesis,
$\vDash\bone\Era\DIA^{0}\fone$,
which is 
$\model{\bone}\subseteq\model{\DIA^{0}\fone}$.
But clearly $\model{\DIA^{0}\fone}=\emptyset$,
so $\model{\bone}=\emptyset$ and
$\mu(\model{\bone})=0$; equally when 
$\lone=\bone\Ela\BOX^{0}\fone$, 
$\mu(\model{\bone}$)=1.)
\end{proof}

\bigskip
\textbf{Derivability Preservation.}
As anticipated, properties and proofs concerning sequents can be easily generalized to corresponding proofs about sets of sequents. Specifically, by considering together Lemma~\ref{lemma:NormalRegular} and 
Lemma~\ref{lemma:RegularValidDerivable}, 
it is established a result, corresponding
to Corollary~\ref{cor:NormalValidDerivable},
namely that a $\dec$-normal set of sequents is derivable 
if and only if it is valid is easily proved. 
\begin{definition}[Derivable Set]\label{derSet}
A set of sequents is \emph{derivable} if and only if all of its sequents 
are derivable in $\GPPLc$. 
\end{definition}
\begin{lemma}\label{DerPres}
Given two sets of sequent $\ssetO,\ssetO'$, 
if $\ssetO \Dec \ssetO'$ and $\ssetO'$ 
is derivable in $\GsCPLc$, 
then $\ssetO$ is derivable in $\GsCPLc$ as well. 
\end{lemma}
\begin{proof}
Assume $\ssetO\Dec\ssetO'$. 
Then, there is a $\vdash\lone\in\ssetO$ such
that it is the ``active'' sequent on which
$\Dec$ is based, which is 
$\Dec$ is based on the following $\dec$-decomposition:
$$
\vdash\lone \ \ \ \dec \ \ \
 \{\vdash\lone_{1},\dots,\vdash\lone_{n}\}.
$$
The proof is by inspection of all possible
forms of $\dec$.
\par 
Let $\lone=\bone\Era\neg\fone$ and
its $\dec$-decomposition be as follows:
$$
\vdash\bone\Era\neg\fone \ \ \ \dec 
\ \ \ \{\vdash\btwo\Ela\fone\}
$$
where $\bone\vDash\neg\btwo$.
Since $\vdash\btwo\Ela\fone \in \ssetO'$,
for hypothesis $\vdash\btwo\Ela\fone$
is derivable by a derivation, call
it $\mathcal{D}$.
Therefore, $\vdash\bone\Era\neg\fone$
can be shown derivable in $\GPPLc$ as follows:
\begin{prooftree}
\AxiomC{$\mathcal{D}$}
\noLine
\UnaryInfC{$\vdash\btwo\Ela\fone$}
\AxiomC{$\bone\vDash\neg\btwo$}
\RightLabel{$\Rnra$}
\BinaryInfC{$\vdash\bone\Era\neg\fone$}
\end{prooftree}

\par
Let $\lone=\bone\Ela\neg\fone$ and its $\dec$-decomposition
be as follows:
$$
\vdash\bone\Ela\neg\fone \ \ \ \dec \ \ \ 
\{\vdash\btwo \Era\fone\}
$$
where $\neg\btwo\vDash\bone$.
Since $\vdash\btwo\Era\fone\in\ssetO'$,
for hypothesis $\vdash\btwo\Era\fone$
is derivable by a derivation, call it
$\mathcal{D}$.
Therefore, $\vdash\bone\Ela\neg\fone$ can
be shown derivable in $\GPPLc$ as follows:
\begin{prooftree}
\AxiomC{$\mathcal{D}$}
\noLine
\UnaryInfC{$\vdash\btwo\Era\fone$}
\AxiomC{$\neg\btwo\vDash\bone$}
\RightLabel{$\Rnla$}
\BinaryInfC{$\vdash\bone\Ela\neg\fone$}
\end{prooftree}

\par
Let $\lone=\bone\Era\fone\vee\ftwo$ and
$\dec$-decomposition be as follows:
$$
\vdash \bone\Era\fone\vee\ftwo \ \ \ \dec \ \ \
\{\vdash \btwo \Era\fone, \vdash\bthree\Era\ftwo\}
$$
where $\bone\vDash\btwo\vee \bthree$.
Since 
$\vdash\btwo\Era\fone,\vdash\bthree\Era\ftwo
\in \ssetO'$, both
$\vdash\btwo\Era\fone$ and 
$\vdash\bthree\Era\ftwo$ are
derivable by two derivations,
call them $\mathcal{D}$ and $\mathcal{D}'$.
Then, $\bone\Era\fone\vee\ftwo$ can be
shown derivable in $\GPPLc$ as follows:
\begin{prooftree}
\AxiomC{$\mathcal{D}$}
\noLine
\UnaryInfC{$\vdash\btwo\Era\fone$}
\RightLabel{$\mathsf{R1}_{\vee}^{\Era}$}
\UnaryInfC{$\vdash\btwo\Era\fone\vee\ftwo$}

\AxiomC{$\mathcal{D}'$}
\noLine
\UnaryInfC{$\vdash\bthree\Era\ftwo$}
\RightLabel{$\mathsf{R2}_{\vee}^{\Era}$}
\UnaryInfC{$\vdash\bthree\Era\fone\vee\ftwo$}

\AxiomC{$\bone\vDash\btwo\vee\bthree$}
\RightLabel{$\Rcup$}
\TrinaryInfC{$\vdash\bone\Era\fone\vee\ftwo$}
\end{prooftree}

\par
Let $\lone=\bone\Ela\fone\vee\ftwo$ and
$\dec$-decomposition be as follows:
$$
\vdash\bone\Ela\fone\vee\ftwo \ \ \ \dec \ \ \
\{\vdash\bone\Ela\fone, \vdash\bone\Ela\ftwo\}.
$$
Since 
$\vdash\bone\Ela\fone,\vdash\bone\Ela\ftwo\in\ssetO'$,
both $\vdash\bone\Ela\fone$ and
$\vdash\bone\Ela\ftwo$
are derivable by two derivations,
call them $\mathcal{D}$ and $\mathcal{D}'$.
Then $\bone\Ela\fone\vee\ftwo$ can
be shown derivable in $\GPPLc$ as follows:
\begin{prooftree}
\AxiomC{$\mathcal{D}$}
\noLine
\UnaryInfC{$\vdash\bone\Ela\fone$}
\AxiomC{$\mathcal{D}'$}
\noLine
\UnaryInfC{$\vdash\bone\Ela\ftwo$}
\RightLabel{$\Rvla$}
\BinaryInfC{$\vdash\bone\Ela\fone\vee\ftwo$}
\end{prooftree}

\par 
Let $\lone=\bone\Era\fone\wedge\ftwo$ and
$\dec$-decomposition be as follows:
$$
\vdash\bone\Era\fone \wedge \ftwo \ \ \ \dec \ \ \ \{\vdash\bone\Era\fone,\vdash\bone\Era\ftwo\}.
$$
Since $\vdash\bone\Era\fone,\vdash\bone\Era\ftwo\in\ssetO'$,
both $\vdash\bone\Era\fone$ and $\vdash\bone\Era\ftwo$ are
derivable by two derivations,
call them $\mathcal{D}$ and $\mathcal{D}'$.
Then, $\bone\Era\fone\wedge\ftwo$ can 
be shown derivable in $\GPPLc$ as follows:
\begin{prooftree}
\AxiomC{$\mathcal{D}$}
\noLine
\UnaryInfC{$\vdash\bone\Era\fone$}
\AxiomC{$\mathcal{D}'$}
\noLine
\UnaryInfC{$\vdash\bone\Era\ftwo$}
\RightLabel{$\Rvla$}
\BinaryInfC{$\vdash\bone\Era\fone\wedge\ftwo$}
\end{prooftree}

\par
Let $\lone=\bone\Ela\fone\wedge\ftwo$ and $\dec$-decomposition
be as follows:
$$
\vdash\bone\Ela\fone\wedge\ftwo \ \ \ \dec \ \ \ \
\{\vdash\btwo\Ela\fone, \vdash
\bthree\Ela\ftwo\}
$$
where $\btwo\wedge\bthree\vDash\bone$.
Since $\vdash\btwo\Ela\fone,\vdash\bthree\Ela\ftwo \in
\ssetO'$,
both $\vdash\btwo\Ela\fone$
and $\vdash\bthree\Ela\ftwo$
are derivable by two derivations, call them $\mathcal{D}$
and $\mathcal{D}'$.
Then, $\bone \Ela\fone\wedge\ftwo$ can be shown
derivable in $\GPPLc$ as follows:
\begin{prooftree}
\AxiomC{$\mathcal{D}$}
\noLine
\UnaryInfC{$\vdash\btwo\Ela\fone$}
\RightLabel{$\mathsf{R1}_{\wedge}^{\Ela}$}
\UnaryInfC{$\vdash\btwo\Ela\fone\wedge\fthree$}
\AxiomC{$\mathcal{D}'$}
\noLine
\UnaryInfC{$\vdash\bthree\Ela\ftwo$}
\RightLabel{$\mathsf{R2}_{\wedge}^{\Ela}$}
\UnaryInfC{$\vdash\bthree\Ela\fone\wedge\ftwo$}
\AxiomC{$\btwo\wedge\bthree\vDash\bone$}
\RightLabel{$\Rcap$}
\TrinaryInfC{$\vdash\bone\Ela\fone\wedge\ftwo$}
\end{prooftree}

\par
Let $\lone=\bone\Era\BOX^{q}\fone$ and
$\dec$-decomposition be as follows:
$$
\vdash\bone\Era\BOX^{q}\fone \ \ \ \dec \ \ \ 
\{\vdash\btwo\Era\fone\}
$$
where $\mu(\model{\btwo})\ge q$. 
Since $\vdash\btwo\Era\fone \in \ssetO'$,
$\vdash\btwo\Era\fone$ is derivable by a derivation,
call it $\mathcal{D}$.
Then, $\vdash\bone\Era\BOX^{q}\fone$ can
be shown derivable in $\GPPLc$ as follows:
\begin{prooftree}
\AxiomC{$\mathcal{D}$}
\noLine
\UnaryInfC{$\vdash\btwo\Era\fone$}
\AxiomC{$\mu(\model{\btwo})\ge q$}
\RightLabel{$\Rbr$}
\BinaryInfC{$\vdash\bone\Era\BOX^{q}\fone$}
\end{prooftree}

Let $\lone=\bone\Ela\BOX^{q}\fone$ and $\dec$-decomposition
be as follows:
$$
\vdash\bone\Ela\BOX^{q}\fone \ \ \ \dec \ \ \ \{\vdash\btwo\Ela\fone\}
$$
where $\mu(\model{\btwo})<q$.
Since $\vdash\btwo \Ela\fone \in \ssetO'$,
$\vdash\btwo\Ela\fone$ is derivable by a derivation,
call it $\mathcal{D}$.
Then, $\vdash\bone\Ela\BOX^{q}\fone$ can be
shown derivable in $\GPPLc$ as follows:
\begin{prooftree}
\AxiomC{$\mathcal{D}$}
\noLine
\UnaryInfC{$\vdash\btwo\Ela\fone$}
\AxiomC{$\mu(\model{\btwo})<q$}
\RightLabel{$\Rbl$}
\BinaryInfC{$\vdash\bone\Ela\BOX^{q}\fone$}
\end{prooftree}

\par 
Let $\lone=\bone\Era\DIA^{q}\fone$ and
$\dec$-decomposition be as follows:
$$
\vdash\bone\Era\DIA^{q}\fone \ \ \ \dec \ \ \ 
\{\vdash\btwo\Ela\fone\}
$$
where $\mu(\model{\btwo})<q$.
Since $\vdash\btwo\Ela\fone\in\ssetO'$,
$\vdash\btwo\Ela\fone$ is derivable by a derivation,
call it $\mathcal{D}$.
Then, $\vdash\bone\Era\DIA^{q}\fone$
 can be shown derivable in $\GPPLc$ as follows:
 \begin{prooftree}
 \AxiomC{$\mathcal{D}$}
 \noLine
 \UnaryInfC{$\vdash\btwo\Ela\fone$}
 \AxiomC{$\mu(\model{\btwo})<q$}
 \RightLabel{$\Rdr$}
 \BinaryInfC{$\vdash\bone\Era\DIA^{q}\fone$}
 \end{prooftree}

Let $\lone=\bone\Ela\DIA^{q}\fone$ and 
$\dec$-decomposition
be as follows:
$$
\vdash\bone\Ela\DIA^{q}\fone \ \ \ \dec \ \ \ \{\vdash\btwo\Era\fone\}
$$
where $\mu(\model{\btwo})\ge q$.
Since $\vdash\btwo\Era\fone\in\ssetO'$,
$\vdash\btwo\Era\fone$ is derivable by a derivation,
call it $\mathcal{D}$.
Then, $\vdash\bone\Ela\DIA^{q}\fone$ can
be shown derivable in $\GPPLc$ as follows:
\begin{prooftree}
\AxiomC{$\mathcal{D}$}
\noLine
\UnaryInfC{$\vdash\btwo\Era\fone$}
\AxiomC{$\mu(\model{\btwo})\ge q$}
\RightLabel{$\Rdl$}
\BinaryInfC{$\vdash\bone\Ela\DIA^{q}\fone$}
\end{prooftree}

 \par 
 Let $\lone=\bone\Era\fone$ and $\dec$-decomposition
 be as follows
 $$
 \vdash\bone\Era\fone \ \ \ \dec \ \ \ \{\}
 $$
 where $\mu(\model{\btwo})=0$.
 Then, $\bone\Era\fone$ can be shown derivable
 in $\GPPLc$
 as follows:
 \begin{prooftree}
 \AxiomC{$\mu(\model{\bone})=0$}
 \RightLabel{$\Rmur$}
 \UnaryInfC{$\vdash\bone\Era\fone$}
 \end{prooftree}
 
 \par
 Let $\lone=\bone\Ela\fone$ and $\dec$-decomposition
 be as follows:
 $$
 \vdash\bone\Ela\fone \ \ \ \dec \ \ \ \{\}
 $$
where $\mu(\model{\btwo})=1$.
Then, $\bone\Ela\fone$ can be shown derivable in
$\GPPLc$ as follows:
\begin{prooftree}
\AxiomC{$\mu(\model{\bone})=1$}
\RightLabel{$\Rmul$}
\UnaryInfC{$\vdash\bone\Ela\fone$}
\end{prooftree}

Notice that if $\vdash\bone\Era\DIA^{0}\fone$ is
derivable, by Lemma~\ref{soundnessAP}
$\vDash\bone\Era\DIA^{0}\fone$, which is 
$\model{\bone}\subseteq \model{\DIA^{0}\fone}$.
But $\model{\DIA^{0}\fone}=0$,
so $\model{\bone}\subseteq \emptyset$,
which is $\model{\bone}=0$.
Thus, when $\vdash\bone\Era\DIA^{0}\fone$ is derivable,
it cannot be decomposed as follows:
$\vdash\bone\Era\DIA^{0}\fone\dec\{\vdash\bot\}$.
The same holds for $\bone\Ela\BOX^{0}\fone$.
Therefore, all the possible forms of $\dec$ have
been considered.
\end{proof}

\bigskip
\textbf{The Completeness of $\GsCPLc$.}
Putting these results together, it is possible to 
conclude that $\GsCPLc$ is complete with respect
to the semantics of $\PPLc$.
\primeCompleteness*
\begin{proof}
If the sequent is valid, by Lemma \ref{ExVal}, it has a \emph{valid} $\Dec$-normal form. 
For Lemma \ref{NormValDer}, a $\Dec$-normal form is valid if and only if it is derivable, so the given $\Dec$-normal form must be derivable as well. 
Therefore, by Lemma \ref{DerPres}, the given (valid) sequent must be derivable in $\GPPLc$. 
\end{proof}

\subsection{$\PPLc$-validity is in $\Psharp$}\label{appendix3.3}

\begin{lemma}\label{app3.3}
For every \emph{closed} formula of $\PPLc$, 
call it $\fone$, either $\model{\fone}$ = $\twoOm$ or $\model{\fone}$ = $\emptyset$.
\end{lemma}
\begin{proof}
The proof is by induction on the structure of $\fone$.
\begin{itemize}
\item \emph{Base case.} If $\fone$ = $\BOX^{q}\ftwo$ or $\fone$ = $\DIA^{q}\ftwo$, then for Definition~\ref{semantics} either $\model{\fone}$ = $\twoOm$ or $\model{\fone}$ = $\emptyset$. 
\item \emph{Inductive case}. There are four possible cases:
\par 
$\fone = \neg\ftwo$. 
Since $\ftwo$ is a closed formula, for IH, either $\model{\ftwo}$ = $\twoOm$ or $\model{\ftwo}$ = $\emptyset$. 
In the former case, 
$$
\model{\fone} = \big(\twoOm \ – \ \model{\ftwo}\big) 
= \big(\twoOm \ – \ \twoOm\big)
= \emptyset.
$$
In the latter case, 
$$
\model{\fone} = \big(\twoOm \ – \ \model{\ftwo}\big) = 
\big(\twoOm \ – \ \emptyset\big) = \twoOm.
$$

\par 
$\fone = \ftwo \vee \fthree$. 
Since $\ftwo$ and $\fthree$ are closed formulas,
 for IH, both $\model{\ftwo}$ and $\model{\fthree}$ 
 are either $\twoOm$ or $\emptyset$. 
If both $\model{\ftwo}$ and $\model{\fthree}$ are $\emptyset$,
then
$$ 
\model{\fone} = \model{\ftwo \vee \fthree} = 
\model{\ftwo} \cup \model{\fthree} = \emptyset. 
$$
Otherwise, 
$$ 
\model{\fone} = \model{\ftwo \vee \fthree} = 
\model{\ftwo} \cup \model{\fthree} = 
\twoOm.
$$

\par 
$\fone = \ftwo \wedge \fthree$. 
Since $\ftwo$ and $\fthree$ are closed formulas, 
for IH, both $\model{\ftwo}$ and $\model{\fthree}$ 
are either $\twoOm$ or $\emptyset$. 
Thus, if both $\model{\ftwo}$ and $\model{\fthree}$ are 
$\twoOm$, then:
$$
\model{\fone} = \model{\ftwo \wedge \fthree} =  
\model{\ftwo} \cap \model{\fthree} = \twoOm.
$$
Otherwise,
$$
\model{\fone} = \model{\ftwo \wedge \fthree} =  
\model{\ftwo} \cap \model{\fthree} = \emptyset.
$$

\par 
$\fone = \BOX^{q}\ftwo$ or $\fone = \DIA^{q}\ftwo$. 
Then the proof is trivial for Definition~\ref{semantics}
of $\BOX$- and $\DIA$-formulas.
\end{itemize}
\end{proof}

\section{Proofs from Section \ref{section4}}\label{appendix4}

\subsection{Characterizing the Semantics of $\PPL$ using Boolean Formulas}\label{appendix4.1}

\begin{definition}[Named Boolean Formulas]\label{booleanFormulas}
\emph{Named Boolean formulas} are defined by the following grammar:
$$
\bone,\btwo:= \bvar_{i}^{a}\midd \top\midd \bot\midd \lnot \bone\midd \bone\land \btwo\midd \bone\vee \btwo
$$
The interpretation $\model \bone_{X}$ of a Boolean formula with
$\FN(\bone)\subseteq X$ is defined in a straightforward way following Definition \ref{PPLsem}.

\begin{minipage}{\linewidth}
\begin{minipage}[t]{0.5\linewidth}
\begin{align*}
\model{\bvar_{i}^{a}}_{X} &:= \{f : X \rightarrow \twoOm \ | \ f(a)(\atom{i}) = 1\} \\
\model{\top}_{X} &:= (\twoOm)^{X} \\
\model{\bot}_{X} &:= \emptyset^{X} 
\end{align*}
\end{minipage}
\hfill
\begin{minipage}[t]{0.5\linewidth}
\begin{align*}
\model{\neg \bone}_{X} &:= (\twoOm)^{X} – \ \model{\bone}_{X} \\
\model{\bone \wedge \btwo}_{X} &:= \model{\bone}_{X} \cap \model{\btwo}_{X} \\
\model{\bone \vee \btwo}_{X} &:= \model{\bone}_{X} \cup \model{\btwo}_{X}.
\end{align*}
\end{minipage}
\end{minipage}
\end{definition}

\begin{lemma}\label{BoolDec}
Any named Boolean formula $\bone$ with $\FN(\bone)\subseteq
X\cup \{a\}$ admits an $a$-decomposition in $X$.
\end{lemma}
\begin{proof}
We will actually prove a stronger statement: any named Boolean formula $\bone$ admits an $a$-decomposition $\bigvee_{i=1}^{k}\bthree_{i}\land \bfour_{i}$ where $\model{\bigvee_{i=1}^{j}\bfour_{i}}_{X}=\model{\top}_{X}$.
We argue by induction on $\bone$:
\begin{itemize}
\itemsep0em
\item if $\bone=\bvar_{i}^{a}$ or $\bone=\lnot \bvar_{i}^{a}$, then $k=1$, $\bthree_{i}=\bone$ and $\bfour_{i}=\top$
\item if $\bone=\bvar_{i}^{b}$, where $b\neq a$, then $k=2$, $\bthree_{0}=\top, \bthree_{1}=\bot$ and $\bfour_{0}=\bone, \bfour_{1}=\lnot \bone$

\item 
if $\bone=\bone_{1}\vee \bone_{2}$ thenm by IH,
$\bone_{1}=\bigvee_{i=0}^{k_{1}-1} \bthree^{1}_{i}\land \bfour^{1}_{i}$ and 
$\bone_{2}=\bigvee_{j=0}^{k_{2}-1} \bthree^{2}_{j}\land \bfour^2_{j}$.  Then,

\medskip
\adjustbox{scale=0.8}{
\begin{minipage}{\linewidth}
\begin{align*}
\bone& \equiv \left (\bigvee_{i=0}^{k_{1}-1} \bthree^{1}_{i}\land \bfour^{1}_{i}\right) \vee \left(\bigvee_{j=0}^{k_{2}-1} \bthree^{2}_{j}\land \bfour^2_{j}\right)
\\
&\equiv
 \left (\bigvee_{i=0}^{k_{1}-1} \bthree^{1}_{i}\land \bfour^{1}_{i}\land\top\right) \vee \left(\bigvee_{j=0}^{k_{2}-1} \bthree^{2}_{j}\land \bfour^{2}_{j}\land\top\right)
\\
& \equiv
\left (\bigvee_{i=0}^{k_{1}-1} \bthree^{1}_{i}\land \bfour^{1}_{i} \land \bigvee_{j}\bfour^2_{j}\right) 
\vee \left(\bigvee_{j=0}^{k_{2}-1} \bthree^{2}_{j}\land \bfour^{2}_{j} \land \bigvee_{i}\bfour^{1}_{i}\right)
\\
&\equiv
\left (\bigvee_{i=0,j=0}^{k_{1}-1, k_{2}-1} ( \bthree^{1}_{i}\land \bfour^{1}_{i}\land \bfour^2_{j})\right) 
\vee \left(\bigvee_{j=0,i=0 }^{k_{2}-1, k_{1}-1} \bthree^{2}_{j}\land \bfour^{2}_{j} \land \bfour^{1}_{i}\right)
\\
&\equiv
\bigvee_{i=0,j=0}^{k_{1}-1, k_{2}-1} ( \bthree^{1}_{i}\vee \bthree^{2}_{j})\land (\bfour^{1}_{i}\land \bfour^2_{j}).
\end{align*}
\end{minipage}
}
\medskip

Let $k= k_{1}\cdot k_{2}$. We can identify any $l\leq k-1$ with a pair $(i,j)$, $i<k_{1}$ and $j<k_{2}$. 
Let $\bthree_{i,j}=\bthree^{1}_{i}\lor\bthree^{2}_{j}$ and 
$\bfour_{i,j}= \bfour^{1}_{i}\lor \bfour^{2}_{j}$. We have then that 
$$
\bone= \bigvee_{i=0,j=0}^{k_{1}-1,k_{2}-1} 
\bthree_{i,j}\land \bfour_{i,j}.
$$
Observe that for $(i,j)\neq (i',j')$, $\bfour_{i,j}\land \bfour_{i',j'}\equiv\bot$. Moreover, 
$\bigvee_{i,j}\bfour_{i,j}\equiv \bigvee_{i,j}\bfour^{1}_{i}\lor \bfour^{2}_{j}\equiv
(\bigvee_{i}\bfour^{1}_{i})\lor (\bigvee_{j}\bfour^{2}_{j})\equiv \top\lor \top\equiv\top$.

\item 
if $\bone=\bone_{1}\land \bone_{2}$, then, by IH,
$\bone_{1}\equiv\bigvee_{i=0}^{k_{1}-1} \bthree^{1}_{i}\land \bfour^{1}_{i}$ and 
$\bone_{2}\equiv\bigvee_{j=0}^{k_{2}-1} \bthree^{2}_{j}\land \bfour^{2}_{j}$.  
Then,
\begin{align*}
\bone & \equiv \left (\bigvee_{i=0}^{k_{1}-1} \bthree^{1}_{i}\land \bfour^{1}_{i}\right) \land \left(\bigvee_{j=0}^{k_{2}-1} \bthree^{2}_{j}\land \bfour^{2}_{j}\right)
\\
&
\equiv
 \bigvee_{i=0,j=0}^{k_{1}-1,k_{2}-1} \bthree^{1}_{i}\land \bfour^{1}_{i}\land \bthree^{2}_{j}\land \bfour^{2}_{j}\\
& \equiv\bigvee_{i=0,j=0}^{k_{1}-1,k_{2}-1} (\bthree^{1}_{i}\land \bthree^{2}_{j})\land (\bfour^{1}_{i}\land \bfour^{2}_{j}).
\end{align*}
As in the case above let $ k_{1}\cdot k_{2}$. We can identify any $l\leq k-1$ with a pair $(i,j)$, $i<k_{1}$ and $j<k_{2}$. 
Let $\bthree_{i,j}=\bthree^{1}_{i}\land \bthree^{2}_{j}$ and 
$\bfour_{i,j}= \bfour^{1}_{i}\land \bfour^{2}_{j}$. We have then that 
$$
\bone\equiv \bigvee_{i=0,j=0}^{k_{1}-1,k_{2}-1} \bthree_{i,j}\land \bfour_{i,j}.
$$
As in the previous case we have that for $(i,j)\neq (i',j')$, $\bfour_{i,j}\land \bfour_{i',j'}\equiv \bot$ and  
$\bigvee_{i,j}\bfour_{i,j}\equiv \bigvee_{i,j}\bfour^{1}_{i}\land \bfour^{2}_{j}\equiv
(\bigvee_{i}\bfour^{1}_{i})\land (\bigvee_{j}\bfour^{2}_{j})\equiv\top\land \top \equiv\top$.
\end{itemize}
\end{proof}
\primeFundamental*

\subsection{The Proof Theory of $\PPL$}\label{appendix4.2}
\begin{definition}[Named External Hypothesis]
A \emph{named external hypothesis} is an expression of one of the following forms:
\begin{itemize}
\item $\vDash \{a\} \in X$
\item $\mu\big(\model{\bone}_{X}\big) = 0$ or $\mu\big(\model{\bone}_{X}\big) = 1$
\item $\bone \vDash^{X} \btwo$
\item $\bone \vDash^{X} \bigvee \{\btwo_{i} \ | \ \mu\big(\model{\bthree_{i}}_{Y}\big)$ $\triangleright$ $q\}$, $\bigvee\{\btwo_{i}$ $|$ $\mu\big(\model{\bthree_{i}}_{Y}\big) \triangleright q\} \vDash^{X} \bone$, with $\triangleright \in \{\ge, \leq, >, <, =\}$
\end{itemize}
where $i \in \Nat,$ $X$ and $Y$ are name sets, $\bone, \btwo_{i}, \bthree_{i}$ are named Boolean formulas, and $q \in \mathbb{Q}_{[0,1]}$.
\end{definition}
\begin{definition}[Named Labelled Formulas]
A \emph{named labelled formula} is a formula of the form $\bone \Era \fone$ or $\bone \Ela \fone$, such that $\fone$ is a formula of $\PPL$ and $\bone$ is a named Boolean formula.
\end{definition}
\noindent
For the sake of simplicity, given a labelled formula $\lone$ = $\bone \Era \fone$ or $\lone$ = $\bone \Ela \fone$, let us introduce the notion of $\Var{\lone}$ = $\Var{\bone} \cup \Var{\fone}$.
\begin{definition}[Named Sequent]
A \emph{named sequent} of $\GPPL$ is an expression of the form $\vdash^{X} \lone$, where $\lone$ is a named labelled formula, and $\Var{\lone} \subseteq X$.
\end{definition}

Let us now define a one-sided, single-succedent, and 
labelled sequent calculus, call it $\GPPL$.
\begin{definition}[The Proof System $\GPPL$]\label{CPLsystem}
$\GPPL$ is defined by the set of rules displayed in Fig.~\ref{fig:PPL}.

\begin{figure}[h!]
\framebox{
\parbox[t][12cm]{15cm}{
\begin{minipage}{\textwidth}
\small
\begin{center}
Initial Sequents
\end{center}
\normalsize
\begin{minipage}{\linewidth}
\adjustbox{scale=0.8,center}{
$
\AxiomC{$\vDash \{a\} \in X$}
\AxiomC{$\bone \vDash^{X} \bvar_{i}^{a}$}
\RightLabel{$\AxO$}
\BinaryInfC{$\vdash^{X} \bone \Era \atm{i}{a}$}
\DP
\qquad \qquad \qquad
\AxiomC{$\vDash \{a\} \in X$}
\AxiomC{$\bvar_{i}^{a} \vDash^{X} \bone$}
\RightLabel{$\AxT$}
\BinaryInfC{$\vdash^{X} \Delta, \bone  \Ela \atm{i}{a}$}
\DP
$}
\end{minipage}
\small
\vskip2mm
\begin{center}
Union and Intersection Rules
\end{center}
\normalsize
\begin{minipage}{\linewidth}
\adjustbox{scale=0.8,center}{
$
\AxiomC{$\vdash^{X} \btwo \Pto \fone$}
\AxiomC{$\vdash^{X} \bthree \Pto \fone$}
\AxiomC{$ \bone  \vDash^{X} \btwo\lor \bthree $}
\RightLabel{$\Rcup$}
\TrinaryInfC{$\vdash^{X} \bone  \Pto \fone$}
\DP
\qquad
\AxiomC{$\vdash^{X} \btwo \Pfrom \fone$}
\AxiomC{$\vdash^{X} \bthree  \Pfrom \fone$}
\AxiomC{$\btwo\land \bthree\vDash^{X}   \bone $}
\RightLabel{$\Rcap$}
\TrinaryInfC{$\vdash^{X} \bone \Pfrom \fone$}
\DP
$}
\end{minipage}
\vskip2mm
\small
\begin{center}
Logical Rules
\end{center}
\normalsize
\begin{minipage}{\linewidth}
\adjustbox{scale=0.8,center}{
$
\AxiomC{$\vdash^{X} \btwo \Pfrom \fone$}
\AxiomC{$\bone \vDash^{X}   \lnot \btwo$}
\RightLabel{$\Rnra$}
\BinaryInfC{$\vdash^{X} \bone  \Pto \neg \fone $}
\DP
\qquad \qquad \qquad
\AxiomC{$\vdash^{X} \btwo \Pto \fone$}
\AxiomC{$ \lnot{\btwo} \vDash^{X} \bone $}
\RightLabel{$\Rnla$}
\BinaryInfC{$\vdash^{X} \bone  \Pfrom \neg \fone$}
\DP
$
}
\end{minipage}

\bigskip
\scriptsize{
\begin{minipage}{\linewidth}
\begin{minipage}[t]{0.3\linewidth}
\begin{prooftree}
\AxiomC{$\vdash^{X}\bone\Era\fone$}
\RightLabel{$\mathsf{R1}_{\vee}^{\Era}$}
\UnaryInfC{$\vdash^{X}\bone\Era\fone\vee\ftwo$}
\end{prooftree}
\end{minipage}
\hfill
\begin{minipage}[t]{0.2\linewidth}
\begin{prooftree}
\AxiomC{$\vdash^{X}\bone\Era\ftwo$}
\RightLabel{$\mathsf{R2}_{\vee}^{\Era}$}
\UnaryInfC{$\vdash^{X}\bone\Era\fone\vee\ftwo$}
\end{prooftree}
\end{minipage}
\hfill
\begin{minipage}[t]{0.35\linewidth}
\begin{prooftree}
\AxiomC{$\vdash^{X}\bone\Ela\fone$}
\AxiomC{$\vdash^{X}\bone\Ela\ftwo$}
\RightLabel{$\Rvla$}
\BinaryInfC{$\vdash^{X}\bone\Ela\fone\vee\ftwo$}
\end{prooftree}
\end{minipage}
\end{minipage}
}

\bigskip
\begin{minipage}{\linewidth}
\begin{minipage}[t]{0.35\linewidth}
\begin{prooftree}
\AxiomC{$\vdash^{X}\bone\Era\fone$}
\AxiomC{$\vdash^{X}\bone\Era\ftwo$}
\RightLabel{$\Rwra$}
\BinaryInfC{$\vdash^{X}\bone\Era\fone\wedge\ftwo$}
\end{prooftree}
\end{minipage}
\hfill
\begin{minipage}[t]{0.3\linewidth}
\begin{prooftree}
\AxiomC{$\vdash^{X}\bone\Ela\fone$}
\RightLabel{$\mathsf{R1^{\Ela}_{\wedge}}$}
\UnaryInfC{$\vdash^{X}\bone\Ela\fone\wedge\ftwo$}
\end{prooftree}
\end{minipage}
\hfill
\begin{minipage}[t]{0.3\linewidth}
\begin{prooftree}
\AxiomC{$\vdash^{X}\bone\Ela\ftwo$}
\RightLabel{$\mathsf{R2^{\Ela}_{\wedge}}$}
\UnaryInfC{$\vdash_{X}\bone\Ela\fone\wedge\ftwo$}
\end{prooftree}
\end{minipage}
\end{minipage}

\small
\begin{center}
Counting Rules
\end{center}
\normalsize

\begin{minipage}{\linewidth}
\adjustbox{scale=0.8,center}{
$\AxiomC{$\mu(\model{\bone}_{X}) = 0$}
\RightLabel{$\Rmur$}
\UnaryInfC{$\vdash^{X} \bone  \Pto \fone$}
\DP
\qquad \qquad \qquad
\AxiomC{$\mu(\model{\bone}_{X}) = 1$}
\RightLabel{$\Rmul$}
\UnaryInfC{$\vdash^{X}\bone  \Pfrom \fone$}
\DP
$}
\end{minipage}

\bigskip
\begin{minipage}{\linewidth}
\adjustbox{scale=0.8,center}{
$
\AxiomC{$\vdash^{X\cup\{a\}}  \btwo \Pto \fone$}
\AxiomC{$\bone \vDash^{X} \bigvee_{i}\{\bfour_{i} \ | \ \mu(\model{\bthree_{i}}_{\{a\}}) \ge q\}$}
\RightLabel{$\Rbr$}
\BinaryInfC{$\vdash^{X}  \bone \Pto \BOX^{q}_{a}\fone$}
\DP
\quad
\AxiomC{$\vdash^{X\cup\{a\}}  \btwo \Ela \fone$}
\AxiomC{$\bigvee_{i}\{\bfour_{i} \ | \ \mu(\model{\bthree_{i}}_{\{a\}}) \ge q\} \vDash^{X} \bone$}
\RightLabel{$\Rbl$}
\BinaryInfC{$\vdash^{X}  \bone \Ela \BOX^{q}_{a}\fone$}
\DP
$}
\end{minipage}

\bigskip
\begin{minipage}{\linewidth}
\adjustbox{scale=0.8,center}{
$
\AxiomC{$\vdash^{X\cup\{a\}}  \btwo \Ela \fone$}
\AxiomC{$\bone \vDash^{X} \bigvee_{i}\{\bfour_{i} \ | \ \mu(\model{\bthree_{i}}_{\{a\}}) < q\}$}
\RightLabel{$\Rdr$}
\BinaryInfC{$\vdash^{X}  \bone \Era \DIA^{q}_{a}\fone$}
\DP
\quad
\AxiomC{$\vdash^{X\cup\{a\}}  \btwo \Era \fone$}
\AxiomC{$\bigvee_{i}\{\bfour_{i} \ | \ \mu(\model{\bthree_{i}}_{\{a\}}) < q\} \vDash^{X} \bone$}
\RightLabel{$\Rdl$}
\BinaryInfC{$\vdash^{X} \bone \Ela \DIA^{q}_{a}\fone$}
\DP
$}
\end{minipage}
\small
\begin{center}
where $\bigvee_{i}\bfour_{i}\wedge \bthree_{i}$ is an $a$-decomposition of $\btwo$.
\end{center}
\end{minipage}
}}
\caption{\small{Proof system for $\PPL$.}}
\label{fig:PPL}
\end{figure}
\end{definition}
\noindent
For the sake of simplicity, let $\Rmur$ and $\Rmul$ be called
$\mu$-rules. 
The notions of derivation and of derivation height for $\GPPL$ 
are defined in a standard way.
\begin{definition}[Derivations in $\GPPL$]
A \emph{derivation in $\GPPL$} is either an initial sequent or a $\mu$-rule, or it is obtained by applying one of the remaining rules to derivations concluding its premisses.
\end{definition}
\begin{definition}[Derivation Height in $\GPPL$]
The \emph{height of a derivation in $\GPPL$} is the greatest number of successive applications of rules in it, where initial sequents and $\mu$-rules have height 0. 
\end{definition}

$\GPPL$ si sound and complete with respect to
the semantics of $\PPL$. 
 Let us start by presenting the notions of validity for labelled formulas
and sequents of $\GPPL$.
 Similarly to $\PPLc$, given a named Boolean formula $\bone$,
  a formula of $\PPL$, call it $\fone$, 
  and an $X$, with $\Var{\bone} \cup \Var{\fone} \subseteq X$,  
  the labelled formula $\bone \Era \fone$ (resp. $\bone \Ela \fone$) 
  is \emph{X-valid}, when $\model{\bone}_{X} \subseteq \model{\fone}_{X}$ (resp. $\model{\fone}_{X}\subseteq \model{\bone}_{X}$). 
  Formally,
 \begin{definition}[Validity of Labelled Formulas]\label{labVal}
The validity of a labelled formula is defined as follows:
\begin{itemize}
\itemsep0em
\item $\bone \Era \fone$ is $X$-\emph{valid} iff $\model{\bone}_{X} \subseteq \model{\fone}_{X}$
\item $\bone \Ela \fone$ is $X$-\emph{valid} iff $\model{\fone}_{X} \subseteq \model{\bone}_{X}$.
\end{itemize}
\end{definition}
\noindent
An external hypothesis containing $\bone_{1}, \dots, \bone_{n}$ 
named Boolean variables is valid 
if and only if, 
for every $X$ such that 
$\bigcup_{i=1}^{m}\Var{\bone_{i}} \subseteq X$ the hypothesis is valid. 
The notion of $X$-validity for external hypothesis is as predictable. 
For example, $\bone \vDash^{X} \bvar_{\atom{i}}^{a}$ is valid if and only if $\Var{\bone} \cup \Var{\bvar_{\atom{i}}^{a}} \subseteq X$ and $\model{\bone}_{X} \subseteq \model{\bvar_{\atom{i}}^{a}}_{X}$ and $\bvar_{\atom{i}}^{a} \vDash^{X} \bone$ is valid if and only if $\Var{\bone} \cup \Var{\bvar_{\atom{i}}^{a}} \subseteq X$ and
$\model{\bvar_{\atom{i}}^{a}}_{X} \subseteq \model{\bone}_{X}$.
A named sequent $\vdash^{X} \lone$ is valid,
denoted $\vDash^{X} \lone$, 
if $\lone$ is $X$-valid. 
\begin{definition}[Validity of Sequents]
A sequent of $\GPPL$, $\vdash^{X} \lone$, is \emph{valid} if and only if $\vDash^{X} \lone$.
A set of such sequents, $\{ \vdash^{X_{1}} \lone_{1}, \dots, \vdash^{X_{m}} \lone_{m}\}$ is \emph{valid} if and only if all its sequents are valid. 
\end{definition}
\begin{notation}
We write $\fone \equiv_{X} \ftwo$ when $\model{\fone}_{X} = \model{\ftwo}_{X}$.
\end{notation}
\begin{theorem}[Soundness of $\GPPL$]\label{theorem:soundness}
If a sequent is derivable in $\GPPL$, then it is valid.
\end{theorem}
\begin{proof}
The proof is by induction on the height of the derivation
of $\vdash^{X} \lone$, call it $n$.
\begin{itemize}
\itemsep0em
\item \emph{Base case.} 
The sequent is either an initial sequent or it is derived by a $\mu$-rule. Let us consider the case of $\AxO$ only. The other cases are analogously proved.
\begin{prooftree}
\AxiomC{$\vDash \{a\} \in X$}
\AxiomC{$\bone \vDash^{X} \bvar_{\atom i}^{a}$}
\RightLabel{$\AxO$}
\BinaryInfC{$\vdash^{X} \bone \Era \atom{i}_{a}$}
\end{prooftree}
As seen, if $\bone \vDash^{X} \bvar_{\atom i}^{a}$, then $\Var{\bone} \cup \Var{\bvar_{\atom i}^{a}} \subseteq X$, and $\model{\bone}_{X} \subseteq \model{\bvar_{\atom i}^{a}}_{X}$, which is $\model{\bone}_{X} \subseteq \{f \in (\twoOm)^{X}$ $|$ $f$($a$)($i$) = 1$\}$. 
But then, 
for Definition~\ref{PPLsem}, 
$\model{\atm{i}{a}}_{X}$ = $\{f \in (\twoOm)^{X}$ $|$ $f$($a$)($i$) = 1$\}$, so $\model{\bone}_{X} \subseteq \model{\atm{i}{a}}_{X}$.
By Definition~\ref{labVal}, 
$\bone \Era \atm{i}{a}$ is $X$-valid, so $\vdash^{X} \bone \Era \atm{i}{a}$ must be valid as well.
\par 
The proof for $\AxT$ proof is equivalent to the one for $\AxO$, 
whereas $\mu$-cases are treated similarly to the corresponding proofs in $\GPPLc$.
\item \emph{Inductive case.} Let us assume soundness to hold for derivations of height up to $n$ and show it holds for derivations of height $n$+1. 
The proof is equivalent to that for $\GPPLc$, so let us just consider a few examples. 
All the other cases are similar. 
\par $\Rvla$. 
The derivation has the following form:
\begin{prooftree}
\AxiomC{$\vdots$}
\noLine
\UnaryInfC{$\vdash^{X} \bone \Ela \fone$}
\AxiomC{$\vdots$}
\noLine
\UnaryInfC{$\vdash^{X} \bone \Ela \ftwo$}
\RightLabel{$\Rvla$}
\BinaryInfC{$\vdash^{X} \bone \Ela \fone \vee \ftwo$}
\end{prooftree}
For IH, the premisses are valid, so $\vDash^{X} \bone \Ela \fone$ and $\vDash^{X} \bone \Ela \ftwo$, which is $\model{\fone}_{X} \subseteq \model{\bone}_{X}$ and $\model{\ftwo}_{X} \subseteq \model{\bone}_{X}$. But then also $\model{\fone}_{X} \cup \model{\ftwo}_{X} \subseteq \model{\bone}_{X}$, which is $\model{\fone \vee \ftwo}_{X} \subseteq \model{\bone}_{X}$ and, thus, $\vdash^{X} \bone \Ela \fone \vee \ftwo$ is valid.

\par $\Rbr.$ 
Let $\bigvee_{i}\bfour_{i} \wedge \bthree_{i}$ be an $a$-decomposition of $\btwo$.
Then, the derivation has the following form:
\begin{prooftree}
\AxiomC{$\vdash^{X\cup\{a\}} \btwo \Era \fone$}
\AxiomC{$\bone \vDash^{X} \bigvee \{\bfour_{i} \ | \ \mu(\model{\bthree_{i}}_{\{a\}}) \ge q\}$}
\RightLabel{$\Rbr$}
\BinaryInfC{$\vdash^{X} \bone \Era \BOX^{q}_{a}\fone$}
\end{prooftree}
For IH, $\btwo \Era \fone$ is $X\cup\{a\}$-valid.
Furthermore, since $\bone \vDash^{X} \bigvee \{\bfour_{i} \ | \ \mu(\model{\bthree_{i}}_{\{a\}}) \ge q\}$, $\model{\bone}_{X} \subseteq \bigcup_{i}\{\model{\bfour_{i}}_{X}$ $|$ $\mu$($\model{\bthree_{i}}_{\{a\}}) \ge q\}$,
which is by Lemma~\ref{lemma:fundamental} 
and given that $\bigvee_{i}\bfour_{i} \wedge \bthree_{i}$ be an $a$-decomposition of $\btwo$, $\model{\bone}_{X} \subseteq \{f \in (\twoOm)^{X} \ | \ \mu(\PROJ_{f}(\model{\btwo}_{X \cup\{a\}}$)) $\ge q\}$. 
But then, since for IH, $\model{\btwo}_{X\cup\{a\}} \subseteq \model{\fone}_{X\cup\{a\}}$, also $\model{\bone}_{X} \subseteq \{f \in (\twoOm)^{X} \ | \ \mu(\PROJ_{f}(\model{\fone}_{X\cup\{a\}})) \ge q\}$ (monotonicity).
Thus, $\{f \in (\twoOm)^{X} \ | \ \mu(\PROJ_{f}(\model{\fone}_{X \cup\{a\}}$))$ \ge q\}$ = $\model{\BOX^{q}_{a}\fone}_{X}$.
 Therefore, $\model{\bone}_{X} \subseteq \model{\BOX^{q}_{a}\fone}_{X}$, which is $\vdash^{X} \bone \Era \BOX^{q}_{a}\fone$ is valid.
  
 \par 
 $\Rbl.$ Let $\bigvee_{i}\bfour_{i} \wedge \bthree_{i}$ be an $a$-decomposition of $\btwo$. Then, the derivation has the following form:
 \begin{prooftree}
 \AxiomC{$\vdots$}
 \noLine
 \UnaryInfC{$\vdash^{X\cup\{a\}} \btwo \Ela \fone$}
 \AxiomC{$\bigvee\{\bfour_{i} \ | \ \mu(\model{\bthree_{i}}_{\{a\}}) \ge q\} \vDash^{X} \bone$}
 \RightLabel{$\Rbl$}
 \BinaryInfC{$\vdash^{X} \bone \Ela \BOX^{q}_{a}\fone$}
 \end{prooftree}
For IH, $\btwo \Ela \fone$ is $X\cup\{a\}$-valid. Furthermore, $\bigcup_{i}\{\model{\bfour_{i}}_{X} \ | \ \mu(\model{\bthree_{i}}_{\{a\}}) \ge q\} \subseteq \model{\bone}_{X}$, that, given $a$-decomposition for $\btwo$ and Lemma~\ref{lemma:fundamental}, 
 $\{f \in (\twoOm)^{X} \ | \ \mu(\PROJ_{f}(\model{\btwo}_{X\cup\{a\}})) \ge q\} \subseteq \model{\bone}_{X}$.
 For IH, $\model{\fone}_{X}$ $\subseteq$ $\model{\btwo}_{X}$, so $\{f \in (\twoOm)^{X}$ $|$ $\mu(\PROJ_{f}(\model{\fone}_{X\cup\{a\}})) \ge q\} \subseteq \model{\bone}_{X}.$ Thus,
 $\model{\BOX^{q}_{a}\fone}_{X} \subseteq \model{\bone}_{X}$, which is $\bone \Ela \BOX^{q}_{a}\fone$ is $X$-valid and, thus, $\vdash^{X} \bone \Ela \BOX^{q}_{a}\fone$ is valid. 
 
 \par 
 $\Rdr.$ Let $\bigvee_{i}\bfour_{i} \wedge \bthree_{i}$ be an $a$-decomposition of $\btwo$. Then, the derivation has the following form:
\begin{prooftree}
\AxiomC{$\vdots$}
\noLine
\UnaryInfC{$\vdash^{X\cup\{a\}} \btwo \Ela \fone$}
\AxiomC{$\vdots$}
\noLine
\UnaryInfC{$\bone \vDash^{X} \bigvee\{\bfour_{i} \ | \ \mu(\model{\bthree_{i}}_{\{a\}}) < q\}$}
\RightLabel{$\Rdr$}
\BinaryInfC{$\vdash^{X}\bone \Era \DIA^{q}_{a}\fone$}
\end{prooftree}
So, $\bone \Ela \fone$ is $X\cup\{a\}$-valid. 
Furthermore, $\model{\bone}_{X} \subseteq \bigcup_{i}\{\model{\bfour_{i}}_{X} \ |$ $\mu(\model{\bthree_{i}}_{\{a\}})$ $< q\}$, which, given $a$-decomposition of $\btwo$ and Lemma~\ref{lemma:fundamental}, is $\model{\bone}_{X} \subseteq \{f \in (\twoOm)^{X} \ |$ $\mu(\PROJ_{f}(\model{\btwo}_{X\cup\{a\}})) < q\}$.
For IH, $\model{\fone}_{X\cup\{a\}} \subseteq \model{\btwo}_{X\cup\{a\}}$, so also $\model{\bone}_{X} \subseteq \{f \in (\twoOm)^{X} \ | \ \mu(\PROJ_{f}(\model{\fone}_{X\cup\{a\}})) < q\}$. Thus,
$\model{\bone}_{X} \subseteq \model{\DIA^{q}_{a}\fone}_{X}$. Therefore, $\bone \Era \DIA^{q}_{a}\fone$ is $X$-valid, which is 
$\vdash^{X} \bone \Era \DIA^{q}_{a}\fone$ is valid. 

\par $\Rdl.$ 
Let $\bigvee_{i}\bfour_{i}\wedge \bthree_{i}$ be an $a$-decomposition of $\btwo$. The derivation has the following form:
\begin{prooftree}
\AxiomC{$\vdots$}
\noLine
\UnaryInfC{$\vdash^{X\cup\{a\}} \btwo \Era \fone$}
\AxiomC{$\vdots$}
\noLine
\UnaryInfC{$\bigvee\{\bfour_{i} \ | \ \mu(\model{\bthree_{i}}_{\{a\}}) < q\} \vDash^{X} \bone$}
\RightLabel{$\Rdl$}
\BinaryInfC{$\vdash^{X} \bone \Ela \DIA^{q}_{a}\fone$}
\end{prooftree}
So, $\btwo \Era \fone$ is $X\cup\{a\}$-valid.
Furthermore,
$\bigcup_{i}\{\model{\bfour_{i}}_{X} \ | \ \mu(\model{\bthree_{i}}_{\{a\}} < q\} \subseteq \model{\bone}_{X}$
Given the $a$-decomposition of $\btwo$ and 
Lemma~\ref{lemma:fundamental}, $\{f \in (\twoOm)^{X}$ $|$ $\mu(\PROJ_{f}(\model{\btwo}_{X\cup\{a\}})) < q\} \subseteq \model{\bone}_{X}$.
Since, for IH, $\model{\btwo}_{X\cup\{a\}} \subseteq \model{\fone}_{X\cup\{a\}}$, also $\{f \in (\twoOm)^{X}$ $|$ $\mu(\PROJ_{f}(\model{\fone}_{X\cup\{a\}}) < q\} \subseteq \model{\bone}_{X}$. Therefore, 
$\model{\DIA^{q}_{a}\fone}_{X} \subseteq \model{\bone}_{X}$, 
which is $\vdash^{X} \bone \Ela \DIA^{q}_{a}\fone$ is valid.
\end{itemize}
\end{proof}

The proof of the completeness of $\GPPL$ is obtained similarly to the one of $\GPPLc$. 
Let us start by defining the decomposition relation $\Dec$ between sets of sequents, which is based on $\dec$. 
\begin{definition}[Decomposition Rewriting Reduction, $\dec$]
The decomposition rewriting reduction, $\dec$, from a sequent to a set of sequents
(both in the language of $\GPPL$), is defined by the following decomposition rewriting rules:
\begin{align*}
\text{if  } \bone \vDash^{X} \neg \btwo, \ \vdash^{X} \bone \Era \neg \fone \ \ &\dec \ \ \{\vdash^{X} \btwo \Ela \fone\} \\
\text{if  } \neg \btwo \vDash^{X} \bone, \ \vdash^{X}  \bone \Ela \neg \fone \ \ &\dec \ \ \{\vdash^{X} \btwo \Era \fone\} \\ 
\text{if  } \bone \vDash^{X} \btwo \vee \bthree, \ \vdash^{X}  \bone \Era \fone \vee \ftwo \ \ &\dec \ \ \{\vdash^{X}  \btwo \Era \fone, \ \vdash^{X}  \bthree \Era \ftwo\} \\
\vdash^{X}  \bone \Ela \fone \vee \ftwo \ \ &\dec \ \ \{\vdash^{X}  \bone \Ela \fone, \ \vdash^{X}  \bone \Ela \ftwo\} \\
\vdash^{X}  \bone \Era \fone \wedge \ftwo \ \ &\dec \ \ \{\vdash^{X}  \bone \Era \fone, \ \vdash^{X} \bone \Era \ftwo\} \\
\text{if  } \btwo \wedge \bthree \vDash^{X} \bone, \ \vdash^{X} \bone \Ela \fone \wedge \ftwo \ \ &\dec \ \ \{\vdash^{X} \btwo \Ela \fone, \ \vdash^{X} \bthree \Ela \ftwo\} \\
\text{if } \btwo = \bigvee\bfour_{i}\wedge \bthree_{i} \ a\text{-decomposition of } \btwo \\
\text{and } \bone \vDash^{X} \bigvee \{\bfour_{i} \ | \ \mu(\model{\bthree_{i}}_{\{a\}}) \ge q\}, \ \vdash^{X} \bone \Era \BOX^{q}_{a}\fone \ \ &\dec \{\vdash^{X\cup\{a\}}  \btwo \Era \fone\} \\
\text{if } \btwo = \bfour_{i}\wedge \bthree_{i} \ a\text{-decomposition of } \btwo \\
\text{and } \bigvee \{\bfour_{i} \ | \ \mu(\model{\bthree_{i}}_{\{a\}}) \ge q\} \vDash^{X} \bone, \ \vdash^{X} \bone \Ela \BOX^{q}_{a}\fone \ \ &\dec \{\vdash^{X\cup\{a\}} \btwo \Ela \fone\} \\
\text{if } \btwo = \bigvee\bfour_{i}\wedge \bthree_{i} \ a\text{-decomposition of } \btwo \\
\text{and } \bone \vDash^{X} \bigvee \{\bfour_{i} \ | \ \mu(\model{\bthree_{i}}_{\{a\}}) < q\}, \ \vdash^{X} \bone \Era \DIA^{q}_{a}\fone \ \ &\dec \{\vdash^{X\cup\{a\}} \btwo \Ela \fone\} \\
\text{if } \btwo = \bigvee\bfour_{i}\wedge \bthree_{i} \ a\text{-decomposition of } \btwo \\
\text{and } \bigvee \{\bfour{i} \ | \ \mu(\model{\bthree_{i}}_{\{a\}}) < q\} \vDash^{X} \bone, \ \vdash^{X}  \bone \Ela \DIA^{q}_{a}\fone \ \ &\dec \{\vdash^{X\cup\{a\}} \btwo \Era \fone\} 
\end{align*}
\end{definition}
\normalsize
\noindent
As anticipated, we define set-decomposition reduction, $\Dec$, from a set of sequents to a set of sequents, basing on $\dec$.
  \begin{definition}[Set Decomposition, $\Dec$]
  The \emph{set-decomposition reduction}, $\Dec$, from a set of
  sequents to another set of sequents, is defined as follows:
\begin{prooftree}
\AxiomC{$\vdash^{X} \lone_{i} \dec \{\vdash^{X'_{1}} \lone_{i_{1}}, \dots, \vdash^{X'_{m}} \lone_{i_{m}}\}$}
\UnaryInfC{$\{\vdash^{X_{1}} \lone_{1}, \dots, \vdash^{X_{i}} \lone_{i}, \dots, \vdash^{X_{n}} \lone_{n}\} \Dec \{\vdash^{X_{1}} \lone_{1}, \dots, \vdash^{X'_{1}} \lone_{i_{1}}, \dots, \vdash^{X'_{m}} \lone_{i_{m}}, \dots, \vdash^{X_{n}} \lone_{n}\}$}
\end{prooftree}
   \end{definition}
 \noindent
As for $\PPLc$, $\Dec$ is the natural lifting of $\dec$ to a relation between sets of sequents. Again, predicate about sequents can be generalized to ones on sets. 
Also the definitions of corresponding sets, $\Dec$-normal form, and normalizations are the equivalent to the ones given for $\PPLc$. Also the notions of basic and regular sequents are very close to the ones
defined for $\GPPLc$:
\begin{definition}[Regular Sequent]
A \emph{basic formula} is a named labelled formula $\bone \Era \fone, \bone \Ela \fone,$ where $\fone$ is atomic.
A \emph{regular sequent} of $\GPPL$ is a sequent of the form $\vdash^{X} \lone$, such that, $\lone$ is a (named) basic formula.
\end{definition}

All the following lemmas and their respective proofs are analogous 
to those for proving the completeness of
$\GPPLc$. 
\begin{lemma}\label{CPLRegularValidDerivable}
A regular sequent is valid if and only if it is derivable in $\GPPL$.
\end{lemma}
\begin{proof} Let $\vdash^{X}\lone$ be an arbitrary, \emph{regular}
sequent, which is let $\lone$ be basic.
\begin{itemize}
\item[$\Rightarrow$] Assume that $\lone$ is an $X$-valid formula. There are two possible cases.

\par Let $\lone=\bone\Era\atom{i}_{a}$ for some $i\in\Nat$.
Then, $\model{\bone}_{X}\subseteq\model{\atom{i}_{a}}_{X}$,
which is $\model{\bone}_{X}\subseteq\{f\in(\twoOm)^{X}$ $|$ 
$f(a)(i)=1\} = \model{\bvar_{i}^{a}}_{X}$, with
$\FN(\atom{i}_{a})\subseteq X$.
Thus, $\vdash^{X}\bone\Era\atom{i}_{a}$ can be
derived by means of $\AxO$ as follows:
\begin{prooftree}
\AxiomC{$a\in X$}
\AxiomC{$\bone\vDash^{X}\bvar_{i}^{a}$}
\RightLabel{$\AxO$}
\BinaryInfC{$\vdash^{X}\bone\Era\atom{i}_{a}$}
\end{prooftree}

\par Let $\lone=\bone\Ela\atom{i}_{a}$ for some $i\in\Nat$.
Then, $\model{\atom{i}_{a}}_{X}\subseteq
\model{\bone}_{X}$,
which is $\model{\bvar_{i}^{a}}_{X}=\{f \in (\twoOm)^{X}$
$|$ $f(a)(i)=1\}$ $\subseteq$ $\model{\bone}_{X}$,
with $\FN(\atom{i}_{a}\subseteq X$.
Thus, $\vdash^{X}\bone\Ela\atom{i}_{a}$ is derivable by
means of $\AxT$ as follows:
\begin{prooftree}
\AxiomC{$a\in X$}
\AxiomC{$\bvar_{i}^{a}\vDash^{X}\bone$}
\RightLabel{$\AxO$}
\BinaryInfC{$\vdash^{X}\bone\Ela\atom{i}_{a}$}
\end{prooftree}

\item[$\Leftarrow$] By Theorem~\ref{theorem:soundness}.
\end{itemize}
\end{proof}

The following auxiliary result is established by exhaustive inspection. Since the proof idea is substantially the same as the
one used for $\PPLc$, just a few examples will be explicitly taken into account.
\begin{lemma}\label{CPLNormalRegular}
If a (non-empty) sequent is $\dec$-normal, then it is regular.
\end{lemma}
\begin{proof}[Proof Sketch]
The lemma is proved by contraposition. 
Assume that $\vdash^{X} \lone$ is an arbitrary, non-regular
sequent, which is that $\lone$ is non-basic.
Let us prove that the sequent is not $\Dec$-normal.
The proof is by inspection. Let us consider just a few cases.

\par 
$\lone=\bone\Era\neg\fone$. 
Then since for every $\bone$,
$\bot\vDash^{X}\bone$, the following $\dec$-decomposition
is well-defined:
$$
\vdash^{X}\bone\Era\neg\fone \ \ \ \dec \ \ \ \{\vdash^{X}\top\Ela\fone\}.
$$
 
\par
$\lone=\bone\Era\fone\vee\ftwo$. 
Then, since for every
$\bone$, $\bone\vDash^{X}\top\vee\top$, the following
$\dec$-decomposition is well-defined:
$$
\vdash^{X}\bone\Era\fone\vee\ftwo \ \ \dec \ \ \{\vdash^{X}\top\Era\fone, \ \vdash^{X} \top\Era\ftwo\}.
$$

\par 
$\lone=\bone\Ela\fone\vee\ftwo$. 
Then, the following
$\dec$-decomposition is well-defined:
$$
\vdash^{X}\bone\Ela\fone\vee\ftwo \ \ \dec \ \ \{\vdash^{X}\bone\Ela\fone, \vdash^{X}\bone\Ela\ftwo\}.
$$

\par 
$\lone=\bone\Era\BOX^{q}_{a}\fone$.
Let $\top=\bigvee\top\wedge\top$ be the $a$-decomposition of
$\top$.
Then, for every $\bone$ and $q\in[0,1]$, 
$\bone\vDash^{X}\bigvee\{\top$ $|$ $\mu(\model{\top})\ge q\}$.
Thus, the following $\dec$-decomposition is well-defined:
$$
\vdash^{X} \bone\Era\BOX^{q}_{a}\fone \ \ \dec \ \ \{\vdash^{X\cup\{a\}} \top\Era\fone\}.
$$
\end{proof}

\begin{corollary}\label{cor:NormalValidDerivable}
If a sequent is $\dec$-normal, then it is valid if and only if it is derivable in $\GPPL$.
\end{corollary}
\begin{proof}
It is a straightforward consequence of Lemma~\ref{CPLRegularValidDerivable} and Lemma~\ref{CPLNormalRegular}.
\end{proof}

\begin{lemma}\label{theorem:normalizing}
Reduction $\Dec$ is strongly normalizing.
\end{lemma}
\begin{proof}[Proof Sketch]
The proof is again based on the notion of sequent measure. 
It is shown that if $\{\vdash^{X_{1}}\lone_{1},\dots,\vdash^{X_{m}}\lone_{m}\}$ $\Dec$ $\{\vdash^{X_{1}'}\vdash\lone_{1}',\dots,\vdash^{X_{m}'}\lone_{m}'\}$,
then $\ms\big(\{\vdash^{X_{1}}\lone_{1},\dots,\vdash^{X_{m}}\lone_{m}\}\big) > \ms\big(\{\vdash^{X_{1}'}\lone_{1}',\dots,\vdash^{X_{m}'}\lone_{m}'\}\big)$.
This property is established by the exhaustive analysis of all 
possible forms of $\dec$-reduction applicable
to the set,
which is by dealing with all possible forms of $\dec$-reduction
of $\vdash^{X_{i}}\lone_{i}$, with $i\in\{1,\dots,m\}$, where
$\vdash^{X_{i}}\lone_{i}$ is the active sequent on which the
$\Dec$-reduction is based.
These cases, are very closed to the corresponding ones
in $\PPLc$ proof, so just a few examples are considered.

$\lone_{i}=\bone\Ela\neg\fone$. 
Assume that $\vdash^{X}\bone\Ela\neg\fone$
is the sequent on which the considered $\Dec$-reduction
of $\{\vdash^{X_{1}}\lone_{1},\dots,\vdash^{X_{m}}\lone_{m}\}$
is based, and specifically that the sequent is
$\dec$-reduced as follows:
$$
\vdash^{X}\bone\Ela\neg\fone \ \ \dec \ \ \{\vdash^{X}\btwo\Era\fone\}
$$
for some $\btwo$ such that $\neg\btwo\vDash_{X}\bone$.
By Definition~\ref{def:cn}, $\cn(\bone\Ela\neg\fone)=\cn(\btwo\Era\fone)+1$.
Thus, since the considered $\dec$-step reduces the active
sequent $\vdash^{X}\bone\Ela\neg\fone$ only,
for Definition~\ref{def:ms}
$$
\ms\big(\{\vdash^{X_{1}}\lone_{1},\dots,\vdash^{X}\bone\Ela\neg\fone,\dots, \vdash^{X_{m}}\lone_{m}\}\big)
>
\ms\big(\{\vdash^{X_{1}}\lone_{1},\dots,\vdash^{X}\btwo\Era\fone,
\dots,\vdash^{X_{m}}\lone_{m}\}\big).
$$

$\lone_{i}=\bone\Era\fone\vee\ftwo$. 
Assume that $\vdash^{X}\fone\vee\ftwo$ is the sequent
on which the considered $\dec$-reduction of
$\{\vdash^{X_{1}}\lone_{1},\dots,\vdash^{X_{m}}\lone_{m}\}$
is based, and specifically that the sequent is
$\dec$-reduced as follows:
$$
\vdash^{X}\bone\Era\fone\vee\ftwo \ \ \dec \ \ \{\btwo\Era\fone,\vdash^{X}\bthree\Era\ftwo\}
$$
for some $\btwo$ and $\bthree$ such that $\bone\vDash^{X}\btwo\vee\bthree$.
Thus, by Definition~\ref{def:cn} and~\ref{def:ms},
\small{
\begin{align*}
\ms\big(\{\vdash^{X_{1}}\lone_{1},\dots,\vdash^{X}\btwo\Era\fone,
\vdash^{X}\bthree\Era\ftwo,\dots,\vdash^{X_{m}}\lone_{m}\}\big)
&= 3^{\cn(\lone_{1})}+\dots+3^{\cn(\btwo\Era\fone)}+3^{\cn(\bthree\Era\ftwo)}+\dots+3^{\cn(\lone_{m})} \\
&= 3^{\cn(\lone_{1})}+\dots+3^{\cn(\fone)}+3^{\cn(\ftwo)}+\dots+3^{\cn(\lone_{m})} \\
&< 3^{\cn(\lone_{1})}+\dots+3^{\cn(\fone)+\cn(\ftwo)+1}+\dots+3^{\cn(\lone_{m})} \\
&= 3^{\cn(\lone_{1})}+\dots+3^{\cn(\bone\Era\fone\vee\ftwo)}+\dots+3^{\cn(\lone_{m})} \\
&= \ms\big(\{\vdash^{X_{1}}\lone_{1},\dots,\vdash^{X}\bone\Era\fone\vee\ftwo,\dots,\vdash^{X_{m}}\lone_{m}\}\big).
\end{align*}
}

\normalsize
\par
$\lone_{i}=\bone\Ela\BOX^{q}_{a}\fone$.
Assume that $\vdash^{X}\bone\Ela\BOX^{q}_{a}\fone$ is the
sequent on which the considered $\dec$-reduction of
$\{\vdash^{X_{1}}\lone_{1},\dots,\vdash^{X_{m}}\lone_{m}\}$
is based, and specifically that the sequent is $\dec$-reduced
as follows:
$$
\vdash^{X}\bone\Ela\BOX^{q}_{a}\fone \ \ \dec \ \ \{\vdash^{X\cup\{a\}}\btwo\Ela\fone\}
$$
for some $\btwo$ such that 
$\btwo=\bigvee\bfour_{i}\wedge\bthree_{d}$
is the $a$-decomposition of $\btwo$ and 
$\bigvee\{\bfour_{i}$ $|$ $\mu(\model{\bthree_{i}}_{\{a\}})\ge q\}$
$\vDash^{X}$ $\bone$.
By Definition~\ref{def:cn}, $\cn(\bone\Ela\BOX^{q}_{a}\fone)=\cn(\btwo\Ela\fone)+1$.
Thus, since the considered $\dec$-step reduces the active
sequent $\vdash\bone\Ela\BOX^{q}_{a}\fone$ only, for
Definition~\ref{def:ms},
$$
\ms\big(\{\vdash^{X_{1}}\lone_{1},\dots,\vdash^{X}\bone\Ela\neg\fone,\dots,\vdash^{X_{m}}\lone_{m}\}\big)
>
\ms\big(\{\vdash^{X_{1}}\lone_{1},\dots,\vdash^{X}\btwo\Era\fone,
\dots,\vdash^{X_{m}}\lone_{m}\}\big).
$$
\end{proof}

In order to prove that validity is existentially preserved through
$\Dec$-decomposition, it is required to introduce the following
auxiliary lemma.
\begin{lemma}\label{lemma:bone}
For every formula of $\PPL$, call it $\fone$, there is a named Boolean formula $\bone_{\fone}$, such that, for every $X$, with $\Var{\fone} \cup \Var{\ftwo} \subseteq X$:
$$
\model{\fone}_{X} = \model{\bone_{\fone}}_{X}.
$$
\end{lemma}
\begin{proof}
The proof is by induction on the structure of $\bone$.
\begin{itemize}
\item \emph{Base case.} $\fone = \atm{i}{a}$. Let $\Var{\fone} \subseteq X$. For Definition~\ref{PPLsem},
$\model{\atm{i}{a}}_{X} = \{f \in (\twoOm)^{X} \ | \ f(a)(i) = 1\} = \model{\bvar_{\atom i}^{a}}_{X}$. Thus, $\bone_{\fone}$ = $x^{a}_{i}$.
\item \emph{Inductive case.} There are four possible cases.

\par 
$\fone = \neg \ftwo.$ 
By IH, there is a $\bone_{\ftwo}$ such that, for every $X$ ($\Var{\bone} \cup \Var{\bone_{\ftwo}} \subseteq X$), $\model{\fone}_{X} = \model{\bone_{\fone}}$. Thus, let $\bone_{\fone}$ = $\neg \bone_{\ftwo}$. Then,
$$
\model{\neg \bone_{\ftwo}}_{X} = 
(\twoOm)^{X} \ – \ \model{\bone_{\ftwo}}_{X} 
\stackrel{IH}{=} (\twoOm)^{X} \ – \ \model{\ftwo}_{X} 
= \model{\neg \ftwo}_{X} = 
\model{\fone}_{X}.
$$

\par $\fone = \ftwo \vee \fthree$. 
By IH, there are $\bone_{\ftwo}$ and $\bone_{\fthree}$ such that, for every $X$ ($\Var{\ftwo} \cup \Var{\fthree} \subseteq X$), $\model{\bone_{\ftwo}}_{X} = \model{\ftwo}_{X}$ and $\model{\bone_{\fthree}}_{X} = \model{\fthree}_{X}$. 
Thus, let $\bone_{\fone}$ = $\bone_{\ftwo} \vee \bone_{\fthree}$. Then,
$$
\model{\bone_{\ftwo} \vee \bone_{\fthree}}_{X} = 
\model{\bone_{\ftwo}}_{X} \cup \model{\bone_{\fthree}}_{X} 
\stackrel{IH}{=}
\model{\ftwo}_{X} \cup \model{\fthree}_{X} = 
\model{\ftwo \vee \fthree}_{X} = 
\model{\fone}_{X}.
$$

\par $\fone = \ftwo \wedge \fthree$. Equivalent to the previous case.

\par $\fone = \BOX^{q}_{a}\ftwo$ and $\fone = \DIA^{q}_{a}\ftwo$. By Lemma~\ref{lemma:fundamental}.
\end{itemize}
\end{proof}

\begin{lemma}\label{existentialPreservation}
For each sequent of $\GPPL$, if it is valid, then it has a valid
$\Dec$-normal form.
\end{lemma}
\begin{proof}
The proof is by analytic inspection of all possible cases.
Given an arbitrary single-succedent sequent
$\seqO = \ \vdash^{X}\lone$, there are two main cases
to be taken into account:
\begin{itemize}
\item Let $\seqO$ be such that no $\dec$-decomposition 
reduction can be applied on it. Thus, since the sequent is $\dec$-normal, it is either empty or regular, by Lemma~\ref{CPLNormalRegular}. 
Thus, it is already a $\Dec$-normal form, and for hypothesis it is valid.

\item Let $\seqO$ be $\dec$-reducible. 
By Theorem~\ref{theorem:normalizing}, $\Dec$
is strongly normalizing, 
so no infinite reduction sequence is possible.
Consequently, it is enough to prove that each reduction
step existentially preserves validity, which is that
for each possible $\Dec$ based on $\dec$-reduction,
if $\seqO$ is valid, then there is a set of sequents
$\{\seqO_{1},\dots,\seqO_{m}\}$, such that 
$\seqO \dec \{\seqO_{1},\dots,\seqO_{m}\}$
and $\{\seqO_{1},\dots,\seqO_{m}\}$ is valid.
The proof consists in taking into account all possible forms
of $\dec$, on which the reduction step can be based.
Since these cases are very similar to the corresponding
one considered for $\PPLc$, just a few examples will be
taken into account.

\par
$\lone=\bone\Ela \neg\fone$. Assume
$\seqO= \vdash^{X}\bone\Ela\neg\fone$.
Let us show that there is a well-defined $\dec$-reduction
of $\seqO$ of the following form such that the reduced set
is valid.
For hypothesis $\vdash^{X}\bone\Ela\neg\fone$ is valid,
which is $\bone\Ela\neg\fone$ is $X$-valid.
Let us consider a Boolean expression $\btwo=\neg\bone$.
Then, 
$$
\model{\neg\btwo} = (\twoOm)^{X} \ – \ \model{\btwo}_{X}
= (\twoOm)^{X} \ – \ \big((\twoOm)^{X}–\model{\bone}_{X}\big) 
= \model{\bone}_{X}
$$ 
and in particular
$\model{\neg\btwo}_{X}\subseteq\model{\bone}_{X}$,
which is $\neg\btwo\vDash^{X}\bone$.
Consequently, the following $\dec$-reduction is well-defined,
 $$
 \vdash^{X}\bone\Ela\neg\fone \ \ \dec \ \ 
 \{\vdash^{X}\btwo\Era\fone\}.
 $$
Since $\bone\Ela \neg\fone$ is $X$-valid, also 
$\model{\neg\fone}_{X}\subseteq\model{\bone}_{X}$.
Thus, for basic set theory, 
$\model{\neg\bone}_{X}\subseteq\model{\fone}_{X}$.
But for construction $\btwo=\neg\bone$, so also
$\model{\btwo}_{X}\subseteq\model{\fone}_{X}$,
which is $\btwo\Era\fone$ is $X$-valid.
Consequently, $\vdash^{X}\bone\Era\fone$ is 
valid as well.

\par
$\lone=\bone\Era\BOX^{q}_{a}\fone$.
Assume 
$\seqO = \ \vdash^{X}\bone\Era\BOX^{q}_{a}\fone$.
Let us show that there is a well-defined $\dec$-reduction
of $\seqO$ of the following form such that
the reduced set is valid.
For hypothesis,
$\vdash^{X}\bone\Era\BOX^{q}_{a}\fone$ is valid,
which is $\bone\Era\BOX^{q}_{a}\fone$ is
$X$-valid. 
For Lemma~\ref{lemma:bone},
$\model{\fone}_{X\cup\{a\}}=\bone_{\fone}$ and,
by Lemma~\ref{lemma:fundamental},
there is an $a$-decomposition of 
$\bone_{\fone}:\btwo = 
\bigvee\bthree_{i}\wedge\bfour_{i}$.
Furthermore,
$\model{\BOX^{q}_{a}\fone}_{X}$ = 
$\{f\in(\twoOm){X}$ $|$
$\mu\big(\Pi_{f}(\model{\fone}_{X\cup\{a\}})\big)\ge q\}$
and $\{f \in (\twoOm)^{X}$ $|$
$\mu\big(\Pi_{f}(\model{\bone_{\fone}}_{X\cup\{a\}})\big)\ge q\}=\bigcup_{i}\{\model{\bfour_{i}}_{X}$ $|$ 
$\mu(\model{\bthree}_{\{a\}})\ge q\}$.
Therefore, $\model{\BOX^{q}_{a}\fone}_{X}$ = 
$\bigcup_{i}\{\model{\bfour_{i}}_{X}$ $|$ $\mu(\model{\bthree_{i}}_{\{a\}})\ge q\}$,
the following $\dec$-decomposition 
is well-defined:
$$
\vdash^{X}\bone\Era\BOX^{q}_{a}\fone \ \ \dec \ \ \{\vdash^{X\cup\{a\}}\btwo\Era\fone\}.
$$
For construction $\btwo$ is an $a$-decomposition
of $\bone_{\fone}$ so $\model{\bone_{\fone}}_{X\cup\{a\}}$
= $\model{\btwo}_{X\cup\{a\}}$.
In particular, $\model{\btwo}_{X\cup\{a\}}\subseteq\model{\fone}_{X\cup\{a\}}$,
which is $\btwo\Era\fone$ is $X\cup\{a\}$-valid.
Then, $\{\vdash^{X\cup\{a\}}\btwo\Era\fone\}$ is valid
as well.
\end{itemize}
\end{proof}

Finally, it is proved that derivability in $\GPPL$ is
preserved in the following sense.
\begin{lemma}\label{lemma:derivability}
Given the set of sequents $\ssetO,\ssetO'$, 
if $\ssetO \Dec \ssetO'$ and $\ssetO'$ is derivable in $\GPPL$, 
then $\ssetO$ is derivable in $\GPPL$ as well. 
\end{lemma}
\begin{proof}
For hypothesis $\ssetO \Dec \ssetO'$, which is, for some
$\seqO\in\ssetO$, there is a $\dec$-decomposition
$$
\seqO \ \ \dec \ \ \{\seqO_{1},\dots,\seqO_{m}\}
$$
on which the considered $\dec$ is based, so that
$\seqO_{1},\dots,\seqO_{m}\in \ssetO'$.
Since the number of possible $\dec$ is finite, the proof
is obtained by considering all forms of reduction
applicable to the set.

$\lone=\bone\Era\neg\fone$. 
Assume $\seqO=\vdash^{X}\bone\Era\neg\fone$
and that the considered $\dec$ is based on the following
$\dec$-decomposition
$$
\vdash^{X}\bone\Era\neg\fone \ \ \dec \ \ \{\vdash^{X}\btwo\Ela\fone\}
$$
for some $\btwo$ such that $\bone\vDash^{X}\neg\btwo$.
For hypothesis $\ssetO'$ is derivable, so each
of its sequents is.
Thus, in this case, $\vdash^{X}\btwo\Ela\fone$ ($\in\ssetO'$)
is derivable.
Given that $\bone\vDash^{X}\neg\btwo$,
it is easily possible to derive $\vdash^{X}\bone\Era\neg\fone$
by applying $\Rnra$ as follows:
\begin{prooftree}
\AxiomC{$\vdots$}
\noLine
\UnaryInfC{$\vdash^{X}\btwo\Ela\fone$}
\AxiomC{$\bone\vDash^{X}\neg\btwo$}
\RightLabel{$\Rnra$}
\BinaryInfC{$\vdash^{X}\bone\Era\neg\fone$}
\end{prooftree}

$\lone=\bone\Ela\neg\fone$. Assume $\seqO=\vdash^{X}\bone\Ela\neg\fone$
and that the considered $\Dec$ is based on the following
$\dec$-decomposition
$$
\vdash^{X}\bone\Ela\neg\fone \ \ \dec \ \ \{\vdash^{X}\btwo\Era\fone\}
$$
for some $\btwo$ such that $\neg\btwo\vDash^{X}\bone$.
For hypothesis $\ssetO'$ is derivable, so each of its
sequents is.
Thus, $\vdash^{X}\btwo\Era\fone$ is derivable and,
since $\neg\btwo\vDash^{X}\bone$, it is easily
possible to derive $\vdash^{X}\bone\Era\neg\fone$
by applying $\Rnla$ as follows:
 \begin{prooftree}
 \AxiomC{$\vdots$}
 \noLine
 \UnaryInfC{$\vdash^{X}\btwo\Era\fone$}
 \AxiomC{$\neg\btwo\vDash^{X}\bone$}
 \RightLabel{$\Rnla$}
 \BinaryInfC{$\vdash^{X}\bone\Ela\neg\fone$}
 \end{prooftree}

 $\lone=\bone\Era\fone\vee\ftwo$.
 Assume $\seqO=\vdash^{X}\bone\Era\fone\vee\ftwo$
 and that the considered $\Dec$ is based on the following
 $\dec$-decomposition
 $$
 \vdash^{X}\bone\Era\fone\vee\ftwo \ \ \dec \ \ \{\vdash^{X}\btwo\Era\fone, \vdash^{X}\bthree\Era\ftwo\}
 $$
 where $\btwo$ and $\bthree$ are two Boolean formulas
 such that $\bone\vDash^{X}\btwo\vee\bthree$.
 For hypothesis $\ssetO'$ is derivable, so each of its 
 sequents is.
 Therefore, both $\vdash^{X}\btwo\Era\fone$ ($\in\ssetO'$)
 and $\vdash^{X}\bthree\Era\fone$ ($\in\ssetO'$)
 are derivable and, since $\bone\vDash^{X}\btwo\vee\bthree$,
 then also $\vdash^{X}\bone\Era\fone\vee\ftwo$ is
 derivable in $\GPPL$ in the following way:
 \begin{prooftree}
 \AxiomC{$\vdots$}
 \noLine
 \UnaryInfC{$\vdash^{X}\btwo\Era\fone$}
 \RightLabel{$\mathsf{R1}_{\vee}^{\Era}$}
 \UnaryInfC{$\vdash^{X}\btwo\Era\fone\vee\ftwo$}
 \AxiomC{$\vdots$}
 \noLine
 \UnaryInfC{$\vdash^{X}\bthree\Era\ftwo$}
 \RightLabel{$\mathsf{R2}_{\vee}^{\Era}$}
 \UnaryInfC{$\vdash^{X}\bthree\Era\fone\vee\ftwo$}
 \AxiomC{$\bone\vDash^{X}\btwo\vee\bthree$}
 \RightLabel{$\Rcup$}
 \TrinaryInfC{$\vdash^{X}\btwo\Era\fone\vee\ftwo$}
 \end{prooftree}

$\lone=\bone\Ela\fone\vee\ftwo$.
Assume $\seqO=\vdash^{X}\bone\Ela\fone\vee\ftwo$
and that the given $\Dec$ step is based on the following
$\dec$-decomposition:
$$
\vdash^{X}\bone\Ela\fone\vee\ftwo \ \ \dec \ \ \{\vdash^{X}\bone\Ela\fone,\vdash^{X}\bone\Ela\ftwo\}.
$$
For hypothesis, $\ssetO'$ is derivable, so each of its sequent is.
Therefore, both $\vdash^{X}\bone\Ela\fone$ ($\in \ssetO'$) 
and $\vdash^{X}\bone\Ela\ftwo$ ($\in \ssetO'$)
are derivable as well.
Then, $\vdash^{X}\bone\Ela\fone\vee\ftwo$ is easily shown to
be derivable by applying $\Rvla$ as follows:
\begin{prooftree}
\AxiomC{$\vdots$}
\noLine
\UnaryInfC{$\vdash^{X}\bone\Ela\fone$}
\AxiomC{$\vdots$}
\noLine
\UnaryInfC{$\vdash^{X}\bone\Ela\ftwo$}
\RightLabel{$\Rvla$}
\BinaryInfC{$\vdash^{X}\bone\Ela\fone\vee\ftwo$}
\end{prooftree}

$\lone = \bone\Era\fone \wedge \ftwo$.
Assume $\seqO=\vdash^{X}\bone\Era\fone\wedge\ftwo$
and that the considered $\Dec$ is based on the following
$\dec$-decomposition:
$$
\vdash^{X}\bone\Era\fone\wedge \ftwo \ \ \dec \ \ \{\vdash^{X}\bone\Era\fone, \vdash^{X}\bone\Era\ftwo\}.
$$
For hypothesis, $\ssetO'$ is derivable, so each of its sequents is.
Therefore, both $\vdash^{X}\bone\Era\fone$ ($\in\ssetO'$)
and $\vdash^{X}\bone\Era\ftwo$ ($\in\ssetO'$) are derivable as well.
Then, $\vdash^{X}\bone\Era\fone\wedge\ftwo$ is
easily shown to be derivable by applying $\Rwra$ as follows:
\begin{prooftree}
\AxiomC{$\vdots$}
\noLine
\UnaryInfC{$\vdash^{X}\bone\Era\fone$}
\AxiomC{$\vdots$}
\noLine
\UnaryInfC{$\vdash^{X}\bone\Era\ftwo$}
\RightLabel{$\Rwra$}
\BinaryInfC{$\vdash^{X}\bone\Era\fone\wedge\ftwo$}
\end{prooftree}

$\lone=\bone\Ela\fone\wedge\ftwo$. 
Assume $\seqO=\vdash^{X}\bone\Ela\fone\wedge\ftwo$
and that the given $\Dec$ is based on the following
$\dec$-decomposition
$$
\vdash^{X}\bone\Ela\fone\wedge\ftwo \ \ \dec \ \ \{\vdash^{X}\btwo\Ela\fone, \vdash^{X}\bthree\Ela\ftwo\}
$$
where $\btwo$ and $\bthree$ and two Boolean expressions such
that $\btwo\wedge\bthree \vDash^{X}\bone$.
For hypothesis $\ssetO'$ is derivable, so each of its sequents is.
Therefore, both $\vdash^{X}\btwo\Ela\fone$ ($\in\ssetO'$)
and $\vdash^{X}\bthree\Ela\fone$ ($\in\ssetO'$) are derivable
and, since $\bone\vDash^{X}\btwo\wedge\bthree$, then
also $\vdash^{X}\bone\Ela\fone\wedge\ftwo$ is derivable
in $\GPPL$ in the following way:
\begin{prooftree}
\AxiomC{$\vdots$}
\noLine
\UnaryInfC{$\vdash^{X}\btwo\Ela\fone$}
\RightLabel{$\mathsf{R1}_{\wedge}^{\Ela}$}
\UnaryInfC{$\vdash^{X}\btwo\Ela\fone\wedge\ftwo$}
\AxiomC{$\vdots$}
\noLine
\UnaryInfC{$\vdash^{X}\bthree\Ela\ftwo$}
\RightLabel{$\mathsf{R2}_{\wedge}^{\Ela}$}
\UnaryInfC{$\vdash^{X}\bthree\Ela\fone\wedge\ftwo$}
\AxiomC{$\btwo\wedge\bthree\vDash^{X}\bone$}
\RightLabel{$\Rcap$}
\TrinaryInfC{$\vdash^{X}\bone\Ela \fone\wedge\ftwo$}
\end{prooftree}

$\lone=\bone\Era\BOX^{q}_{a}\fone$.
Assume $\seqO=\vdash^{X}\bone\Era\BOX^{q}_{a}\fone$
and that the considered $\Dec$ is based on the following
$\dec$-decomposition
$$
\vdash^{X}\bone\Era\BOX^{q}_{a}\fone \ \ \dec \ \ \{\vdash^{X\cup\{a\}}\btwo\Era\fone\}
$$
for some $\btwo$ such that 
$\bigvee \bfour_{i}\wedge\bthree_{i}$
is an $a$-decomposition of $\btwo$ and that
$\btwo\vDash^{X}\bigvee\{\bfour_{i}$ $|$ 
$\mu(\model{\bthree_{i}})\ge q\}$.
For hypothesis $\ssetO'$ is derivable, so also
$\vdash^{X\cup\{a\}} \btwo\Era\fone$ ($\in\ssetO'$)
is and, for the given definition of
$\btwo$, $\vdash^{X}\btwo\Era\BOX^{q}_{a}\fone$
is derived by $\Rbr$ as follows:
\begin{prooftree}
\AxiomC{$\vdots$}
\noLine
\UnaryInfC{$\vdash^{X\cup\{a\}}\btwo\Era\fone$}
\AxiomC{$\bone\vDash^{X}\bigvee\{\bfour_{i} \ | \ \mu\big(\model{\bthree_{i}}_{\{a\}}\big)\ge q\}$}
\RightLabel{$\Rbr$}
\BinaryInfC{$\vdash^{X}\bone\Era\BOX^{q}_{a}\fone$}
\end{prooftree}

$\lone=\bone\Ela\BOX^{q}_{a}\fone$.
Assume $\seqO=\vdash^{X}\bone\Ela\BOX^{q}_{a}\fone$
and that the considered $\Dec$ is based on the following
$\dec$-decomposition
$$
\vdash^{X}\bone\Ela\BOX^{q}_{a}\fone \ \ \dec \ \ \{\vdash^{X\cup\{a\}}\btwo\Ela\fone\}
$$
for some $\btwo$ such that $\bigvee\bfour_{i}\wedge\bthree_{i}$
is an $a$-decomposition of $\btwo$ and
that $\bigvee\{\bfour_{i}$ $|$ $\mu(\model{\bthree_{i}})\} \vDash^{X}\bone$. 
For hypothesis, $\ssetO'$ is derivable, so also
$\vdash^{X\cup\{a\}}\btwo\Ela\fone$ ($\in\ssetO'$)
is and, for the given definition of $\btwo$,
$\vdash^{X}\btwo\Ela\BOX^{q}_{a}\fone$
is derived by $\Rbl$ as follows:
\begin{prooftree}
\AxiomC{$\vdots$}
\noLine
\UnaryInfC{$\vdash^{X\cup\{a\}}\btwo\Ela\fone$}
\AxiomC{$\bigvee\{\bfour_{i} \ | \ \mu\big(\model{\bthree_{i}}_{\{a\}}\big)\vDash^{X}\bone$}
\RightLabel{$\Rbr$}
\BinaryInfC{$\vdash^{X}\bone\Ela\BOX^{q}_{a}\fone$}
\end{prooftree}

$\lone=\bone\Era\DIA^{q}_{a}\fone$.
Assume $\seqO=\vdash^{X}\bone\Era\DIA^{q}_{a}\fone$
and that the considered $\Dec$ is based on the following
$\dec$-decomposition
$$
\vdash^{X}\bone\Era\DIA^{q}_{a}\fone \ \ \dec \ \ \{\vdash^{X\cup\{a\}} \btwo \Ela\fone\}
$$
for some $\btwo$ such that $\bigvee\bfour_{i}\wedge\bthree_{i}$
is an $a$-decomposition of $\btwo$ and that
$\btwo\vDash^{X}\bigvee\{\bfour_{i} \ | \ \mu(\model{\bthree_{i}})<q\}$.
For hypothesis, $\ssetO'$ is derivable, so
also $\vdash^{X\cup\{a\}}\btwo\Ela\fone$ ($\in\ssetO'$)
is and, for the given definition of
$\btwo$, $\vdash^{X}\btwo\Era\DIA^{q}_{a}\fone$ is derivable
by $\Rdr$ as follows:
\begin{prooftree}
\AxiomC{$\vdots$}
\noLine
\UnaryInfC{$\vdash^{X\cup\{a\}}\btwo\Ela\fone$}
\AxiomC{$\bone\vDash^{X}\bigvee\{\bfour_{i} \ | \ \mu\big(\model{\bthree_{i}}_{\{a\}}\big)<q$}
\RightLabel{$\Rdr$}
\BinaryInfC{$\vdash^{X}\bone\Era\DIA^{q}_{a}\fone$}
\end{prooftree}

$\lone=\bone\Ela\DIA^{q}_{a}\fone$.
Assume $\seqO=\vdash^{X}\bone\Ela\DIA^{q}_{a}\fone$
and that the considered $\Dec$ is based on the following
$\dec$-decomposition
$$
\vdash^{X}\bone\Ela\DIA^{q}_{a}\fone \ \ \dec \ \ \{\vdash^{X\cup\{a\}}\btwo\Era\fone\}
$$
for some $\btwo$ such that $\bigvee\bfour_{i}\wedge\bthree_{i}$
is an $a$-decomposition of $\btwo$
and that $\bigvee\{\bfour_{i} \ | \ \mu(\model{\bthree_{i}})<q\}\vDash^{X}\bone$.
For hypothesis, $\ssetO'$ is derivable, so also
$\vdash^{X\cup\{a\}}\btwo\Era\fone$ ($\in\ssetO'$)
is and, for the given definition of $\btwo$,
$\vdash^{X}\btwo\Ela\DIA^{q}_{a}\fone$
is derived by $\Rdl$ as follows:
\begin{prooftree}
\AxiomC{$\vdots$}
\noLine
\UnaryInfC{$\vdash^{X\cup\{a\}}\btwo\Era\fone$}
\AxiomC{$\bigvee\{\bfour_{i} \ | \ \mu\big(\model{\bthree_{i}}_{\{a\}}\big) \vDash^{X}\bone$}
\RightLabel{$\Rdl$}
\BinaryInfC{$\vdash^{X}\bone\Ela\DIA^{q}_{a}\fone$}
\end{prooftree}
\end{proof}

\begin{theorem}[Completeness of $\GPPL$]\label{CPLcompleteness}
If a sequent of $\GPPL$ is valid, then it is derivable
in $\GPPL$.
\end{theorem}
\begin{proof}
By combining the given Lemma~\ref{existentialPreservation}, 
Corollary~\ref{cor:NormalValidDerivable}, and
Lemma~\ref{lemma:derivability}.
\end{proof}

\section{Proofs from Section \ref{section5}}\label{appendix5}

\subsection{Prenex Normal Forms}\label{appendix5.1}

\begin{lemma}\label{Lq0}
For every $X$ such that $\Var{\fone} \subseteq X$ and $a \in X$,
it holds that:
$$
\model{\BOX^{0}_{a}\fone}_{X} = (\twoOm)^{X} \qquad \qquad \qquad \qquad \model{\DIA^{0}_{a}\fone}_{X} = \emptyset^{X}.
$$ 
\end{lemma}
\begin{proof}
Let us consider $\model{\BOX^{0}_{a}\fone}_{X} = (\twoOm)^{X}.$ 
Since for every $\PROJ_{f}(\model{\fone}_{X\cup\{a\}})$, also $\mu\big(\PROJ_{f}(\model{\fone}_{X\cup\{a\}})\big) \ge0$ holds. 
So, trivially,
\begin{align*}
\model{\BOX^{0}_{a}\fone}_{X} &= \{f \in (\twoOm)^{X} \ | \ \mu\big(\PROJ_{f}(\model{\fone}_{X\cup\{a\}})\big) \ge 0\} \\
&= (\twoOm)^{X}.
\end{align*}
Let us now consider
$\model{\DIA^{0}_{a}\fone}_{X} = \emptyset^{X}$. 
Since for not $\PROJ_{f}(\model{\fone}_{X\cup\{a\}})$, 
also $\mu\big(\PROJ_{f}(\model{\fone}_{X\cup\{a\}})\big)<0$ holds. 
So, trivially,
\begin{align*}
\model{\DIA^{0}_{a}\fone}_{X} &= \{f \in (\twoOm)^{X} \ | \ \mu\big(\PROJ_{f}(\model{\fone}_{X\cup\{a\}})\big) < 0\} \\
&= \emptyset^{X}.
\end{align*}
\end{proof}
\noindent
Without loss of generality, let us assume $a \not \in \Var{\fone}$.
\begin{corollary}
For every $X$ such that $\Var{\fone} \cup \Var{\ftwo} \subseteq X$ and $a \in X$, it holds that:
\begin{align*}
\fone \wedge \BOX^{0}_{a}\ftwo &\equiv_{X} \fone &&& \fone \vee \BOX^{0}_{a}\ftwo &\equiv_{X} (\twoOm)^{X} \\
\fone \wedge \DIA^{0}_{a}\ftwo &\equiv_{X} \emptyset^{X} &&& \fone \vee \DIA^{0}_{a}\ftwo &\equiv_{X} \fone.
\end{align*}
\end{corollary}
\begin{proof}
Let us consider
$\fone \wedge \BOX^{0}_{a}\ftwo \equiv_{X} \fone$. 
Thus,
$$
\model{\fone \wedge \BOX^{0}_{a}\ftwo}_{X} = \model{\fone}_{X}.
$$
By Definition~\ref{PPLsem},
\begin{align*}
\model{\fone \wedge \BOX^{0}_{a}\ftwo}_{X} = \ \ \ &\model{\fone}_{X} \cap \model{\BOX^{0}_{a}\ftwo}_{X} \\
\stackrel{\text{L}~\ref{Lq0}}{=} &\model{\fone}_{X} \cap (\twoOm)^{X} \\
= \ \ \ &\model{\fone}_{X}.
\end{align*}
Let us now consider $\fone \vee \BOX^{0}_{a}\ftwo \equiv_{X} (\twoOm)^{X}$. 
Thus,
$$
\model{\fone \vee \BOX^{0}_{a}\ftwo}_{X} = (\twoOm)^{X}.
$$
Again for Definition~\ref{PPLsem},\footnote{Well-definition is ensured by the fact that $a \not \in X$.}
\begin{align*}
\model{\fone \vee \BOX^{0}_{a}\ftwo}_{X} = \ \ \ &\model{\fone}_{X} \cup \model{\BOX^{0}_{a}\ftwo}_{X} \\
\stackrel{\text{L}~\ref{Lq0}}{=} &\model{\fone}_{X} \cup (\twoOm)^{X} \\
= \ \ \ &(\twoOm)^{X}.
\end{align*}
Let us consider 
$\fone \wedge \DIA^{0}_{a}\ftwo \equiv_{X} \emptyset^{X}$. 
Thus,
$$
\model{\fone \wedge \DIA^{0}_{a}\ftwo}_{X} = \emptyset^{X}.
$$
For Definition~\ref{PPLsem},
\begin{align*}
\model{\fone \wedge \DIA^{0}_{a}\ftwo}_{X} = \ \ \ &\model{\fone}_{X} \cap \model{\DIA^{0}_{a}\ftwo}_{X} \\
\stackrel{\text{L}~\ref{Lq0}}{=} &\model{\fone}_{X} \cap \emptyset^{X} \\
= \ \ \ &\emptyset^{X}.
\end{align*}
Let us finally consider
$\fone \vee \DIA^{0}_{a}\ftwo \equiv_{X} \fone$. 
Thus,
$$
\model{\fone \vee \DIA^{0}_{a}\ftwo}_{X} = \model{\fone}_{X}.
$$
For Definition~\ref{PPLsem},
\begin{align*}
\model{\fone \vee \DIA^{0}_{a}\ftwo}_{X} = \ \ \ &\model{\fone}_{X} \cup \model{\DIA^{0}_{a}\ftwo}_{X} \\
\stackrel{\text{L}~\ref{Lq0}}{=} &\model{\fone}_{X} \cup \emptyset^{X} \\
= \ \ \ &\model{\fone}_{X}.
\end{align*}
\end{proof}

In order to prove Lemma~\ref{lemma:commutations}, some auxiliary lemmas has to be preliminarily considered.

\begin{lemma}\label{L1}
Given $\mathcal{Z} \subseteq (\twoOm)^{X\cup Y}$ and $f \in \mathcal{X} \subseteq (\twoOm)^{X}$, then $\PROJ_{f}(\mathcal{Z}) = \PROJ_{f}(\mathcal{Z} \cap \EXT{X}{Y})$.
\end{lemma}
\begin{proof}
\begin{itemize}
\item[$\subseteq$] If $g \in \PROJ_{f}$($\mathcal{Z}$), then $f + g \in \mathcal{Z}$. Since $f \in \mathcal{X},$ by $\EXT{X}{Y}$ definition, $f + g \in \EXT{X}{Y}$. As a consequence, $f + g \in \mathcal{Z} \cap \EXT{X}{Y}$ and $g \in \PROJ_{f}$\big($\mathcal{Z} \cap \EXT{X}{Y}$\big).
\item[$\supseteq$] Trivial, since the projection operator is monotone.
\end{itemize}
\end{proof}

\begin{lemma}\label{L2}
Let $f \in (\twoOm)^{X}$, 
$\mathcal{Z} \subseteq (\twoOm)^{X\cup Y}$, and 
$\mathcal{X} \subseteq (\twoOm)^{X}$, 
with $X\cap Y = \emptyset$. 
If $\mu\big(\PROJ_{f}(\mathcal{Z} \cap \EXT{\mathcal{X}}{Y})\big) > 0$, then $f \in \mathcal{X}$.
\end{lemma}
\begin{proof}
By contraposition. Assume $f \not \in \mathcal{X}$. There are two possible cases.
\begin{itemize}
\itemsep0em
\item If $\mathcal{Z} \cap \EXT{X}{Y}$ = $\emptyset$, then
\begin{align*}
\mu(\PROJ_{f}\big(\mathcal{Z}\cap \EXT{X}{Y})\big) &= \mu\big(\PROJ_{f}(\emptyset)\big) \\
&= \mu(\emptyset) \\
&= 0.
\end{align*}
\normalsize
\item Otherwise,
\begin{align*}
\mu\big(\PROJ_{f}(\mathcal{Z} \cap \EXT{X}{Y})\big) &\ = \mu\big(\{g \in (\twoOm)^{Y} \ | \ f + g \in (\mathcal{Z} \cap \EXT{X}{Y})\}\big) \\
&\stackrel{f \not \in \mathcal{X}}{=} \mu(\emptyset) \\
&\ = 0.
\end{align*}
\normalsize
\noindent
Indeed by definition, $\EXT{X}{Y}$ = $\{h + y \in (\twoOm)^{X\cup Y}$ $|$ $h \in \mathcal{X}\}$. So (since $X \cap Y$ = $\emptyset$), if $f \not \in \mathcal{X}$, $\{g \in$ ($\twoOm$)$^{Y}$ $|$ $f + g \in \EXT{X}{Y}\}$ = $\emptyset$. Trivially, also $\{g \in$ ($\twoOm$)$^{Y}$ $|$ $f + g \in \mathcal{Z} \cap \EXT{X}{Y}\}$ = $\emptyset$. 
\end{itemize}
\end{proof}
\noindent
It is possible to prove that, for $a \not \in X, \Var{\fone} \subseteq X,$ and $f \in$ ($\twoOm$)$^{X}$, if $f \in \model{\fone}_{X},$ then $\PROJ_{f}$($\model{\fone}_{X\cup\{a\}}$) = ($\twoOm$)$^{\{a\}}$. Intuitively, since $a \not \in \Var{\fone}, \model{\fone}_{X\cup\{a\}}$ is a set of $f \in$ ($\twoOm$)$^{X\cup\{a\}}$, the characterization of which \emph{does not} rely on $\{a\}$. So, for each $\alpha \in Y$, either $f$($\alpha$) $\in \model{\fone}_{X}$, and $f + g \in \model{\fone}_{X\cup\{a\}}$, or $f$($\alpha$) $\not \in \model{\fone}_{X}$, and $f + g \not \in \model{\fone}_{X\cup\{a\}}$. Therefore if $f \in \model{\fone}_{X}, \PROJ_{f}(\model{\fone}_{X\cup\{a\}}$) = ($\twoOm$)$^{\{a\}}$

\begin{lemma}\label{L3}
Let $a \not \in X, \Var{\fone} \subseteq X,$ and $f \in (\twoOm)^{X}$:
\begin{itemize}
\itemsep0em
\item if $f \in \model{\fone}_{X},$ then $\PROJ_{f}(\model{\fone}_{X\cup\{a\}}) = (\twoOm)^{\{a\}}$
\item if $f \not \in \model{\fone}_{X},$ then $\PROJ_{f}(\model{\fone}_{X\cup\{a\}}) = \emptyset^{\{a\}}$
\end{itemize}
\big(or equivalently, $\model{\fone}_{X\cup\{a\}} = \model{\fone}_{X}^{\Uparrow\{a\}}$\big).
\end{lemma}
\begin{proof}
Since $a \not \in \Var{\fone}$ for every $a \in Y$, there are two possible cases:
\begin{itemize}
\itemsep0em
\item if $f$($\alpha$) $\in \model{\fone}_{X}$, then $f + g \in \model{\fone}_{X\cup\{a\}}$
\item if $f$($\alpha$) $\not \in \model{\fone}_{X}$, then $f + g \not \in \model{\fone}_{X\cup\{a\}}$.
\end{itemize}
Thus, Let $f \in \model{\fone}_{X}$. Then,
\begin{align*}
\PROJ_{f}(\model{\fone}_{X\cup\{a\}}) &= \{g \in (\twoOm)^{\{a\}} \ | \ f + g \in \model{\fone}_{X\cup\{a\}}\} \\
&= \{g \in (\twoOm)^{\{a\}} \ | \ f \in \model{\fone}_{X}\} \\
&= (\twoOm)^{\{a\}}.
\end{align*}
\normalsize
Let $f \not \in \model{\fone}_{X}$. Then,
\begin{align*}
\PROJ_{f}(\model{\fone}_{X\cup\{a\}}) &= \{g \in (\twoOm)^{\{a\}} \ | \ f + g \in \model{\fone}_{X\cup\{a\}}\} \\
&= \{g \in (\twoOm)^{\{a\}} \ | \ f \in \model{\fone}_{X}\} \\
&= \emptyset^{\{a\}}.
\end{align*}
\normalsize
\end{proof}

\begin{corollary}\label{C1}
Let $a \not \in X, \Var{\fone} \subseteq X, f \in (\twoOm)^{X},$ and $q > 0$. Then:
\begin{align*}
\mu\big(\PROJ_{f}(\model{\fone}_{X\cup\{a\}})\big) \ge q &\text{ iff } f \in \model{\fone}_{X} \\
\mu\big(\PROJ_{f}(\model{\fone}_{X\cup\{a\}}) \cap \PROJ_{f}(\model{\ftwo}_{X\cup\{a\}})\big) \ge q \wedge \mu\big(\PROJ_{f}(\model{\ftwo}_{X\cup\{a\}})\big) \ge q &\text{ iff } f \in \model{\fone}_{X} \\
\mu\big(\PROJ_{f}(\model{\fone}_{X\cup\{a\}}) \cup \PROJ_{f}(\model{\ftwo}_{X\cup\{a\}})\big) \ge q \vee \mu\big(\PROJ_{f}(\model{\ftwo}_{X\cup\{a\}})\big) \ge q &\text{ iff } f \in \model{\fone}_{X} \\
\mu\big(\PROJ_{f}(\model{\fone}_{X\cup\{a\}}) \cup \PROJ_{f}(\model{\ftwo}_{X\cup\{a\}})\big) < q \wedge 
\mu\big(\PROJ_{f}(\model{\ftwo}_{X\cup\{a\}})\big) < q &\text{ iff } f \not \in \model{\fone}_{X}
\end{align*}
\end{corollary}
\begin{proof} First, let us consider:
$$
\mu\big(\PROJ_{f}(\model{\fone}_{X\cup\{a\}})\big) \ge q \text{ iff } f \in \model{\fone}_{X}
$$
For every $\alpha' \in X$, there are two possible cases:
\begin{itemize}
\item $f$($\alpha'$) $\in \model{\fone}_{X}$. Then $\PROJ_{f}$($\model{\fone}_{X\cup\{a\}}$) = $\{g \in$ ($\twoOm$)$^{\{a\}}$ $|$ $f$($\alpha'$) $\in \model{\fone}_{X}\}$ = ($\twoOm$)$^{\{a\}}$. So, for every $q$, $\PROJ_{f}$($\model{\fone}_{X\cup\{a\}}$) $\ge q$
\item $f$($\alpha'$) $\not \in \model{\fone}_{X}$. Then $\PROJ_{f}$($\model{\fone}_{X\cup\{a\}}$) = $\{g \in$ ($\twoOm$)$^{\{a\}}$ $|$ $f$($\alpha'$) $\in \model{\fone}_{X}\}$ = $\emptyset^{\{a\}}$. So, for every $q$ ($q \neq$ 0), $\PROJ_{f}$($\model{\fone}_{X\cup\{a\}}$) $< q$.
\end{itemize}
Let us now consider:
$$
\mu\big(\PROJ_{f}(\model{\fone}_{X\cup\{a\}}) \cap \PROJ_{f}(\model{\ftwo}_{X\cup\{a\}})\big) \ge q \wedge 
\mu\big(\PROJ_{f}(\model{\ftwo}_{X\cup\{a\}})\big) \ge q \text{ iff } f \in \model{\fone}_{X}.
$$
There are two cases to be taken into account:
\begin{itemize}
\itemsep0em
\item $f \in \model{\fone}_{X}$.
For Lemma~\ref{L3} ($a \not \in X$), $\PROJ_{f}(\model{\fone}_{X\cup\{a\}})$
= $(\twoOm)^{\{a\}}$, so
\begin{align*}
\mu\big(\PROJ_{f}(\model{\fone}_{X\cup\{a\}}) \cap \PROJ_{f}(\model{\ftwo}_{X\cup\{a\}})\big) 
&= \mu\big((\twoOm)^{\{a\}} \cap \PROJ_{f}(\model{\ftwo}_{X\cup\{a\}})\big) \\
&= \mu\big(\PROJ_{f}(\model{\ftwo}_{X\cup\{a\}})\big).
\end{align*}
\normalsize
\item $f \not \in \model{\fone}_{X}$.
For Lemma~\ref{L3}, $\PROJ_{f}$($\model{\fone}_{X\cup\{a\}}$) = $\emptyset^{X}$, so
\begin{align*}
\mu\big(\PROJ_{f}(\model{\fone}_{X\cup\{a\}}) \cap \PROJ_{f}(\model{\ftwo}_{X\cup\{a\}})\big) 
&= \mu\big(\emptyset^{\{a\}} \cap \PROJ_{f}(\model{\ftwo}_{X\cup\{a\}})\big) \\
&= \mu\big(\emptyset^{\{a\}}\big) \\
&= 0.
\end{align*}
\normalsize
\end{itemize}
Therefore,
\begin{itemize}
\item[$\Leftarrow$] If $f \in \model{\fone}_{X}$ and $\mu\big(\PROJ_{f}(\model{\ftwo}_{X\cup\{a\}})\big)\ge q$, 
then for the first clause 
$\mu\big(\PROJ_{f}(\model{\fone}_{X\cup\{a\}}) \cap \PROJ_{f}(\model{\ftwo}_{X\cup\{a\}})\big)$ = 
$\mu\big(\PROJ_{f}(\model{\fone}_{X\cup\{a\}}))$. 
Therefore, $\mu\big(\PROJ_{f}(\model{\fone}_{X\cup\{a\}}) \cap \PROJ_{f}(\model{\ftwo}_{X\cup\{a\}})\big) \ge q$.
\item[$\Rightarrow$] If $f \not \in \model{\fone}_{X},$ then for the second clause 
$\mu\big(\PROJ_{f}(\model{\fone}_{X\cup\{a\}}) \cap
\PROJ_{f}(\model{\ftwo}_{X\cup\{a\}})\big) = 0$. 
Therefore, $\mu\big(\PROJ_{f}(\model{\fone}_{X\cup\{a\}}) \cap \PROJ_{f}(\model{\ftwo}_{X\cup\{a\}})\big) < q$.
\\ If $f \in \model{\fone}_{X}$ and 
$\mu\big(\PROJ_{f}(\model{\ftwo}_{X\cup\{a\}})\big) < q$, 
for the first clause $\mu\big(\PROJ_{f}(\model{\fone}_{X\cup\{a\}} \cap \PROJ_{f}(\model{\ftwo}_{X\cup\{a\}})\big)$ = 
$\mu\big(\PROJ_{f}(\model{\ftwo}_{X\cup\{a\}})\big)$. 
So, $\mu\big(\PROJ_{f}(\model{\fone}_{X\cup\{a\}}) \cap \PROJ_{f}(\model{\ftwo}_{X\cup\{a\}})\big) < q$. 
\end{itemize}
It is now possible to take into account:
$$
\mu\big(\PROJ_{f}(\model{\fone}_{X\cup\{a\}}) \cup \PROJ_{f}(\model{\ftwo}_{X\cup\{a\}})\big) \ge q \vee 
\mu\big(\PROJ_{f}(\model{\ftwo}_{X\cup\{a\}})\big) \ge q 
\text{ iff } f \in \model{\fone}_{X}.
$$
There are two cases to be considered:
\begin{itemize}

\item $f \in \model{\fone}_{X}$. For Lemma~\ref{L3}, ($a \not \in X$), $\PROJ_{f}$($\model{\fone}_{X\cup\{a\}}$) = ($\twoOm$)$^{\{a\}}$, so
\begin{align*}
\mu\big(\PROJ_{f}(\model{\fone}_{X\cup\{a\}}) \cup \PROJ_{f}(\model{\ftwo}_{X\cup\{a\}})\big) 
&= \mu\big((\twoOm)^{\{a\}} \cup \PROJ_{f}(\model{\ftwo}_{X\cup\{a\}})\big) \\
&= \mu\big((\twoOm)^{\{a\}}\big) \\
&= 1.
\normalsize
\end{align*}
\item $f \not \in \model{\fone}_{X}$. For Lemma~\ref{L3}, $\PROJ_{f}$($\model{\fone}_{X\cup\{a\}}$) = $\emptyset^{X}$, so
\begin{align*}
\mu\big(\PROJ_{f}(\model{\fone}_{X\cup\{a\}}) \cup \PROJ_{f}(\model{\ftwo}_{X\cup\{a\}})\big) &= 
\mu\big(\emptyset^{\{a\}} \cup \PROJ_{f}(\model{\ftwo}_{X\cup\{a\}})\big) \\
&= \mu\big(\PROJ_{f}(\model{\ftwo}_{X\cup\{a\}})\big).
\end{align*}
\end{itemize}
\normalsize
Therefore,
\begin{itemize}
\item[$\Leftarrow$] If $f \in \model{\fone}_{X},$ then, for the first clause, $\mu$($\PROJ$($\model{\fone}_{X\cup\{a\}}$) $\cup$ $\PROJ_{f}$($\model{\ftwo}_{X\cup\{a\}}$)) = 1. 
So, for every $q$, $\mu$($\PROJ_{f}$($\model{\fone}_{X\cup\{a\}}$) $\cup$ $\PROJ$($\model{\ftwo}_{X\cup\{a\}}$)) $\ge q$.
\\ If $\mu$($\PROJ_{f}$($\model{\ftwo}_{X\cup\{a\}}$)) $\ge q$, then also $\mu$($\PROJ_{f}$($\model{\fone}_{X\cup\{a\}}$) $\cup$ $\PROJ_{f}$($\model{\ftwo}_{X\cup\{a\}}$)) $\ge q$.
\item[$\Rightarrow$] If $f \not \in \model{\fone}_{X}$ and $\mu$($\PROJ_{f}$($\model{\ftwo}_{X\cup\{a\}}$)) $< q$, then, for the second clause, $\mu$($\PROJ_{f}$($\model{\fone}_{X\cup\{a\}}$) $\cup$ $\PROJ_{f}$($\model{\ftwo}_{X\cup\{a\}}$)) = $\mu$($\PROJ_{f}$($\model{\ftwo}_{X\cup\{a\}}$)) $< q$.
\end{itemize}
It is now possible to deal with:
$$
\mu\big(\PROJ_{f}(\model{\fone}_{X\cup\{a\}}) \cup \PROJ_{f}(\model{\ftwo}_{X\cup\{a\}})\big) < q 
\wedge \mu\big(\PROJ_{f}(\model{\ftwo}_{X\cup\{a\}})\big) < q 
\text{ iff } f \not \in \model{\fone}_{X}.
$$
Let us take into account two cases:
\begin{itemize}

\item $f \in \model{\fone}_{X}$.
By Lemma~\ref{L3} ($a \not \in X$), 
$\PROJ_{f}\big(\model{\fone}_{X\cup\{a\}}\big)$ = ($\twoOm$)$^{\{a\}}$, so
\begin{align*}
\mu\big(\PROJ_{f}(\model{\fone}_{X\cup\{a\}}) \cup \PROJ_{f}(\model{\ftwo}_{X\cup\{a\}})\big) 
&= \mu\big((\twoOm)^{\{a\}} \cup \PROJ_{f}(\model{\ftwo}_{X\cup\{a\}})\big) \\
&= \mu\big((\twoOm)^{\{a\}}\big) \\
&= 1
\end{align*}

\item $f \not \in \model{\fone}_{X}$. 
By Lemma~\ref{L3}, $\PROJ_{f}\big(\model{\fone}_{X\cup\{a\}}\big)$ = $\emptyset^{\{a\}}$, so
\begin{align*}
\mu\big(\PROJ_{f}(\model{\fone}_{X\cup\{a\}}) \cup \PROJ_{f}(\model{\ftwo}_{X\cup\{a\}})\big) 
&= \mu\big(\emptyset^{\{a\}} \cup \PROJ_{f}(\model{\ftwo}_{X\cup\{a\}})\big) \\
& = \mu\big(\PROJ_{f}(\model{\ftwo}_{X\cup\{a\}})\big).
\end{align*}
\normalsize
\end{itemize}
Therefore,
\begin{itemize}
\item[$\Leftarrow$] Let $f \not \in \model{\fone}_{X}$ and 
$\mu\big(\PROJ_{f}(\model{\fone}_{X\cup\{a\}})\big) < q$. 
For the second clause, 
$\mu\big(\PROJ_{f}(\model{\fone}_{X\cup\{a\}}) \cup \PROJ_{f}(\model{\ftwo}_{X\cup\{a\}})\big)$ 
= $\mu\big(\PROJ_{f}(\model{\ftwo}_{X\cup\{a\}})\big)$. 
Since for hypothesis $\mu\big(\PROJ_{f}(\model{\ftwo}_{X\cup\{a\}})\big) < q$, also
$\mu\big(\PROJ_{f}(\model{\fone}_{X\cup\{a\}}) \cup \PROJ_{f}(\model{\ftwo}_{X\cup\{a\}})\big) < q$.
\item[$\Rightarrow$] Let $f \in \model{\fone}_{X}$. 
For the first clause 
$\mu\big(\PROJ_{f}(\model{\fone}_{X\cup\{a\}}) \cup \PROJ_{f}(\model{\ftwo}_{X\cup\{a\}})\big) = 1$. 
Thus, for every $q$, 
$\mu\big(\PROJ_{f}(\model{\fone}_{X\cup\{a\}})
\cup \PROJ_{f}(\model{\ftwo}_{X\cup\{a\}})\big) > q$.
\\ Let $f \not \in \model{\fone}_{X}$ and $\mu\big(\PROJ_{f}(\model{\ftwo}_{X\cup \{a\}})\big) \ge q$. 
For the second clause, $\mu\big(\PROJ_{f}(\model{\fone}_{X\cup\{a\}}) \cup \PROJ_{f}(\model{\ftwo}_{X\cup\{a\}})\big)$ 
= $\mu\big(\PROJ_{f}(\model{\ftwo}_{X\cup\{a\}})\big)$. 
For hypothesis, $\mu\big(\PROJ_{f}(\model{\fone}_{X\cup\{a\}})\big) \ge q$, 
thus also $\mu\big(\PROJ_{f}(\model{\fone}_{X\cup\{a\}}) \cup \PROJ_{f}(\model{\ftwo}_{X\cup\{a\}})\big) \ge q$.
\end{itemize}
\end{proof}

\primeCommutations*
\begin{proof}
Let us show:
 $$
 \model{\fone \wedge \BOX^{q}_{a}\ftwo}_{X} = \model{\BOX^{q}_{a}(\fone \wedge \ftwo)}_{X}.
 $$
 For the Definition~\ref{PPLsem},\footnote{That $\BOX^q_{a}$ is well-defined is guaranteed by the fact that $X$ = $X_{\fone} \cup (X_{\ftwo}/\{a\})$ so, for hypothesis, $a \not \in X$.}
$$
\model{\fone \wedge \BOX^{q}_{a}\ftwo}_{X} = \model{\fone}_{X} \cap \model{\BOX^{q}_{a}\ftwo}_{X}.
$$
Let $f \in$ ($\twoOm$)$^{X}$. There are two cases:
\begin{itemize}

\item $\mu\big(\PROJ_{f}(\model{\ftwo}_{X\cup\{a\}})\big) \ge q$.
Since $\{f \in (\twoOm)^{X} \ |$ $\mu\big(\PROJ_{f}(\model{\ftwo})_{X\cup\{a\}})\big) \ge q\}$ = ($\twoOm$)$^{X}$,
\begin{align*}
\model{\fone \wedge \BOX^{q}_{a}\ftwo}_{X} &= \model{\fone}_{X} \cap \model{\BOX^{q}_{a}\ftwo}_{X} \\
&= \model{\fone}_{X} \cap (\twoOm)^{X} \\
&= \model{\fone}_{X}
\end{align*}

\item $\mu\big(\PROJ_{f}(\model{\fone}_{X\cup\{a\}})\big) < q$. Since $\{f \in$ ($\twoOm$)$^{X}$ $|$ $\mu\big(\PROJ_{f}(\model{\fone}_{X})\big) \ge q\}$ = $\emptyset^{X}$,
\begin{align*}
\model{\fone \wedge \BOX^{q}_{a}\ftwo}_{X} &= \model{\fone}_{X} \cap \model{\BOX^{q}_{a}\ftwo}_{X} \\
&= \model{\fone}_{X} \cap \emptyset^{X} \\
&= \emptyset^{X}.
\end{align*}
\end{itemize}
So,
$$
\model{\fone \wedge \BOX^{q}_{a}\ftwo}_{X} = \begin{cases} \model{\fone}_{X} \ \ \ &\text{ if } \mu(\PROJ_{f}(\model{\ftwo}_{X\cup\{a\}}) \ge q \\ \emptyset^{X} \ \ \ &\text{ otherwise} \end{cases}
$$
On the other hand,
\begin{align*}
\model{\BOX^{q}_{a}(\fone \wedge \ftwo)}_{X} &= \{f \in (\twoOm)^{X} \ | \ \mu\big(\PROJ_{f}(\model{\fone \wedge \ftwo}_{X\cup\{a\}})\big) \ge q\} \\
&= \{f \in (\twoOm)^{X} \ | \ \mu\big(\PROJ(\model{\fone}_{X\cup\{a\}} \cap \model{\ftwo}_{X\cup\{a\}})\big) \ge q \} \\
&= \{f \in (\twoOm)^{X} \ | \ \mu\big(\PROJ_{f}(\model{\fone}_{X\cup\{a\}}) \cap \PROJ_{f}(\model{\ftwo}_{X\cup\{a\}})\big) \ge q\}
\end{align*}
\normalsize
There are two possible cases:
\begin{itemize}
\itemsep0em
\item Let $\mu\big(\PROJ_{f}(\model{\ftwo}_{X\cup\{a\}})\big) \ge q$. Then,
\begin{align*}
\model{\BOX^{q}_{a}(\fone \wedge \ftwo)}_{X} = \ \ &\{f \in (\twoOm)^{X} \ | \ \mu(\PROJ_{f}(\model{\fone}_{X\cup\{a\}})  \cap \PROJ_{f}(\model{\ftwo}_{X\cup\{a\}})) \ge q\} \\
\stackrel{\text{C}~\ref{C1}}{=} &\{f \in (\twoOm)^{X} \ |  \ f \in \model{\fone}_{X}  \wedge \mu(\PROJ_{f}(\model{\ftwo}_{X\cup\{a\}})) \ge q\} \\
= \ \ &\{f \in (\twoOm)^{X} \ | \ f \in \model{\fone}_{X}\} \\
= \ \ &\model{\fone}_{X}
 \end{align*}
\item Let $\mu\big(\PROJ_{f}(\model{\ftwo}_{X\cup\{a\}})\big)< q$. 
Since $\mu\big(\PROJ_{f}(\model{\fone}_{X\cup\{a\}}) \cap (\PROJ_{f}(\model{\ftwo}_{X\cup\{a\}})\big)$ $\leq$ 
$\mu\big(\PROJ_{f}(\model{\ftwo}_{X\cup\{a\}})\big)$, 
then 
$\mu\big(\PROJ_{f}(\model{\fone}_{X\cup\{a\}}) \cap \PROJ_{f}(\model{\ftwo}_{X\cup\{a\}})\big) < q$. 
So, $\{f \in (\twoOm$)$^{X}$ $|$ 
$\mu\big(\PROJ_{f}(\model{\fone}_{X\cup\{a\}}) \cap \PROJ_{f}(\model{\ftwo}_{X\cup\{a\}})\big) \ge q\} = \emptyset$:
\begin{align*}
\model{\BOX^{q}_{a}(\fone \wedge \ftwo)}_{X} &= \{f \in (\twoOm)^{X} \ | \ \mu\big(\PROJ_{f}(\model{\fone}_{X\cup\{a\}}) \cap \PROJ_{f}(\model{\ftwo}_{X\cup\{a\}})\big) \ge q\} \\
&= \emptyset^{X}.
\end{align*}
\end{itemize}
\normalsize
Therefore,
\begin{align*}
\model{\BOX^{q}_{a}(\fone \wedge \ftwo)}_{X} &= \begin{cases} \model{\fone}_{X} \ \ \ &\text{ if } \mu(\PROJ_{f}(\model{\ftwo}_{X\cup \{a\}})) \ge q \\ \emptyset^{X} \ \ \ &\text{ otherwise} 
\end{cases} \\
&= \model{\fone \wedge \BOX^{q}_{a}\ftwo}_{X}
\end{align*}
\normalsize
Let us now show that:
$$
\model{\fone \vee \BOX^{q}_{a}\ftwo}_{X} = \model{\BOX^{q}_{a}(\fone \vee \ftwo)}_{X}.
$$
By Definition~\ref{PPLsem},
$$
\model{\fone \vee \BOX^{q}_{a}\ftwo}_{X} = \model{\fone}_{X} \cup \model{\BOX_{a}^{q}\ftwo}_{X}
$$
Let $f \in$ ($\twoOm$)$^{X}$. There are two possible cases:
\begin{itemize}

\item $\mu\big(\PROJ_{f}(\model{\ftwo}_{X\cup\{a\}})\big) \ge q$. 
Since $\{f \in$ ($\twoOm$)$^{X}$ $|$ $\mu\big(\PROJ_{f}(\model{\ftwo}_{X\cup\{a\}})\big) \ge q\}$ = ($\twoOm$)$^{X}$,
\begin{align*}
\model{\fone \vee \BOX^{q}_{a}\ftwo}_{X} &= \model{\fone}_{X} \cup \model{\BOX^{q}_{a}\ftwo}_{X} \\
&= \model{\fone}_{X} \cup (\twoOm)^{X} \\
&= (\twoOm)^{X}.
\end{align*}

\item $\mu\big(\PROJ_{f}(\model{\ftwo}_{X\cup\{a\}})\big) < q$. 
Since $\{f \in (\twoOm)^{X} \ | \ \mu\big(\PROJ_{f}(\model{\ftwo}_{X\cup\{a\}})\big) \ge q\}$ = $\emptyset^{X}$,
\begin{align*}
\model{\fone \vee \BOX^{q}_{a}\ftwo}_{X} &= \model{\fone}_{X} \cup \model{\BOX^{q}_{a}\ftwo}_{X} \\
&= \model{\fone}_{X} \cup \emptyset^{X} \\
&= \model{\fone}_{X}.
\end{align*}
\normalsize
\end{itemize}
Thus,
$$
\model{\fone \vee \BOX^{q}_{a}\ftwo}_{X} = \begin{cases}
(\twoOm)^{X} \ \ &\text{ if } \mu(\PROJ_{f}(\model{\ftwo}_{X\cup \{a\}})) \ge q \\ \model{\fone}_{X} \ \ &\text{ otherwise} 
\end{cases}
$$
On the other hand,
\normalsize
\begin{align*}
\model{\BOX^{q}_{a}(\fone \vee \ftwo)}_{X} &= \{f \in (\twoOm)^{X} \ | \ \mu\big(\PROJ_{f}(\model{\fone \vee \ftwo}_{X\cup\{a\}})\big) \ge q\} \\
&= \{f \in (\twoOm)^{X} \ | \ \mu\big(\PROJ_{f}(\model{\fone}_{X\cup\{a\}} \cup \model{\ftwo}_{X\cup\{a\}})\big) \ge q\} \\
&= \{f \in (\twoOm)^{X} \ | \ \mu\big(\PROJ_{f}(\model{\fone}_{X\cup\{a\}} \cup \PROJ_{f}(\model{\ftwo}_{X\cup\{a\}})\big) \ge q\}.
\end{align*}
\normalsize
There are two possible cases:
\begin{itemize}

\item Let $\mu\big(\PROJ_{f}(\model{\ftwo}_{X\cup\{a\}})\big) \ge q$. 
Since trivially $\mu\big(\PROJ_{f}(\model{\fone}_{X\cup\{a\}}) \cup
\PROJ_{f}(\model{\ftwo}_{X\cup\{a\}})\big)$ $\ge$ 
$\mu\big(\PROJ_{f}(\model{\ftwo}_{X\cup\{a\}})\big)$, 
$\mu\big(\PROJ_{f}(\model{\fone}_{X\cup\{a\}}) \cup \PROJ_{f}(\model{\ftwo}_{X\cup\{a\}})\big) \ge q$. 
So, $\{f \in$ ($\twoOm$)$^{X}$ $|$ $\mu\big(\PROJ_{f}(\model{\fone}_{X\cup\{a\}}) \cup \PROJ_{f}(\model{\ftwo}_{X\cup\{a\}})\big)$ = ($\twoOm$)$^{X}$.

\item Let $\mu\big(\PROJ_{f}(\model{\ftwo}_{X\cup\{a\}})\big)< q$.
Then,
\begin{align*}
\model{\BOX^{q}_{a}(\fone \vee \ftwo)}_{X} = \ &\{f \in (\twoOm)^{X} \ | \ \mu\big(\PROJ_{f}(\model{\fone}_{X\cup\{a\}}) \cup \PROJ_{f}(\model{\ftwo}_{X\cup\{a\}})\big) \ge q\} \\
\stackrel{\text{C}~\ref{C1}}{=} &\{f \in (\twoOm)^{X} \ | \ f \in \model{\fone}_{X} \vee \mu\big(\PROJ_{f}(\model{\fone}_{X\cup\{a\}})\big) \ge q\} \\
= \ &\{f \in (\twoOm)^{X} \ | \ f \in \model{\fone}_{X}\} \\
= \ &\model{\fone}_{X}.
\end{align*}
\end{itemize}
Therefore,
\begin{align*}
\model{\BOX^{q}_{a}(\fone \vee \ftwo)}_{X} &= \begin{cases} (\twoOm)^{X} \ \ &\text{ if } \mu(\PROJ_{f}(\model{\ftwo}_{X\cup\{a\}})) \ge q \\ \model{\fone}_{X} \ \ &\text{ otherwise} 
\end{cases} \\
&= \model{\fone \vee \BOX^{q}_{a}\ftwo}_{X}
\end{align*}
Let us now prove:
$$
\model{\fone \wedge \DIA^{q}_{a}\ftwo}_{X} = \model{\DIA^{q}_{a}(\neg \fone \vee \ftwo)}_{X}
$$
For Definition~\ref{PPLsem},\footnote{Well-definition is guaranteed by the fact that, for hypothesis, $a \not \in X$.}
$$
\model{\fone \wedge \DIA^{q}_{a}\ftwo}_{X} = \model{\fone}_{X} \cap \model{\DIA^{q}_{a}\ftwo}_{X}
$$
Let $f \in$ ($\twoOm$)$^{X}$. There are two cases:
\begin{itemize}

\item Let $\mu\big(\PROJ_{f}(\model{\ftwo}_{X\cup\{a\}})\big) \ge q$. 
Since $\{f \in (\twoOm)^{X} \ | \ \mu\big(\PROJ_{f}(\model{\ftwo}_{X\cup\{a\}})\big) < q\}$ = $\emptyset^{X}$,
\begin{align*}
\model{\fone \wedge \DIA^{q}_{a}\ftwo}_{X} &= \model{\fone}_{X} \cap \model{\DIA^{q}_{a}\ftwo}_{X} \\
&= \model{\fone}_{X} \cap \emptyset^{X} \\
&= \emptyset^{X}.
\end{align*}

\item Let $\mu\big(\PROJ_{f}(\model{\ftwo}_{X\cup\{a\}})\big) < q$.
Since $\{f \in (\twoOm)^{X} \ | \ \mu\big(\PROJ_{f}(\model{\ftwo}_{X})\big) < q\}$ = ($\twoOm$)$^{X}$,
\begin{align*}
\model{\fone \wedge \BOX^{q}_{a}\ftwo}_{X} &= \model{\fone}_{X} \cap \model{\DIA^{q}_{a}\ftwo}_{X} \\
&= \model{\fone}_{X} \cap (\twoOm)^{X} \\
&= \model{\fone}_{X}.
\end{align*}
\normalsize
\end{itemize}
Therefore
$$
\model{\fone \wedge \DIA^{q}_{a}\ftwo}_{X} = \begin{cases} \emptyset^{X} \ \ &\text{ if } \mu(\PROJ(\model{\ftwo}_{X\cup\{a\}})) > q \\
\model{\fone}_{X} \ \ &\text{ otherwise}
\end{cases}
$$
On the other hand,
\begin{align*}
\model{\DIA^{q}_{a}(\neg \fone \vee \ftwo)}_{X} &= \{f \in (\twoOm)^{X} \ | \ \mu\big(\PROJ_{f}(\model{\neg \fone \vee \ftwo}_{X\cup\{a\}})\big) < q\} \\
&= \{f \in (\twoOm)^{X} \ | \ \mu\big(\PROJ_{f}(\model{\neg \fone}_{X\cup\{a\}} \cup \model{\ftwo}_{X\cup\{a\}})\big) < q\} \\
&= \{f \in (\twoOm)^{X} \ | \ \mu\big(\PROJ_{f}(\model{\neg\fone}_{X\cup\{a\}}) \cup \PROJ_{f}(\model{\ftwo}_{X\cup\{a\}})\big) < q\}
\end{align*}
\normalsize
There are two cases:
\begin{itemize}

\item Let $\mu\big(\PROJ_{f}(\model{\ftwo}_{X\cup\{a\}})\big) \ge q$. Since $\mu\big(\PROJ_{f}(\model{\neg\fone}_{X\cup\{a\}}) \cup \PROJ_{f}(\model{\ftwo}_{X\cup\{a\}})\big)$ $\ge$ 
$\mu\big(\PROJ_{f}(\model{\ftwo}_{X\cup\{a\}})\big)$, 
$\mu\big(\PROJ_{f}(\model{\neg\fone}_{X\cup\{a\}}) \cup \PROJ_{f}(\model{\ftwo}_{X\cup\{a\}})\big) \ge q$. 
Consequently, $\{f \in$ ($\twoOm$)$^{X}$ $|$ 
$\mu\big(\PROJ_{f}(\model{\neg\fone}_{X\cup\{a\}}) \cup \PROJ_{f}(\model{\ftwo}_{X\cup\{a\}})\big) < q\}$ = $\emptyset^{X}$.

\item Let $\mu\big(\PROJ_{f}(\model{\ftwo}_{X\cup\{a\}})\big) < q$. 
\begin{align*}
\model{\DIA^{q}_{a}(\neg \fone \vee \ftwo)}_{X} = \ &\{f \in (\twoOm)^{X} \ | \ \mu(\PROJ_{f}(\model{\neg \fone}_{X\cup\{a\}}) \\
& \ \ \ \ \cup \PROJ_{f}(\model{\ftwo}_{X\cup\{a\}})) < q\} \\
\stackrel{\text{C}~\ref{C1}}{=} &\{f \in (\twoOm)^{X} \ | \ f \not \in \model{\neg \fone}_{X} \\
& \ \ \ \ \wedge \mu(\PROJ_{f}(\model{\ftwo}_{X\cup\{a\}})) < q\} \\
= \ &\{f \in (\twoOm)^{X} \ | \ f \not \in \model{\neg \fone}_{X}\} \\
= \ &\model{\fone}_{X}
\end{align*}
\end{itemize}
Therefore,
\begin{align*}
\model{\DIA^{q}_{a}(\neg \fone \vee \ftwo)}_{X} &= \begin{cases} \emptyset^{X} \ \ &\text{ if } \mu\big(\PROJ_{f}(\model{\ftwo}_{X\cup\{a\}})\big) \ge q \\
\model{\ftwo}_{X} \ \ &\text{ otherwise}
\end{cases} \\
&= \model{\fone \wedge \DIA^{q}_{a}\ftwo}_{X}.
\end{align*}
Let finally prove,
$$
\model{\fone \vee \DIA^{q}_{a}\ftwo}_{X} = \model{\DIA^{q}_{a}(\neg \fone \wedge \ftwo)}_{X}.
$$
For Definition~\ref{PPLsem},
$$
\model{\fone \vee \DIA^{q}_{a}\ftwo}_{X} = \model{\fone}_{X} \cup \model{\DIA^{q}_{a}\ftwo}_{X}.
$$ 
Let $f \in$ ($\twoOm$)$^{X}$. There are two possible cases,
\begin{itemize}

\item $\mu\big(\PROJ_{f}(\model{\ftwo}_{X\cup\{a\}})\big) \ge q$. 
Since $\{f \in$ ($\twoOm$)$^{X}$ $|$ 
$\mu\big(\PROJ_{f}(\model{\ftwo}_{X\cup\{a\}})\big)\}$ 
= $\emptyset^{X}$,
\begin{align*}
\model{\fone \vee \DIA^{q}_{a}\ftwo}_{X} &= \model{\fone}_{X} \cup \model{\DIA^{q}_{a}\ftwo}_{X} \\
&= \model{\fone}_{X} \cup \emptyset^{X} \\
&= \model{\fone}_{X}
\end{align*}
\item $\mu\big(\PROJ_{f}(\model{\ftwo}_{X\cup\{a\}})\big) < q$. 
Since $\{f \in (\twoOm)^{X} \ | \ \mu\big(\PROJ_{f}(\model{\ftwo}_{X})\big) < q\}$ = $(\twoOm)^{X}$,
\begin{align*}
\model{\fone \vee \DIA^{q}_{a}\ftwo}_{X} &= \model{\fone}_{X} \cup \model{\DIA^{q}_{a}\ftwo}_{X} \\
&= \model{\fone}_{X} \cup (\twoOm)^{X} \\
&= (\twoOm)^{X}.
\end{align*}
\end{itemize}
Therefore, 
$$
\model{\fone \vee \DIA^{q}_{a}\ftwo}_{X} = \begin{cases} \model{\fone}_{X} \ \ &\text{ if } \mu(\PROJ_{f}(\model{\ftwo}_{X\cup\{a\}})) \ge q \\ (\twoOm)^{X} \ \ &\text{ otherwise.}
\end{cases}
$$
On the other hand,
\begin{align*}
\model{\DIA^{q}_{a}(\neg \fone \wedge \ftwo)}_{X} &= \{f \in (\twoOm)^{X} \ | \ \mu\big(\PROJ_{f}(\model{\neg \fone \wedge \ftwo}_{X\cup\{a\}})\big) < q\} \\
&= \{f \in (\twoOm)^{X} \ | \ \mu(\PROJ_{f}\big(\model{\neg \fone}_{X\cup\{a\}} \cap \model{\ftwo}_{X\cup\{a\}})\big) < q\} \\
&= \{f \in (\twoOm)^{X} \ | \ \mu\big(\PROJ_{f}(\model{\neg \fone}_{X\cup\{a\}}) \cap \PROJ_{f}(\model{\ftwo}_{X\cup\{a\}})\big) < q\}
\end{align*}
There are two cases,
\begin{itemize}

\item Let $\mu\big(\PROJ_{f}(\model{\neg \ftwo}_{X\cup\{a\}})\big) \ge q$. Then,
\begin{align*}
\model{\DIA^{q}_{a}(\neg\fone \wedge \ftwo)}_{X} = \ &\{f \in (\twoOm)^{X} \ | \ \mu\big(\PROJ_{f}(\model{\neg\fone}_{X\cup\{a\}}) \cap \PROJ_{f}(\model{\ftwo}_{X\cup\{a\}})\big) < q\} \\
\stackrel{\text{C}~\ref{C1}}{=} &\{f \in (\twoOm)^{X} \ | \ f \not \in \model{\neg \fone}_{X} \vee \mu\big(\PROJ_{f}(\model{\ftwo}_{X\cup\{a\}})\big) < q\} \\
= \ &\{f \in (\twoOm)^{X} \ | \ f \not \in \model{\neg \fone}_{X}\} \\
= \ &\{f \in (\twoOm)^{X} \ | \ f \in \model{\fone}_{X}\} \\
= \ &\model{\fone}_{X}.
\end{align*}

\item Let $\mu\big(\PROJ_{f}(\model{\ftwo}_{X\cup\{a\}})\big) < q$.
Since $\mu\big(\PROJ_{f}(\model{\neg \fone}_{X\cup\{a\}}) \cap \PROJ_{f}(\model{\ftwo}_{X\cup\{a\}})\big)$ $\leq$ 
$\mu\big(\PROJ_{f}(\model{\ftwo}_{X\cup\{a\}})\big)$,
$\mu\big(\PROJ_{f}(\model{\neg \fone}_{X\cup\{a\}}) \cap \PROJ_{f}(\model{\ftwo}_{X\cup\{a\}})\big) < q$.
Consequently, $\{f \in (\twoOm)^{X} \ | \ \mu\big(\PROJ_{f}(\model{\neg\fone}_{X\cup\{a\}}) \cap \PROJ_{f}(\model{\ftwo}_{X\cup \{a\}})\big) < q\}$ 
= ($\twoOm$)$^{X}$.
\end{itemize}
Therefore,
\begin{align*}
\model{\DIA^{q}_{a}(\neg\fone \wedge \ftwo)}_{X} &= \begin{cases} \model{\fone}_{X} \ \ &\text{ if } \mu(\PROJ_{f}(\model{\neg \ftwo}_{X\cup\{a\}})) \ge q \\ (\twoOm)^{X} \ \ &\text{ otherwise} 
\end{cases} \\
&= \model{\fone \vee \DIA^{q}_{a}\ftwo}_{X}.
\end{align*}
\end{proof}

\primeDuality*
\begin{proof}
Let us first deal with:
$$
\model{\neg \DIA^{q}_{a}\fone}_{X} = \model{\BOX^{q}_{a}\fone}_{X}.
$$
There are two cases:
\begin{itemize}
\itemsep0em
\item If $q$ = 0,
\begin{align*}
\model{\neg\DIA^{0}_{a}\fone}_{X} & \ = (\twoOm)^{X} \ – \ \model{\DIA^{0}_{a}\fone}_{X} \\
&\stackrel{\text{L}~\ref{Lq0}}{=} (\twoOm)^{X} \ – \ \emptyset^{X} \\
& \ = (\twoOm)^{X} \\
&\stackrel{\text{L}~\ref{Lq0}}{=} \model{\BOX^{0}_{a}\fone}_{X}
\end{align*}
\item Otherwise, $q >$ 0,
\begin{align*}
\model{\neg\DIA^{q}_{a}\fone}_{X} &= (\twoOm)^{X} – \model{\DIA^{q}_{a}\fone}_{X} \\
&= (\twoOm)^{X} – \{f : X \rightarrow \twoOm \ | \ \mu(\PROJ_{f}(\model{\fone}_{X\cup\{a\}})) < q\} \\
&= \{f : X \rightarrow \twoOm \ | \ \mu(\PROJ_{f}(\model{\fone}_{X\cup\{a\}})) \ge q\} \\
&= \model{\BOX^{q}_{a}\fone}_{X}.
\end{align*}
\end{itemize}

Let us now consider:
$$
\model{\neg \BOX^{q}_{a}\fone}_{X} = \model{\DIA^{q}_{a}\fone}_{X}.
$$
Again, there are two possible cases:
\begin{itemize}

\item Let $q$ = 0. Then,
\begin{align*}
\model{\neg\BOX^{0}_{a}\fone}_{X} = \ &(\twoOm)^{X} \ – \ \model{\BOX^{0}_{a}\fone}_{X} \\
\stackrel{\text{L}~\ref{Lq0}}{=} &(\twoOm)^{X} \ – \ (\twoOm)^{X} \\
= \ &\emptyset^{X} \\
\stackrel{\text{L}~\ref{Lq0}}{=} &\model{\DIA^{0}_{a}\fone}_{X}.
\end{align*}

\item Let $q >$ 0. Then,
\begin{align*}
\model{\neg\BOX^{q}_{a}\fone}_{X} &= (\twoOm)^{X} – \ \model{\BOX^{q}_{a}\fone}_{X} \\
&= (\twoOm)^{X} – \ \{f : X \rightarrow \twoOm \ | \ \mu\big(\PROJ_{f}(\model{\fone}_{X\cup\{a\}})\big) \ge q\} \\
&= \{f : X \rightarrow \twoOm \ | \ \mu\big(\PROJ_{f}(\model{\fone}_{X\cup\{a\}})\big) < q\} \\
&= \model{\DIA^{q}_{a}\fone}_{X}. 
\end{align*}
\end{itemize}
\end{proof}

\primeProposition*
\noindent
In order to prove the given Proposition~\ref{PNF}, 
some auxiliary definitions and lemmas have to be introduced.
\begin{definition}[Negation Simple Normal Form]
A formula of $\PPL$ is a \emph{negation simple normal form} (NSNF) 
if and only if negations only appear in front of atoms.
\end{definition}

\begin{lemma}\label{NSNF}
Every formula of $\PPL$ can be converted into an equivalent formula in NSNF.
\end{lemma}
\begin{proof}
For every formula of $\PPL$, call it $\fone$, there is a NSNF,
call it $\ftwo$, such that, for every $X$ with $\Var{\fone} \cup \Var{\ftwo} \subseteq X$, $\fone \equiv_{X} \ftwo$. 
The proof is by induction on the weight of $\fone$.
\begin{itemize}
\item \emph{Base case.} $\fone = \at{n}$, for some $n \in \Nat$. Let $\ftwo$ = $\fone$, as $\fone$ is already in NSNF.
\item \emph{Inductive case.} 
Assume that the procedure holds for formulas of weight up to $n$, 
and show it to exist for formulas of weight $n$+1. 
If $\fone \neq \neg \fthree$, then the proof is trivial: 
for IH there is a procedure for $\fthree$, so the same holds for $\fone$. Let $\fone = \neg \fthree$. There is a limited number of cases:

 \par $\fthree = \neg \fthree'$. 
For IH, there is a NSNF formula $\ffour$ such that $\fthree' \equiv_{X} \ffour$. Let $\ftwo = \ffour$. Since, $\fone = \neg \fthree = \neg \neg \fthree' \stackrel{IH}{\equiv}_{X} \neg \neg \ffour \equiv_{X} \ffour$, and $\ffour$ is in NSNF for IH. So $\fone$ $\equiv_{X}$ ($\ffour$ =) $\ftwo$ and $\ftwo$ is in NSNF.

 \par $\fthree = \fthree' \vee \fthree''.$ So, $\fone = \neg(\fthree' \vee \fthree'') \equiv_{X} \neg \fthree' \wedge \neg \fthree''$. For IH, there are two NSNF formulas $\ffour'$ and $\ffour''$ such that $\neg \fthree' \equiv_{X} \ffour'$ and $\neg \fthree'' \equiv_{X} \ffour''$. Let $\ftwo = \ffour' \wedge \ffour''$. As $\fone = \neg(\fthree' \vee \fthree'') \equiv_{X} \neg \fthree' \wedge \neg \fthree'' \stackrel{IH}{\equiv}_{X} \ffour' \wedge \ffour''$ = $\ftwo$, $\fone \equiv_{X} \ftwo$ and, since $\ffour'$ and $\ffour''$ are in NSNF (IH) $\ftwo$ is in NSNF as well.
 
 \par $\fthree = \fthree' \wedge \fthree''$. Equivalent to the previous case.
  
 \par $\fthree = \BOX^{q}_{a}\fthree'$. So, $\fone = \neg(\BOX^{q}_{a}\fthree') \stackrel{\text{Lem}~\ref{lemma:pseudoduality}}{\equiv} \DIA^{q}_{a}\fthree'$. 
For IH, there is a NSNF $\ffour$, such that $\fthree' \equiv_{X} \ffour$. Let $\ftwo = \DIA^{q}_{a}\ffour$. As $\fone = \neg(\BOX^{q}_{a}\fthree') 
\stackrel{\text{Lem}~\ref{lemma:pseudoduality}}{\equiv}_{X} \DIA^{q}_{a}(\fthree') \stackrel{IH}{\equiv}_{X} \DIA^{q}_{a}\ffour = \ftwo$, $\fone \equiv_{X} \ftwo$ and, since $\ffour$ is in NSNF, also $\ftwo$ = $\DIA^{q}_{a}\ffour$ is in NSNF.

 \par $\fthree = \DIA^{q}_{a}\fthree'.$ Equivalent to the previous case. 
\end{itemize}
\end{proof}

\begin{lemma}\label{NSNFPNF}
Every formula of $\PPL$ in NSNF can be converted into an equivalent formula in PNF (of the same or smaller weight).
\end{lemma}
\begin{proof}
For every formula of $\PPL$ in NSNF, call it $\fone$, there is a PNF $\ftwo$ such that, for every $X$ with $\Var{\fone} \cup \Var{\ftwo} \subseteq X$, $\fone \equiv_{X} \ftwo$.
Assume $\fone$ is in NSNF.
First, all bound variable are renamed with fresh ones. 
Then, the lemma is proved by induction on the weight of $\fone$.
\begin{itemize}
\item \emph{Base case.} $\fone = \at{n}$. $\fone$ is already in PNF.
\item \emph{Inductive case.} Assume that the procedure holds for formulas of weight up to $n$, and show it to exist for formulas of weight $n$+1.

\par $\fone = \neg \fthree.$ Since $\fone$ is in NSNF, $\fthree$ must be atomic, so $\fone$ = $\neg \at{n}$ (for some $n \in \Nat$) is already in PNF.
\par $\fone = \fthree \vee \ffour.$ For IH, there are two PNF formulas $\fthree'$ and $\ffour'$ such that $\fthree \equiv_{X} \fthree'$ and $\ffour \equiv_{X} \ffour'$. There are some possible cases:
\begin{itemize}
\itemsep0em
\item Both $\fthree'$ and $\ffour'$ do not contain $\triangle$. Let $\ftwo$ = $\fthree' \vee \ffour'$. $\ftwo$ is in PNF and $\fone \equiv_{X} \ftwo$ (as $\fone$ = $\fthree \vee \ffour$ $\stackrel{IH}{\equiv}_{X}$ $\fthree' \vee \ffour'$ = $\ftwo$).
\item $\fthree'$ = $\triangle \fthree''$. If $\triangle$ is of the form $\BOX^{q}_{a}$, for Lemma~\ref{lemma:commutations}, 
$$
(\BOX^{q}_{a}\fthree'') \vee \ffour' \stackrel{\text{L}~\ref{lemma:commutations}}{\equiv} \BOX^{q}_{a}(\fthree'' \vee \ffour')
$$ 
of the same weight. For IH, $\fthree'' \vee \ffour'$ can be converted into a PNF formula, $\ftwo'$. Let $\ftwo$ = $\BOX^{q}_{a}\ftwo'$. Since $\ftwo'$ is in PNF, clearly also $\ftwo$ must be. Indeed, 
$$
\fone = \fthree \vee \ffour \stackrel{IH}{\equiv}_{X} (\BOX^{q}_{a}\fthree'') \vee \ffour' \stackrel{\text{L}~\ref{lemma:commutations}}{\equiv}_{X} \BOX^{q}_{a}(\fthree'' \vee \ffour') \stackrel{IH}{\equiv}_{X} \BOX^{q}_{a}\ftwo' = \ftwo.
$$ 
If $\Delta$ is of the form $\DIA^{q}_{a},$ the proof is the same by applying both Lemma~\ref{lemma:commutations} (with, $\DIA^{q}_{a}\fthree \vee \ffour \equiv_{X} \DIA^{q}_{a}(\neg\fthree \vee \ffour)$ and Lemma~\ref{NSNF}).
\end{itemize}
\par $\fone = \fthree \wedge \ffour$. The proof is similar to the previous one. 

\par $\fone = \BOX^{q}_{a}\ffour.$ For IH, there is a $\ffour'$ such that $\ffour'$ is in PNF and $\ffour \equiv_{X} \ffour'$. Let $\ftwo = \BOX^{q}_{a}\ffour'$. Since $\ffour'$ is in PNF, clearly $\ftwo$ is in PNF as well. Furthermore, since $\fone=\BOX^{q}_{a}$:
$$
\fone = \BOX^{q}_{a}\ffour \equiv_{X} \BOX^{q}_{a}\ffour' = \ftwo.
$$

\par $\fone = \DIA^{q}_{a}\fthree.$ The proof is equivalent to the previous one. 

\end{itemize}
\end{proof}
\noindent
Thus, every formula can be converted into an equivalent formula in PNF.
\begin{proof}[Proof of Proposition~\ref{PNF}]
By combining Lemma~\ref{NSNF} and Lemma~\ref{NSNFPNF}.
\end{proof}

\subsection{Positive Prenex Normal Form}\label{appendix5.2}


\primeEpsilon*
\noindent
To establish Lemma~\ref{lemma:epsilon} we need a few preliminary lemmas. First of all,
for all $k\in \mathbb N$, let $[0,1]_{k}$ indidate the set of rationals of the form
$q= \sum_{i=0}^{k} b_{i}\cdot 2^{-i}$, where $b_{i}\in \{0,1\}$.

\begin{lemma}\label{lemma:base2}
For all $p\leq 2^{k}$, $\frac{p}{2^{k}}\in [0,1]_{k}$.
\end{lemma}
\begin{proof}
Let $b_{0}\dots b_{k}=(p)_{2}$ the base 2 writing of $p$ (with possibly all 0s at the end so that the length is precisely $k$), then  $p= \sum_{i=0}^{k}b_{i}\cdot 2^{i}$, so:
\begin{align*} 
\frac{p}{2^{k}} &=\dfrac{\sum_{i=0}^{k}b_{i}\cdot 2^{i}}{2^{k}} \\
&= \sum_{i=0}^{k}b_{i}\cdot 2^{-k+i} \\
&= \sum_{i=0}^{k}b_{k-i}\cdot 2^{-i}.
\end{align*}
\end{proof}

\begin{lemma}\label{lemma:k}
For all Boolean formula $\bone$ with $\FN(\bone)\subseteq \{a\}$, $\mu(\llbracket \bone\rrbracket_{\{a\}})\in [0,1]_{k}$, where $k$ is the maximum natural number such that $\bvar_{k}^{a}$ occurs in $\bone$.

\end{lemma}
\begin{proof}
By Lemma~\ref{lemma:counti} and Lemma~\ref{lemma:base2}:
$$
\mu(\llbracket \bone\rrbracket_{\{a\}})=\sharp \bone= \dfrac{\sharp\{m:\{\bvar_{0}^{a},\dots, \bvar_{k}^{a}\}\to \{0,1\}\mid m\vDash \bone\}}{2^{k}}\in [0,1]_{k}.
$$
\end{proof}

\begin{lemma}\label{lemma:pif}
For all $S\in \mathscr B\big( (2^{\omega})^{X\cup \{a\}}\big)$ and $f:X\to 2^{\omega}$, 
$\overline{\Pi_{f}(S)}= \Pi_{f}(\overline S)$.
\end{lemma}
\begin{proof}
Indeed:
\begin{align*}
\overline{\Pi_{f}(S)} &=\overline{\{g:\{a\}\to 2^{\omega}\mid f+g\in S\}} \\
&= \{g:\{a\}\to 2^{\omega}\mid f+g\notin S\} \\
&= \{g:\{a\}\to 2^{\omega}\mid f+g\in \overline S\} \\
&= \Pi_{f}(\overline S).
\end{align*}
\end{proof}

\begin{lemma}\label{lemma:leq}
For all $S\in \mathscr B( (2^{\omega})^{X})$ and $r\in [0,1]$, 
$\mu(S)\leq r$ iff $\mu(\overline S)\geq 1-r$.
\end{lemma}
\begin{proof}
The claim follows from: 
$$
1=\mu((2^{\omega})^{X})=\mu(S\cup \overline S)=\mu(S)+\mu(\overline S).
$$
\end{proof}
\noindent
We have now all ingredients to prove Lemma~\ref{lemma:epsilon}.
%
\begin{proof}[Proof of Lemma~\ref{lemma:epsilon}]
Let $\bone_{A}$ be a Boolean formula satisfying $\llbracket \bone\rrbracket_{X\cup\{a\}}=\llbracket \bone_{A}\rrbracket_{X\cup\{a\}}$ (which exists by Lemma \ref{lemma:bone}).
Let $\bone_{A}$ be $a$-decomposable as $\bigvee_{i}^{n}\bthree_{i}\land \bfour_{i}$ and let $k$ be maximum such that $\bvar_{k}^{a}$ occurs in $\bone_{A}$.
By Lemma \ref{lemma:k} for all $i=0,\dots,n$, $\mu(\llbracket \bthree_{i}\rrbracket_{\{a\}})\in [0,1]_{k}$.
This implies in particular that for all $f:X\to 2^{\omega}$, $\mu(\Pi_{f}(\llbracket A\rrbracket_{X\cup \{a\}}))\in [0,1]_{k}$, since $\Pi_{f}(\llbracket A\rrbracket_{X\cup \{a\}}))$ coincides with the unique $\llbracket \bthree_{i}\rrbracket_{\{a\}}$ such that $f\in \llbracket \bfour_{i}\rrbracket_{X}$ (by~Lemma~\ref{lemma:fundamental}).
Now, if $r\notin [0,1]_{k}$, let $\epsilon=0$; and if $r\in[0,1]_{k}$, then 
if $r=1$, let $\epsilon=- 2^{-(k+1)}$, and if $r\neq 1$, let $\epsilon=2^{-(k+1)}$.
In all cases $r+\epsilon \notin [0,1]_{k}$, so we deduce
\noindent
\medskip
\adjustbox{scale=0.95}{
\begin{minipage}{\linewidth}
\begin{align*}
\llbracket \lnot \BOX^{r}_{a}A\rrbracket_{X} & \ =
\{f:X\to 2^{\omega}\mid \mu\big(\Pi_{f}(\llbracket A\rrbracket_{X\cup\{a\}})\big)< r\} \\
& \ =
\{f:X\to 2^{\omega}\mid \mu\big(\Pi_{f}(\llbracket A\rrbracket_{X\cup\{a\}})\big)\leq r+\epsilon\} 
\\ &
\stackrel{\text{L~\ref{lemma:leq}}}{=}
\{f:X\to 2^{\omega}\mid \mu\big(\overline{\Pi_{f}(\llbracket A\rrbracket_{X\cup\{a\}})}\big)\geq 1-(r+\epsilon)\}
\\ &
\stackrel{\text{L~\ref{lemma:pif}}}{=}
\{f:X\to 2^{\omega}\mid \mu\big(\Pi_{f}(\overline{\llbracket A\rrbracket_{X\cup\{a\}})}\big)\geq 1-(r+\epsilon)\}
\\ & \ 
=
\{f:X\to 2^{\omega}\mid \mu\big(\Pi_{f}({\llbracket \lnot A\rrbracket_{X\cup\{a\}})}\big)\geq 1-(r+\epsilon)\}
\\& \ =
\llbracket \BOX^{1-(r+\epsilon)}_{a}\lnot A\rrbracket_{X}.
\end{align*}
\end{minipage}
}

%
\end{proof}

\section{Proofs for Section \ref{section6}}\label{appendix6}

\subsection{The Type System $\TTT$}

\subsubsection{Permutative Normal Forms}

To characterize the $\PNF$ we introduce the following class of terms, called $(S,a)$-trees.

\begin{definition}\label{def:satree}
Let $S$ be a set of name-closed $\lambda$-terms and $a$ be a name.
For all $i\in \mathbb N\cup\{-1\}$, the set of \emph{$(S,a)$-trees of level $i$} is defined as follows:
\begin{itemize}
\item any $t\in S$ is a $(S,a)$-tree of level $-1$;
\item if $t,u$ are $(S,a)$-trees of level $j$ and $k$, respectively, and $j,k<i$, then $t\oplus_{a}^{i}u$ is a $(S,a)$-tree of level $i$.
\end{itemize}
%
The \emph{support} of a $(S,a)$-tree $t$, indicated as $\mathrm{Supp}(t)$, is the finite set of terms in $S$ which are leaves of $t$.
\end{definition}

%

\begin{definition}\label{def:pnf}
The sets $\mathcal T$ and $\mathcal V$ of name-closed terms are defined inductively as follows:
\begin{itemize}
\item all variable $x\in \mathcal V$;

\item if $t\in \mathcal V$, $\lambda x.t\in \mathcal V$;
\item if $t\in \mathcal V$ and $u\in \mathcal T$, then $tu\in \mathcal V$;

\item if $t\in \mathcal V$, then $t\in \mathcal T$;


\item if $ t$ is a $(\mathcal T, a)$-tree, then $\nu a. t\in \mathcal T$.

\end{itemize}
\end{definition}

\begin{lemma}\label{lemma:enne}
For all name-closed term $t\in \mathcal T$, $t\in \mathcal V$ iff it does not start with $\nu$.
\end{lemma}
\begin{proof}
First observe that if $t=\nu a.b$, then $t\notin \mathcal V$. For the converse direction, we argue by induction on $t\neq \nu a.b$:
	\begin{itemize}
	\item if $t=x$ then $t\in \mathcal V$;
	\item if $t= \lambda x.u$ then by IH $u\in \mathcal V$, so $t\in \mathcal V$;
	\item if $t=uv$, then the only possibility is that $t\in \mathcal V$;
	\item if $t=u\oplus_{a}^{i}v$, then $t$ is not name-closed, again the hypothesis.
	\end{itemize}
\end{proof}


For all term $t$ and $a\in \FN(t)$, we let $I_{a}(t)$ be the maximum index $i$ such that $\oplus_{a}^{i}$ occurs in $t$, and $I_{a}(t)=-1$ if $\oplus_{a}^{i}$ does not occur in $t$ for all index $i$.

We now show that Definition \ref{def:pnf} precisely captures $\PNF$.
\begin{lemma}\label{lemma:PNF}
A name-closed term $t$ is in $\PNF$ iff  $t\in \mathcal T$.

\end{lemma}
\begin{proof}
\begin{description}
\item[($\To$)]
We argue by induction on $t$. 
If $t$ has no bound name, then $t$ is obviously in $\mathcal T$. Otherwise:
\begin{itemize}
\item if $t= \lambda x.u$, then $u$ is also name-closed. Hence, by IH $u\in \mathcal T$; observe that $u$ cannot start  with $\nu$ (as $t$ would not be normal) so by Lemma \ref{lemma:enne}, $u\in \mathcal V$, which implies $t\in \mathcal T$; 


\item if $t=uv$, then $u$ and $v$ are both name-closed, and so by induction $u,v\in \mathcal T$; if $u$ started with $\nu$, then $t$ would not be normal, hence, by Lemma \ref{lemma:enne}, $u\in \mathcal V$, and we conclude then that  $t\in \mathcal T$.

\item if $t=u\oplus_{a}^{i}v$, then it cannot be name-closed;

\item if $t= \nu a.u$, then we show, by a sub-induction on $u$, that $u$ is a $(\mathcal T, a)$-tree of level $I_{a}(u)$: first note that $I_{a}(u)$ cannot be $-1$, since otherwise $\nu a.u \pperm u$, so $u$ would not be normal. We now consider all possible cases for $u$:

		\begin{itemize}
		\item $u$ cannot be a variable, or $I_{a}(u)$ would be $-1$;
		\item If $u=\lambda x.v$, then $I_{a}(v)=I_{a}(u)$ and so by sub-IH $v$ is a $(\mathcal T, a)$-tree of level $I_{a}(u)$, which implies that $u= \lambda x. v_{1}\oplus_{a}^{I_{a}(u)} v_{2}$, which is not normal. Absurd.
		
		\item if $u= vw$, then let $J=I_{a}(u)= \max\{ I_{a}(v),I_{a}(w)\}$ cannot be $-1$, since otherwise $I_{a}(u)$ would be $-1$. Hence $J\geq 0$ is either $I_{a}(v)$ or $I_{a}(w)$. We consider the two cases separatedly:
			\begin{itemize}
			\item if $J=I_{a}(v)$, then by sub-IH, $v= v_{1}\oplus_{a}^{J}v_{2}$, and thus 
			$u= ( v_{1}\oplus_{a}^{J}v_{2})w$ is not normal;
			\item if $J=I_{a}(w)$, then by sub-IH, $w= w_{1}\oplus_{a}^{J}w_{2}$, and thus 
			$u= v( w_{1}\oplus_{a}^{J}w_{2})$ is not normal.
			
			\end{itemize}
		In any case we obtain an absurd conclusion.
		
		\item if $u= u_{1}\oplus_{a}^{i}u_{2}$, then if $u_{1},u_{2}$ are both in $\mathcal V$, we are done, since  $u$ is a $(\mathcal T, a)$-tree of level $i=I_{a}(u)$. Otherwise, if $i< I_{a}(u)$, then $J=I_{a}(u)= \max\{ I_{a}(v), I_{a}(w)\}$, so we consider 	two cases:
		\begin{itemize}
		\item if $J=I_{a}(v)>i$, then by sub-IH, $v= v_{1}\oplus_{a}^{J}v_{2}$, and thus 
		$u= ( v_{1}\oplus_{a}^{J}v_{2})\oplus_{a}^{i}u_{2}$ is not normal;
		\item if $J=I_{a}(w)>i$, then by sub-IH, $w= w_{1}\oplus_{a}^{J}w_{2}$, and thus 
		$u= v\oplus_{a}^{i}( w_{1}\oplus_{a}^{J}w_{2})$ is not normal.
			
		\end{itemize}
		We conclude then that $i=I_{a}(u)$;  we must then show that $I_{a}(v),I_{a}(w)<i$. Suppose first $I_{a}(v)\geq i$, then $( v_{1}\oplus_{a}^{J}v_{2})\oplus_{a}^{i}u_{2}$ is not normal. In a similar way we can show that 
		$I_{a}(w)<i$. 

		\item if $u=\nu b.v$, then $I_{a}(u)=I_{a}(v)$, so by sub-IH $v= v_{1}\oplus_{a}^{I_{a}(u)}v_{2}$, and we conclude that $u=\nu b.  v_{1}\oplus_{a}^{i}v_{2}$ is not normal.
		
		\end{itemize}
	
\end{itemize}

\item[($\Leftarrow$)] It suffices to check by induction on $t\in \mathcal T$ that it is in $\PNF$.

\end{description}
\end{proof}

\begin{corollary}\label{cor:pnf}
A name-closed term of the form $\nu a.t$ is in $\PNF$ iff $t$ is a $(\mathcal T, a)$-tree of level $I_{a}(t)\geq 0$.
\end{corollary}

\subsubsection{The Subject Reduction Property}

We now establish the following standard property for $\TTT$:

\subjectreduction*

To establish the subject reduction property, we first need to establish a few auxiliary lemmas.

\begin{lemma}\label{lemma:weak2}
If $\Gamma \vdash^{X,r}t: \bone \Pto \sigma$ and $X\subseteq Y$, then 
$\Gamma \vdash^{Y,r}t: \bone\Pto \sigma$.
\end{lemma}
\begin{proof}
Straightforward induction on a type derivation of $t$.
\end{proof}

\begin{lemma}\label{lemma:weak}
If $\Gamma \vdash^{X,r}t: \bone\Pto \sigma$ holds and $\btwo\vDash^{X} \bone$, then 
$\Gamma \vdash^{X,r}t: \btwo\Pto \sigma$ is derivable.
\end{lemma}
\begin{proof}
By induction on a type derivation of $t$:
\begin{itemize}

\item if the last rule is 

\begin{prooftree}
\AXC{$\FN(t)\subseteq X$}
\AXC{$\bone \vDash^{X} \bot$}
\RL{$\TB$}
\BIC{$\Gamma \vdash^{X,q}t:\bone \Pto \sigma$}
\end{prooftree}

Then, since $\btwo\vDash^{X} \bone$ and $\bone \vDash^{X}\bot$ imply $\btwo\vDash^{X} \bot$, we conclude by an instance of the same rule.

\item if the last rule is

\begin{prooftree}
\AXC{}
\RL{$\TID$}
\UIC{$\Gamma \vdash^{X,q\cdot s}x: \bone \Pto \sigma$}
\end{prooftree}
where $r=q\cdot s$, the claim is immediately deduced by an instance of the same rule.

\item if the last rule is

\begin{prooftree}
\AXC{$\Gamma \vdash^{X,s}t: \bone_{1} \Pto \sigma$}
\AXC{$\Gamma \vdash^{X,s}t: \bone_{2} \Pto \sigma$}
\AXC{$\bone \vDash^{X}\bone_{1}\vee \bone_{2}$}
\RL{$\TU$}
\TIC{$\Gamma \vdash^{X,s}x: \bone \Pto \sigma$}
\end{prooftree}
Then, from $\btwo \vDash^{X} \bone$ we deduce $\btwo \vDash^{X} \bone_{1}\vee \bone_{2}$ so we conclude by applying an instance of the same rule.

\item if the last rule is 

\begin{prooftree}
\AXC{$\Gamma, y:\FF s,  \vdash^{X,r} t: \bone\Pto\rho$}
\RL{$\TLA$}
\UIC{$\Gamma \vdash^{X,r} \lambda y.t:\bone\Pto  \FF s\To  \rho$}
\end{prooftree}
Then, by IH, we deduce $\Gamma, y:\FF s \vdash^{X,r}t[u/x]: \btwo  \Pto \rho$ and we conclude by applying an instance of the same rule.

\item if the last rule is 

\begin{prooftree}
\AXC{$\Gamma\vdash^{X,q} t_{1}: \bone_{1}\Pto \BOX^{r}\sigma \To \tau$}
\AXC{$\Gamma \vdash^{X,r} t_{2}: \bone_{2} \Pto \sigma$}
\AXC{$\bone\vDash^{X} \bone_{1}\land \bone_{2}$}
\RL{$\TA$}
\TIC{$\Gamma\vdash^{X,q} t_{1}t_{2}: \bone \Pto \tau$}
\end{prooftree}

then from $\btwo\vDash^{X} \bone $ and $\bone\vDash^{X} \bone_{1}\land \bone_{2}$ we deduce
$\btwo\vDash^{X} \bone_{1}\land \bone_{2}$, so 
 we conclude by applying an instance of the same rule.

\item if the last rule is 

\begin{prooftree}
\AXC{$\Gamma\vdash^{X\cup\{a\},q} t_{1}:\bone'\Pto \sigma$}
\AXC{$\bone\vdash^{X\cup\{a\}}\bvar_{i}^{a}\land \bone'
$}
\RL{$\TL$}
\BIC{$\Gamma\vdash^{X\cup\{a\},q} t_{1}\oplus^{i}_{a}t_{2}:  \bone\Pto \sigma$}
\end{prooftree}
Then, from $\btwo\vDash^{X} \bone $, we deduce $\btwo \vDash^{X}\bvar_{i}^{a}\land \bone'$, so 
 we conclude by applying an instance of the same rule.

\item if the last rule is 

\begin{prooftree}
\AXC{$\Gamma \vdash^{X\cup\{a\},q} t_{2}:\bone'\Pto \sigma$}
\AXC{$\bone\vDash\lnot\bvar_{i}^{a}\land \bone'$}
\RL{$\TR$}
\BIC{$\Gamma \vdash^{X\cup\{a\},q} t_{1}\oplus^{i}_{a}t_{2}:  \bone\Pto \sigma$}
\end{prooftree}
we can argue similarly to the previous case.

\item if the last rule is

\begin{prooftree}
\AXC{$\Gamma \vdash^{X\cup\{a\},r} t:\bigvee_{i}\btwo_{i}\land \bthree_{i}\Pto \sigma$}
\AXC{$\mu(\model{\btwo_{i}}_{\{a\}})\geq s$}
\AXC{$\bone \vDash^{X} \bigvee_{i}\bthree_{i}$}
\RL{$\TN$}
\TIC{$\Gamma \vdash^{X,r\cdot s} \nu a.t: \bone\Pto  \sigma$}
\end{prooftree}

Then, from $\btwo\vDash^{X} \bone $ and $\bone \vDash^{X} \bigvee_{i}\bthree_{i}$ we deduce $\btwo \vDash^{X}  \bigvee_{i}\bthree_{i}$, so 
 we conclude by applying an instance of the same rule.

\end{itemize}

\end{proof}

\begin{lemma}\label{lemma:betasub}
If $\Gamma, x: \BOX^{s}\sigma \vdash^{X,r} t: \bone \Pto  \tau$ and 
$\Gamma \vdash^{X,s} u: \btwo \Pto \sigma$, then 
$\Gamma \vdash^{X,r} t[u/x]: \bone \land \btwo \Pto \tau$.

\end{lemma}
\begin{proof}
We argue by induction on the typing derivation of $t$:
\begin{itemize}

\item if the last rule is 

\begin{prooftree}
\AXC{$\FN(t)\subseteq X$}
\AXC{$\bone \vDash^{X} \bot$}
\RL{$\TB$}
\BIC{$\Gamma, x:\BOX^{s}\sigma \vdash^{X,q\cdot s}t:\bone \Pto \tau$}
\end{prooftree}
where $r=q\cdot s$, then, since $\FN({t[u/x]})\subseteq X$ and $\bone \land \btwo \vDash^{X} \bone \vDash^{X}\bot$ the claim is deduced using an instance of the same rule.

\item if the last rule is

\begin{prooftree}
\AXC{$\FN(\bone)\subseteq X$}
\RL{$\TID$}
\UIC{$\Gamma, x: \BOX^{s}\sigma \vdash^{X,s}x: \bone \Pto \sigma$}
\end{prooftree}
Then, $t[u/x]=u$ and $r=x$, so the claim is deduced from the hypothesis and Lemma \ref{lemma:weak}.

\item if the last rule is

\begin{prooftree}
\AXC{$\Gamma, x: \BOX^{s}\sigma \vdash^{X,s}t: \bone_{1} \Pto \sigma$}
\AXC{$\Gamma, x: \BOX^{s}\sigma \vdash^{X,s}t: \bone_{2} \Pto \sigma$}
\AXC{$\bone \vDash^{X}\bone_{1}\vee \bone_{2}$}
\RL{$\TU$}
\TIC{$\Gamma, x: \BOX^{s}\sigma \vdash^{X,s}x: \bone \Pto \sigma$}
\end{prooftree}

Then, by IH, we deduce $\Gamma, x: \BOX^{s}\sigma \vdash^{X,s}t[u/x]: \bone_{1}\land \btwo \Pto \sigma$
and $\Gamma, x: \BOX^{s}\sigma \vdash^{X,s}t[u/x]: \bone_{2}\land \btwo \Pto \sigma$ and since 
$\bone \land \btwo \vDash^{X} (\bone_{1}\land \btwo)\vee (\bone_{2}\land \btwo)$ we  conclude by applying an instance of the same rule.

\item if the last rule is 

\begin{prooftree}
\AXC{$\Gamma, y:\FF s, x: \BOX^{s}\sigma \vdash^{X,r} t: \bone\Pto\rho$}
\RL{$\TLA$}
\UIC{$\Gamma, x:\BOX^{s}\sigma \vdash^{X,r} \lambda y.t:\bone\Pto  \FF s\To  \rho$}
\end{prooftree}

Then, by IH, we deduce $\Gamma, y:\FF s \vdash^{X,r}t[u/x]: \bone \land \btwo \Pto \rho$ and since $(\lambda y.t)[u/x]=\lambda y.t[u/x]$ we conclude by applying an instance of the same rule.

\item if the last rule is 

\begin{prooftree}
\AXC{$\Gamma, x:\BOX^{s}\sigma \vdash^{X,q} t_{1}: \bone_{1}\Pto \BOX^{r}\sigma \To \tau$}
\AXC{$\Gamma,x:\BOX^{s}\sigma \vdash^{X,r} t_{2}: \bone_{2} \Pto \sigma$}
\AXC{$\bone\vDash^{X} \bone_{1}\land \bone_{2}$}
\RL{$\TA$}
\TIC{$\Gamma,x:\BOX^{s}\sigma\vdash^{X,q} t_{1}t_{2}: \bone \Pto \tau$}
\end{prooftree}

Then, by IH, we deduce $\Gamma \vdash^{X,q}t_{1}[u/x]: \bone_{1}\land \btwo \Pto \BOX^{r}\sigma \To \tau$ and 
$\Gamma \vdash^{X,r}t_{2}[u/x]: \bone_{2}\land \btwo\Pto \sigma$, and since $(t_{1}t_{2})[u/x]= (t_{1}[u/x])(t_{2}[u/x])$ and  
$\bone\land \btwo \vDash^{X} (\bone_{1}\land \btwo)\land (\bone_{2}\land \btwo)$, 
 we conclude by applying an instance of the same rule.

\item if the last rule is 

\begin{prooftree}
\AXC{$\Gamma, x:\BOX^{s} \sigma\vdash^{X\cup\{a\},q} t_{1}:\bone'\Pto \sigma$}
\AXC{$
\bone\vDash^{X\cup\{a\}}\bvar_{i}^{a}\land \bone'
$}
\RL{$\TL$}
\BIC{$\Gamma, x:\BOX^{s} \sigma\vdash^{X\cup\{a\},q} t_{1}\oplus^{i}_{a}t_{2}:  \bone\Pto \sigma$}
\end{prooftree}

Then, by IH, we deduce 
$\Gamma\vdash^{X\cup\{a\}, q} t_{1}[u/x]: \bone'\land \btwo$.
From $t[u/x]= (t_{1}[u/x])\oplus_{a}^{i}(t_{2}[u/x])$ and the fact that  
$
\bone\land \btwo\vdash^{X\cup\{a\}}\bvar_{i}^{a}\land (\bone'\land \btwo)
$, we deduce the claim by an instance of the same rule.

\item if the last rule is 

\begin{prooftree}
\AXC{$\Gamma, x:\BOX^{s}\sigma \vdash^{X\cup\{a\},q} t_{2}:\bone'\Pto \sigma$}
\AXC{$\bone\vDash^{X\cup\{a\}}\lnot\bvar_{i}^{a}\land \bone'$}
\RL{$\TR$}
\BIC{$\Gamma, x:\BOX^{s}\sigma \vdash^{X\cup\{a\},q} t_{1}\oplus^{i}_{a}t_{2}:  \bone\Pto \sigma$}
\end{prooftree}

we can argue similarly to the previous case.

\item if the last rule is

\begin{prooftree}
\AXC{$\Gamma, x:\BOX^{s}\sigma \vdash^{X\cup\{a\},q} t:\bigvee_{i}\btwo_{i}\land \bthree_{i}\Pto \sigma$}
\AXC{$\mu(\model{\btwo_{i}}_{\{a\}})\geq s$}
\AXC{$\bone \vDash^{X} \bigvee_{i}\bthree_{i}$}
\RL{$\TN$}
\TIC{$\Gamma, x:\BOX^{s}\sigma \vdash^{X,q\cdot r} \nu a.t: \bone\Pto  \sigma$}
\end{prooftree}

Then, by IH,
$\Gamma \vdash^{X\cup\{a\},q} t[u/x]:\left(\bigvee_{i}\btwo_{i}\land \bthree_{i}\right)\land \btwo\Pto \sigma$.
Now, from $\bigvee_{i}\btwo_{i}\land( \bthree_{i}\land \btwo) \vDash^{X\cup\{a\}}\left(\bigvee_{i}\btwo_{i}\land \bthree_{i}\right)\land \btwo$, using Lemma \ref{lemma:efs}, we deduce
$\Gamma \vdash^{X\cup\{a\},q} t[u/x]:\bigvee_{i}\btwo_{i}\land (\bthree_{i}\land \btwo)\Pto \sigma$.
Hence, since
$\bone \land \btwo \vDash^{X}\bigvee_{i}(\bthree_{i}\land \btwo)$, and since 
 $(\nu a.t)[u/x]=\nu a.t[u/x]$, we can deduce the clam by applying an instance of the same rule.

%
%
 
\end{itemize}

\end{proof}

\begin{lemma}\label{lemma:disj}
Let $\bone, \bone_{1},\dots, \bone_{n}$ and $\btwo,\btwo_{1},\dots, \btwo_{n}$ be such that $\FN({\bone}), \FN({\bone_{i}})\subseteq{X}$ and $\FN(\btwo), \FN({\btwo_{i}})\subseteq\{a\}$, where $a\notin X$.
If $\btwo$ is satisfiable, then if $\bone\land \btwo \vDash^{X\cup \{a\}}\bigvee_{i}^{n} \bone_{i}\land \btwo_{i}$ holds, 
also $\bone\vDash^{X}\bigvee_{i}^{n}\bone_{i}$ holds.

\end{lemma}
\begin{proof}
Let $v\in 2^{X}$ be a model of $\bone$. Since $\btwo$ is satisfiable, $v$ can be extended to a model $v'\in 2^{X\cup\{a\}}$ of $\bone \land \btwo$. By hypothesis, then $v'$ satisfies $\bigvee_{i}^{n} \bone_{i}\land \btwo_{i}$, so for some $i_{0}\leq n$, it satisfies $\bone_{i_{0}}\land \btwo_{i_{0}}$. We deduce then that $v$ satisfies $\bone_{i_{0}}$, and thus $v$ satisfies $\bigvee_{i}^{n}\bone_{i}$.
\end{proof}

\begin{lemma}\label{lemma:nota}
If $ \Gamma \vdash^{ X\cup \{a\},r}t: \bone \Pto \rho$ and $\FN(t)\subseteq X$, then for all weak $a$-decomposition $\bone\equiv_{X\cup\{a\}} \bigvee_{i} \btwo_{i}\land \bthree_{i}$ and $s\in[0,1]$,
$\Gamma \vdash^{X, r\cdot s} t: \bigvee_{i}\bthree_{i}\Pto \rho$.
\end{lemma}
\begin{proof}
By induction on a type derivation of $t$:
\begin{itemize}

\item if the last rule is 

\begin{prooftree}
\AXC{$\FN{(t)}\subseteq X\cup\{a\}$}
\AXC{$\bone \vDash^{X\cup\{a\}} \bot$}
\RL{$R_{\bot}$}
\BIC{$\Gamma \vdash^{X\cup\{a\},q}t:\bone \Pto \sigma$}
\end{prooftree}
From $ \bigvee_{i} \btwo_{i}\land \bthree_{i}\vDash^{X\cup\{a\}}\bot$ we deduce
$\btwo_{i}\land \bthree_{i}\vDash^{X\cup\{a\}}\bot$ for all $i$. Since $\btwo_{i}$ and $\bthree_{i}$ have no propositional variable in common, we also deduce that either $\btwo_{i}\vDash^{\{a\}} \bot$ or $\bthree_{i}\vDash^{X} \bot$.
By hypothesis, we deduce then that $\bthree_{i}\vDash^{X}\bot$ holds for all $i$ and thus that 
$\bigvee_{i}\bthree_{i}\vdash^{X}\bot$. Since moreover $\FN{(t)}\subseteq X$, the claim can be deduced by an instance of the same rule.

\item if the last rule is

\begin{prooftree}
\AXC{$\FN(\bone)\subseteq X\cup\{a\}$}
\RL{$\TID$}
\UIC{$\Gamma , x:\BOX^{r'}\sigma\vdash^{X\cup\{a\},r}x: \bone \Pto \sigma$}
\end{prooftree}

where $r=r'\cdot s'$, then the claim can be deduced by an instance of the same rule.

\item if the last rule is
\begin{prooftree}
\AXC{$\Gamma \vdash^{X\cup\{a\},r}t: \bone_{1} \Pto \sigma$}
\AXC{$\Gamma \vdash^{X\cup\{a\},r}t: \bone_{2} \Pto \sigma$}
\AXC{$\bone \vDash^{X\cup\{a\}}\bone_{1}\vee \bone_{2}$}
\RL{$\TU$}
\TIC{$\Gamma \vdash^{X\cup\{a\},r}x: \bone \Pto \sigma$}
\end{prooftree}
Let $\bigvee_{j}\btwo'_{j}\land \bthree'_{j}$ and $\bigvee_{k}\btwo''_{j}\land \bthree''_{k}$ be weak $a$-decompositions of $\bone_{1}$ and $\bone_{2}$.
By IH we deduce $\Gamma\vdash^{X,r\cdot s}t: \bigvee_{j}\bthree'_{j}\Pto \sigma$ and
 $\Gamma\vdash^{X,r\cdot s}t: \bigvee_{k}\bthree''_{k}\Pto \sigma$.

From $\bigvee_{i}\btwo_{i}\land \bthree_{i}\vDash^{X\cup\{a\}} \bone_{1}\vee \bone_{2}$ we deduce 
that for all $i$, 
$\btwo_{i}\land \bthree_{i} \vDash^{X\cup\{a\}} \left (\bigvee_{j}\btwo'_{j}\land \bthree'_{j}\right) \lor \left (\bigvee_{k}\btwo''_{j}\land \bthree''_{k}\right)$, so by Lemma \ref{lemma:disj} we deduce that 
$\bthree_{i} \vDash^{X}\left (\bigvee_{j} \bthree'_{j}\right) \lor \left (\bigvee_{k} \bthree''_{k}\right)$, and finally that 
$\bigvee_{i}\bthree_{i} \vDash^{X}\left (\bigvee_{j} \bthree'_{j}\right) \vee \left ( \bigvee_{k} \bthree''_{k}\right)$
Hence the claim can be deduced by applying an instance of the same rule.

\item if the last rule is 

\begin{prooftree}
\AXC{$\Gamma, y:\FF s \vdash^{X\cup\{a\},r} t: \bone\Pto\rho$}
\RL{$\TLA$}
\UIC{$\Gamma \vdash^{X\cup\{a\},r} \lambda y.t:\bone\Pto  \FF s\To  \rho$}
\end{prooftree}
Then, the claim follows from the IH by applying an instance of the same rule.

\item if the last rule is 

\begin{prooftree}
\AXC{$\Gamma \vdash^{X\cup\{a\},r} t_{1}: \bone_{1}\Pto (\BOX^{s'}\sigma) \To \tau$}
\AXC{$\Gamma \vdash^{X\cup\{a\},s'} t_{2}: \bone_{2} \Pto \sigma$}
\AXC{$\bone\vDash^{X\cup\{a\}} \bone_{1}\land \bone_{2}$}
\RL{$\TA$}
\TIC{$\Gamma\vdash^{X\cup\{a\},r} t_{1}t_{2}: \bone \Pto \tau$}
\end{prooftree}
Then,
let $\bigvee_{j}\btwo'_{j}\land \bthree'_{j}$ and $\bigvee_{k}\btwo''_{j}\land \bthree''_{k}$ be weak $a$-decompositions of $\bone_{1}$ and $\bone_{2}$.
By IH we deduce $\Gamma\vdash^{X,r\cdot s}t: \bigvee_{j}\bthree'_{j}\Pto (\BOX^{s'}\sigma)\To \tau$ and
 $\Gamma\vdash^{X,s'}t: \bigvee_{k}\bthree''_{k}\Pto \sigma$.

From $\bigvee_{i}\btwo_{i}\land \bthree_{i} \vDash^{X\cup\{a\}} \left (\bigvee_{j}\btwo'_{j}\land \bthree'_{j}\right) \land \left(\bigvee_{k}\btwo''_{j}\land \bthree''_{k}\right)$ we deduce that for all $i$,
\begin{align*}
\btwo_{i}\land \bthree_{i}&  \vDash^{X\cup \{a\}}\bigvee_{j}\btwo'_{j}\land \bthree'_{j}  \\ 
\btwo_{i}\land \bthree_{i}& \vDash^{X\cup \{a\}}\bigvee_{k}\btwo''_{k}\land \bthree''_{k}
\end{align*}
Using Lemma \ref{lemma:disj} we deduce then that 
$\bigvee_{i} \bthree_{i} \vDash^{X}\bigvee_{j} \bthree'_{j}$ 
and $\bigvee_{i} \bthree_{i} \vDash^{X}\bigvee_{k} \bthree''_{k}$ hold, that is, that $\bigvee_{i} \bthree_{i} \vDash^{X}\left (\bigvee_{j} \bthree'_{j} \right )\land \left (\bigvee_{k} \bthree''_{k}\right)$ holds.
We can thus conclude by applying an instance of the same rule.

\item if the last rule is

\begin{prooftree}
\AXC{$\Gamma\vdash^{X\cup\{a\},r} t_{1}:\bone'\Pto \sigma$}
\AXC{$\bone\vDash^{X\cup\{a\}}\bvar_{i}^{b}\land \bone'
$}
\RL{$\TL$}
\BIC{$\Gamma\vdash^{X\cup\{a\},q} t_{1}\oplus^{i}_{b}t_{2}:  \bone\Pto \sigma$}
\end{prooftree}
Then, let $\bigvee_{j}\btwo'_{j}\land \bthree'_{j}$ be an $a$-decomposition of $\bone'$ and observe that 
$(\bvar_{0}^{a}\land \bvar_{i}^{b}) \lor (\lnot \bvar_{0}^{a}\land \bvar_{i}^{b})$ is an $a$-decomposition of $\bvar_{i}^{b}$.
By IH we deduce that 
$\Gamma\vdash^{X,r\cdot s} t_{1} : \bigvee_{j}\bthree'_{j}\Pto \sigma$. 

By reasoning as in the previous case, from $\bone \vDash^{X\cup\{a\}} \bvar_{i}^{b}\land \bone'$ we deduce, using Lemma \ref{lemma:disj}, that $\bigvee_{i}\bthree_{i} \vDash^{X} \bvar_{i}^{b} \land \bigvee_{j}\bthree'_{j}$. We can thus conclude  by applying an instance of the same rule.

\item if the last rule is 

\begin{prooftree}
\AXC{$\Gamma \vdash^{X\cup\{a\},r} t_{2}:\bone'\Pto \sigma$}
\AXC{$\bone\vDash\lnot\bvar_{i}^{b}\land \bone'$}
\RL{$\TR$}
\BIC{$\Gamma \vdash^{X\cup\{a\},r} t_{1}\oplus^{i}_{b}t_{2}:  \bone\Pto \sigma$}
\end{prooftree}

Then, we can argue similarly to the previous case. 

\item if the last rule is

\begin{prooftree}
\AXC{$\Gamma \vdash^{X\cup\{a\}\cup\{b\},q} t:\bigvee_{u}\mathscr f_{u}\land \bfour_{u}\Pto \sigma$}
\AXC{$\mu(\model{\mathscr f_{u}}_{\{b\}})\geq s'$}
\AXC{$\bone \vDash^{X\cup\{a\}} \bigvee_{u}\bfour_{u}$}
\RL{$\TN$}
\TIC{$\Gamma \vdash^{X\cup\{a\},q\cdot s'} \nu b.t: \bone\Pto  \sigma$}
\end{prooftree}

Then, let $ \bigvee_{i}\btwo_{i}\land \bthree_{i}$, be a weak $a$- decomposition of $\bone$. 
For all $u$, let $\bigvee_{j} \btwo_{u,j}\land \bfour_{u,j}$ be an $a$-decomposition of $\bfour_{u}$. 
 Then observe that we have
\begin{align*}
\bigvee_{u}\mathscr f_{u}\land \bfour_{u} & \equiv^{X\cup\{a\}\cup\{b\}} \bigvee_{u} \mathscr f_{u} \land \left ( \bigvee_{j} \btwo_{u,j}\land \bfour_{u,j}\right) \\
&\equiv^{X\cup\{a\}\cup\{b\}}
\bigvee_{u} \bigvee_{j}\mathscr f_{u} \land ( \btwo_{u,j}\land \bfour_{u,j}) \\
&\equiv^{X\cup\{a\}\cup\{b\}}
\bigvee_{u,j}  \btwo_{u,j} \land (\mathscr f_{u} \land \bfour_{u,j}) 
\end{align*}
So, by IH, we deduce that $\Gamma \vdash^{X\cup\{b\}, q\cdot s} t: \bigvee_{u,j} \mathscr f_{u}\land \bfour_{u,j}\Pto \sigma$.
Now, from $\bone \vDash^{X\cup\{a\}} \bigvee_{u} \bfour_{u}$ we deduce
$\btwo_{i}\land \bthree_{i} \vDash^{X\cup\{a\}} \bigvee_{u} \bigvee_{j}\btwo_{u,j}\land \bfour_{u,j}$, and using Lemma \ref{lemma:disj} we conclude
$\bigvee_{i}\bthree_{i}\vDash^{X\cup\{a\}} \bigvee_{u} \bigvee_{j} \bfour_{u,j}$. We can thus conclude by applying an instance of the same rule.

%
%

\end{itemize}

\end{proof}

We now have all ingredients to establish the subject reduction property of $\TTT$.

%
\begin{proof}[Proof of Proposition \ref{prop:subject}]
For the case of $\beta$-reduction it suffices to check the claim when $t$ is a redex $(\lambda x.t_{1})t_{2}$ and $u$ is $t_{1}[t_{2}/x]$. From 
$\Gamma\vdash^{X,r}t:\bone \Pto \sigma$ we deduce then that 
$\Gamma, x: \BOX^{s}\tau \vdash^{X,r} t_{1}: \bone_{1}\Pto \sigma$ and 
$\Gamma \vdash^{X,s} t_{2}: \bone_{2}\Pto \tau$ hold,
where $\bone \vDash^{X} \bone_{1}\land \bone_{2}$.  
From Lemma \ref{lemma:betasub} we deduce then
$\Gamma \vdash^{X,s} u: \bone_{1}\land \bone_{2}\Pto \sigma$ and from Lemma \ref{lemma:weak} we deduce
$\Gamma \vdash^{X,s} u: \bone\Pto \sigma$.

%
%
For permutative reduction we consider all rules one by one:

\begin{description}

\item[($t\oplus_{a}^{i}t\pperm t$)]  
The last rule of $t$ is either
\begin{prooftree}
\AXC{$ \Gamma \vdash^{ X,r} t: \bone'\Pto \rho$}
\AXC{$ \bone \vDash^{X} \bvar_{i}^{a}\land \bone'$}
\RL{$\TL$}
\BIC{$ \Gamma \vdash^{X,r} t\oplus^{a}_{i}t:  \bone\Pto \rho$}
\end{prooftree}

or

\begin{prooftree}
\AXC{$ \Gamma \vdash^{ X,r} t: \bone'\Pto \rho$}
\AXC{$ \bone \vDash^{X}\lnot \bvar_{i}^{a}\land \bone'$}
\RL{$\TR$}
\BIC{$ \Gamma \vdash^{X,r} t\oplus^{a}_{i}t:  \bone\Pto \rho$}
\end{prooftree}
Then, in either case, from $ \Gamma \vdash^{ X,r} t: \bone'\Pto \rho$, $\bone \vDash^{X} \bone'$, using Lemma \ref{lemma:weak} we deduce $ \Gamma \vdash^{ X,r} t: \bone\Pto \rho$.

\bigskip

\item[($ (t\oplus_{a}^{i}u)\oplus_{a}^{i}v \pperm t\oplus_{a}^{i}v$)] 
There are three possible sub-cases:
	\begin{enumerate}
	\item the type derivation is as follows:
	
\begin{prooftree}
	\AXC{$\Gamma\vdash^{X,r} t: \bone''\Pto \rho$}
	\AXC{$\bone'\vDash^{X} \bvar_{i}^{a}\land \bone''$}
	\RL{$\TL$}
	\BIC{$\Gamma \vdash^{X,r} t\oplus_{a}^{i} u :\bone'\Pto \rho$}
	\AXC{$\bone\vDash^{X} \bvar_{i}^{a}\land \bone'$}
	\RL{$\TL$}
	\BIC{$\Gamma \vdash^{X,r} (t\oplus_{a}^{i} u)\oplus_{a}^{i} :\bone\Pto \rho$}
\end{prooftree}

	Then from $\Gamma\vdash^{X,r} t: \bone''\Pto \rho$ and since we have $\bone\vdash^{X} \bvar_{i}^{a}\land \bone''$ we deduce $\Gamma \vdash^{X,r}t\oplus_{a}^{i}v : \bone \Pto \rho$.
	
	\item the type derivation is as follows:
\begin{prooftree}
	\AXC{$\Gamma\vdash^{X,r} u: \bone''\Pto \rho$}
	\AXC{$\bone'\vDash^{X} \lnot\bvar_{i}^{a}\land \bone''$}
	\RL{$\TR$}
	\BIC{$\Gamma \vdash^{X,r} t\oplus_{a}^{i} u :\bone'\Pto \rho$}
	\AXC{$\bone\vDash^{X} \bvar_{i}^{a}\land \bone'$}
	\RL{$\TL$}
	\BIC{$\Gamma \vdash^{X,r} (t\oplus_{a}^{i} u)\oplus_{a}^{i}v :\bone\Pto \rho$}
\end{prooftree}

	Then, from $\bone\vDash^{X} \bvar_{i}^{a}\land \bone'$ and $\bone'\vDash^{X} \lnot\bvar_{i}^{a}\land \bone''$ we deduce $\bone \vDash^{X}\bot$, so we conclude $\Gamma \vdash^{X,r} (t\oplus_{a}^{i}v): \bone \Pto \rho$ using one of the initial rules.
	
	\item the type derivation is as follows:
	
\begin{prooftree}
	\AXC{$\Gamma\vdash^{X,r} v: \bone''\Pto \rho$}
	\AXC{$\bone\vDash^{X} \lnot \bvar_{i}^{a}\land \bone''$}
	\RL{$\TR$}
	\BIC{$\Gamma \vdash^{X,r} (t\oplus_{a}^{i} u)\oplus_{a}^{i}v :\bone\Pto \rho$}
\end{prooftree}

	Then from $\Gamma\vdash^{X,r} u: \bone'\Pto \rho$ and  $\bone\vDash^{X} \lnot \bvar_{i}^{a}\land \bone''$ we deduce $\Gamma \vdash^{X,r}t\oplus_{a}^{i}v : \bone \Pto \rho$.
	\end{enumerate}

	\bigskip

\item[($t\oplus_{a}^{i}(u\oplus_{a}^{i}v)\pperm t\oplus_{a}^{i}v$)] Similar to the case above.

\bigskip

\item[($\lambda x.(t\oplus_{a}^{i}u)\pperm (\lambda x.t)\oplus_{a}^{i}(\lambda x.u)$)]
There are two possible sub-cases:
	\begin{enumerate}
	\item 
	
\begin{prooftree}
	\AXC{$\Gamma, x:\FF s \vdash^{X,r} t: \bone'\Pto \rho$}
	\AXC{$\bone\vDash^{X} \bvar_{i}^{a}\land \bone'$}
	\RL{$\TL$}
	\BIC{$\Gamma, x:\FF s \vdash^{X,r} t\oplus_{a}^{i}u: \bone\Pto \rho$}
	\RL{$\TLA$}
	\UIC{$\Gamma \vdash^{X,r} \lambda x.(t\oplus_{a}^{i}u): \bone \Pto (\FF s\To \rho)$}
\end{prooftree}

	Then, we deduce 
	
\begin{prooftree}
	\AXC{$\Gamma, x:\FF s \vdash^{X,r} t: \bone'\Pto \rho$}
	\RL{$\TLA$}
	\UIC{$\Gamma \vdash^{X,r} \lambda x.t: \bone' \Pto (\FF s\To \rho)$}
	\AXC{$\bone\vDash^{X} \bvar_{i}^{a}\land \bone'$}
	\RL{$\TL$}
	\BIC{$\Gamma \vdash^{X,r} (\lambda x.t)\oplus_{a}^{i}(\lambda x.u): \bone \Pto (\FF s\To \rho)$}
\end{prooftree}

	\item 
	
\begin{prooftree}
	\AXC{$\Gamma, x:\FF s \vdash^{X,r}u: \bone'\Pto \rho$}
	\AXC{$\bone\vDash^{X} \lnot\bvar_{i}^{a}\land \bone'$}
	\RL{$\TR$}
	\BIC{$\Gamma, x:\FF s \vdash^{X,r} t\oplus_{a}^{i}u: \bone\Pto \rho$}
	\RL{$\TLA$}
	\UIC{$\Gamma \vdash^{X,r} \lambda x.(t\oplus_{a}^{i}u): \bone \Pto (\FF s\To \rho)$}
\end{prooftree}

	Then, we can argue similarly to the previous case.
	
	\end{enumerate}

	\item[($(t\oplus_{a}^{i}u)v \pperm (tv)\oplus_{a}^{i}(uv)$)] 
	There are two possible sub-cases:
	\begin{enumerate}
	\item 
	\footnotesize{
\begin{prooftree}
	\AXC{$\Gamma \vdash^{X,r} t: \bone''\Pto (\BOX^{s}\sigma)\To \rho$}
	\AXC{$\bone'\vDash^{X} \bvar_{i}^{a}\land \bone''$}
	\RL{$\TL$}
	\BIC{$\Gamma\vdash^{X,r} t\oplus_{a}^{i} u: \bone' \Pto (\BOX^{s}\sigma)\To \rho$}
	\AXC{$\Gamma \vdash^{X,s}v: \btwo \Pto \sigma$}
	\AXC{$\bone \vDash^{X} \bone'\land \btwo$}
	\RL{$\TA$}
	\TIC{$\Gamma\vdash^{X,r}(t\oplus_{a}^{i}u)v: \bone \Pto \rho$} 
\end{prooftree}
}
\normalsize
	Then, we deduce 
	\footnotesize{
\begin{prooftree}
	\AXC{$\Gamma \vdash^{X,r} t: \bone''\Pto (\BOX^{s}\sigma)\To \rho$}
	\AXC{$\Gamma \vdash^{X,s}v: \btwo \Pto \sigma$}
	\AXC{$\bone'\land \btwo \vDash^{X} \bone''\land \btwo$}
	\RL{$\TA$}
	\TIC{$\Gamma \vdash^{X,r} tv : \bone'\land \btwo \Pto \rho$}
	\AXC{$\bone \vDash^{X} \bvar_{i}^{a}\land \bone'\land \btwo $}
	\RL{$\TL$}
	\BIC{$\Gamma \vdash^{X,r} (tv)\oplus_{a}^{i} (uv): \bone \Pto \rho$}
\end{prooftree}
}
\normalsize

	\item 
	\footnotesize
\begin{prooftree}
	\AXC{$\Gamma \vdash^{X,r} u: \bone''\Pto (\BOX^{s}\sigma)\To \rho$}
	\AXC{$\bone'\vDash^{X} \lnot\bvar_{i}^{a}\land \bone''$}
	\RL{$\TR$}
	\BIC{$\Gamma\vdash^{X,r} t\oplus_{a}^{i} u: \bone' \Pto (\BOX^{s}\sigma)\To \rho$}
	\AXC{$\Gamma \vdash^{X,s}v: \btwo \Pto \sigma$}
	\AXC{$\bone \vDash^{X} \bone'\land \btwo$}
	\RL{$\TA$}
	\TIC{$\Gamma\vdash^{X,r}(t\oplus_{a}^{i}u)v: \bone \Pto \rho$} 
\end{prooftree}
\normalsize

	Then, we can argue similarly to the previous case.
	\end{enumerate}

\bigskip 
\item[{($t(u\oplus_{a}^{i}v)\pperm (tu)\oplus_{a}^{i}(tv)$)}]
	As before, there are two sub-cases:
	\begin{enumerate}
	\item 
	
	\footnotesize{
\begin{prooftree}
	\AXC{$\Gamma \vdash^{X,r} t: \bone'\Pto (\BOX^{s}\sigma)\To \rho$}
	\AXC{$\Gamma \vdash^{X,s}u: \btwo' \Pto \sigma$}
	\AXC{$\btwo\vDash^{X} \bvar_{i}^{a}\land \btwo'$}
	\RL{$\TL$}
	\BIC{$\Gamma \vdash^{X,s} u\oplus_{a}^{i} v: \btwo \Pto \sigma$}
	\AXC{$\bone \vDash^{X} \bone'\land \btwo$}
	\RL{$\TA$}
	\TIC{$\Gamma\vdash^{X,r}t(u\oplus_{a}^{i}v): \bone \Pto \rho$} 
\end{prooftree}
}
\normalsize
	Then, we deduce that 
\footnotesize{	
\begin{prooftree}
	\AXC{$\Gamma \vdash^{X,r} t: \bone'\Pto (\BOX^{s}\sigma)\To \rho$}
	\AXC{$\Gamma \vdash^{X,s}u: \btwo' \Pto \sigma$}
	\AXC{$\bone \vDash^{X} \bone'\land \btwo'$}
	\RL{$\TA$}
	\TIC{$\Gamma\vdash^{X,r}tu: \bone \Pto \rho$} 
	\AXC{$\bone\vDash^{X} \bvar_{i}^{a}\land \bone$}
	\RL{$\TR$}
	\BIC{$\Gamma \vdash^{X,s} (tu)\oplus_{a}^{i}(tv): \bone \Pto \sigma$}
\end{prooftree}
}
\normalsize
	\item 
\footnotesize{	
\begin{prooftree}
	\AXC{$\Gamma \vdash^{X,r} t: \bone'\Pto (\BOX^{s}\sigma)\To \rho$}
	\AXC{$\Gamma \vdash^{X,s}v: \btwo' \Pto \sigma$}
	\AXC{$\btwo\vDash^{X}\lnot \bvar_{i}^{a}\land \btwo'$}
	\RL{$\TR$}
	\BIC{$\Gamma \vdash^{X,s} u\oplus_{a}^{i} v: \btwo \Pto \sigma$}
	\AXC{$\bone \vDash^{X} \bone'\land \btwo$}
	\RL{$\TA$}
	\TIC{$\Gamma\vdash^{X,r}t(u\oplus_{a}^{i}v): \bone \Pto \rho$} 
\end{prooftree}
}
\normalsize
	Then, we can argue similarly to the previous case.
	\end{enumerate}

\bigskip

\item[($(t\oplus_{a}^{i}u)\oplus_{b}^{j}v \pperm (t\oplus_{b}^{j}v)\oplus_{a}^{i}(u\oplus_{b}^{j}v)$)]  
We suppose here $a\neq b$ or $i< j$. There are three sub-cases:
		\begin{enumerate}
		\item 
		
\begin{prooftree}

		\AXC{$\Gamma \vdash^{X,r} t: \bone''\Pto \rho$}
		\AXC{$\bone'\vDash^{X} \bvar_{a}^{i} \land \bone''$}
		\RL{$\TL$}
		\BIC{$\Gamma \vdash^{X,r}t\oplus_{a}^{i}u: \bone'\Pto \rho$}
		\AXC{$\bone\vDash^{X} \bvar_{b}^{j} \land \bone'$}
		\RL{$\TL$}
		\BIC{$\Gamma \vdash^{X,r}(t\oplus_{a}^{i}u)\oplus_{b}^{j}v: \bone'\Pto \rho$}
\end{prooftree}

		Then, we deduce 		
\begin{prooftree}
		\AXC{$\Gamma \vdash^{X,r} t: \bone''\Pto \rho$}
		\AXC{$\bone\vDash^{X} \bvar_{b}^{j} \land \bone''$}
		\RL{$\TL$}
		\BIC{$\Gamma \vdash^{X,r}t\oplus_{b}^{j}v: \bone'\Pto \rho$}
		\AXC{$\bone\vDash^{X} \bvar_{a}^{i} \land \bone$}
		\RL{$\TL$}
		\BIC{$\Gamma \vdash^{X,r}(t\oplus_{b}^{j}u)\oplus_{a}^{i}(u\oplus_{b}^{j}v): \bone'\Pto \rho$}
\end{prooftree}

		\item 
		
\begin{prooftree}
		\AXC{$\Gamma \vdash^{X,r} u: \bone''\Pto \rho$}
		\AXC{$\bone'\vDash^{X} \lnot\bvar_{a}^{i} \land \bone''$}
		\RL{$\TR$}
		\BIC{$\Gamma \vdash^{X,r}t\oplus_{a}^{i}u: \bone'\Pto \rho$}
		\AXC{$\bone\vDash^{X} \bvar_{b}^{j} \land \bone'$}
		\RL{$\TL$}
		\BIC{$\Gamma \vdash^{X,r}(t\oplus_{a}^{i}u)\oplus_{b}^{j}v: \bone'\Pto \rho$}
\end{prooftree}

		Then, we can argue similarly to the previous case.
		\item 
		
\begin{prooftree}
		\AXC{$\Gamma \vdash^{X,r} v: \bone'\Pto \rho$}
		\AXC{$\bone\vDash^{X} \bvar_{b}^{j} \land \bone'$}
		\RL{$\TR$}
		\BIC{$\Gamma \vdash^{X,r}(t\oplus_{a}^{i}u)\oplus_{b}^{j}v: \bone\Pto \rho$}
\end{prooftree}
		Then, we deduce (using the fact that $\bone\equiv_{X} (\bvar_{i}^{a}\land \bone)\vee(\lnot \bvar_{i}^{a}\land \bone)$)

		\medskip
	
		\begin{center}
		\resizebox{0.88\textwidth}{!}{
		$
		\AXC{$\Gamma \vdash^{X,r} v: \bone'\Pto \rho$}
		\AXC{$\bone\vDash^{X} \bvar_{b}^{j} \land \bone'$}
		\RL{$\TR$}
		\BIC{$\Gamma \vdash^{X,r}t\oplus_{b}^{j}v: \bone\Pto \rho$}
		\AXC{$\bvar_{i}^{a}\land \bone \vDash^{X} \bvar_{i}^{a}\land \bone$}
		\RL{$\TR$}
		\BIC{$\Gamma \vdash^{X,r} (t\oplus_{b}^{j}v)\oplus_{a}^{i}(u\oplus_{b}^{j}v): \bvar_{i}^{a}\land \bone\Pto \rho$}
		\AXC{$\Gamma \vdash^{X,r} v: \bone'\Pto \rho$}
		\AXC{$\bone\vDash^{X} \bvar_{b}^{j} \land \bone'$}
		\RL{$\TR$}
		\BIC{$\Gamma \vdash^{X,r}u\oplus_{b}^{j}v: \bone\Pto \rho$}
		\AXC{$\lnot\bvar_{i}^{a}\land \bone \vDash^{X} \bvar_{i}^{a}\land \bone$}
		\RL{$\TR$}
		\BIC{$\Gamma \vdash^{X,r} (t\oplus_{b}^{j}v)\oplus_{a}^{i}(u\oplus_{b}^{j}v): \lnot\bvar_{i}^{a}\land \bone\Pto\rho$}
		\AXC{$\bone\vDash^{X} (\bvar_{i}^{a} \land \bone)\lor (\lnot \bvar_{i}^{a}\land \bone)$}
		\RL{$\TU$}
		\TIC{$\Gamma \vdash^{X,r} (t\oplus_{b}^{j}v)\oplus_{a}^{i}(u\oplus_{b}^{j}v): \bone\Pto \rho$}
		\DP
		$}
		\end{center}
		\end{enumerate}

\item[($t\oplus_{b}^{j}(u\oplus_{a}^{i}v) \pperm (t\oplus_{b}^{j}u)\oplus_{a}^{i}(t\oplus_{b}^{j}v)$)]  

Similar to the case above.

\bigskip

\item[{($\nu b. (t\oplus_{a}^{i}u)\pperm (\nu b.t)\oplus_{a}^{i}(\nu b.u)$)}] 
We suppose $a\neq b$.
There are two sub-cases:
	\begin{enumerate}
	\item
	
\medskip
\begin{center}
\adjustbox{scale=0.7}{
\begin{minipage}{\linewidth}
$\AXC{$ \Gamma\vdash^{X\cup\{a\}\cup\{b\}, r} t: \bone''\Pto \sigma$}
	\RL{$\TL$}
	\AXC{$\bigvee_{i} \bfour_{i}\land \bthree_{i}\vDash^{X\cup\{a\}\cup\{b\}} \bvar_{i}^{a}\land \bone''$}
	\BIC{$\Gamma\vdash^{X\cup\{a\}\cup\{b\},r} t\oplus_{a}^{i}u: \bigvee_{i} \bfour_{i}\land \bthree_{i}\Pto \sigma$}
	\AXC{$\left ( \mu(\model{\bfour_{i}}_{\{b\}})\geq s\right)_{i\leq k}$}
	\AXC{$ \bone\vDash^{X\cup\{a\}} \bigvee_{i}\bthree_{i}$}
	\RL{$\TN$}
	\TIC{$\Gamma \vdash^{X,r\cdot s} \nu b. (t\oplus_{a}^{i}u): \bone \Pto \sigma$}
	\DP
	$
\end{minipage}
}
\end{center}
\medskip

	then from $\bigvee_{i} \bfour_{i}\land \bthree_{i}\vDash^{X\cup\{b\}}\bvar_{i}^{a}\land \bone''$,	we deduce $\bigvee_{i} \bfour_{i}\land \bthree_{i}\vdash^{X\cup\{a\}\cup\{b\}}\bigvee_{i}\bfour_{i}\land \bvar_{i}^{a} \land \bthree_{i}$, 
	and thus $\bfour_{i}\land \bthree_{i}\vdash^{X\cup\{a\}\cup\{b\}}\bigvee_{i}\bfour_{i}\land \bvar_{i}^{a} \land \bthree_{i}$, and since $a\neq b$,  $\bthree_{i}\vdash^{X\cup\{a\}\cup\{b\}}(\bigvee_{i}\bfour_{i} \land \bthree_{i})\land \bvar_{i}^{a}$.
Hence from $ \bone\vDash^{X\cup\{a\}}  \bigvee_{i}\bthree_{i}$ we deduce
$ \bone\vDash^{X\cup\{a\}}  \bigvee_{i}\bthree_{i}\land \bvar_{i}^{a} $.
	Moreover, from $\bigvee_{i} \bfour_{i}\land \bthree_{i}\vDash^{X\cup\{b\}}\bone''$, using Lemma \ref{lemma:efs} we deduce $ \Gamma\vdash^{X\cup\{a\}\cup\{b\}, r} t: \bigvee_{i} \bfour_{i}\land \bthree_{i}\Pto \sigma$. 
	We finally deduce
	
	\tiny{
\begin{prooftree}
	\AXC{$ \Gamma\vdash^{X\cup\{a\}\cup\{b\}, r} t: \bigvee_{i} \bfour_{i}\land \bthree_{i}\Pto \sigma$}
	\AXC{$\left ( \mu(\model{\bfour_{i}}_{\{b\}})\geq s\right)_{i\leq k}$}
	\AXC{$ \bone\vDash^{X\cup\{a\}} \bigvee_{i}\bthree_{i}$}
	\RL{$\TN$}
	\TIC{$\Gamma \vdash^{X\cup\{a\},r\cdot s} \nu b. t: \bone \Pto \sigma$}
	\AXC{$\bone\vDash^{X\cup\{a\}} \bvar_{i}^{a}\land \bone$}
	\RL{$\TL$}
	\BIC{$\Gamma \vdash^{X\cup\{a\},r\cdot s}( \nu b. t)\oplus_{a}^{i}(\nu b.u): \bone \Pto \sigma$}
\end{prooftree}
}

	\item
	
\begin{prooftree}
	\AXC{$ \Gamma\vdash^{X\cup\{a\}\cup\{b\}, r} u: \bone''\Pto \sigma$}
	\AXC{$\bigvee_{i} \bfour_{i}\land \bthree_{i}\vDash^{X\cup\{a\}\cup\{b\}} \lnot\bvar_{i}^{a}\land \bone''$}
	\RL{$\TR$}
	\BIC{$\Gamma\vdash^{X\cup\{a\}\cup\{b\},r} t\oplus_{a}^{i}u: \bigvee_{i} \bfour_{i}\land \bthree_{i}\Pto \sigma$}
	\AXC{$\left ( \mu(\model{\bfour_{i}}_{\{b\}})\geq s\right)_{i\leq k}$}
	\AXC{$ \bone\vDash^{X\cup\{a\}} \bigvee_{i}\bthree_{i}$}
	\RL{$\TN$}
	\TIC{$\Gamma \vdash^{X,r\cdot s} \nu b. (t\oplus_{a}^{i}u): \bone \Pto \sigma$}
\end{prooftree}

\normalsize
	The we can argue similarly to the previous case.
	\end{enumerate}

\bigskip

\item[($\nu a.t\pperm t$)] 
We have 
\begin{prooftree}
\AXC{$\Gamma \vdash^{X\cup\{a\},r}t: \bigvee_{i}\btwo_{i}\land \bthree_{i}\Pto \sigma$}
\AXC{$\mu(\model{\btwo_{i}}_{\{a\}})\geq s$}
\AXC{$\bone \vDash^{X} \bigvee_{i}\bthree_{i}$}
\RL{$\TN$}
\TIC{$\Gamma\vdash^{X,r\cdot s}\nu a.t:\bone \Pto \sigma$}
\end{prooftree}
then by Lemma \ref{lemma:nota} we deduce $\Gamma \vdash^{X,r\cdot s}t:\bigvee_{i}\bthree_{i} \Pto \sigma$, and from 
$\bone \vDash^{X} \bigvee_{i}\bthree_{i}$ the claim can be deduced using Lemma \ref{lemma:efs}.

\bigskip 
\item[($\lambda x.\nu a.t\pperm \nu a.\lambda x.t$ )]  
We have

\begin{prooftree}
\AXC{$\Gamma, x: \FF s \vdash^{X\cup\{a\}, r} t: \bigvee_{i}\btwo_{i}\land \bthree_{i} \Pto \sigma$}
\AXC{$\mu(\model{\btwo_{i}}_{\{a\}})\geq s$}
\AXC{$\bone \vDash^{X} \bigvee_{i}\bthree_{i}$}
\RL{$\TN$}
\TIC{$\Gamma, x:\FF s\vdash^{X,r\cdot s}\nu a.t : \bone \Pto \sigma$}
\RL{$\TLA$}
\UIC{$\Gamma \vdash^{X,r\cdot s}\lambda x.\nu a.t: \bone \Pto \FF s\To \sigma$}
\end{prooftree}
from which we deduce:

\begin{prooftree}
\AXC{$\Gamma, x: \FF s \vdash^{X\cup\{a\}, r} t: \bigvee_{i}\btwo_{i}\land \bthree_{i} \Pto \sigma$}
\RL{$\TLA$}
\UIC{$\Gamma \vdash^{X\cup\{a\},r}\lambda x.t: \bigvee_{i}\btwo_{i}\land \bthree_{i}  \Pto \FF s\To \sigma$}
\AXC{$\mu(\model{\btwo_{i}}_{\{a\}})\geq s$}
\AXC{$\bone \vDash^{X} \bigvee_{i}\bthree_{i}$}
\RL{$\TN$}
\TIC{$\Gamma, x:\FF s\vdash^{X,r\cdot s}\nu a.\lambda x.t : \bone \Pto \FF s\To\sigma$}
\end{prooftree}

\item[($ (\nu a.t)u \pperm \nu a.(tu)$)] 
We have
\tiny{
\begin{prooftree}
\AXC{$\Gamma\vdash^{X\cup\{a\}, r} t: \bigvee_{i}\btwo_{i}\land \bthree_{i}\Pto (\BOX^{s'}\tau)\To \sigma $}
\AXC{$\mu(\model{\btwo_{i}}_{\{a\}})\geq s$}
\AXC{$\bone' \vDash^{X} \bigvee_{i}\bthree_{i}$}
\RL{$\TN$}
\TIC{$\Gamma \vdash^{X,r\cdot s}\nu a.t : \bone' \Pto (\BOX^{s'}\tau)\To \sigma $}
\AXC{$\Gamma \vdash^{X,s'}u: \bfour\Pto \tau$}
\AXC{$\bone\vDash^{X} \bone'\land \bfour$}
\RL{$\TA$}
\TIC{$\Gamma \vdash^{X,r\cdot s} (\nu a.t)u: \bone\Pto \sigma$}
\end{prooftree}
}
\normalsize
From Lemma \ref{lemma:weak2} we deduce then $\Gamma\vdash^{X\cup\{a\},s'}u: \btwo \Pto \tau$; moreover, 
 we have
 $\bigvee_{i}\btwo_{i}\land (\bthree_{i}\land \bfour) \vDash^{X\cup\{a\}}\left( \bigvee_{i}\btwo_{i}\land \bthree_{i}\right)\land \bfour$ and 
 $\bone \vDash^{X}\bigvee_{i}(\bthree_{i}\land \bfour)$, 
 so we deduce
 
 \medskip

\begin{center}
\resizebox{0.93\textwidth}{!}{
$
\AXC{$\Gamma\vdash^{X\cup\{a\}, r} t: \bigvee_{i}\btwo_{i}\land \bthree_{i}\Pto (\BOX^{s'}\tau)\To \sigma $}
\AXC{$\Gamma \vdash^{X\cup\{a\},s'}u: \bfour\Pto \tau$}
\AXC{$\bigvee_{i}\btwo_{i}\land (\bthree_{i}\land \bfour) \vDash^{X\cup\{a\}}\left( \bigvee_{i}\btwo_{i}\land \bthree_{i}\right)\land \bfour$}
\RL{$\TA$}
\TIC{$\Gamma\vdash^{X\cup\{a\},r} tu: \bigvee_{i}\btwo_{i}\land (\bthree_{i}\land \bfour) \Pto \sigma $}
\AXC{$\mu(\model{\btwo_{i}}_{\{a\}})\geq s$}
\AXC{$\bone \vDash^{X} \bigvee_{i}(\bthree_{i}\land \bfour)$}
\RL{$\TN$}
\TIC{$\Gamma \vdash^{X,r\cdot s}\nu a.(tu) : \bone \Pto (\BOX^{s'}\tau)\To \sigma $}
\DP
$
}
\end{center}

\end{description}
\end{proof}


\subsubsection{The Normalization Theorem}

We now establish our main result about $\TTT$:

\normalization*

Before getting to the actual proof of Theorem \ref{thm:normalization} we need to establish some preliminary lemmas.
For all $r\in[0,1]$, let $\mathrm{Norm}^{r}$ indicate the set of name-closed terms which reduce to a $\PNF$ in normal form with $\PROB\geq r$

\begin{lemma}\label{lemma:x}
If $tx\in \mathrm{Norm}^{r}$, then $t\in \mathrm{Norm}^{r}$.
\end{lemma}
\begin{proof}
It is clear that for all variables $x,y$, $u$ is a $\PNF$ normal with probability $\geq r$ iff $u[x\mapsto y]$ is a $\PNF$ normal with probability $\geq r$. 
Suppose $tx$ reduces to a $\PNF$ $u$ which is normal with probability $\geq r$. Then two cases arise:
\begin{itemize}
\item $u=u'x$ where $t$ reduces to $u'$. Then $u'$ is also a $\PNF$ in normal form with same probability as $u'x$, and thus $t\in \mathrm{Norm}^{r}$;
\item $tx$ reduces to $(\lambda y.t')x$, and $t'[y\mapsto x]$ reduces to $u$. Then 
$t$ reduces to the $\PNF$ $u[x\mapsto y]$.
\end{itemize}

\end{proof}

\begin{definition}
A name-closed term $t$ is said \emph{neutral} when it is of the form $xt_{1}\dots t_{n}$, where $t_{1},\dots, t_{n}\in \mathcal T$. 
\end{definition}

Observe that a neutral term $t$ is in head normal form, and thus $\sum_{v\in \mathcal{HN}}\mathcal D_{t}(v)=1$.

\begin{lemma}\label{lemma:uno}
If $t\in\RED_{\rho}^{X,r}(S)$ then for all $f\in S$, $\pi^{f}(t)\in  \mathrm{Norm}^{r}$.
\end{lemma}
\begin{proof}
We show by induction on $\rho$ the following facts:
\begin{enumerate}
\item for all
$t\in 
\RED_{\rho}^{X,r}(S)$ and $f\in S$, $\pi^{f}(S)\in \mathrm{Norm}^{r}$;
\item for all neutral term $t$, $t\in \RED_{\rho}^{X,r}(S)$.
\end{enumerate}
\begin{itemize}
\item if $\rho=o$, then by definition $t\in \RED_{\rho}^{X,r}(S)$ iff for all $f\in S$, $\pi^{f}(t)\in \mathrm{Norm}^{r}$. Moreover if $t$ is neutral then for all $f\in S$, $\pi^{f}(t)=t\in \mathcal{HN}$, so $t\in \RED_{\rho}^{X,r}(S)$.

\item if $\rho=(\BOX^{s}\rho_{1})\To \rho_{2}$ and $t\in \RED_{\rho}^{X, r}(S)$, then since by the IH $x\in \RED_{\rho_{1}}^{X,s}(S)$, $tx\in \RED_{\rho_{2}}^{X,r}(S)$, so by the induction hypothesis for all $f\in S$, 
$\pi^{f}(tx)=\pi^{f}(t)x\in \mathrm{Norm}^{r}$, whence $\pi^{f}(t)\in \mathrm{Norm}^{r}$ by Lemma \ref{lemma:x}.

Moreover, for all
neutral term $xt_{1}\dots t_{n}$, $S'\subseteq S$, $f\in S'$, and $u\in\RED_{\rho_{1}}^{X,s}(S')$,
 $xt_{1}\dots t_{n}\pi^{f}(u)$ is also neutral; by the induction hypothesis we have then 
$\pi^{f}(xt_{1}\dots t_{n}u)=xt_{1}\dots t_{n}\pi^{f}(u)\in \RED_{\rho_{2}}^{X,r}(S')$, and we conclude that $xt_{1}\dots t_{n}\in \RED_{\rho}^{X,r}(S)$.
\end{itemize}
\end{proof}

\begin{lemma}\label{lemma:zero}
$\RED_{\rho}^{X,r}(\emptyset)=\RED_{\rho}^{X,0}(S)=\{t\in \Lambda_{\nu}\mid \FN(t)\subseteq X\}$.
\end{lemma}
\begin{proof}
Let $\Lambda_{\nu}^{X}=\{t\in \Lambda_{\nu}\mid \FN(t)\subseteq X\}$. 
By induction on $\rho$: if $\rho=o$, then trivially for all $f\in \emptyset$ and $t\in \Lambda_{\nu}^{X}$, $\pi^{f}(t)\in \mathrm{Norm}^{r}$, so $\Lambda_{\nu}^{X}\subseteq \RED_{\rho}^{X,r}(\emptyset)$. 
Moreover, since  for all $t\in\Lnu^{X}$ and $f\in S$, $\pi^{f}(t)\in \mathrm{Norm}^{0}$, also $\Lambda_{\nu}^{X}\subseteq \RED_{\rho}^{X,0}(S)$ holds.  The converse direction is immediate.
 
If $\rho=(\BOX^{s}\rho_{1})\To \rho_{2}$, then for all $t\in \Lambda_{\nu}^{X}$, and $u\in \RED_{\rho_{1}}^{X,s}(\emptyset)$, $tu\in \Lambda_{\nu}^{X}$, so by IH $tu\in \RED_{\rho_{2}}^{X,r}(\emptyset)$. 
This shows in particular that for all $S'\subseteq \emptyset$, $tu\in \RED_{\rho_{2}}^{X,r}(S')$, so we conclude that $t\in \RED_{\rho}^{X,r}(\emptyset)$. The converse direction is again immediate.
Moreover, for all $t\in \Lambda_{\nu}^{X}$, and $u\in \RED_{\rho_{1}}^{X,s}(S)$, $tu\in \RED_{\rho_{2}}^{X,0}(S)$, since by IH $\Lambda_{\nu}^{X}\subseteq \RED_{\rho_{2}}^{X,0}(S)$, and certainly $tu\in \Lambda_{\nu}^{X}$. 
This shows in particular that for all $S'\subseteq S$, $tu\in \RED_{\rho_{2}}^{X,0}(S')$, so we conclude that $t\in \RED_{\rho}^{X,0}(S)$. The converse direction is again immediate.

\end{proof}

\begin{lemma}\label{lemma:efs}
For all term $t$ such that $\FN(t)\subseteq X$, 
$t\in \RED_{\rho}^{X,r}(S)$ iff for all $f\in S$, $\pi^{f}(t)\in \RED_{\rho}^{\emptyset, r}$.
\end{lemma}
\begin{proof}
If $\rho=o$, the claim follows from the definition.
If $\rho=(\BOX^{s}.\rho_{1})\To \rho_{2}$, suppose $t\in \RED_{\rho}^{X,r}(S)$, let $f\in S$ and $u\in \RED^{\emptyset,s}_{\rho_{1}}$. Then $u$ is name-closed hence for all $g\in S$, $\pi^{g}(u)=u\in \RED_{\rho}^{\emptyset, r}$, which by the IH, implies $u\in \RED_{\rho_{1}}^{X,s}(S)$; we deduce then that $tu\in \RED_{\rho_{2}}^{X,r}(S)$, so by IH $\pi^{f}(tu)=\pi^{f}(t)u\in \RED_{\rho_{2}}^{\emptyset, r}$. We can thus conclude that $\pi^{f}(t)\in \RED_{\rho}^{\emptyset,r}$.

Conversely, suppose that for all $f\in S$, $\pi^{f}(t)\in \RED_{\rho}^{\emptyset, r}$, let $S'\subseteq S$ and $u\in \RED_{\rho_{1}}^{X,s}(S')$; if $f\in S'$, then by IH $\pi^{f}(u)\in \RED_{\rho_{1}}^{\emptyset,s}$, so $\pi^{f}(tu)=\pi^{f}(t)\pi^{f}(u)\in \RED_{\rho_{2}}^{\emptyset, r}$; we have thus proved that for all $f\in S'$, $\pi^{f}(tu) \in \RED_{\rho_{2}}^{\emptyset, r}$, which by IH implies that $tu\in \RED_{\rho_{2}}^{X,r}(S')$. We conclude then that $t\in \RED_{\rho}^{X,r}(S)$.
\end{proof}

\begin{lemma}\label{lemma:perm}
If $t\in \RED_{\rho}^{X,r}(S)$ and $t'\pperm t$, then $t'\in \RED_{\rho}^{X,r}(S)$.
\end{lemma}
\begin{proof}
Simple induction on $\rho$.
\end{proof}

We still need a few lemmas that will be essential for the inductive steps of the proof of Theorem \ref{thm:normalization}.

\begin{lemma}\label{lemma:nu1}
For all terms $t, u_{1},\dots, u_{n}$ with $\FN(t)\subseteq X\cup \{a\}$, $\FN(u_{1}),\dots,\FN(u_{n})\subseteq X$, and measurable sets $S\in \mathcal B((2^{\omega})^{X\cup\{a\}})$ and $S'\in \mathcal B((2^{\omega})^{X})$, if   
\begin{enumerate}
\item for all $f\in S$, $\pi^{f}(tu_{1}\dots u_{n})\in  \mathrm{Norm}^{r}$;
\item 
for all $f\in S'$,
 $\mu( \Pi^{f}(S))\geq s$;
\end{enumerate}
then for all $f\in S'$, $\pi^{f}((\nu a.t)u_{1}\dots u_{n}) \in \mathrm{Norm}^{ r\cdot s}$.

\end{lemma}
\begin{proof}
We suppose $s>0$ since otherwise the claim is trivial. Let us first consider the case where $n=0$.
Let $f\in S'$. Since any term reduces to a (unique) $\PNF$, we can suppose w.l.o.g.~that the name-closed term $\pi^{f}(\nu a.t)=\nu a.t^{*}$ is in $\PNF$. By Corollary~\ref{cor:pnf} $t^{*}$ is a $(\mathcal T, a)$-tree of level $I_{a}(t^{*})$. 
Observe that for all $g\in (2^{\omega})^{\{a\}}$, 
$\pi^{g}(t^{*})=\pi^{g}(\pi^{f}(t))= \pi^{g+f}(t)$ .

If $N$ is the cardinality of $\mathrm{Supp}(t^{*})$, by 2. and the fact that $s>0$ we deduce that there exists a finite number $K$ of elements $w_{1},\dots, w_{K}$ of $\mathrm{Supp}(t^{*})$ which are in $\mathrm{Norm}^{r}$, so that  
$$
 \sum_{w\in \mathrm{Norm}^{r}}
\mu\{ g\mid \pi^{g}(t^{*}) =w\}= 
\sum_{{i}=1}^{K} 
\mu\{ g\mid \pi^{g}(t^{*}) =w_{i}\}\geq s$$

Now, by reducing each such terms $w_{i}$ in $\nu a.t$ to some $\PNF$ $w'_{i}$ which is in normal form with probability $\geq r$ (that is, such that $\sum_{u\in \mathcal{HN}}\mathcal D_{w'_{i}}(u)\geq r$), we obtain a new $\PNF$ $\nu a.t^{\sharp}$ and we can compute

\resizebox{0.95\linewidth}{!}{
\begin{minipage}{\linewidth}
\begin{align*}
\sum_{u\in \mathcal{HN}}\mathcal D_{\nu a.t^{\sharp}}(u) & =
\sum_{u\in \mathcal{HN}}\left (\sum_{t'\in \mathrm{Supp}(t^{\sharp})}\mathcal D_{t'}(u)\cdot 
\mu\{ g\mid \pi^{g}(t^{\sharp}) =t'\} \right) \\
&= 
\sum_{t'\in \mathrm{Supp}(t^{\sharp})}
\left (\sum_{u\in \mathcal{HN}} \mathcal D_{t'}(u)\cdot 
\mu\{ g\mid \pi^{g}(t^{\sharp}) =t'\} \right) \\
&= 
\sum_{t'\in \mathrm{Supp}(t^{\sharp})}
\left (\sum_{u\in \mathcal{HN}} \mathcal D_{t'}(u)\right)\cdot 
\mu\{ g\mid \pi^{g}(t^{\sharp}) =t'\}  \\
&\geq 
\sum_{{i}=1}^{K}
\left (\sum_{u\in \mathcal{HN}} \mathcal D_{w'_{i}}(u)\right)\cdot 
\mu\{ g\mid \pi^{g}(t^{\sharp}) =w'_{i}\}  \\
& \geq 
\sum_{{i}=1}^{K}
r\cdot 
\mu\{ g\mid \pi^{g}(t^{\sharp}) =w'_{i}\}  \\
& = 
r\cdot
\sum_{{i}=1}^{K} 
\mu\{ g\mid \pi^{g}(t^{\sharp}) =w'_{i}\}  \\
& =
r\cdot \sum_{{i}=1}^{K} 
\mu\{ g\mid \pi^{g}(t^{*}) =w_{i}\}  \\
& \geq r\cdot s
\end{align*}
\end{minipage}
}
\medskip

For the case in which $n>0$, we argue as follows:
from the hypotheses we deduce by the first point that 
$ \nu a.(tu_{1}u_{2}\dots u_{n})\in \mathrm{Norm}^{r\cdot s}$. We can then conclude by observing that 
$(\nu a.t)u_{1}u_{2}\dots u_{n}\pperm \nu a.(tu_{1}u_{2}\dots u_{n})$.

\end{proof}

%

\begin{lemma}\label{lemma:nu2}
For all terms $t,u_{1},\dots, u_{n}$ with $\FN(t)\subseteq X\cup \{a\}$, $\FN(u_{i})\subseteq X$,
and measurable sets  $S\in \mathscr B((2^{\omega})^{X\cup\{a\}})$ and $S'\subseteq \mathscr B((2^{\omega})^{X})$, if
\begin{enumerate}
\item $tu_{1}\dots u_{n}\in \RED_{\rho}^{X\cup\{a\}, r}(S)$;
\item for all $f\in S'$, $\mu( \Pi^{f}(S))\geq s$;
\end{enumerate}
then $(\nu a.t)u_{1}\dots u_{n}\in \RED_{\rho}^{X,r\cdot s}(S')$.

\end{lemma}
\begin{proof}
If $s=0$ the claim follows from Lemma~\ref{lemma:zero}, so suppose $s>0$.
We argue by induction on $\rho$:
\begin{itemize}
\item if $\rho=o$, then the claim follows from Lemma~\ref{lemma:nu1}.
\item if $\rho=(\BOX^{s'}\rho_{1})\To \rho_{2}$, then let $S''\subseteq S'$ and $u\in \RED_{\rho_{1}}^{X,s'}(S'')$. 
Since $s>0$ for all $f\in S''\subseteq S'$ there exists at least one $g\in (2^{\omega})^{\{a\}}$, such that $g+f\in \model{\bone}_{X}$.
Let $T=\{ g+f\in S \mid f\in S''\}$. $T$ is measurable, since it is the counter-image of a measurable set through a measurable function (the projection function  from $(2^{\omega})^{X\cup\{a\}}$ to $ (2^{\omega})^{X}$).

We have then that $u\in \RED_{\rho_{1}}^{X,s'}(T)$: for all $g+f\in S''$, $\pi^{g+f}(u)=\pi^{f}(u)\in \RED_{\rho_{1}}^{\emptyset, s'}$, so by Lemma~\ref{lemma:efs}, $u\in \RED_{\rho_{1}}^{X,s'}(T)$.
Then using the hypothesis 1. we deduce that $tu_{1}\dots u_{n}u\in \RED_{\rho_{2}}^{X,s'}(T)$; moreover, for all $f\in S''$, $\mu(\Pi^{f}(T))=\mu(\Pi^{f}(S))\geq s$. So, by the induction hypothesis $(\nu a.t)u_{1}\dots u_{n}u\in \RED_{\rho_{2}}^{X,r\cdot s}(S'')$. 
We can thus conclude that $(\nu a.t)u_{1}\dots u_{n}\in \RED_{\rho}^{X,r\cdot s}(S')$.

\end{itemize}
\end{proof}

\begin{lemma}\label{lemma:tre}
If $t[u/x]u_{1}\dots u_{n}\in \RED_{\rho}^{X,r}(S)$, then $(\lambda x.t)uu_{1}\dots u_{n}\in \RED_{\rho}^{X,r}(S)$.
\end{lemma}
\begin{proof}
By induction on $\rho$.
 If $\rho=o$, then if $t[u/x]u_{1}\dots u_{n}\in \RED_{\rho}^{X,r}(S)$, then for all $f\in S$, 
 $\pi^{f}(t[u/x]u_{1}\dots u_{n})=\pi^{f}(t)[\pi^{f}(u)/x]\pi^{f}(u_{1})\dots \pi^{f}(u_{n})\in \mathrm{Norm}^{r}$, that is $t'$ normalizes to a $r$-head normal form. It is clear then that $\pi^{f}(t'')$, where 
 $t''=(\lambda x.t)uu_{1}\dots u_{n})$, normalizes to the same $r$-normal form, so $\pi^{f}(t'')\in \mathrm{Norm}^{r}$, and we conclude that 
 $t''\in \RED_{\rho}^{X,r}(S)$. 
 
 If $\rho=(\BOX^{s}.\rho_{1})\To \rho_{2}$, then let $S'\subseteq S$ and $v\in \RED_{\rho_{1}}^{X, s}(S')$; then $
t[u/x]u_{1}\dots u_{n}v \in \RED_{\rho_{2}}^{X,s}(S')$ so by  IH $(\lambda x.t)uu_{1}\dots u_{n}v\in \RED_{\rho_{2}}^{X,s}(S')$. This proves in particular that $(\lambda x.t)uu_{1}\dots u_{n}\in \RED_{\rho}^{X,s}(S)$.
\end{proof}

\begin{lemma}\label{lemma:quattro}
For all $t$ such that $\FN(t)\subseteq X$ and measurable sets $S, S'\in \mathscr B( (2^{\omega})^{X})$, if for all $S''\subseteq S'$ and 
$u\in \RED_{\rho_{1}}^{X,s}( S'')$, $t[u/x]\in \RED_{\rho_{2}}^{X, r}(S\cap S'')$, then $\lambda x.t\in \RED_{(\BOX^{s}.\rho_{1})\To \rho_{2}}^{X, r}( S\cap S')$.

\end{lemma}
\begin{proof}
Let $T\subseteq S\cap S'$ and  $u \in \RED_{\rho_{1}}^{X,s}( T)$. 
Then, by hypothesis, $t[u/x]\in \RED_{\rho_{2}}^{X,r}(S\cap T)= \RED_{\rho_{2}}^{X,r}(S\cap S')$ and we conclude by Lemma \ref{lemma:tre} that $(\lambda x.t)u\in \RED_{\rho_{2}}^{X,r}(S\cap S')$.
This proves that $\lambda x.t\in \RED_{(\BOX^{s}.\rho_{1})\To \rho_{2}}^{X,r}(S\cap S')$.

\end{proof}

We can now finally prove Theorem \ref{thm:normalization}, that will descend from the following:

\begin{proposition}
If $\Gamma \vdash^{X,r}t: \bone \Pto \rho$, where $\Gamma=\{x_{1}:\BOX^{s_{1}}\rho_{1},\dots, x_{n}:\BOX^{s_{n}}\rho_{n}\}$, then for all $S\subseteq \model{\bone}_{X}$ and $u_{1}\in \RED_{\rho_{1}}^{X,s_{1}}(S),\dots,
u_{n}\in \RED_{\rho_{n}}^{X,s_{n}}(S)$, $t[u_{1}/x_{1},\dots, u_{n}/x_{n}]\in \RED_{\rho}^{X,r}(S)$.

\end{proposition}
\begin{proof}
The argument is by induction on a type derivation of $t$:
\begin{itemize}

\item If the last rule is:
\begin{prooftree}
\AXC{$\mathrm{FV}(t)\subseteq \Gamma,\FN(t)\subseteq X$}
\AXC{$\bone\vDash^{X}\bot$}
\RL{$\TB$}
\BIC{$
\Gamma\vdash^{X,r}t: \bone\Pto \rho$}
\end{prooftree}
then from $S\subseteq \model\bone_{X}\subseteq\emptyset$ we deduce $S=\emptyset$, so can conclude by Lemma~\ref{lemma:zero}.

\item if the last rule is  
\begin{prooftree}
\AXC{$\FN(\bone)\subseteq X$}
\RL{$\TID$}
\UIC{$\Gamma, x:(\BOX^{r}.\rho) \vdash^{X,r} x: \bone\Pto \rho$}
\end{prooftree}
Then for all $S\subseteq \model{\bone}_{X}$, and for all
 $u_{1}\in \RED^{X,s_{1}}_{\rho_{1}}(S),\dots
u_{n}\in \RED^{X,s_{n}}_{\rho_{n}}(S)$,  and 
$u\in \RED^{X,r}_{\rho}(S)$, 
$t[u_{1}/x_{1},\dots, u_{n}/x_{n}, u/x]= u \in \RED_{\rho}^{X, s}(\model{\bone}_{X}\cap S)= \RED_{\rho}^{X,s}(S)$.

\item if the last rule is  
\begin{prooftree}
\AXC{$\Gamma \vdash^{X,r} t: \bone_{1}\Pto \rho$}
\AXC{$\Gamma \vdash^{X,r} t: \bone_{2}\Pto \rho$}
\AXC{$\bone\vDash^{X} \bone_{1}\vee \bone_{2}$}
\RL{$\TU$}
\TIC{$\Gamma \vdash^{X,r} t: \bone\Pto \rho$}
\end{prooftree}
then by IH and Lemma \ref{lemma:uno} for all $S\subseteq \model{\bone_{1}}_{X}$, for all $u_{1}\in \RED_{\rho_{1}}^{X,s_{1}}(S),\dots, u_{n}\in \RED_{\rho_{n}}^{X,s_{n}}(S)$ and for all $f\in S$, 
$\pi^{f}( t[u_{1}/x_{1},\dots, u_{n}/x_{n}])\in \RED_{\rho}^{\emptyset, r}$; similarly,
for all $S\subseteq \model{\bone_{2}}_{X}$, for all $u_{1}\in \RED_{\rho_{1}}^{X,s_{1}}(S),\dots, u_{n}\in \RED_{\rho_{n}}^{X,s_{n}}(S)$ and for all $f\in S$, 
$\pi^{f}( t[u_{1}/x_{1},\dots, u_{n}/x_{n}])\in \RED_{\rho}^{\emptyset, r}$.

Let now $S\subseteq \model{\bone}_{X} \subseteq \model{\bone_{1}}_{X}\cup \model{\bone_{2}}_{X}$ and $f\in S$.
If $f\in  \model{\bone_{1}}_{X}$, then first observe that from $u_{i}\in \RED_{\rho_{i}}^{X,s_{i}}(S)$ it follows in particular that 
$u_{i}\in \RED_{\rho_{i}}^{X,s_{i}}(S\cap \model{\bone_{1}}_{X})$: for all $g\in S\cap  \model{\bone_{1}}_{X}$, since $g\in S$, $\pi^{g}(u_{i})\in \RED_{\rho_{i}}^{\emptyset,s_{i}}$, so by Lemma \ref{lemma:efs} we conclude that $u_{i}\in \RED_{\rho_{i}}^{X,s_{i}}(S\cap  \model{\bone_{1}}_{X})$. 
Hence we deduce that 
$\pi^{f}( t[u_{1}/x_{1},\dots, u_{n}/x_{n}])\in \RED_{\rho}^{\emptyset, r}$.
If $f\in  \model{\bone_{2}}_{X}$ by a similar argument we deduce $\pi^{f}( t[u_{1}/x_{1},\dots, u_{n}/x_{n}])\in \RED_{\rho}^{\emptyset, r}$.
Hence for all $f\in S$, $\pi^{f}( t[u_{1}/x_{1},\dots, u_{n}/x_{n}])\in \RED_{\rho}^{\emptyset, r}$, and we conclude that 
$ t[u_{1}/x_{1},\dots, u_{n}/x_{n}]\in \RED_{\rho}^{X,r}(S)$.

\item $t=\lambda x.u$ is obtained by
\begin{prooftree}
\AXC{$ \Gamma, x: \BOX^{s}.\rho \vdash^{X,r} u: \bone\Pto \rho'$}
\RL{$\TLA$}
\UIC{$\Gamma\vdash^{X,r} \lambda x.u: \bone\Pto (\BOX^{s}.\rho)\To \rho'$}
\end{prooftree}
By IH, for all $S\subseteq \model{\bone}_{X}$, for all 
$u_{1}\in \RED_{\rho_{1}}^{X,s_{1}}(S),\dots, u_{n}\in \RED_{\rho_{n}}^{X,s_{n}}(S)$, 
and $v\in \RED_{\rho}^{X,s}(S)$, 
$u[u_{1}/x_{1},\dots, u_{n}/x_{n},v/x]\in \RED_{\rho'}^{X,r}(\model{\bone}_{X}\cap S)$.
In particular, for all choice of $u_{1},\dots, u_{n}$ and variable $y$ not occurring free in $u_{1},\dots, u_{n}$, by letting $t'=u[u_{1}/x_{1},\dots, u_{n}/x_{n},y/x]$, we have that for all
 $v\in \RED_{\rho}^{X,s}(S)$, 
 $t'[v/y]=u[u_{1}/x_{1},\dots, u_{n}/x_{n},v/x] \in \RED_{\rho'}^{X,r}( \model{\bone}_{X}\cap S)$.
 By Lemma \ref{lemma:quattro}, then, we can conclude that $\lambda y.t'=(\lambda y.t[y/x])[u_{1}/x_{1},\dots, u_{n}/x_{n}]= (\lambda x.t)[u_{1}/x_{1},\dots, u_{n}/x_{n}]\in \RED_{(\BOX^{s}.\rho)\To \rho'}^{X,r}(\model{\bone}_{X}\cap S)$.

\item $t= uv$ is obtained by
\begin{prooftree}
\AXC{$ \Gamma \vdash^{X,r} u: \btwo\Pto (\BOX^{s}.\rho_{1})\To \rho_{2}$}
\AXC{$ \Gamma\vdash^{X,s}v: \bthree\Pto \rho_{1}$}
\AXC{$\bone\vDash^{X}\btwo\land \bthree$}
\RL{$\TA$}
\TIC{$\Gamma\vdash^{X,r} uv: \bone\Pto \rho_{2}$}
\end{prooftree}
By IH, for all
$S \subseteq \model{\btwo}_{X}$,
for all  $u_{1}\in \RED^{X,s_{1}}_{\rho_{1}}(S),\dots
u_{n}\in \RED^{X,s_{n}}_{\rho_{n}}(S)$ and $w\in \RED_{\rho_{1}}^{X,s}(S)$, 
$u[u_{1}/x_{1},\dots, u_{n}/x_{n}] w\in \RED_{ \rho_{2}}^{X, r}(S)$, and 
for all
$S \subseteq \model{\bthree}_{X}$,
for all  $u_{1}\in \RED^{\theta,s_{1}}_{\rho_{1}}(S),\dots
u_{n}\in \RED^{\theta,s_{n}}_{\rho_{n}}(S)$, 
$v[u_{1}/x_{1},\dots, u_{n}/x_{n}] \in \RED_{\rho_{1}}^{X, s}(S)$.
Let now $S\subseteq\model{\bone}_{X}\subseteq \model{\bone}_{X}\cap \model{\btwo}_{X}$ and  $u_{1}\in \RED^{\theta,s_{1}}_{\rho_{1}}(S),\dots
u_{n}\in \RED^{\theta,s_{n}}_{\rho_{n}}(S)$. Then we deduce 
$u[u_{1}/x_{1},\dots, u_{n}/x_{n}] v[u_{1}/x_{1},\dots, u_{n}/x_{n}]=
uv[u_{1}/x_{1},\dots, u_{n}/x_{n}]
\in \RED_{ \rho_{2}}^{X, r}(S)$.

\item if $t= t_{1}\oplus^{i}_{a}t_{2}$, and the last rule is 
\begin{prooftree}
\AXC{$ \Gamma \vdash^{X,r} t_{1}: \bone\Pto \rho$}
\AXC{$ \btwo\vDash  (\lnot\bvar_{i}^{a}\land \bone)$}
\RL{$\TL$}
\BIC{$ \Gamma \vdash^{X,r}t_{1}\oplus^{a}_{n} t_{2}:\btwo\Pto \rho$}
\end{prooftree}
By IH and Lemma~\ref{lemma:efs}, for all 
$S\subseteq \model{\bone}_{X}$, $f\in S$, 
$u_{1}\in \RED^{X,s_{1}}_{\rho_{1}}(S),\dots
u_{n}\in \RED^{X,s_{n}}_{\rho_{n}}(S)$, 
$\pi^{f}(t_{1}[u_{1}/x_{1},\dots, u_{n}/x_{n}]) \in \RED^{\emptyset, r}_{\rho}$.
Let now $S\subseteq \model{\btwo}_{X}\subseteq \model{\bvar_{i}^{a}}_{X}\land \model{\bone}_{X}$ and $u_{1}\in \RED^{X,s_{1}}_{\rho_{1}}(S),$ $\dots,$
$u_{n}\in \RED^{X,s_{n}}_{\rho_{n}}(S)$.
If $f\in S$, then $f(a)=0$ so in particular $\pi^{f}(t[u_{1}/x_{1},\dots, u_{n}/x_{n}])= \pi^{f}(t_{1}[u_{1}/x_{1},\dots, u_{n}/x_{n}]) \in \RED^{\emptyset, r}_{\rho}$. 
Hence, using Lemma~\ref{lemma:efs} we conclude that $t[u_{1}/x_{1},\dots,$ $u_{n}/x_{n}]$ $\in$ $\RED_{\rho}^{X,r}(\model{c}_{X}\cap S)$.

\item if $t= t_{1}\oplus^{i}_{a}t_{2}$, and the last rule is 
\begin{prooftree}
\AxiomC{$ \Gamma \vdash^{X,r} t_{2}: \bone\Pto \rho$}
\AxiomC{$ \btwo\vDash  (\bvar_{i}^{a}\land \bone)$}
\RightLabel{$\TR$}
\BinaryInfC{$ \Gamma \vdash^{X,r}t_{1}\oplus^{a}_{n} t_{2}:\btwo\Pto \rho$}
\end{prooftree}
then we can argue similarly to the previous case.

\item If $t=\nu a.u$ and the last rule is
\begin{prooftree}
\AxiomC{$\Gamma\vdash^{X\cup\{a\},r} u: \bigvee_{i}\btwo_{i}\land \bthree_{i}\Pto \rho$}
\AxiomC{$\vDash\mu(\model{c_{i}}_{X})\geq s$}
\AxiomC{$\bone \vDash \bigvee_{i}\bthree_{i}$}
\RightLabel{$\TN$}
\TrinaryInfC{$\Gamma\vdash^{X, r\cdot s} \nu a.u: \bone\Pto \rho$}
\end{prooftree}
where $a\notin FV(\Gamma)\cup FV(\bthree_{i})$ and $\btwo=\bigvee_{i}\btwo_{i}\land \bthree_{i}$, then let $S\subseteq \model{\bone}_{X}$ and $u_{1}\in \RED^{X\cup\{ a\},s_{1}}_{\rho_{1}}(S),\dots
u_{n}\in \RED^{X\cup\{a\} ,s_{n}}_{\rho_{n}}(S)$.
Let $T=\{g+f\mid f\in S\}$, which is measurable since counter-imaged of a measurable set through the projection function; then since $\FN(u_{i})\subseteq X$, we deduce using Lemma \ref{lemma:uno} that $u_{i}\in \RED_{\rho_{i}}^{X,s_{i}}(T\cap\model{\btwo}_{X\cup\{a\}})$. 
By IH and the hypothesis we deduce then that 
\begin{itemize}
\item $u[u_{1}/x_{1},\dots, u_{n}/x_{n}]\in \RED^{X\cup\{a\}, r}_{\rho}(T\cup\model{\btwo}_{X\cup\{a\}})$;
\item for all $f\in S$, $\mu(\Pi^{f}(T)\cap\model{\btwo}_{X\cup\{a\}})\geq s$;

\end{itemize}

Hence, by Lemma~\ref{lemma:nu2} we conclude that 
$\nu a.u[u_{1}/x_{1},\dots, u_{n}/x_{n}]= (\nu a.u)[u_{1}/x_{1},\dots, u_{n}/x_{n}] \in \RED_{\rho}^{X,r\cdot s}(S)$.

\end{itemize}
\end{proof}

\subsection{The Curry-Howard Correspondence for $\TTT$}

\begin{proposition}\label{prop:embed}
If $\Gamma\vdash^{X,r}t: \bone \Pto \sigma$ is derivable in $\TTT$, then 
$\vdash^{X}_{\PPL}\bone \Pto \BOX^{r}_{a}\sigma^{*}, \bot\Pfrom \Gamma^{*}$.
\end{proposition}
\noindent
The embedding relies on some properties of $\PPL$ established in previous sections, as well as on the following few simple lemmas.
Preliminarily notice that, as anticipated, 
we will define an embedding between $\lambda \PPL$ and
$\PPL$, by considering a multi-succedent calculus, which is
sound and complete for the semantics $\PPL$, but is slightly different with respect to the one introduced in Section~\ref{section4}.
We defined its rules by simply adding a set of labelled formulas,
call it $\Delta$, to the succedent of each premiss and conclusion.
The only rules which change a bit are $\mathsf{R1_{\vee}^{\Era}}/\mathsf{R2_{\vee}^{\Era}}$ and $\mathsf{R1_{\wedge}^{\Ela}}/\mathsf{R2_{\wedge}^{\Ela}}$.
These are, respectively, substituted with the following multi-succedent rules:

\begin{minipage}{\linewidth}
\begin{minipage}[t]{0.4\linewidth}
\begin{prooftree}
\AxiomC{$\vdash^{X}\bone\Era\fone,\bone\Era\ftwo,\Delta$}
\RightLabel{$\Rvra$}
\UnaryInfC{$\vdash^{X}\bone\Era\fone\vee\ftwo,\Delta$}
\end{prooftree}
\end{minipage}
\hfill
\begin{minipage}[t]{0.5\linewidth}
\begin{prooftree}
\AxiomC{$\vdash^{X}\bone\Ela\fone,\bone\Ela\ftwo,\Delta$}
\RightLabel{$\Rwla$}
\UnaryInfC{$\vdash^{X}\bone\Ela\fone\wedge\ftwo,\Delta$}
\end{prooftree}
\end{minipage}
\end{minipage}


\begin{lemma}\label{lemma:true}
For all simple type $\sigma$, $q\in \mathbb{Q}_{[0,1]}$, name $a$, 
and $\bone$ with $\FN(\bone)\subseteq X$, 
$\vdash^{X}_{\PPL}\bone \Pto \BOX^{q}_{a}\sigma^{*},\Gamma$.
\end{lemma}
\begin{proof}
Simple check by induction on $\sigma$.
\end{proof}

\begin{lemma}\label{lemma:count}
The following rule is admissible in $\GPPL$:
\begin{prooftree}
\AxiomC{$\vdash^{X\cup\{a\}} \bone'\Pto A, \Delta$}
\AxiomC{$\left ( \mu(\model{\btwo_{i}}_{\{a\}})\geq r\right)_{i\leq k}$}
\AxiomC{$\bone \vDash^{X} \bigvee_{i}^{k}\bthree_{i}$}
\TrinaryInfC{$\vdash^{X} \bone \Pto \BOX^{r}_{a}A, \Delta$}
\end{prooftree}
where $\bigvee_{i}^{k}\btwo_{i}\land \bthree_{i}$ is a weak $a$-decomposition of $\bone'$.
\end{lemma}
\begin{proof}
First observe that if $r=0$, $\model{\BOX^{r}_{a}A}_{X}=(2^{\omega})^{X}$, so the rule is trivially admissible. Suppose then $r>0$, let the premises of the rule be valid, and let all labelled formulas in $\Delta$ not be valid.
We first show that $\model \bone_{X}\subseteq \model{\BOX^{r}_{a}A}_{X}$:
if $f\in \model{\bone}_{X}$, then for some $i$, $f\in \model{\bthree_{i}}_{X}$, hence by hypothesis 
$\mu\big(\Pi_{f}(\model{\bigvee_{i}^{k}\btwo_{i}\land \bthree_{i}}_{\{a\}})\big)\geq r$. 
From $\model{\bigvee_{i}^{k}\btwo_{i}\land \bthree_{i}}_{X\cup\{a\}}\subseteq \model A_{X\cup\{a\}}$, 
we deduce then also 
$\mu\big(\Pi_{f}(\model A_{X\cup\{a\}})\big)\geq r$, 
so we conclude $f\in \model{\BOX^{r}_{a}A}_{X}$.
From this fact, using Lemma~\ref{lemma:fundamental}, we deduce that, if $\bigvee_{j}\btwo'_{j}\land \bthree'_{j}$ is an $a$-decomposition of $\model{A}_{X\cup\{a\}}$, then $\model{\bone}_{X}\subseteq \{ \model{\bthree'_{j}}_{X}\mid \mu(\model{\btwo'_{i}}_{\{a\}})\geq r\}$. 
So we can deduce the conclusion of the rule using an instance of the rule  $\Rbr$.
\end{proof}

\begin{lemma}\label{lemma:qs}
If $a\notin \FN(A)$, then $\BOX^{s}_{a}\BOX^{q}_{b}A \vDash \BOX^{q\cdot s}_{b}A$.
\end{lemma}
\begin{proof}
If $a\notin \FN(A)$, we have that $\BOX^{q}_{a}A\equiv^{X}A$.
In fact,
$\model{\BOX^{q}_{a}A}_{X}=\{ f\in (2^{\omega})^{X}\mid \mu\big(\Pi_{f}(\model{A}_{X\cup\{a\}})\big)\geq q\}=
\{f\in (2^{\omega})^{X}\mid f\in \model{A}_{X}\}= \model{A}_{X}$. 
Hence, we have
$\BOX^{s}_{a}\BOX^{q}_{b}A \equiv_{X} \BOX^{q}_{b}A$. Moreover, since $q\cdot s\leq q$, 
$\model{\BOX^{q}_{b}A}_{X}\subseteq \model{
 \BOX^{q\cdot s}_{b}A}_{X}$.
\end{proof}

\begin{proof}[Proof of Proposition \ref{prop:embed}]
We define the embedding by induction on the typing rules of $\TTT$:

\begin{description}

\item[Initial Sequents]

\begin{center}
\adjustbox{scale=0.9,center}{
$\begin{matrix}
\AXC{$\mathrm{FN}(t)\subseteq \Gamma,\FN(t)\subseteq X$}
\RL{$\TB$}
\UIC{$\Gamma \vdash^{X,r}t: \bot \Pto \sigma$}
\DP
& \  & 
\AXC{$\FN(\bone)\subseteq X$}
\RL{$\TID$}
\UIC{$\Gamma, x: \BOX^{r}\sigma \vdash^{X,r\cdot s}x:\bone \Pto \sigma$}
\DP
\\
\ & & 
\\
\Downarrow & & \Downarrow
\\
\ & & 
\\
\AXC{$\mu(\model\bot_{X})=0$}
\RL{$\Rmur$}
\UIC{$ \vdash^{X} \bot \Pto \BOX^{r}_{a}\sigma^{*}, \bot\Pfrom \Gamma^{*}$}
\DP
&\ & 
\AXC{$\FN(\bone)\subseteq X$}
\RL{\small Lemma \ref{lemma:true}}
\UIC{$\vdash^{X} \bone \Pto \BOX^{r\cdot s}_{a}\sigma^{*}, 
\bot\Pfrom \BOX^{r}_{a}\sigma^{*}, \bot\Pfrom \Gamma^{*}$}
\DP
\end{matrix}
$
}
\end{center}
\medskip

\item[Union Rule]

\medskip
\begin{center}
\adjustbox{scale=0.9,center}{
$\begin{matrix}
\AXC{$\vdots$}\noLine\UIC{$\Gamma \vdash^{X,r} t:\btwo\Pto \sigma$}
\AXC{$\vdots$}\noLine\UIC{$\Gamma \vdash^{X,r} t:\bthree\Pto \sigma$}
\AXC{$\bone\vDash^{X,r}\btwo \vee \bthree$}
\RL{$\TU$}
\TIC{$\Gamma \vdash^{X,r}t: \bone\Pto \sigma$}
\DP
\\
\  
\\
\Downarrow \\
\ 
\\
\AXC{$\vdots$}
\noLine
\UIC{$ \vdash^{X} \btwo \Pto \BOX^{r}_{a}\sigma^{*}, \bot\Pfrom \Gamma^{*}$}
\AXC{$\vdots$}
\noLine
\UIC{$ \vdash^{X} \bthree \Pto \BOX^{r}_{a}\sigma^{*}, \bot\Pfrom \Gamma^{*}$}
\AXC{$\bone\vDash^{X,r}\btwo \vee \bthree$}
\RL{$\Rcup$}
\TIC{$ \vdash^{X} \bone \Pto \BOX^{r}_{a}\sigma^{*}, \bot\Pfrom \Gamma^{*}$}
\DP
\end{matrix}
$
}
\end{center}
\medskip

\item[Abstraction Rule]

\medskip
\begin{center}
\adjustbox{scale=0.9,center}{
$\begin{matrix}
\AXC{$\vdots$}
\noLine
\UIC{$\Gamma ,x:\FF s\vdash^{X,r}t:\bone \Pto \sigma$}
\RL{$\TLA$}
\UIC{$\Gamma \vdash^{X,r}\lambda x.t:\bone \Pto  (\FF s\To \sigma)$}
\DP
\\
\ 
\\
\Downarrow \\
\ 
\\
\AXC{$\vdots$}
\noLine
\UIC{$\vdash^{X}\bone \Pto \BOX^{r}_{a}\sigma^{*}, \bot\Pfrom \FF s^{*}, \bot\Pfrom \Gamma^{*}$}
\AXC{$\top\vDash^{X}\lnot\bot$}
\RL{$\mathbf \Rnra$}
\BIC{$\vdash^{X}\bone \Pto \BOX^{r}_{a}\sigma^{*}, \top\Pto \lnot \FF s^{*}, \bot\Pfrom \Gamma^{*}$}
\AXC{$\bone\vDash^{X}\top$}
\RL{$\Rcup$}
\BIC{$\vdash^{X}\bone \Pto \BOX^{r}_{a}\sigma^{*}, \bone\Pto \lnot \FF s^{*}, \bot\Pfrom \Gamma^{*}$}
\RL{$ \Rnra$}
\UIC{$ \vdash^{X}\bone \Pto  \lnot {\FF s}^{*}\vee \BOX^{r}_{a}\sigma^{*}, \bot\Pfrom \Gamma^{*}$}
\RL{Lemma \ref{lemma:commutations}}
\UIC{$ \vdash^{X}\bone \Pto  \BOX^{r}_{a}(\lnot {\FF s}^{*}\vee \sigma^{*}), \bot\Pfrom \Gamma^{*}$}
\DP
\end{matrix}
$
}
\end{center}
\medskip

\item[Application Rule]

\medskip
\begin{center}
\adjustbox{scale=0.9,center}{
$\begin{matrix}

\AXC{$\vdots$}
\noLine
\UIC{$\Gamma\vdash^{X,r}t:\bone_{1} \Pto (\BOX^{s}\sigma\To \tau)$}
\AXC{$\vdots$}
\noLine
\UIC{$\Gamma\vdash^{X,s}u:\bone_{2} \Pto \sigma$}
\AXC{$\bone\vDash^{X}\bone_{1}\land \bone_{2}$}
\RL{$\TA$}
\TIC{$\Gamma \vdash^{X,r}tu: \bone \Pto  \tau$}
\DP
\\
\  
\\
\Downarrow 
\\
\ 
\\
\AXC{$\vdots$}
\noLine
\UIC{$ \vdash^{X} \bone_{1}\Pto (\BOX^{r}_{a}(\lnot (\BOX^{s}_{b}\sigma^{*})\vee \tau^{*}), \bot\Pfrom \Gamma^{*}$}
\RL{Lemma \ref{lemma:commutations}}
\UIC{$ \vdash^{X} \bone_{1}\Pto \lnot (\BOX^{s}_{b}\sigma^{*})\vee 
(\BOX^{r}_{a} \tau^{*}), \bot\Pfrom \Gamma^{*}$}
\AXC{$\vdots$}
\noLine
\UIC{$\vdash^{X}\bone_{2} \Pto \BOX^{s}_{b}\sigma^{*}, \bot\Pfrom\Gamma^{*}$}
\AXC{$\bone\vDash^{X}\bone_{1}\land \bone_{2}$}
\RL{cut}
\TIC{$\vdash^{X}\bone \Pto \BOX^{r}_{a}\tau^{*}, \bot\Pfrom\Gamma^{*}$}
\DP
\end{matrix}
$
}
\end{center}
\medskip

\item[Choice Rules]

\medskip
\begin{center}
\adjustbox{scale=0.9,center}{
$\begin{matrix}
\AXC{$\vdots$}\noLine\UIC{$\Gamma \vdash^{X\cup\{a\},r}t: \btwo\Pto \sigma$}
\AXC{$\bone\vDash^{X\cup\{a\}} \bvar_{i}^{a}\land \btwo$}
\RL{$\TL$}
\BIC{$\Gamma\vdash^{X\cup\{a\},r}t\oplus_{a}^{i} u: \bone \Pto \sigma$}
\DP
& 
\AXC{$\vdots$}\noLine\UIC{$\Gamma \vdash^{X\cup\{a\},r}u: \btwo\Pto \sigma$}
\AXC{$\bone\vDash^{X\cup\{a\}} \lnot\bvar_{i}^{a}\land \btwo$}
\RL{$\TR$}
\BIC{$\Gamma\vdash^{X\cup\{a\},r}t\oplus_{a}^{i} u: \bone \Pto \sigma$}
\DP
\\
\   & 
\\
\Downarrow & \Downarrow \\
\  & 
\\
\AXC{$\vdots$}
\noLine
\UIC{$ \vdash^{X\cup\{a\}} \btwo \Pto \BOX^{r}_{a}\sigma^{*}, \bot \Pfrom \Gamma^{*}$}
\AXC{$\bone\vDash^{X\cup\{a\}}  \btwo$}
\RL{$\Rcup$}
\BIC{$ \vdash^{X\cup\{a\}} \bone \Pto \BOX^{r}_{a}\sigma^{*}, \bot \Pfrom \Gamma^{*}$}
\DP
& \AXC{$\vdots$}
\noLine
\UIC{$ \vdash^{X\cup\{a\}} \btwo \Pto \BOX^{r}_{a}\sigma^{*}, \bot \Pfrom \Gamma^{*}$}
\AXC{$\bone\vDash^{X\cup\{a\}}  \btwo$}
\RL{$\Rcup$}
\BIC{$ \vdash^{X\cup\{a\}} \bone \Pto \BOX^{r}_{a}\sigma^{*}, \bot \Pfrom \Gamma^{*}$}
\DP
\end{matrix}
$
}
\end{center}
\medskip

\item[Counting Rule]

\medskip
\begin{center}
\adjustbox{scale=0.9,center}{
$\begin{matrix}
\AXC{$\vdots$}
\noLine
\UIC{$\Gamma\vdash^{X\cup\{a\},r}t:\bigvee_{i}\btwo_{i}\land \bthree_{i}\Pto \sigma$}
\AXC{$\mu(\model{\btwo_{i}}_{\{a\}})\geq s$}
\AXC{$\bone \vDash^{X} \bigvee_{i}\bthree_{i}$}
\RL{$\TN$}
\TIC{$\Gamma \vdash^{X,r\cdot s}\nu a.t:\bone \Pto  \sigma$}
\DP

\\
\ 
\\
\Downarrow 
\\
\ 
\\
\AXC{$\vdots$}
\noLine
\UIC{$\vdash^{X\cup\{a\}}\bigvee_{i}\btwo_{i}\land \bthree_{i}\Pto \BOX^{r}_{b}\sigma^{*}, \bot\Pfrom \Gamma^{*}$}
\AXC{$\mu(\model{\btwo_{i}}_{\{a\}})\geq s$}
\AXC{$\bone \vDash^{X} \bigvee_{i}\bthree_{i}$}
\RL{Lemma \ref{lemma:count}}
\TIC{$ \vdash^{X}\bone \Pto  \BOX^{s}_{a}\BOX^{r}_{b} \sigma^{*}, \bot\Pfrom \Gamma^{*}$}
\RL{Lemma \ref{lemma:qs}}
\UIC{$\vdash^{X}\bone \Pto  \BOX^{r\cdot s}_{a} \sigma^{*}, \bot\Pfrom \Gamma^{*}$}
\DP
\end{matrix}
$
}
\end{center}
\medskip

\end{description}
\end{proof}


%
%

\subsection{Minimal $\PPL$}

To obtain a minimal system we cannot take as atoms the ``classical'' atoms $i_{a}$ and $\lnot i_{a}$. For simplicity, we can directly take as atoms the constants
 $\FLIP_{i}$ from $\TTT$.
 
The formulas $A,B$ of $\IPPL$ are defined by the grammar below:
\begin{align*}
A_{0}, B_{0}  = \FLIP_{i} \mid A \To  B_{0}   \qquad
A,B  = \BOX^{q}A_{0}
\end{align*}
Observe that, as in $\TTT$, counting quantifiers are not named.

An \emph{intuitionistic labelled sequent} is of the form $\Delta\vdash^{X} \bone\Pto A$, where $\Delta$ is a finite multiset of \emph{non-labelled} formulas and $\FN(\bone)\subseteq X$. 
The rules of the system $\IPPL$ are shown in Fig.~\ref{fig:intuitionistic}.
Observe that, while the formulas of $\IPPL$ have no free-names, Boolean formulas \emph{do} have names, and rules do depend on them. Notably, the Flipcoin Rule intuitively tells that one can construct a new proof of $A$ from two proofs of the same formula by flipping a (named) coin. 

In Fig.~\ref{fig:decoratedrules} we illustrate how any derivation of $\IPPL$ can be decorated with terms from $\Lnu$. 
To any formula $A$ of $\IPPL$, we can straightforwardly associate a type $A^{\bullet}$, and to any multiset $\Gamma=A_{1},\dots, A_{n}$ of formulas of $\IPPL$ we associate the context $\Gamma^{\bullet}=\{x_{1}:A_{1}^{\bullet},\dots, x_{n}:A_{n}^{\bullet}\}$. Using the decorations in Fig.~\ref{fig:decoratedrules} we deduce then:

\begin{proposition}
To any $\IPPL$-derivation $\mathscr D$ of $\Gamma\vdash^{X}\bone \Pto \BOX^{r}A_{0}$ is associated a term $ \bb{ t}_{\mathscr D}$ such that $\Gamma^{\bullet}\vdash^{X,r}\bb{ t}_{\mathscr D}: \bone \Pto (A_{0})^{\bullet}$.
\end{proposition}

Observe that from this fact, through the embedding of $\TTT$ inside $\PPL$, we deduce an embedding of $\IPPL$ inside $\PPL$:

\begin{corollary}[embedding of $\IPPL$ inside $\PPL$]
If $\Delta\vdash_{\IPPL}^{X} \bone \Pto A$ then $\vdash_{\PPL}^{X} \bone \Pto A^{\bullet *}, \bot\Pfrom \Delta^{\bullet*}$.
\end{corollary}

\begin{figure}[t]
\fbox{
\begin{minipage}{0.95\linewidth}
\begin{center}
\begin{center}
Initial Sequents
\end{center}

\adjustbox{scale=0.8}{
$
\AxiomC{$\phantom{A}$}
\RL{$\TBx$}
\UnaryInfC{$\Delta\vdash^{X}   \bot  \Pto A$}
\DP
\qquad
\AxiomC{$\FN(\bone)\subseteq X$}
\RL{$\TIDx$}
\UnaryInfC{$\Gamma,A\vdash^{X} \bone  \Pto A$}
\DP
$}

\vskip4mm

\begin{center}
Logical Rules
\end{center}

\adjustbox{scale=0.8}{
$
\AXC{$\Gamma,A \vdash^{X} \bone\Pto\BOX^{q}B_{0}$}
\RL{$\TLAx$}
\UIC{$\Gamma \vdash^{X} \bone\Pto \BOX^{q}(  A\To B_{0})$}
\DP
$}

\vskip4mm

\adjustbox{scale=0.8}{
$
\AXC{$\Gamma \vdash^{X}  \btwo\Pto \BOX^{r}( A \To B_{0})$}
\AXC{$\Gamma \vdash^{X}  \bthree \Pto  A$}
\AXC{$\bone\vDash \btwo\land \bthree$}
\RL{$\TAx$}
\TIC{$\Gamma\vdash^{X} \bone \Pto \BOX^{r}B_{0}$}
\DP
$}

\vskip4mm

\begin{center}
Flipcoin Rule
\end{center}

\adjustbox{scale=0.68}{
$
\AXC{$\Gamma, A \vdash^{X\cup\{a\}} \btwo\Pto A$}
\AXC{$\Gamma \vdash^{X\cup\{a\}} \bthree\Pto A$}
\AXC{$\bone\vDash^{X\cup\{a\}}(\btwo\land \bvar_{i}^{a})\vee( \bthree\land \lnot \bvar_{i}^{a})$}
\RL{$\TUx$}
\TIC{$\Gamma \vdash^{X\cup\{a\}} \bone\Pto A$}
\DP
$}


\begin{center}
Counting Rule
\end{center}

\adjustbox{scale=0.8}{
$
\AXC{$\Gamma \vdash^{X\cup\{a\}} \btwo\Pto \BOX^{q} A_{0}$}
\AXC{$\left( \model{\btwo_{i}}_{\{a\}})\geq r\right)_{i\leq k}$}
\AXC{$\bone \vDash \bigvee_{i}^{k}\bthree_{i}$}
\RL{$\TNx$}
\TIC{$\Gamma \vdash^{X} \bone\Pto  \BOX^{q\cdot r} A_{0}$}
\DP
$}

\medskip

\adjustbox{scale=0.8}{
(where $\bigvee_{i}\btwo_{i}\land \bthree_{i}$ is a weak $a$-decomposition of $\btwo$)

}

\end{center}
\end{minipage}
}

\caption{Rules of $\IPPL$.}
\label{fig:intuitionistic}
\end{figure}

\begin{figure}[t]
\fbox{
\begin{minipage}{0.95\linewidth}
\begin{center}
\begin{center}
Initial Sequents
\end{center}

\adjustbox{scale=0.8}{
$
\AxiomC{$\phantom{A}$}
\RL{$\TBx$}
\UnaryInfC{$\Delta\vdash^{X}\bb \Omega:\bot  \Pto A$}
\DP
\qquad
\AxiomC{$\FN(\bone)\subseteq X$}
\RL{$\TIDx$}
\UnaryInfC{$\Gamma, \bb x:A\vdash^{X} \bb x: \bone  \Pto A$}
\DP
$}

\vskip4mm

\begin{center}
Logical Rules
\end{center}

\adjustbox{scale=0.8}{
$
\AXC{$\Gamma,\bb x:A \vdash^{X} \bb t:\bone\Pto\BOX^{q}B_{0}$}
\RL{$\TLAx$}
\UIC{$\Gamma \vdash^{X}\bb{\lambda x.t}: \bone\Pto \BOX^{q}(  A\To B_{0})$}
\DP
$}

\vskip4mm

\adjustbox{scale=0.8}{
$
\AXC{$\Gamma \vdash^{X} \bb t: \btwo\Pto \BOX^{r}( A \To B_{0})$}
\AXC{$\Gamma \vdash^{X} \bb u: \bthree \Pto  A$}
\AXC{$\bone\vDash \btwo\land \bthree$}
\RL{$\TAx$}
\TIC{$\Gamma\vdash^{X}\bb{tu}: \bone \Pto \BOX^{r}B_{0}$}
\DP
$}

\vskip4mm

\begin{center}
Flipcoin Rule
\end{center}

\adjustbox{scale=0.8}{
$
\AXC{$\Gamma, A \vdash^{X\cup\{a\}}\bb t: \btwo\Pto A$}
\AXC{$\Gamma \vdash^{X\cup\{a\}}\bb u: \bthree\Pto A$}
\AXC{$\bone\vDash^{X\cup\{a\}}(\btwo\land \bvar_{i}^{a})\vee( \bthree\land \lnot \bvar_{i}^{a})$}
\RL{$\TUx$}
\TIC{$\Gamma \vdash^{X\cup\{a\}} \bb{t\oplus_{a}^{i}u}:\bone\Pto A$}
\DP
$}


\begin{center}
Counting Rule
\end{center}

\adjustbox{scale=0.7}{
$
\AXC{$\Gamma \vdash^{X\cup\{a\}} \bb t:\btwo\Pto \BOX^{q} A_{0}$}
\AXC{$\left( \model{\btwo_{i}}_{\{a\}})\geq r\right)_{i\leq k}$}
\AXC{$\bone \vDash \bigvee_{i}^{k}\bthree_{i}$}
\RL{$\TNx$}
\TIC{$\Gamma \vdash^{X} \bb{\nu a.t}:\bone\Pto  \BOX^{q\cdot r} A_{0}$}
\DP
$}

\medskip

\adjustbox{scale=0.8}{
(where $\bigvee_{i}\btwo_{i}\land \bthree_{i}$ is a weak $a$-decomposition of $\btwo$)
}

\end{center}
\end{minipage}
}
\caption{Rules of $\IPPL$ decorated with terms.}
\label{fig:decoratedrules}
\end{figure}